\documentclass{ws-rmpmod}

\usepackage{undertilde}
\usepackage{amssymb}
\usepackage{mathrsfs}
\usepackage{cite}
\usepackage{bm}

\newcommand{\p}{\partial}

\newcommand{\dd}{{\rm d}}

\begin{document}

\markboth{E. Minguzzi}{Causality theory for closed cone structures with applications}

%
\catchline{}{}{}{}{}
%

\title{Causality theory for closed cone structures with applications
}

\author{Ettore Minguzzi}

\address{Dipartimento di Matematica e Informatica ``U. Dini'', \\ Universit\`a
degli Studi di Firenze, \\Via S. Marta 3,  I-50139 Firenze, Italy.\\
\email{ettore.minguzzi@unifi.it}}
%

\maketitle


\begin{abstract}
We develop causality theory for upper semi-continuous distributions of cones over manifolds generalizing results from mathematical relativity in two directions: non-round cones and non-regular differentiability assumptions. We prove the validity of most results of the regular Lorentzian causality theory including: causal ladder, Fermat's principle, notable singularity theorems in their causal formulation, Avez-Seifert theorem, characterizations of stable causality and global hyperbolicity by means of (smooth) time functions. For instance, we give the first proof for these structures of the equivalence between stable causality, $K$-causality and existence of a time function. The result implies that closed cone structures that admit continuous increasing functions also admit smooth ones. We also study proper cone structures, the fiber bundle analog of proper cones. For them we obtain most results on domains of dependence. Moreover, we prove that horismos and Cauchy horizons are generated by lightlike geodesics, the latter being defined through the achronality property. Causal geodesics and steep temporal functions are obtained with a powerful product trick. The paper also contains  a  study of Lorentz-Minkowski spaces under very weak regularity conditions. Finally, we introduce the concepts of stable distance and stable spacetime solving two well known problems (a) the characterization of Lorentzian manifolds embeddable in Minkowski spacetime, they turn out to be the stable spacetimes, (b) the proof that topology, order and distance (with a formula a la Connes) can be represented by the smooth steep temporal functions. The paper is self-contained, in fact we do not  use any advanced result from mathematical relativity.
\end{abstract}



\newpage
%
\section*{Contents}
\contentsline {section}{\numberline {1}Introduction}{2}
\contentsline {subsection}{\numberline {1.1}Lorentzian embeddings into Minkowski spacetime}{6}
\contentsline {subsection}{\numberline {1.2}The distance formula}{9}
\contentsline {subsection}{\numberline {1.3}Notations and conventions}{10}
\contentsline {section}{\numberline {2}Causality for cone structures}{10}
\contentsline {subsection}{\numberline {2.1}Causal and chronological relations}{16}
\contentsline {subsection}{\numberline {2.2}Notions of increasing functions}{26}
\contentsline {subsection}{\numberline {2.3}Limit curve theorems}{27}
\contentsline {subsection}{\numberline {2.4}Peripheral properties and lightlike geodesics}{29}
\contentsline {subsection}{\numberline {2.5}Future sets and achronal boundaries}{33}
\contentsline {subsection}{\numberline {2.6}Imprisoned causal curves}{36}
\contentsline {subsection}{\numberline {2.7}Stable causality}{37}
\contentsline {subsection}{\numberline {2.8}Reflectivity and distinction}{41}
\contentsline {subsection}{\numberline {2.9}Domains of dependence and Cauchy horizons}{43}
\contentsline {subsection}{\numberline {2.10}Global hyperbolicity and its stability}{46}
\contentsline {subsection}{\numberline {2.11}The causal ladder}{57}
\contentsline {subsection}{\numberline {2.12}Fermat's principle}{59}
\contentsline {subsection}{\numberline {2.13}Lorentz-Finsler space}{62}
\contentsline {subsection}{\numberline {2.14}Stable distance and stable spacetimes}{71}
\contentsline {subsection}{\numberline {2.15}Singularity theorems}{76}
\contentsline {section}{\numberline {3}Special topics}{87}
\contentsline {subsection}{\numberline {3.1}Proper Lorentz-Minkowski spaces and Legendre transform}{87}
\contentsline {subsection}{\numberline {3.2}Stable recurrent set}{97}
\contentsline {subsection}{\numberline {3.3}Hawking's averaging for closed cone structures}{99}
\contentsline {subsection}{\numberline {3.4}Anti-Lipschitzness and the product trick}{102}
\contentsline {subsection}{\numberline {3.5}Smoothing anti-Lipschitz functions}{107}
\contentsline {subsection}{\numberline {3.6}Equivalence between $K$-causality and stable causality}{112}
\contentsline {subsection}{\numberline {3.7}The regular ($C^{1,1}$) theory}{119}
\contentsline {section}{\numberline {4}Applications}{120}
\contentsline {subsection}{\numberline {4.1}Functional representations and the distance formula}{120}
\contentsline {subsection}{\numberline {4.2}Lorentzian embeddings}{130}

\section{Introduction}

In this work we shall generalize causality theory, a by now well known chapter of mathematical relativity \cite{hawking73,beem96,minguzzi06c,chrusciel12}, in two directions: non-round cones and weak differentiability assumptions. Ultimately we use the generalized theory to prove results in {\em Lorentzian} geometry: namely we characterize the Lorentzian submanifolds of (flat) Minkowski spacetime, they turn out to be the stable spacetimes, and prove the smooth Lorentzian distance formula.




 Concerning the weakening of differentiability conditions, Hawking and Ellis \cite[Sec.\ 8.4]{hawking73} already discussed the validity of singularity theorems under a $C^{1,1}$ assumption on the Lorentzian cone distribution. They were concerned that the (geodesic) singularities predicted by  the singularity theorems could  just  signal a violation of the assumed differentiability conditions. If so the spacetime continuum would survive the singularity in a rougher form. Since the optimal differentiability condition for the existence and uniqueness of geodesics is $C^{1,1}$ it was particularly important to weaken the differentiability assumption from $C^2$ to $C^{1,1}$. Furthermore, since the Einstein's equations relate the Ricci tensor to the stress-energy tensor, and since the energy density is  discontinuous at the surface of a gravitational body, say a planet, mathematically one would naturally consider metrics with  second derivative in $L^\infty_{loc}$ which suggests again to consider $C^{1,1}$ metrics.
 Senovilla \cite{senovilla97} stressed this point emphasizing that the $C^2$ condition enters at several key places of causality theory. In fact, the existence of convex neighborhoods, which was continuously used in local arguments, seemed to require that assumption.

 The problem was solved in \cite{minguzzi13d,kunzinger13,kunzinger13b} where it was shown that under a $C^{1,1}$ differentiability assumption convex neighborhoods do exist and the exponential map provides a local lipeomorphism. From here most results of causality theory follow \cite{minguzzi13d}; Kunzinger and collaborators explored the validity of the singularity theorems under weak differentiability assumption \cite{kunzinger15,kunzinger15b,graf17}, while the author  considered the non-isotropic case \cite{minguzzi15}.

At the time some results had already signaled these possibilities. It was clear that causality  theory had to be quite robust. Most arguments were topological in nature, and it was understood that several results really belonged to more abstract theories. For instance, we used Auslander-Levin's theorem on closed relations to infer the existence of time functions, or to prove the equivalence between $K$-causality and stable causality \cite{minguzzi09c}.
Time functions had little to do with Lorentzian cones, rather they were a byproduct of the Seifert relation $J_S$ being closed.
 Meanwhile, Fathi and Siconolfi \cite{fathi12,fathi15} showed that some results of causality theory connected to the existence of smooth time functions in stably causal or globally hyperbolic spacetimes really could be generalized to $C^0$ cone structures. They used methods from weak KAM theory. Recently, Bernard and Suhr \cite{bernard16} have used methods from dynamical system theory, particularly Conley theory, to prove similar results under upper semi-continuity assumptions on the cone distribution.

 Different smoothing techniques which reached the same results in the $C^2$ theory had been developed by Chru\'sciel, Grant and the author \cite{chrusciel13}. They were in line with the traditional strategy associated to the names of Geroch, Seifert and Hawking, who used volume functions to build time functions \cite{geroch70,hawking68,seifert77,hawking73} (Seifert's paper is generally regarded as flawed, but  our work which is much in his spirit, showed the usefulness of some of his ideas on the smoothing problem). The main strategy was to smooth anti-Lipschitz time functions where anti-Lipschitzness was a property naturally shared by Hawking's average time function.

 A first question that we wish to ask in this work is the following: are these volume functions methods still valuable under low differentiability assumptions? We shall prove that they are. We shall obtain all the standard result of the $C^2$ theory under an upper semi-continuity assumption on the cone distribution using volume functions. In fact we shall prove some important results that so far have not appeared in the literature, such as the equivalence between (i) stable causality, (ii)  $K$-causality, and (iii) the existence of a time function, cf.\ Th.\ \ref{nio}. We shall also obtain some known equivalences for global hyperbolicity clarifying the role of Cauchy hypersurfaces cf.\ Th.\ \ref{xxy}.

 The proofs will require some modifications since we met the following difficulties.
 Hawking's average time function is no more anti-Lipschitz, in fact its anti-Lipschitzness was proved using the existence of convex neighborhoods which now are no more at our disposal.
 The problem is solved  constructing an averaged volume function in $M\times \mathbb{R}$, showing that one level set $S_0$ intersects every $\mathbb{R}$ fiber, and taking the graphing function of $S_0$ as time function. This product trick will prove to be extremely powerful, giving optimal conditions for the existence of steep time function and leading to the solution of some other problems that we present in the last section. Another difficulty that might be mentioned is the following: in the globally hyperbolic case the simpler Geroch's time function construction does not work anymore. In order to get the equivalence of global hyperbolicity with the existence of  a Cauchy smooth steep time functions, we improved the proof of the stability of global hyperbolicity and the smoothing technique for anti-Lipschitz functions, which now provides a bound on the derivative of the smooth approximation.

Of course, causality theory is not just time functions. We have mentioned that it is possible to make sense of  most of the theory under a $C^{1,1}$ assumption. Even before a satisfactory theory for the $C^{1,1}$  case was available Chru\'sciel and Grant \cite{chrusciel13}  approached Lorentzian causality theory under a $C^0$ assumption. They met some important difficulties connected to the failure of some standard results of causality theory, such as the result  $I\circ J\cup J\circ I\subset I$, on the composition of the chronological and causal relations. Their theory seemed to work well only under locally Lipschitz regularity and did not include results involving lightlike geodesics. It was an important limitation since many interesting results of causality theory are connected with the study of lightlike geodesics, particularly those running on the Cauchy horizons. Some of the questions were addressed by S\"amann \cite{samann16} who obtained results on global hyperbolicity and stable causality for $C^0$ Lorentzian structures and proved a version of the Avez-Seifert theorem. Related applications also followed, for instance with the $C^0$ inextendibility studies \cite{sbierski15,galloway17}.  However, most questions, particularly those connected to  geodesics, remained unanswered.

The present work solves many of these problems by showing that most of causality theory holds for {\em  closed }(upper semi-continuous) {\em cone structures}.
Probably, the most characteristic  result of causality theory concerns the validity of the causal ladder of spacetimes \cite{hawking73,beem96,minguzzi06c}. This classical result confers the theory an order and beauty which would justify by itself interest in causality. We prove that the whole causal ladder holds true for closed cone structures. Of course, many proofs differ from the Lorentzian $C^2$ ones.



Next we define the lightlike geodesics using the local achronality property (which in the $C^{1,1}$ theory is derived \cite[Th.\ 6]{minguzzi13d}) and show that horismos are indeed generated by lightlike geodesic.

The study of achronal boundaries suggests to work with {\em proper cone structures}. They are slightly more restrictive than closed cone structures, and represent the bundle analog of proper cones (sharp convex closed cones with non-empty interior). We show that  most classical result on Cauchy developments pass to the proper cone structure case, for instance Cauchy horizons are generated by lightlike geodesics. These results seem remarkable since proper cone structures are again upper semi-continuous cone distributions and several properties which were believed to be essential for causality theory, including  $I\circ J\cup J\circ I\subset I$, still fail for them.

So far we did not mention how to introduce the metrical properties, and have been concerned with just the causal (one would say conformal in the Lorentzian setting)  properties. Here we use repeatedly this idea: the metrical theory can be regarded as a causality theory  on a manifold with one additional dimension $M^\times=M\times \mathbb{R}$. The so called Lorentz-Finsler function $\mathscr{F}\colon C\to [0,+\infty)$, which provides the length of causal vectors, is regarded as  defining a
cone structure $C^\times$ or $C^\downarrow$ on $M^\times$, cf.\ Eqs.\ (\ref{nhz}) and (\ref{xyb}). A Lorentz-Finsler space (spacetime) is just a cone structure on $M^\times$. So we do not need to develop some new theory, rather we work out a causality theory on $M^\times$. For instance, causal geodesics are defined as the projections of the lightlike geodesics defined through the local $C^\times$-achronality property on $M^\times$.

Using these ideas we are able to  give  a version of the Avez-Seifert theorem  and of Fermat's principle, and also to prove causal versions of Penrose's, Hawking's, and Hawking and Penrose's  singularity theorems. The differentiability assumptions for the validity of these causality results are really much weaker than those to be found in previous literature and, furthermore, they hold for anisotropic cones, see Sec.\ \ref{sin} for a discussion.


Some important more specific topics require many pages for their proper study. We have placed them in  Chap.\ \ref{ca3} where they do not distract from the  main line of development devoted to causality theory. The first section concerns the study of Lorentz-Minkowski spaces and the proof that the reverse triangle inequality,  reverse Cauchy-Schwarz inequality, and the duality between Finsler Lagrangian and Hamiltonian  hold under very weak regularity conditions. These results motivate some of our terminology which refers to Lorentz-Finsler spaces. The subsequent sections are devoted to the smoothing techniques and to the construction of anti-Lipschitz and steep time functions. Here Sec.\ \ref{fir}-\ref{las} must be read in this order. The last section \ref{xxo} summarizes what is gained by passing to the  regular theory, but can be skipped on first reading.

In the last section we show that causality theory for anisotropic cones has something important to say on apparently unrelated questions. We shall use it as a tool to solve two well known problems in whose formulations anisotropic cones do not appear. They are the problem of characterizing the Lorentzian submanifolds of Minkowski spacetime, and the problem of proving the Lorentzian distance formula. We devote the next two subsections of this Introduction to their presentation, here we just mention that their solutions use the notions of stable distance and stable spacetime  which we introduce in Sec.\ \ref{mvb}. We shall show that the stable distance is the most convenient distance for stably causal spacetimes.




As a last observation, this work is  self-contained. References are provided mostly for acknowledgment, so the work could be used as an introduction, though advanced, to causality theory.

%
%


\subsection{Lorentzian embeddings into Minkowski spacetime}
The Nash embedding theorem for non-compact manifolds states
\begin{theorem}
Any Riemannian $n$-dimensional manifold with $C^k$ metric, $k \ge 3$, admits a $C^k$ isometric imbedding  into some $N$-dimensional Euclidean space $E^{N}$.
\end{theorem}

The optimal value $N_0(n)$ will be referred as Nash dimension.
It must be mentioned that according to Bob Solovay, and as acknowledged by Nash, the proof of the original bound $N_0\le \frac{1}{2}(3 n^3 + 7 n^2 + 11n)$ for the non-compact case contained a small error. Once amended Solovay obtained the slightly worse bound $N_0\le \frac{1}{2}(3 n^3 + 7 n^2 + 11n) + 2n+1$. Greene \cite{greene70},  Gromov and Rokhlin \cite{gromov70},  and G{\"u}nther \cite{gunther89} obtained better bounds under stronger differentiability assumptions.

One could have expected the embedding to be $C^{k+1}$, however it is really $C^k$, see the review by Andrews for a discussion of this subtle point \cite{andrews02}.

It was also proved by Clarke \cite{clarke70}, Greene \cite{greene70}, Gromov and Rokhlin \cite{gromov70}, and Sokolov \cite{sokolov71,sokolov71b}, that  pseudo-Riemannian manifolds $(M,g)$ with metrics $g$ of signature $(p,q)$ can be  isometrically emebedded into pseudo-Euclidean space $E^{p',q'}$, for some $p'>p$, $q'>q$.

The Lorentzian signature $(-,+, \cdots, +)$ has peculiar properties.
Any pseudo-Riemannian metric splits the tangent space $T_pM\backslash 0$, into what might be called the {\em causal} $g(y,y)\le 0$, $y\ne 0$, and the {\em spacelike} $g(y,y)> 0$ vectors, however only under Lorentzian signature the set of causal vectors is disconnected in the induced topology. In fact it is the union of two convex sharp cones. The Lorentzian manifold is said to be time orientable if it admits the existence of a continuous  causal vector field $V$. In that case one can call the cone containing $V$, {\em future} (denoted $C$ by us) while calling  {\em past} the opposite one. In so doing the  Lorentzian manifold gets  {\em time oriented}. Connected time oriented Lorentzian manifolds are called {\em spacetimes}.  Thus the Lorentzian signature brings into the manifold a causal order induced by a distribution of convex sharp cones. Of course this peculiarity stays at the very foundation of Einstein's general relativity where connected time oriented Lorentzian manifolds are used as model spacetimes.

In this work we shall be concerned with the existence of isometric embeddings of Lorentzian manifolds into the {\em Lorentzian} space $E^{N,1}$. The latter space is connected and can be trivially given a time orientation, in which case it is called {\em Minkowski spacetime}.   We shall solve the problem of characterizing those Lorentzian manifolds that can be regarded as submanifolds of $E^{N,1}$. Equivalently, we shall solve the problem of characterizing those spacetimes $(M,g)$ which can be regarded as submanifolds of Minkowski spacetime. Clearly, not all spacetimes can be so embedded, for instance, those that admit closed causal curves cannot. As a consequence, the  solution will call for metric and causality conditions on $(M,g)$.
 Given the relevance of Lorentzian spacetimes for general relativity, it has to be expected that the class of spacetimes isometrically embeddable in Minkowski could play a significative role in Physics.

Our final result can be formulated in a very simple way:
\begin{center}
A spacetime is  isometrically embeddable in Minkowski iff it is stable.
\end{center}
 Here a spacetime is {\em stable} if (a) its causality and (b) the finiteness of the Lorentzian distance, are stable under small perturbations of the metric i.e.\ in the $C^0$ topology on metrics.
This is a rather large class of spacetimes, much larger than that of globally hyperbolic spacetimes. For instance, we shall prove that the stably causal spacetimes for which the Lorentzian distance is finite and continuous are of this type.

The problem of isometrically embedding a spacetime into a Minkowski spacetime of a certain dimension is an old one. Clarke \cite{clarke70} proved that globally hyperbolic manifolds can be so embedded. The proof relied on some smoothness issues that had yet to be fully settled at the time, so a complete proof was really obtained only recently by M\"uller and S\'anchez \cite{muller11} through a different strategy.

As a preliminary step they observed that the embedding of $(M,g)$ into  Minkowski spacetime is equivalent to the existence of a steep temporal function on $(M,g)$. In particular, $(M,g)$ has to be stably causal.
We recall that  a spacetime is stably causal if causality is stable in the  $C^0$ topology on metrics. Moreover, a smooth steep temporal function is just  a function $t$ such that $\dd t$ is positive on the future cone $C$, and $-g^{-1}(\dd t,\dd t)\ge 1$. Using the reverse Cauchy-Schwarz inequality the latter condition can be replaced by $\dd t (y)\ge \sqrt{-g(y,y)}$ for every $y\in C$. In short, they are functions which increase sufficiently fast over causal curves.

The argument for the mentioned equivalence is simple. Let $\{x^0,x^1,\cdots x^N\}$ be the canonical coordinates on $E^{N,1}$, $\dd s^2=-(\dd x^0)^2+\sum_{i\ge 1}(\dd x^i)^2$. 
One direction follows observing that the restriction of $x^0$ to the submanifold provides the steep temporal function (the steepness condition for a function passes to submanifolds as can be easily seen from its second characterization given above). For the converse, let $\bar g$ be the semi-definite metric coincident with $g$ on $\ker \dd t$, and which annihilates $\nabla^g t$. Then $g=-\beta \dd t^2+\bar g$, with $\beta^{-1}=-g^{-1}(\dd t, \dd t)\ge 1$. Consider the Riemannian metric $g_R=(4-\beta^2) \dd t^2+ \bar g$. If the Nash embedding of $(M,g_R)$ is $i_n\colon M\to E^{N}$, then the map $i\colon M\to E^{N, 1}$, $p\mapsto (2t(p), i_n(p))$ is an isometric embedding of $(M,g)$ on Minkowski space.

This result moves the problem to that of characterizing those spacetimes which admit a smooth steep temporal function. In the same article M\"uller and S\'anchez proved that globally hyperbolic spacetimes do admit such functions, thus establishing the embedding result foreseen by Clarke (another existence proof can be found in \cite{minguzzi16a}). However, it is easy to convince oneself that global hyperbolicity is just a sufficient condition, and certainly not the optimal one. In fact, consider a submanifold $M$ of  Minkowski spacetime, globally hyperbolic in its induced metric $g$. Then the submanifold $(N,g\vert_N)$ obtained by removing a point from $M$ will still be a Lorentzian submanifold of Minkowski but no more globally hyperbolic in the induced metric (see also the more interesting   Examples \ref{cuw} and \ref{syu}).
One might naively hope that globally hyperbolic spacetimes could be characterized as the {\em closed} submanifolds of some Minkowski spacetime. This is not the case,  a simple counterexample has been  provided by M\"uller  \cite[Example 1]{muller13}. Thus through the notion of  embedding the natural objects that are singled out are the stable spacetimes rather than the globally hyperbolic ones.

Summarizing one can contemplate two natural ways of adding a metric structure to a manifold. In the extrinsic approach the manifold is embedded in a reference affine space, say $E^N$ or $E^{N,1}$, while in the intrinsic approach the associated reference vector space is used just as a model for the tangent space of the manifold. In positive signature both methods lead to the same structure, that of Riemannian manifold, this is the content of Nash's theorem, but in the Lorentzian signature the former leads to the notion of stable spacetime while the latter leads to that of general spacetime.

Our idea for constructing a steep time function over the larger class of stable spacetimes is as follows. We introduce a (non-Lorentzian) cone structure $C^\downarrow$ on the product spacetime $M^\times=M\times \mathbb{R}$, and show that every temporal function $F$ on $M^\times$, whose zero level set $F^{-1}(0)$ intersects every $\mathbb{R}$-fiber, provides a steep time function $f$ on $M$ whose graph is $F^{-1}(0)$. The problem is moved to the construction of a temporal function on the product, and there the main difficulty is connected to the proof that the zero level set intersects every $\mathbb{R}$-fiber exactly once.
 We solve  this problem by constructing the function through an averaging procedure reminiscent, though not exactly coincident,  to that first employed by Hawking (in fact we do not open the cones in the direction of the fiber). Here the stability condition on the finite  Lorentzian distance comes into play to guarantee that every fiber is intersected at least once. Actually, the averaging procedure produces just a continuous anti-Lipschitz function so we apply to it a smoothing argument to get the desired steep function.

  A peculiar feature of the proof is that it uses a causality result for non-isotropic cone structures to infer results for Lorentzian spacetimes. This fact confirms  that the most convenient framework for causality theory is indeed that of general cone structures as it is  proved in the first sections of this work.

\subsection{The distance formula}
As it is well known Connes developed a program for the unification of fundamental forces based on non-commutative geometry \cite{connes94,landi97,connes08}. He
 focused on the so called interior geometry and was able to recover much of the Standard Model of particle physics within that framework. The derivation of the spacetime geometry was not as successful. The idea was to use an approach a la Gelfand, by regarding the manifold as the spectra of a certain algebra of functions. The family of functions to be considered had to encode the topology and more generally the distance. This was made possible through Connes' distance formula which, however, was really proved for Riemannian rather than Lorentzian manifolds.

Parfionov and Zapatrin \cite{parfionov00} proposed to consider the more physical Lorentzian version and for that purpose they introduced the notion of steep time function which we already met in the embedding problem.
 Let $d$ denote the Lorentzian distance,  and let $\mathscr{S}$ be the family of $C^1$ steep time functions. The Lorentzian version of the Connes distance formula would be, for every $p,q\in M$
\begin{equation} \label{dap}
 d(p,q)=\mathrm{inf} \big\{[f(q)-f(p)]^+\colon \ f \in \mathscr{S}\big\}.
\end{equation}
where $c^+=\max\{0, c\}$.  There arises the fundamental problem of finding the conditions that a spacetime should satisfy for (\ref{dap}) to hold true.
They called such spacetimes, {\em simple}, but did not provide any characterization for them.

Moretti \cite[Th.\ 2.2]{moretti03} proved a version of the formula for globally hyperbolic spacetimes in which the functions on the right-hand side are steep almost everywhere and only inside some compact set, not being defined outside the compact set.

Rennie  and Whale  gave a version with no causality assumption \cite{rennie16}, however the family of functions on the right-hand side of their Lorentzian distance formula includes  discontinuous functions. In order to have any chance to represent  also the topology, the representing functions must be continuous.
Moreover, due to the continuity of the representing functions the causality condition in the distance formula cannot be too weak, as we shall see (cf.\ Th.\ \ref{xhg}).

 For globally hyperbolic spacetimes  the most interesting version so far available is due to Franco \cite[Th.\ 1]{franco10}. It holds on globally hyperbolic spacetimes and on the right-hand side one finds  globally defined continuous causal functions differentiable and steep almost everywhere.
However, since in Connes' recipe one acts over the representing functions with the Dirac operator their $C^1$ differentiability is important.

In this work we shall prove not only that the formula holds for globally hyperbolic spacetimes  in the {\em smooth} version, but that, more generally, the formula holds precisely for the  stably causal spacetimes which admit a continuous and finite Lorentzian distance (hence they are causally continuous).
The continuity requirement on the Lorentzian distance  might seem strong. However, in Lorentzian geometry Equation (\ref{dap}) imposes the continuity of $d$ since the left-hand side is lower semi-continuous while the right-hand side is upper semi-continuous. So the mentioned result is really optimal.

Still, stably causal spacetimes  are  central in causality theory so it could be disappointing that the formula does not hold for them. All this suggests that a further improvement of the formula could be possible but that it should   pass through the improvement of the very definition of Lorentzian distance. We shall show that there is a better definition of distance which we call {\em stable distance}. This novel distance $D$ has wider applicability, and then the spacetimes for which the distance formula holds are precisely the stable ones met in the embedding problem. We shall also prove that for these spaces the family of steep time functions allows one to recover not only the distance, but also the causal order and topology of the spacetime and that such results hold for the general Lorentz-Finsler theory under very weak differentiability conditions.

These results should be useful for the development of any genuine Lorentzian version of Connes' program. Among the mathematical physics works that have explored such a direction we  mention \cite{strohmaier06,besnard09,franco13,vandendungen13,franco14b,besnard15}.

\subsection{Notations and conventions}

The manifold $M$ has dimension $n+1$. A bounded subset $S\subset M$, is one with compact closure. Greek indices run from $0$ to $n+1$. Latin indices from $1$ to $n$. The Lorentzian signature is $(-,+,\cdots,+)$. The Minkoski metric is denoted $\eta_{\alpha \beta}$, so $\eta_{00}=-1$, $\eta_{ii}=1$. The subset symbol $\subset$ is reflexive. The boundary of a set $S$ is denoted $\p S$.
``Arbitrarily small'' referred to a neighborhood $U$ of $p\in M$, means that for every neighborhood $V\ni p$ we can find $U$ inside $V$. A coordinate open neighborhood of $M$  is an element of the atlas, namely one diffeomorphic with some open set of $\mathbb{R}^{n+1}$.
Sometimes a subsequence  of a sequence $x_n$ is denoted with a change of index, e.g.\ $x_k$ instead of $x_{n_k}$. In order to simplify the notation we often use the same symbol for a curve or its image. Many statements of this work admit, often without notice, time dual versions obtained by reversing the time orientation of the spacetime.

\section{Causality for  cone structures} \label{wda}

In this work the manifold $M$
 is assumed to be connected, Hausdorff, second countable (hence paracompact) and of dimension $n+1$. Furthermore, it is $C^r$, $1\le r\le \infty$.

 The differentiability degree of the manifold determines the maximum degree of differentiability of the  objects living on $M$, and conversely it makes sense to speak of certain differentiable object only provided the manifold has a sufficient degree of differentiability.  So whenever considering Lipschitz vector fields or Lipschitz Riemannian metrics, the manifold has to be assumed  $C^{1,1}$. It is worth recalling that every $C^r$ manifold, $1\le r< \infty$, is $C^r$ diffeomorphic to a $C^\infty$ manifold \cite[Th.\ 2.10]{hirsch76}, so in proofs one can choose a smooth atlas whenever convenient. Of course at the end of the argument one has to return to the original atlas. If the adjective {\em smooth} is used in some statement, then it should be understood as the maximal degree of differentiability compatible with the original manifold atlas.

Let $V$ be a finite $n$+1-dimensional vector space, e.g.\ $V=T_xM$, for $x\in M$.
A cone $C$ is  a subset of $V\backslash 0$ which satisfies: if $s>0$ and $y\in C$ then $s y\in C$. The topological notions, such as the closure operator, will refer to the topology induced by $V$ on $V\backslash 0$.
 In particular, our closed cones do not contain the origin and $\p C$ does not  contain the origin.
All our cones will be closed and convex. Since $C$ does not include the origin, convexity implies sharpness of $C\cup \{0\}$, namely this set does not contain any line passing through the origin. So all our cones will be sharp. Although redundant according to our definitions, for clarity we shall add the adjective {\em sharp} in many statements.


\begin{definition}
A {\em proper cone} is a closed sharp  convex cone with non-empty interior.
\end{definition}

\begin{remark}
Notice that for a proper cone $C-C=V$ in the sense of Minkowski sum, namely $C$ is a {\em generating cone}.
We mention that in Banach space theory sharp convex cones are simply called cones. In finite dimension the generating cones are precisely those with non-empty interior \cite[Lemma 3.2,Th.\ 3.5]{aliprantis07}. Moreover, the cones with non-empty interior are closed iff they are Archimedean \cite[Lemma 2.4]{aliprantis07}.
\end{remark}

We write $C_1<C_2$ if $C_1 \subset \mathrm{Int} C_2$ and $C_1\le C_2$ if $C_1\subset C_2$. For a proper cone  $C=\overline{\mathrm{Int} C}$ and any compact section of $C$ is homeomorphic to an $n$-dimensional closed ball.

We mention a  property   which introduces the concept of convex combination of cones relative to a hyperplane.  Its straightforward proof is omitted.
Let $C_i\subset V$, $i=1,\cdots, m$ be  proper cones and suppose that  there is an affine hyperplane $P$ cutting them in compact convex sets with non-empty interior (convex bodies) $\tilde C_i$  (there is always such hyperplane if $\sum_i C_i$ is sharp). The  combination of $\{C_i\}$ relative to the weights $w_i\in[0,1]$, $\sum_i w_i=1$, and hyperplane $P$ is the cone $C_{(P,\{w_i\})}$  whose intersection with $P$ is given by $\sum_i w_i \tilde C_i:=\{\sum_i w_i c_i\colon \forall i, \ c_i\in \tilde C_i\}$.
\begin{proposition} \label{doo}
The convex combination $C_{(P,\{w_i\})}$ is itself a proper cone which coincides with $C_1$ for $w_1=1$. Moreover, let $C$ be a convex closed cone, let $C'$ be a proper cone and let $C_i$ be proper cones such that for all $i$, $C<C_i<C'$, then $C<C_{(P,\{w_i\})}<C'$. Finally, a strict convex combination of two proper cones $C_1$, $C_2$, $w_1,w_2>0$, with $C_1<C_2$ is such that   $C_1<C_{(P,\{w_1,w_2\})}<C_2$.
\end{proposition}

In this work we shall study the global properties of distributions of cones over manifolds.
\begin{definition}
A {\em cone structure}  is a multivalued map $x \mapsto C_x$, where $C_x\subset T_xM\backslash 0$ is a closed sharp convex non-empty  cone.
\end{definition}




The cone structures might enjoy various degrees of regularity.
Causality theory for cone structures under $C^{1,1}$ regularity has been already investigated. The reader can find a summary in  Sec.\ \ref{xxo}.
This work will be devoted to weaker assumptions for whose formulation we need some  local considerations.

Let $x\mapsto F(x)\subset \mathbb{R}^l$ be a set valued map defined on some open set  $D\subset\mathbb{R}^k$. It is said to be  {\em upper semi-continuous} if for every $x\in D$ and for every neighborhood  $U\supset F(x)$ we can find a neighborhood $N \ni x$ such that $F(N):=\cup_{x\in N} F(x)\subset U$, cf.\ \cite{aubin84}.

It is said to be {\em lower semi-continuous} if for every $x$, and open set $V\subset\mathbb{R}^l$, intersecting $F(x)$, $V\cap F(x)\ne \emptyset$,  the inverse image $F^{-1}(V):=\{w\in D: F(w)\cap V\ne \emptyset\}$ is a neighborhood of $x$. Equivalently, \cite[Prop.\ 1.4.4]{aubin84} for any $y\in F(x)$ and for any sequence of elements $x_n \to x$, there exists a sequence $y_n\in F(x_n)$ converging to $y$.
The map is continuous if it is both upper and lower semi-continuous.

We say that $F$ has convex values if $F(x)$ is convex for every $x$. We shall need the following result.
\begin{proposition} \label{low}
Suppose that $F$ has convex values.
If $F$ is lower semi-continuous  then for every $x$ and for every compact set $K\subset \mathrm{Int} F(x)$ we can find a neighborhood $N\ni x$, such that for every $w\in N$, $K\subset \mathrm{Int} F(w)$. The converse holds provided $F$ is also closed and $\mathrm{Int} F(x)\ne \emptyset$ for every $x$.
\end{proposition}

\begin{proof}
Let $F$ be lower semi-continuous and with convex values.
By compactness it is sufficient to prove that for every $y\in \mathrm{Int} F(x)$ we can find  neighborhoods $V\ni y$ and $N\ni x$, such that for every $w\in N$, $V\subset F(w)$. Indeed, with obvious meaning of the notation,  we can cover $K$ with a finite selection of open sets $\{V_1, \cdots, V_j\}$, hence $N=\cap_i N_i$ has the desired property. In fact,  for every $w\in N$, $K\subset \cup_i V_i\subset \cup_i F(w)=F(w)$.
By convexity, given $y \in \mathrm{Int} F(x)$ we can find $l+1$ points $y^i\in \mathrm{Int} F(x)$ such that $y$ belongs to the interior of a simplex with vertices $\{y^i\}$. By continuity we can find a neighborhood $V\ni y$ and neighborhoods $O_i\ni y^i$, $O_i\subset \mathrm{Int} F(x)$, such that $V$ is contained in any simplex obtained by replacing the original vertices with the  perturbed vertices  $y'{}^i\in O_i$. Let $N=\cap_i F^{-1}(O_i)$, then by the lower semi-continuity of $F$, for every $w\in N$, $F(w)$ has non-empty intersection with every $O_i$ and so contains one perturbed simplex and hence $V$.

For the converse, it is well known that for a closed convex set $C=\overline{\mathrm{Int}(C)}$. If $V$ is an open set such that $V\cap F(x)\ne \emptyset$, then $V$ includes some point $y\in \mathrm{Int} F(x)$. We can find a compact neighborhood  $K\ni y$,  such that $K\subset V\cap\mathrm{Int} F(x)$ thus there is a neighborhood $N\ni x$ such that for every $w\in N$, $K\subset \mathrm{Int} F(w)$, in particular $V\cap F(w)\ne \emptyset$, that is $F^{-1}(V)\supset N$, which proves that $F$ is lower semi-continuous. $\empty \ \ $
\end{proof}

Finally, we shall say that $F$ is locally Lipschitz if for every $x$, we can find a neighborhood $D\ni x$ and a constant $l>0$, such that
\begin{equation} \label{jsw}
\forall \ x_1,x_2\in D, \qquad F(x_1)\subset F(x_2)+l \Vert x_1-x_2\Vert B ,
\end{equation}
 where $B$ is the unit ball of  $\mathbb{R}^l$. It is easy to check that local Lipschitzness implies continuity.

Let us return to the continuity properties of our cone structure. At every $p\in M$ we have a  local coordinate system $\{x^\alpha\}$ over a neighborhood $U\ni p$. The local coordinate system induces a splitting $U\times \mathbb{R}^{n+1}$ of the tangent bundle $TU$   by which sets over different tangent spaces can be compared. Let $F(x)=[C_x\cup \{0\}]\cap \bar B$ where $\bar B$ is the closed unit ball of $\mathbb{R}^{n+1}$, then the notions of upper/lower semi-continuous, continuous and locally Lipschitz cone structures follow from the previous definitions. Of course, they do not depend on  the coordinate system chosen  (they make sense if the manifold is $C^1$ in the former cases, and $C^{1,1}$ in the latter Lipschitz case).

An equivalent approach is as follows.  We have the coordinate sphere subbundle $U\times \mathbb{S}^{n}$, so when it comes to compare $C_q$ with $C_r$, $q,r\in U$, we can just compare $\hat C_q:=C_q \cap\mathbb{S}^{n}$ with $\hat C_r:= C_r \cap\mathbb{S}^{n}$. Since $\mathbb{S}^{n}$ with its canonical distance is a metric space, we can define a notion of Hausdorff distance $\hat d_H$ for its closed subsets and a related topology. The distribution of cones is continuous on $U$ if the map $q\mapsto \hat C_q$ is continuous, and it is locally Lipschitz if the map is   locally Lipschitz \cite{fathi12,fathi15}.


 We are now going to define more specific cone structures. The most natural approach seems that of defining them through properties of the cone bundle as follows.  We recall that we use the topology of the slit tangent bundle and that our cones do not contain the origin.

\begin{definition}
A {\em closed cone structure} $(M,C)$ is a cone structure which is a closed subbundle of the slit tangent bundle.
\end{definition}
The previous definition does not coincide with that given by Bernard and Suhr \cite{bernard16}. Indeed our condition on the cone structure is more restrictive since our cones are non-empty and sharp (non-degenerate and regular in their terminology). One reason is that we shall be mostly interested in causality theory, where it is customary to assume that spacetime is locally non-imprisoning, cf.\ Prop.\ \ref{iiu}. This assumption brings some simplifications, for instance the parametrization of curves is less relevant in our treatment than in theirs.

\begin{proposition} \label{yhh}
A multivalued map $x \mapsto C_x$, where $C_x\subset T_xM\backslash 0$ is a {\em closed cone structure} iff for all $x\in M$, $C_x$ is closed, sharp, convex, non-empty cone and the multivalued map is upper semi-continuous (namely, it is an upper semi-continuous cone structure).
\end{proposition}

\begin{proof}
It is sufficient to prove that the result holds true in any local coordinate chart of $M$. We need to  consider the continuity properties of the cone bundle cut by the unit coordinate balls. That is, we are left with a compact convex distribution for which  the equivalence  follows from well known results, in fact one direction follows from \cite[Prop.\ 1.1.2]{aubin84}, while the other follows from \cite[Th.\ 1.1.1]{aubin84} by letting $F$ be the distribution of unit coordinate closed balls.
\end{proof}



\begin{example}
A time oriented Lorentzian manifold  $(M,g)$ has an associated canonical cone structure given by the distribution of causal cones
\[
C_x=\{y\in T_xM\backslash\{0\}: g(y,y)\le 0, \ y \ \textrm{future directed} \}.
\]
The next results clarifies that some notable regularity properties of the metric $g$ pass to  the cone structure.
\begin{proposition} \label{jss}
Let $(M,g)$ be a time oriented Lorentzian manifold. If $g$ is continuous (locally Lipschitz) then $x\mapsto C_x$ is continuous (resp.\ locally Lipschitz).
\end{proposition}

The proof in the locally Lipschitz case can be adapted to different regularities, say H\"older, provided the corresponding regularity is defined for cone structures, e.g.\ by generalizing Eq.\ (\ref{jsw}).

\begin{proof}
Let $w$ be a global continuous future directed timelike vector field. Let $\bar x\in M$, and let $U$ be a  coordinate neighborhood of $\bar x$. Let us consider the trivialization of the bundle $T U$, as induced by the coordinates.
The function $f(x,y)=\max [g_{\alpha \beta}(x) y^\alpha y^\beta, g_{\alpha \beta}(x) w^\alpha(x)  y^\beta]$ is continuous on $U\times \mathbb{R}^{n+1}$ and is negative precisely on future timelike  vectors.

Let us prove the lower semi-continuity of the cone structure. Since $\overline{\mathrm{Int} C_x}=C_x\cup \{0\}$ it is sufficient to prove the lower semi-continuity of $F(x)=\mathrm{Int} C_x$. Let $(\bar x,y)$ be a future directed timelike vector, hence $f(\bar x,y)<-\eta<0$ for some $\eta>0$, and let $x_n\to \bar x$, then there is an integer $N>0$ such that  for $n>N$, $\vert f(x_n,y)- f(\bar x,y)\vert <\eta$, thus $f(x_n,y)<0$, which implies $(x_n,y)\in \mathrm{Int} C_{x_n}$. Now redefine the sequence $\{y_k=y\}$ for $k\le N$, so that it is timelike for every $n$.

%

For the upper semi-continuity, notice that $[C\cup \{0\} ]\cap TU=\{(x,y)\colon x\in U, f(x,y)\le 0\}$ which by the continuity of $f$ is closed in the topology of $TU$. From here closure of  $C\cup \{0\}$ follows easily and hence upper semi-continuity of the cone structure, cf.\ Prop.\ \ref{yhh}.

For the locally Lipschitz property, let us choose coordinates such that $g_{\alpha \beta}(\bar x)=\eta_{\alpha \beta}$, i.e.\ the Minkowski metric. We are going to focus on the subbundle of $TU$ of vectors that in coordinates read as follows $(x^\alpha, y^\alpha)$ where $y^0=1$, i.e. we are going to work on $U\times \mathbb{R}^n$. It will be sufficient to prove the locally Lipschitz property for this distribution of sliced cones, namely for a distribution of ellipsoids determined by the equation $0=g_{\alpha \beta}(x) y^\alpha y^\beta=g_{00}(x)+2g_{0i}(x) y^i+g_{ij}(x) y^i y^j$, where $i,j=1,\ldots, n$. The ellipsoid is a unit circle for $x=\bar x$. Let $\Vert \cdot \Vert$ be the Euclidean norm on $\mathbb{R}^n$. Let us consider two ellipsoids relative to the points $x_1$ and $x_2$. Let $y_1$ and $y_2$ be two points that realize the Hausdorff distance $D(x_1,x_2)$ between the ellipsoids, i.e. $D(x_1,x_2)= \Vert \delta y \Vert$, $\delta y=y_1-y_2$, where the vector $\delta y= y_1-y_2$ can be identified with a vector of $\mathbb{R}^n$ since its 0-th component vanishes. The definition of Hausdorff distance easily implies that $\delta y$ is orthogonal to one of the ellipsoids  which we assume, without loss of generality, to be that relative to $x_2$, (otherwise switch the labels 1 and 2).
 Then $\delta y$  is proportional to the gradient of the function $h(w^i)=g_{\alpha \beta}(x_2)w^\alpha w^\beta=g_{00}(x_2)+2g_{0i}w^i+g_{ij} w^i w^j$ at $y_2$, namely $2y_2^\alpha g_{\alpha i} $, hence $\vert y_2^\alpha g_{\alpha \beta} \delta y^\beta \vert=\Vert y_2^\alpha g_{\alpha i} \Vert \Vert \delta y\Vert$. Now we observe that
 \begin{align*}
 0&=g_{\alpha \beta}(x_1) y^\alpha_1y^\beta_1-g_{\alpha \beta}(x_2) y^\alpha_2y^\beta_2=[g_{\alpha \beta}(x_1)-g_{\alpha \beta}(x_2)] y_1^\alpha y_1^\beta+g_{\alpha \beta}(x_2) [ y^\alpha_1y^\beta_1- y^\alpha_2 y^\beta_2]\\
 &=[g_{\alpha \beta}(x_1)-g_{\alpha \beta}(x_2)] y_1^\alpha y_1^\beta+ 2 g_{\alpha \beta}(x_2) y_2^\alpha \delta y^\beta+g_{\alpha \beta}(x_2) \delta y^\alpha \delta y^\beta.
 \end{align*}
By the already proved continuity property, as $x_2,x_1\to \bar x$, we have $\delta y\to 0$, $y_1^i$ and $y_2^i$ have norms that converge to one  and (by assumption) $g_{\alpha \beta}(x_i)\to \eta_{\alpha \beta}$, so we have also $\Vert y_2^\alpha g_{\alpha i} \Vert\to 1$. We conclude that the last term on the right-hand side becomes negligible with respect to the penultimate term. Moreover, provided $x_1,x_2$ belong to  a small neighborhood of $\bar x$ where $\Vert y_1^i\Vert \le c$, for some $c>1$ (we already have $y_1^0=1$) we have  $\vert [g_{\alpha \beta}(x_2)-g_{\alpha \beta}(x_1)] y_1^\alpha y_1^\beta \vert\le c^2 \sum_{\alpha,\beta} \vert g_{\alpha \beta}(x_2)-g_{\alpha \beta}(x_1) \vert \le c^2 L \Vert x_2-x_1\Vert$, where $L$ is the Lipschitz constant of the metric.
In conclusion, for every $C >1$ we can find a neighborhood of $\bar x$ such that for $x_1,x_2$ in the neighborhood
\[
\Vert \delta y\Vert \le \frac{C L}{2} \Vert x_2-x_1\Vert,
\]
which proves that the cone distribution is locally Lipschitz.
\end{proof}
\end{example}

\begin{definition}
A {\em proper cone structure} is a closed cone structure in which, additionally, the cone bundle  is proper, in the sense that $(\mathrm{Int} \,C)_x\ne \emptyset$  for every $x$.
\end{definition}
The terminology is justified in that the adjectives entering ``proper'' (that is, sharp, convex, closed and with non-empty interior) are applied fiberwise, whereas those mentioning topological properties have to be interpreted using the topology of the cone bundle, e.g.\ ${(\bar C)_x}=C_x$ for every $x$ which is equivalent to $C$ being closed.
The non-emptyness condition should not be confused with $\mathrm{Int} \,(C_x)\ne \emptyset$ for every $x$, see also Prop.\ \ref{cob} and subsequent examples.

As for cones, given two cone structures we write $C_1<C_2$ if $C_1 \subset \mathrm{Int} C_2$ and $C_1\le C_2$ if $C_1\subset C_2$, where the interior is with respect to the topology of the slit tangent bundle $TM\backslash 0$. Notice that for a proper cone structure $C=\overline{\mathrm{Int} C}$ does not necessarily hold.

\begin{proposition}
A multivalued map $x \mapsto C_x\subset T_xM\backslash 0$ is a {\em proper cone structure} iff $C_x$ is proper and the multivalued map is upper semi-continuous and such that the next property holds true
\begin{quote}
(*): \ $C$ contains a continuous field of proper cones.
\end{quote}
\end{proposition}
\begin{proof}
It is clear that (*) implies  $(\mathrm{Int} \,C)_x\ne \emptyset$ for every $x$.
The converse follows from the fact that $(\mathrm{Int} \,C)_x\ne \emptyset$ at $x$ implies, recalling the definition of product topology, that there is a local continuous cone structure at $x$ contained in $C$ (actually a product in a splitting induced by local coordinates). By sharpness and upper semi-continuity one can find a local smooth 1-form field $\omega$ positive on $C$. Such field can be globalized using a partition of unity, thus providing a distribution of hyperplanes $P=\omega^{-1}(1)$. Still using the partition of unity the local $C^0$ cone structures can be used to form a global continuous field of proper cones by means of Prop.\ \ref{doo} (see also the proof of Prop.\ \ref{ohg} or Th.\ \ref{ddo} for a similar argument).
\end{proof}

Fathi and Siconolfi \cite{fathi12,fathi15}  investigated the problem of the existence of increasing functions for proper cone fields under a $C^0$ assumption. It is clear that a $C^0$ proper cone structure is just a $C^0$ distribution of proper  cones.

For a distribution of proper  cones
\begin{quote}
locally Lipschitz  $\Rightarrow$ continuous $\Rightarrow$ (*) and upper semi-continuous (proper cone structure) $\Rightarrow$ upper semi-continuous (closed).
\end{quote}

The condition (*) is a kind of selection property. Observe that a Lorentzian manifold is time orientable if  there exists a continuous selection on the bundle of timelike vectors. Since reference frames (observers) are modeled with such selections, their existence is fundamental for the physical interpretation of the theory. The condition (*) might be regarded in a similar fashion as it implies that there are continuous selections which can be perturbed remaining selections. Another view on condition (*) is obtained by passing to the dual cone bundle. Then (*) can be read as a continuous sharpness condition.
\begin{example}
On the manifold $\mathbb{R}^2$ endowed with coordinates $(x,t)$, let us consider the cone distribution $\mathbb{R}_+(\dot x,1)$ where $ \dot x \in [-2,-k]$ for $x<0$,  $ \dot x \in [-2,2]$ for $x=0$ and $ \dot x \in [k,2]$ for $x>0$. It is upper semi-continuous for $-2\le k\le 2$, but it does not admit a continuous selection for $0<k\le 2$. For $k=0$ it admits the continuous selection $\p_t$ but it still does not satisfy (*). For $-2\le k<0$ it satisfies (*).
\end{example}

\begin{remark}
Most results of causality theory require two tools for their derivation. The {\em limit curve theorem} and the (*) condition.
The limit curve theorem holds under upper semi-continuity and its usefulness will be pretty clear. As for the (*) condition, many  arguments  use the fact that for $p\in M$ an arbitrarily close point  $q$ can be found in the causal future of $p$ such that the causal past of $q$ contains $p$ in its interior. This property holds under (*). In other arguments one needs to show that some achronal boundaries are Lipschitz hypersurfaces. This result holds again under (*).

Insistence upon upper semi-continuity is justified not only on mathematical grounds; discontinuities have to be taken into account, for instance, in the study of light propagation in presence of a discontinuous refractive index, e.g.\ at the interface of two different media, cf.\ Sec.\ \ref{fer}.

Moreover, upper semi-continuity turns out to be the natural assumption for the validity of most results. Assuming better differentiability properties might obscure part of the theory. For instance, at this level of differentiability the chronological relation loses some of its good properties but most results can be
proved anyway by using the causal relation, a fact which clarifies that the latter relation is indeed more fundamental. Hopefully the exploration of the mathematical limits of causality theory  might eventually tell us something on the very nature of gravity.
\end{remark}

\subsection{Causal and chronological relations}

Causality theory concerns the study of the global qualitative properties of solutions to the differential inclusion
\begin{equation} \label{zoq}
\dot x (t)\in C_{x(t)} \, ,
\end{equation}
where $x\colon I \to M$, $I$ interval of the real line.
If $x\in C^1(I)$ and (\ref{zoq}) is satisfied everywhere we speak of {\em classical solution}.

Of course, a key point is the identification of a more general and convenient notion of solution.
It has to be sufficiently weak to behave well under a  suitable notion of limit, however not too weak since it should   retain much of the qualitative behavior of $C^1$ solutions. The correct choice turns out to be the following: a {\em solution} is a map  $x$ which is  locally absolutely continuous, namely for every connected compact interval $[a,b]=:K\subset I$, $x\vert_K\in AC(K)$. The inclusion (\ref{zoq}) must be satisfied almost everywhere, that is in a subset of the differentiability points of $x$.

The notion of absolute continuity  can be understood in two equivalent ways, given $t \in I$ either we introduce a coordinate  neighborhood $U\ni x(t)$, and demand that the component maps $t \mapsto x^\alpha (t)$ be absolutely continuous real functions, or we introduce a Riemannian metric $h$ on $\bar  U$, and regard the notion of absolute continuity as that of maps to the metric space $(U,d^h)$. (It can be useful to recall that every  manifold admits a complete Riemannian metric \cite{nomizu61}. The Riemannian metric can be found Lipschitz provided  the manifold is $C^{1,1}$.)
Since on compact subsets any two Riemannian metrics  are Lipschitz  equivalent, the latter notion of absolute continuity is independent of the chosen Riemannian metric. Similarly, the former notion is independent of the coordinate system, as the changes of coordinates are $C^1$ and the composition $f \circ g$ with $f$ Lipschitz and $g$ absolutely continuous is absolutely continuous.

A solution to (\ref{zoq}) will also be called a {\em parametrized continuous causal curve}.
The image of a solution to (\ref{zoq}) will also be called a {\em continuous causal curve}.

\begin{remark}{\em Convenient reparametrizations.}
Over every compact set $A\subset U$ we can find a constant $a>0$ such that for every $x\in A$, $y\in T_xM$, $\Vert y\Vert_h=\sqrt{h_{\alpha \beta} y^\alpha y^\beta} \le a \sum_\mu \vert y^\mu\vert$. As each component $x^\mu(t)$ is absolutely continuous, each derivative $\dot x^\alpha$ is integrable and so $\Vert \dot x\Vert_h$ is integrable. The integral
\[
s(t)=\int_0^t \Vert \dot x\Vert_h(t') \dd t' \, ,
\]
is the Riemannian $h$-arc length. Observe that our condition (\ref{zoq}) together with the fact that $C$ does not contain the origin imply that the argument is positive almost everywhere so the map $t\mapsto s(t)$ is increasing and absolutely continuous. Its inverse $s\mapsto t(s)$ is differentiable wherever $t\mapsto s(t)$ is with $\dot s\ne 0$, in fact $t'= \dot s^{-1}=\Vert \dot x\Vert_h^{-1}$ at those points, where a prime denotes differentiation with respect to $s$. By Sard's theorem for  absolutely continuous functions \cite{montesinos15} and by the Luzin N property of absolutely continuous functions, a.e.\ in the $s$-domain  the map   $s\mapsto t(s)$ is differentiable and $ \dot x(t(s))\in C_{x(t(s))}$.
 At those points $ x'= \dot x/\Vert \dot x\Vert_h\in C_{x(t(s))}$ so $\Vert x'\Vert_h=1$ and the map $s
\mapsto x(t(s))$ is really Lipschitz.
Thus, by a change of parameter we can pass from absolutely continuous solutions to Lipschitz solutions parametrized with respect to $h$-arc length (see also the discussion in \cite[Sec.\ 5.3]{petersen06}).
\end{remark}

 In causality theory the parametrization is not that important; most often one uses the $h$-arc length where $h$ is a complete Riemannian metric, for that way the inextendibility of the solution is reflected in the unboundedness of the domain, cf.\ Cor.\ \ref{dox}. However, general absolutely continuous parametrizations are better behaved under limits, as we shall see. Finally, since the parametrization is not that important, we can replace the original cone $x\mapsto C_x$ structure with the compact convex replacements
\[
\check C_x= \{y\in C_x\cup 0\colon \Vert y\Vert_h\le 1\}.
\]

As we shall see, we shall be able to import several results from the theory of differential inclusion, by considering the distribution  $\check C_x$ in place of $C_x$.
In fact, we shall need some important results on differential inclusions under low regularity due to Severi, Zaremba, Marchaud,  Filippov, Wa\v zewski, and other mathematicians. As far as I know this is the first work which applies systematically differential inclusion theory to causality theory. Good accounts of the general theory of differential inclusions can be found in the books  by Clarke
\cite[Chap.\ 3]{clarke83}, Aubin and Cellina \cite{aubin84}, Filippov \cite{filippov88}, Tolstonogov \cite{tolstonogov00} and Smirnov \cite{smirnov02}. For a review see also \cite{kikuchi67,davy72}.

For every  subset $U$ of $M$ we define the causal relation
\begin{align*}
J(U)&=\{(p,q)\in U\times U\colon p=q \textrm{ or there is a continuous causal} \\ & \qquad \qquad \qquad  \quad  \ \quad  \textrm{
  curve contained in }  U \textrm{  from } p \textrm{ to } q\}.
\end{align*}
For $p\in U$ we write
\[
J^+(p, U)=\{q\in U\colon (p,q)\in J(U)\}, \ \ \textrm{ and } \ \ J^{-}(p, U)=\{q\in U\colon (q,p)\in J(U)\}. \]
 For  $S\subset U$, we write $J^+(S, U)=\cup_{p\in S}  J^+(p, U)$, and similarly in the past case.   For every set $S$ we  introduce the {\em horismos}
 \[
 E^\pm (S, U)=J^\pm (S, U)\backslash \mathrm{Int}_U (J^\pm (S, U)),
 \]
  where the interior uses the topology induced on $U$.

An element of $T_xM$ is a {\em timelike vector} if it belongs to $(\mathrm{Int} C)_x$. It is easy to prove that the cone of timelike vectors $(\mathrm{Int} C)_x$ is an open convex cone.  A {\em timelike curve} is the image of a piecewise  $C^1$ solution to the differential inclusion
\begin{equation} \label{tim}
\dot x(t)\in (\mathrm{Int} C)_{x(t)} \, .
\end{equation}
The chronological relation of $U\subset M$ is defined as follows
\begin{align*}
I(U)&=\{(p,q)\in U\times U\colon \textrm{ there is a  timelike }   \textrm{curve contained in } U  \textrm{ from } p \textrm{ to } q\}.
\end{align*}
 The bundle of lightlike vectors is $\p C$, thus
a {\em lightlike vector} at $x$ is an element of $\p C\cap \pi^{-1}(x) =(\p C)_x=C_x\backslash (\mathrm{Int} C)_x$, where $\pi\colon TM\to M$.


Thus a vector is timelike if sufficiently small perturbations of the vector preserve its causal character, i.e.\ timelike vectors are elements of $\mathrm{Int} C$. In general, a vector $v\in \mathrm{Int} C_x$ for some $x\in M$ might not have this property, for the perturbation changes the base point. For instance, consider the  Minkowski spacetime  with its canonical  cone distribution, but replace the cone at the origin with a wider cone, then for the modified cone structure $(\mathrm{Int} C)_{o}\subsetneq \mathrm{Int} C_o$ where $o$ is the origin, $\mathrm{Int} C_o$ is the wider timelike cone and $(\mathrm{Int} C)_{o}$ is the original timelike cone.
\begin{proposition} \label{cob}
For a $C^0$ proper cone structure $(\mathrm{Int} C)_{x}=\mathrm{Int} C_x$ for every $x$.
\end{proposition}

  So the naive definition of timelike cone as $\mathrm{Int} C_x$ works in the continuous case. Also for a  $C^0$ proper cone structure  the lightlike vectors at $x$ are the elements of  $\p C_x$.

\begin{proof}
Let $v\in \mathrm{Int} C$, $\pi(v)=x$, then $\mathrm{Int} C\cap T_xM\subset C_x$ is a neighborhood of $v$ for the topology of $T_xM$,  thus $v\in  \mathrm{Int} (C_x)$. Conversely, let us introduce a coordinate neighborhood $U\ni p$, so that $TU$ can be identified with $U\times \mathbb{R}^{n+1}$ and hence different fibers can be compared. Let $v\in \mathrm{Int} C_x$ and let $K\subset \mathrm{Int} C_x$  be a compact neighborhood of  $v$, then by  Prop.\ \ref{low} there is a neighborhood $N\ni x$ such that $K\subset C_w$ for every $w\in N$, namely $N\times K$ is a neighborhood of $v$ contained in $C$, hence $v\in \mathrm{Int} C$.
\end{proof}
For a proper cone structure we have
\begin{equation} \label{mxl}
I(U)=\cup_{\tilde C\le C} \tilde I(U) ,
\end{equation}
where $\tilde C$ runs over the   $C^0$ proper cone structures $\tilde C \le C$.
 This family is non-empty thanks to the (*) condition. Equation (\ref{mxl}) can be obtained by noticing that  any $C$-timelike curve is a $\tilde C$-timelike curve  for some $C^0$ proper cone structure, $\tilde C \le C$.
In general a proper cone structure $C$ will not contain a maximal $C^0$ cone structure.



For $p\in U$ we write
\[
I^+(p, U)=\{q\in U\colon (p,q)\in I(U)\}, \ \  \textrm{ and } \ \  I^{-}(p, U)=\{q\in U\colon (q,p)\in I(U)\}.
\]
 For $S\subset U$, we write $I^+(S, U)=\cup_{p\in S}  I^+(p, U)$, and similarly in the past case.
If $U=M$, the argument $U$ is dropped in the previous notations, so the causal relation is $J$ and the chronological relation is $I$. They will also be denoted $\le_J$ or just $\le$ and $\ll$. Also, we write $p<q$ if there is a continuous causal curve joining $p$ to $q$.

\begin{proposition} \label{oaw}
Let $(M,C)$ be a  proper cone structure, then the corners in a timelike curve can be rounded off so as to make it a $C^1$ solution to (\ref{tim}) connecting the same endpoints.  As a consequence, $I$ can be built from $C^1$ solutions.
\end{proposition}

\begin{proof}
Let $\sigma$ be a $C^1$ timelike curve ending at $p$ and $\gamma$ a $C^1$ timelike curve starting from $p$, then they can be modified in an arbitrarily small neighborhood  of $p$ to join into a $C^1$ timelike curve. In fact, let $\dot \sigma, \dot \gamma\in (\mathrm{Int} C)_p$ be the tangent vectors to the curves at $p$ in some parametrizations. We can find an open round cone  $\bar R_p\subset (\mathrm{Int} C)_p$ containing $\dot \sigma, \dot \gamma$ in its interior and a coordinate neighborhood $U\ni p$ such that $U\times R_p\subset \mathrm{Int} C$, where the product comes from the splitting of the tangent bundle induced by the coordinates. Thus we can find a Minkowski metric in a neighborhood of $p$ with cones narrower than $(\mathrm{Int} C)_q$, $q\in U$. But as it is well known the corner can be rounded off in Minkowski spacetime, \cite{lerner72,hawking73,minguzzi18b} and the modified curve has tangent contained in the Minkowski cone and hence in $\mathrm{Int} C$ in a neighborhood of $p$, as we wished to prove.
\end{proof}

\begin{proposition}
Let $(M,C)$ be a proper cone structure, then $I$ is open, transitive and contained in $J$.
\end{proposition}

\begin{proof}
Transitivity is clear. By Eq.\ (\ref{mxl}) it is sufficient to prove openness under the $C^0$ assumption.  $I$ is open because any $C$-timelike curve is also a timelike curve for a round cone structure $R$ with smaller cones, $R<C$, where the openness of the chronological relation is well known in  Lorentzian geometry \cite{hawking73}.
\end{proof}

\begin{example} \label{ekg}
In a closed cone structure the causal future of a point might have empty interior though $\mathrm{Int} C_x\ne 0$ for every $x$.  Consider a manifold $\mathbb{R}^2$ of coordinates $(x,t)$, endowed with the stationary (i.e.\ independent of $t$) cone structure $\mathbb{R}_+(\dot x,1)$ given by $\vert \dot x\vert\le 1$ at $x=0$ and $\vert \dot x\vert\le \vert x\vert$, for $\vert x\vert >0$. On the region $\vert x\vert >0$ the fastest continuous causal curves satisfy $\log [ x(t_1)/x(t_0)]=\pm (t_1-t_0)$, thus since the left-hand side diverges for $x(t_0)\to 0$, no solution starting from $(0,t_0)$ can reach the region $x>0$ and similarly the region $x<0$. Thus $J^+((0,t_0))=\{(0,t)\colon t\ge t_0\}$, which has empty interior. Notice that in this example it is not true that $(\mathrm{Int} C)_x\ne 0$ for every $x$, thus this is not a proper cone structure.
\end{example}
\begin{example} \label{mik}
In a proper cone structure a $C^1$ curve can be non-timelike even if $\dot x(t)\in \mathrm{Int} C_{x(t)}$.  Consider a manifold $\mathbb{R}^2$ of coordinates $(x,t)$, endowed with the stationary round cone structure $\mathbb{R}_+(\dot x,1)$: $ x \le \dot x\le  -x+ 1$ for $x<0$, $\vert \dot x\vert\le 1$ for $x=0$;  $  -x\le \dot x \le x+ 1$ for $x>0$. Notice that the $C^0$ proper cone structure $\tilde C$ defined by $\dot x \in [0,1]$ is contained in the given one. The curve $t \mapsto (0,t)$ is not timelike.
\end{example}

\begin{example} \label{mik2}
As another example, consider the manifold $\mathbb{R}^2$ of coordinates $(x,t)$, endowed with the stationary round cone structure $\mathbb{R}_+(\dot x,1)$: $\dot x\in [1,3]$ for $x<0$, $\dot x\in [-4,4]$ for $x=0$ and $\dot x\in [2,4]$ for $x>0$. Then the  $C^0$ proper cone structure $\tilde C$ defined by $\dot x \in [2,3]$ is contained in the given one. The curve $t \mapsto (0,t)$ is not timelike.
\end{example}

It is interesting to explore the properties of the relation  $\mathring{J}:=\mathrm{Int} J$ which will be used to define the notion of geodesic.

\begin{proposition} \label{jon}
 The relation $\mathring{J}$ is open, transitive and contained in $J$.  Moreover, in a proper cone structure $I\subset \mathring{J}$, $\overline{\mathring{J}}=\bar J$ and $\p \mathring{J}=\p J$.
\end{proposition}

One should be careful because in general $\mathring{J}^+(p)\subsetneq \mathrm{Int}( J^+(p))$.

\begin{proof}
It is open by definition, so let us prove its transitivity. Let $(p,q)\in \mathring{J}$ and $(q,r)\in \mathring{J}$, then there are is a product neighborhood which satisfies $(p,q)\in U\times V_1\subset J$, and a product neighborhood which satisfies $(q,r)\in V_2\times W\subset J$. Since $U\times \{q\} \cup \{q\}\times W\subset J$ we have by composition $U\times W\subset J$, thus $(p,r)\in \mathring{J}$. For the last statement of the proposition we need only to prove $J\subset \overline{\mathring{J}}$. Let $(p,q)\in J$ and let $p'\ll p$, $q'\gg q$, then since $I$ is open and contained in  $\mathring{J}$, $(p',q') \in \mathring{J}$. Since $p'$ can be taken arbitrarily close to $p$, and analogously, $q'$ can be taken arbitrarily close to $q$, we have $(p,q)\in \overline{\mathring{J}}$.
\end{proof}

The local causality of closed cone structures  is no different from that of Minkowski spacetime due to the next observation.

\begin{proposition} \label{iiu}
Let $(M,C)$ be a  closed cone structure. For every $x\in M$ we can find a relatively compact coordinate open neighborhood $U\ni x$, and a flat Minkowski metric $g$ on $U$ such that at every $y\in U$, $C_y\subset (\mathrm{Int} C^g)_y$ (that is $C\vert_U<C^g\vert_U$). Furthermore, for every Riemannian metric $h$ there is a constant $\delta_h(U)>0$ such that all continuous causal curves in $\bar U$  have $h$-arc length smaller than $\delta_h$.
\end{proposition}

We shall see later that the constructed neighborhood is really globally hyperbolic, (Remark \ref{roc}). Particularly important will be the local non-imprisoning property of this neighborhood which will follow by joining the last statement with Corollary \ref{dox}.

\begin{proof}
Since $C_x$ is sharp we can find a round cone $R_x$ in $T_xM$ containing $C_x$ in its interior. Thus we can find coordinates $\{ x^\alpha \}$ in a neighborhood $\tilde U\ni x$ such that the cone $R_x$ is that  of the Minkowski metric $g=-(\dd x^0)^2+\sum_i (\dd x^i)^2$, where $\dd x^0$ is positive on $C_x$. By upper semi-continuity all these properties are preserved in a sufficiently small neighborhood of the form $I_g^+(p,\tilde U)\cap I_g^-(q,\tilde U):=U\ni x$, $\bar U\subset \tilde U$, in particular the timelike cones of $g$ contain the causal cones of $C$. The continuous causal curves for $C$ in $\bar U$ are continuous causal curves for $C^g$, thus  the last statement follows from the Lorentzian version \cite[p.\ 75]{beem96}.
\end{proof}

Since every continuous  $C$-causal curve is continuous $g$-causal, there cannot be closed continuous $C$-causal curves in $U$.

\begin{remark} \label{nff}
Using standard arguments \cite{minguzzi06c} one can show that the closed cone structure admits at every point a basis  for the topology $\{U_k\}$, $\overline U_{k+1}\subset U_k$, with the properties mentioned by the previous proposition. In fact, the neighborhoods can be set to be nested chronological diamonds for $C^g>C$, so that $U_{k}$ is $C$-causally convex  in $U_1$ for each $k$  (furthermore, they are globally hyperbolic for both $g$ and $C$).
\end{remark}

\begin{proposition} \label{ohg}
Let $(M,C)$ be a closed cone structure. Then there is a locally Lipschitz 1-form $\omega$ such that $C$ is contained in $\omega>0$. Moreover, there is a locally Lipschitz proper cone structure $C'>C$  contained in $\omega>0$.
\end{proposition}

\begin{proof}
In Prop.\ \ref{iiu} we have shown that every $x\in M$ admits an open coordinate neighborhood $U$ such that $\omega^U:=\dd x^0$ is positive on $C\vert_U$, and the round cone $R^U$ of the Minkowski metric contains $C$ and is also in the positive domain of $\omega^U$. The use of a Lipschitz partition of unity and Prop.\ \ref{doo} gives the desired global result.
\end{proof}

A consequence of the Hopf-Rinow theorem and Prop.\ \ref{iiu} is
\begin{corollary} \label{dox}
Let $(M,C)$ be a closed cone structure and let $h$ be a complete Riemannian metric. A continuous causal curve $x\colon [0,a)\to M$ is future inextendible iff its $h$-arc length is infinite.
\end{corollary}




Concerning the existence of solutions we have the next results. Under upper semi-continuity we have \cite[Cor.\ 4.4]{smirnov02}
\cite[Th.\  2.1.3,4]{aubin84}
\begin{theorem} \label{zar} (Zaremba, Marchaud)
Let $(M,C)$ be a closed cone structure. Every point $p\in M$ is the starting point of  an inextendible continuous causal curve. Every continuous causal curve can be made inextendible through extension.
\end{theorem}

For a proper cone structure we have also
\begin{theorem} \label{mmz}
Let $(M,C)$ be a proper cone structure. For every $x_0\in M$ and timelike vector $y_0\in (\mathrm{Int} C)_{x_0} $, there is a timelike curve passing through $x_0$ with velocity $y_0$.
\end{theorem}

\begin{proof}
Since $(\mathrm{Int} C)_{x_0}$ is open there is a closed round cone $R_{x_0}\ni y_0$ contained in $(\mathrm{Int} C)_{x_0}$. Thus we can find coordinates $\{ x^\alpha \}$ in a neighborhood $U\ni x_0$ such that the cone $R_{x_0}$ is that of the Minkowski metric $g=-(\dd x^0)^2+\sum_i (\dd x^i)^2$, where  $\p_0 \in (\mathrm{Int} C)_{x_0}$. By continuity all these properties are preserved in a sufficiently small neighborhood $U\ni x_0$, in particular the timelike cones of $g$ are contained in $\mathrm{Int} C$. Then the integral line of $\p_0$ passing through $x_0$ is a timelike curve.
\end{proof}

Under stronger regularity conditions it can be improved as follows \cite[Th.\ 4]{filippov67} (the non-convex valued version in \cite[p.\ 118]{aubin84} has to assume Lipschitzness).
\begin{theorem} \label{sbs}
Let $(M,C)$ be a $C^0$ closed cone structure. For every $x_0\in M$ and $y_0\in C_{x_0} $, there is a $C^1$ causal curve passing through $x_0$ with velocity $y_0$. If the cone structure is proper and $y_0$ is timelike the curve can be found timelike.
\end{theorem}
Continuous causal curves can be characterized using the local causal relation, in fact we have the following  manifold translation of 
\cite{ghuila65}\cite[p.\ 99, Lemma 1]{aubin84}.
\begin{theorem}
Let $(M,C)$ be a closed cone structure.  A continuous curve $\sigma$ is a continuous causal curve if and only if for  every $p\in \sigma$ there is a coordinate neighborhood $U\ni p$, such that for every $t\le t'$ with $\sigma([t,t'])\subset U$ we have $\sigma(t')\in J^+(\sigma(t),U)$.
\end{theorem}

It turns out that upper semi-continuity and Lipschitz continuity are the most interesting weak differentiability conditions that can be placed on the cone structure.

We recall a key, somehow little known result by Filippov \cite[Th.\ 6]{filippov67} \cite[Th.\ 3.1]{wolenski90b}. Here $\Vert \gamma-\sigma\Vert =\sup_t \Vert \gamma(t)-\sigma(t)\Vert$ and the meaning of {\em solution} has been clarified after Eq.\ (\ref{zoq}).
\begin{theorem} \label{azs}
Let $U$ be an open subset of $\mathbb{R}^n$, and let $x\mapsto \check C_x\subset \mathbb{R}^n$ be a Lipschitz multivalued map defined on $U$ with non-empty compact convex values. Let $\sigma\colon [0,a]\to U$,  be a solution of $\dot x\in \check C_{x(t)}$ with initial condition $ \sigma(0)= p\in U$. For any
$\epsilon >0$ there exists a $C^1$ solution  $ \gamma\colon [0,a]\to U$ to $\dot x\in \check C_{x(t)}$ with initial condition $\gamma(0)= p$,
such that $\Vert \gamma-\sigma\Vert \le \epsilon$.
\end{theorem}

It has the following important consequence.

\begin{theorem} \label{sbn}
Let $(M,C)$ be a locally Lipschitz proper cone structure and let $h$ be a Riemannian metric.
 Every point admits an open neighborhood $U$ with the following property.
  Every  $h$-arc length parametrized continuous causal curve in $U$ with starting point $p\in U$ can be uniformly approximated by a $C^1$ timelike solution with the same starting point, and time dually. In particular, $\overline{I^+(p,U)}\supset J^+(p,U)$ and $\overline{I^-(p,U)}\supset J^-(p,U)$.
\end{theorem}

With Th.\ \ref{dao} we shall learn that the last inclusions are actually equalities. It is worth to mention that the neighborhood $U$ is constructed as in Prop.\ \ref{iiu}.


\begin{proof}
Let $U$ be a coordinate neighborhood endowed with coordinates $\{x^\alpha\}$ constructed as in the proof of Prop.\ \ref{iiu}, where additionally  $\p_0\in \mathrm{Int} C$ and the 1-form $\omega=\dd x^0$ is Lipschitz and positive over $C\vert_U$.
Theorem \ref{azs} applies with
\[
\check C_x= \{y\in  C_x\colon \Vert y\Vert_h\le 1, \textrm{ and } \omega(y)
\ge \delta \},
\]
where $\delta>0$ can be chosen so small on $\bar U$ that $\check C_x\supset \{y\in  C_x\colon \Vert y\Vert_h= 1\}\ne \emptyset$. Every $h$-arc length parametrized solution $\sigma\colon [0,a]\to U$ to (\ref{zoq}) is a solution to $\dot x(t)\in \check C_{x(t)}$ since its velocity is almost everywhere $h$-normalized. Moreover, for every $q\in U$, $\check C_q$ is non-empty, compact and convex. By Theorem \ref{azs} for every $\epsilon >0$ there is  classical solution   $ \gamma \colon [0,a]\to U$ to $\dot x\in \check C_{x(t)}$ with initial condition $\gamma(0)= p$,
such that $\Vert \gamma-\sigma\Vert \le \epsilon/2$, where the norm is the Euclidean norm induced by the coordinates. But this solution is also a $C^1$ solution to $\dot x(t)\in C_{x(t)}$ (since $\delta >0$, we have  $\dot \gamma (t) \ne 0$ for every $t$), namely $\gamma$ is a $C^1$ causal  curve. Let us consider the curve $\eta$ whose components are $\eta^i(t)=\gamma^i(t)$, $\eta^0(t)= \gamma^0(t)+ \epsilon \frac{t}{2a}$, then  $\Vert \eta- \gamma \Vert \le \epsilon/2$, thus $\Vert \eta-  \sigma \Vert \le \epsilon$, but $\dot
\eta^i= \dot \gamma^i $, $\dot \eta^0=\dot \gamma ^0+ \frac{\epsilon}{2a}$, that is $\dot \eta=\dot \gamma+ \frac{\epsilon}{2a}\p_0$ which is timelike. $\empty \ \, $
\end{proof}

The previous result establishes that under Lipschitz regularity, at least locally the solutions to the differential inclusion $\dot x(t)\in F(x(t))$, with $F(x)=\mathrm{Int} C_x$ in our case, are dense in the solutions to the {\em relaxed}  differential inclusion $\dot x(t)\in \overline{\textrm{co}} \, F(x(t))$, where $\overline{\textrm{co}} \, F(x)$ is the smallest closed convex set containing $F(x)$. Results of this type are called {\em relaxation theorems}  the first versions being proved by Filippov and Wa\v zewski \cite{clarke83} \cite[Th.\ 2, Sec.\ 2.4]{aubin84}. In the Lorentzian framework the importance of the Lipschitz condition for the validity of the inclusion $\overline{I^+(p,U)}\supset J^+(p,U)$   was recognized   by Chru\'sciel and Grant \cite{chrusciel12}. They termed {\em causal bubbles} the sets of the form $J^+(p,U)\backslash \overline{I^+(p,U)}$.

We arrive at a classical result of causality theory.

\begin{theorem} \label{soa}
Let $(M,C)$ be a locally Lipschitz proper cone structure. Let $\gamma$ be a continuous causal curve obtained by joining a continuous causal curve $\eta$ and a timelike curve $\sigma$ (or with order exchanged). Then $\gamma$ can be deformed  in an arbitrarily small neighborhood $O\supset \gamma$ to give a timelike curve $\bar \gamma$ connecting the same endpoints of $\gamma$. In particular, $J\circ I\cup I\circ J\subset I$, $\bar J=\bar I$, $\p J=\p I$, $I=\mathring{J}$. For every subset $S$,  $\overline{ J^+(S)}=\overline{ I^+(S)}$, $\p J^+(S)=\p I^+(S)$, $I^+(S)=\mathrm{Int} (J^+(S))$, and time dually.
\end{theorem}

A word of caution. One might wish to consider causal and chronological relations $J(B)$, $I(B)$, where $B$ is not necessarily open. However, in this case  $J(B)\circ I(B)\cup I(B)\circ J(B)\subset I(B)$ would not hold since the deformed curve mentioned in the theorem might not stay in $B$.

\begin{proof}
Let $O$ be an open subset containing $\gamma$. Let $p=\eta(0)$ and $q=\eta(1)$ be the endpoints of $\eta\colon [0,1] \to O$  and let $q$ and $r$ be the endpoints of $\sigma$. Let $A\subset  [0,1] $ be given by those $t$ such that $\eta(t)$ can be connected to $r$ with a timelike curve contained in $O$. Clearly $1\in A$ and since $I(O)$ is open there is a maximal open connected subset of $A$ containing $1$. It cannot have infimum $a\ge 0$, $a\notin A$, indeed by contradiction,  $x=\eta(a)$ admits a neighborhood  $U\ni x$, $U\subset O$ with the properties of Theorem \ref{sbn}. So we can find $y\in U$, $y=\eta(b)\in \eta$, $b>a$, and a timelike curve in $U$ starting from $x$ with endpoint  arbitrarily close to $y$. But $I(O)$ is open and $y\ll_O r$ so there is a timelike curve from $x$ to $r$, a contradiction. Thus $A=[0,1]$ and there is a timelike curve from $p$ to $r$ contained in $O$.

For the penultimate statement we have only to show that $\bar J\subset \bar I$, but this follows immediately if for every continuous causal curve $\gamma$ and every neighborhood $O\supset \gamma$ we can find a timelike curve $\bar \gamma \subset O$ with endpoints arbitrarily close to the endpoints of $\gamma$.
  Let $U$ be an arbitrarily small neighborhood, of the type mentioned in Theorem \ref{sbn},  of the future endpoint $r$ of $\gamma$. Then we can find $q\in \gamma\cap U$, $q< r$, and a timelike curve $\sigma$ in $U$ with future endpoint $r'$ close to $r$ as much as desired. Then by the first part of this theorem we can find a timelike curve $\bar \gamma\subset O$, with endpoints $p$ and $r'$, which concludes the proof.

The inclusion $I \subset \mathring{J}$ was proved in Prop.\ \ref{jon}. For the other direction let $(p,q)\in \mathring{J}$ and let  $q'\ll q$ be a point sufficiently close to $q$ that $(p,q')\in J$. By Th.\ \ref{sbn} we can find $r\gg p$ sufficiently close to $q'$ that $r\ll q$, thus $p \ll q$.

The last statement has a proof  very similar to that of the penultimate statement, just observe that $\bar \gamma$ has the same starting point as $\gamma$.
\end{proof}

In the next theorem we say that a property holds locally if there is a covering $\{V_\alpha\}$ of $M$, consisting of relatively compact open sets such that the property holds for every cone structure $(V_\alpha,C\vert_{V_\alpha})$.

\begin{theorem}
Let $(M,C)$ be a proper cone structure. The conditions
\begin{itemize}
\item[(a)] $I\circ J\cup J\circ I\subset I$, (causal space condition)
\item[(b)] for both sign choices and for all $ p, \ J^\pm(p)\backslash\overline{ I^\pm (p)}=\emptyset$ (no causal bubbling),
 \end{itemize}
 are equivalent. Moreover, the local versions imply the global versions, while the other direction holds provided $(M,C)$ is strongly causal. Finally, they imply $\bar I=\bar J$, $I=\mathring J$, $\p I=\p J$, and that for  every subset $S$,  $\overline{ J^+(S)}=\overline{ I^+(S)}$, $\p J^+(S)=\p I^+(S)$, $I^+(S)=\mathrm{Int} (J^+(S))$, and time dually.
\end{theorem}

Condition (a) is the main characterizing property of Kronheimer and Penrose's {\em causal spaces}  \cite{kronheimer67}, which can be defined as triples $(M,I,J)$ where $M$ is a set, $I\subset J\subset M\times M$ are transitive relations which satisfy property (a), where additionally $I$ is irreflexive and $J$ is reflexive and antisymmetric. Hence our terminology.

It  can be noticed that the proof (a) $\Rightarrow$ (b) uses only the transitivity of $I$ and $J$, and the openness of $I$, while (b) $\Rightarrow$ (a) uses also the fact that $I^\pm(p) \cap V$ is non-empty for every point $p\in M$ and open set $V\ni p$.

In some of the next results we shall assume that the proper cone structure is   locally Lipschitz when in fact, as the comparison of this theorem and the previous one suggests, we could have just imposed properties (a) and (b).

\begin{proof}
(a) $\Rightarrow$ (b). Indeed, if $q\in J^+(p)$ it is sufficient to take $r\in I^+(q)$ and notice that $r$ can be chosen arbitrarily close to $q$. Then $r\in I^+(p)$ implies $q\in \overline{ I^+(p)}$.

(b) $\Rightarrow$ (a).   Let $(p,q)\in J$ and $r\in I^+(q)$, from the assumption $J^+(p)\subset \overline{I^+(p)}$, but we know that $I^+(p)\subset {J^+(p)}$, thus $I^+(p)= \mathrm{Int}{J^+(p)}$. Now $I^+(q)$ is an open neighborhood of $r$ contained in $J^+(p)$, thus $r\in \mathrm{Int}{J^+(p)}=I^+(p)$. The similar case with $p\in I^-(q)$ and $(q,r)\in J$ is treated similarly, so $I\circ J\cup J\circ I\subset I$.

Suppose that every point admits a neighborhood $V$ with the properties of the theorem and such that (a) holds, $I(V)\circ J(V)\cup J(V)\circ I(V)\subset I(V)$. Let us prove that (a) holds globally. Indeed, if not there are  a timelike curve  $\gamma\colon [0,1]\to M$, and a continuous causal curve $\sigma\colon [0,1]\to M$ where $\gamma(1)=\sigma(0)$, such that, (recall that $I^+(\gamma(0))$ is open) there is a first point $p:=\sigma(t)$, $t>0$, of exit from $I^+(\gamma(0))$ (the case in which the first curve is timelike and the second is causal is treated in the time-dual way). Let $V\ni p$ be  a neighborhood with the mentioned properties, then for sufficiently small $0<\epsilon<t$, $q:=\sigma(t-\epsilon)\in I^+(\gamma(0))\cap V$ and the $\sigma$-segment between $q$ and $p$ is contained in $V$. Let $r\in V$ be a point in a timelike curve $\eta$ connecting $\gamma(0)$ to $q$, sufficiently close to $q$ that the segment of $\eta$ between $r$ and $q$ stays in $V$, then $(r,q)\in I(V)$ and $(q,p)\in J(V)$ which by the local assumption imply $(r,p)\in I(V)\subset I$, and so $p\in I^+(\gamma(0))$, a contradiction.

Conversely, suppose that (a) holds and that $(M,C)$ is strongly causal, and let us consider a covering of open causally convex relatively compact sets. If $\gamma$ and $\sigma$, $\gamma(1)=\sigma(0)$, are timelike and continuous causal curves contained in one such set $V$, then their concatenation joins points in $V$ which, by assumption, can be joined by a timelike curve.  By causal convexity the timelike curve has to be contained in $V$, hence $I(V)\circ J(V)\subset J(V)$.


Let us prove $\bar I=\bar J$, for the other two identities follow from that. Since $I\subset J$, $\bar I \subset \bar J$, so we have to prove  the condition $\bar J\subset \bar I$, or equivalently $J\subset \bar I$. Assume that the are no causal bubbles, let $(p,q)\in J$, then $q\in J^+(p)\subset \overline{I^+(p)} $ which implies $(p,q)\in \bar I$. The proof of the identity $I=\mathring{J}$ is as in Th.\ \ref{soa}.
The results on the subset $S$, follow from ${ J^+(S)}\subset \overline{ I^+(S)}$, so let $q\in J^+(p)$, $p\in S$, moreover let $r\in I^+(q)$ so that $r\in I^+(p)$, then the limit $r\to q$ gives $q\in \overline{I^+(p)}\subset \overline{ I^+(S)}$ as desired.
\end{proof}

The next result on the arc-connectedness of the space of  solutions is a manifold reformulation of a Kneser's type theorem for differential inclusions  \cite[Cor.\ 4.2, 4.6]{smirnov02} \cite{davy72}
\begin{theorem}
Let $(M,C)$ be a locally Lipschitz proper cone structure. Any point of $M$ admits an open neighborhood $U$ such that for any $p\in U$, any two parametrized continuous causal curves starting (or ending) at $p$ contained in $U$ are joined by a continuous homotopy of continuous causal curves starting (resp.\ ending) at $p$.
\end{theorem}

\subsection{Notions of increasing functions} \label{nug}
 We shall make use of various notions of increasing function for a closed cone structure $(M,C)$. For future reference we list them here.
A continuous function $\tau\colon M\to \mathbb{R}$ is
\begin{itemize}
\item[(a)]  {\em causal} or {\em isotone}, if  $(p,q)\in J \Rightarrow \tau(p)\le \tau(q)$,
\item[(b)] a {\em time function}, if it increases over every continuous causal curve,
\item[(c)] {\em Cauchy} if restricted to any inextendible continuous causal curve it has image $\mathbb{R}$,
\item[(d)] a {\em temporal function}, if it is $C^1$ and such that for every $p\in M$, $\dd \tau$ is positive on the (future) causal cone $C_p$, (it would be called a (minus) {\em Lyapounov} function in the study of dynamical systems)
    \item[(e)]   {\em  locally anti-Lipschitz}, if there is a Riemannian metric $h$ such that  for every compact set $K$, there is a constant $C_K>0$ such that $\tau(\gamma(1))-\tau(\gamma(0))\ge C_K \ell^h(\gamma)$ for every continuous causal curve $\gamma\colon [0,1]\to K$  (this property does not depend on $h$). By $\sigma$-compactness if $\tau$ is locally anti-Lipschitz there is a Riemannian metric $\hat h$ such that  $\tau(\gamma(1))-\tau(\gamma(0))\ge \ell^{\hat h}(\gamma)$ for every $\gamma\colon [0,1]\to M$. We also say that $\tau$ is $\hat h$-{\em anti-Lipschitz}. We say that $\tau$ is {\em  stably locally anti-Lipschitz} if it is   locally anti-Lipschitz with respect to some wider $C^0$ proper cone structure $C'>C$ (it exists by Prop.\ \ref{ohg}).
\item[(f)] $f$-{\em steep}, if there is a continuous function $f\colon C\to [0,+\infty)$ positive homogeneous of degree one,  $\tau$ is $C^1$ and $\dd \tau(y) \ge f(y)$ for every $y\in C$ ({\em strictly} steep if the inequality is strict). Thus strictly $f$-steep functions are temporal. With some abuse of notation we say that $\tau$ is $h$-{\em steep}, if  with respect to the Riemannian metric $h$, for every $y\in C$, we have $\dd \tau(y) \ge \Vert y\Vert_h $ (hence $h$-anti-Lipschitz and temporal). If $h$ is complete then it is Cauchy.

\end{itemize}


\noindent The last claim in (f) is due to the fact that  over every inextendible continuous causal curve $x\colon I \to M$, the $h$-arc length $\int_a^b \Vert \dot x\Vert_h \dd t$ diverges in both directions \cite{bernard16}.

The next results, which are the cone structure version of \cite[Prop.\ 4.3]{chrusciel13}, suggest that in order to construct temporal functions one has to focus on anti-Lipschitz functions.
\begin{theorem} \label{oqm}
Let $(M,C)$ be a $C^0$ proper cone structure. The  $C^1$ locally anti-Lipschitz functions are precisely the temporal functions.
\end{theorem}

\begin{proof}
Let $\tau$ be temporal, then at every $p$, $(\dd \tau\vert_p)^{-1}(1)\cap C_p$ is compact, so we can find a Riemannian metric $h$ whose unit balls contain it.  Then for every $v\in C$, $\dd \tau (v) \ge \Vert v\Vert_h$ which implies $h$-anti-Lipschitzness.

Let $\tau$ be a $C^1$ locally anti-Lipschitz function, then by $\sigma$-compactness there is a Riemannian metric $h$ such that $\tau$ is $h$-anti-Lipschitz. Let us consider a ($C^1$) timelike curve $x\colon [0,1)\to M$ and let us set $v=\dot x(0)$. We know that $\tau(x(t))-\tau(x(0)) \ge  \ell^h(x([0,t)))=\int_0^t \Vert \dot x(s)\Vert_h \dd s$, thus dividing by $t$ and taking the limit $t\to 0$, we get $\dd \tau(v) \ge \Vert v\Vert_h$. By Th.\ \ref{mmz} the inequality is true for every $v\in \mathrm{Int} C_{x(0)}$ and hence, by continuity, for every $v\in C$.
\end{proof}

\begin{theorem}
Let $(M,C)$ be a closed cone structure. The  $C^1$ stably locally anti-Lipschitz functions are precisely the temporal functions.
\end{theorem}

\begin{proof}
Let $\tau$ be temporal. Since $\dd \tau$ is positive on $C$ we can find the locally Lipschitz proper cone structure $C'>C$ of Prop.\ \ref{ohg} so close to $C$ that $\dd \tau$ is positive on $C'$.  By Th.\ \ref{oqm} $\tau$ is locally anti-Lipschitz with respect to $C'$ hence a $C^1$ stably locally anti-Lipschitz function.

Let $\tau$ be a $C^1$ stably locally anti-Lipschitz function, then there is a $C^0$ proper cone structure $C'>C$ such that $\tau$ is  $C^1$  locally anti-Lipschitz
 with respect to $C'$, and by Th.\ \ref{oqm} a  temporal function for $C'$ and hence for $C$.
\end{proof}

As we shall see (Remark \ref{sof}), we shall obtain  temporal functions for closed cone structures by passing through the preliminary construction of stably locally anti-Lipschitz functions.

\subsection{Limit curve theorems}

One of the most effective tools used in causality theory is the limit curve theorem \cite{hawking73,beem96,minguzzi07c}. The theory of differential inclusions clarifies that it is very robust, as it holds  under upper semi-continuity of the cone structure.

The next result follows easily from \cite[Th.\ 4.6]{smirnov02} \cite[Cor.\ 2.7.1]{filippov88}.

\begin{theorem} \label{mxe}
Let $(M,C)$ and $(M,C_k)$, $k\ge 1$, be  closed cone structures, $C_{k+1}\le C_k$,  $C=\cap_k C_k$,  and let $h$ be a Riemannian metric on
$M$. If the continuous $C_k$-causal curves $x_k\colon I_k \to M$,
parametrized with respect to $h$-arc length, converge $h$-uniformly on
compact subsets to $x\colon I \to M$, then $x$ is a  continuous
$C$-causal curve.
\end{theorem}

\begin{proof}
The theorem is true for any constant sequence $C_k=C$ by \cite[Cor.\ 2.7.1]{filippov88}. So for every $s$, the sequence $x_n$ consists of continuous $C_s$-causal curves for $n\ge s$, thus $x$ is a continuous $C_s$-causal curve. So for every $s$, $\dot x \in C_s$ a.e., which implies $\dot x \in C$ a.e., namely $x$ is a continuous $C$-causal curve.
\end{proof}

The next result is the manifold version of \cite[Cor.\ 4.5]{smirnov02}.

\begin{theorem} \label{son}
Let $(M,C)$ be a closed cone structure, and let $h$ be a Riemannian metric. Let $K\subset M$ be compact and let $x_k\colon [0,L]\to K$ be a sequence of $h$-arc length parametrized continuous causal curves, then there is a subsequence converging uniformly on $[0,L]$ to a continuous causal curve $x$ (whose parametrization is not necessarily the $h$-arc length parametrization).
\end{theorem}

The bound on the $h$-arc length of $x_k$ is necessary, without it counterexamples can easily be found on the Lorentzian 2-dimensional spacetime $\mathbb{R}\times S^1$ whose metric is $g=-\dd t\dd \theta$.

\begin{proof}
Consider a finite covering $\{U_i\}$ of $K$ by coordinate neighborhoods and let $\delta>0$ be a Lebesgue number relative to the metric $d^h$.
 A subsequence $x_k^1$ of $x_k$ is such that  the points $x^1_k(0)$ converge to some point $x(0)\in U_i$ for some $i$. Apply to the sequence $x^1_k\vert_{[0,\delta]}$ the mentioned result \cite[Cor.\ 4.5]{smirnov02}, thus obtaining a convergent sequence $x_k^2$, then focus on the convergence of  $x^2_k(\delta)$ and repeat the argument proceeding in $[0,\delta]$ steps. Since $L/\delta$ is bounded by some natural number $N$, in $N$-steps one constructs the desired converging sequence.
\end{proof}

As a corollary we obtain the limit curve lemma familiar from (Lorentzian) mathematical relativity
\cite{galloway86b} \cite[Lemma 14.2]{beem96} under much weaker assumptions.
\begin{lemma}(Limit curve lemma) \\
Let $(M,C)$ and $(M,C_n)$ be  closed cone structures, where $C=\cap_n C_n$ and  for every $n$, $C_{n+1}\le C_n$, and let $h$ be a complete
Riemannian metric. \\
Let $x_n\colon (-\infty,+\infty) \to M$,  be a
sequence of inextendible continuous causal curves parametrized with
respect to $h$-arc length, and suppose that $p\in M$ is an accumulation
point of the sequence $x_n(0)$. There is  an  inextendible
continuous causal curve $x\colon (-\infty,+\infty)  \to M$, such
that $x(0)=p$ and a subsequence $x_k$ which converges
$h$-uniformly on compact subsets to $x$.
\end{lemma}

Using the previous results  we obtain a version which is especially useful when we have causal segments for which both endpoints are converging, see  \cite{minguzzi07c} for the Lorentzian version.

\begin{theorem} \label{main} (Limit curve theorem)\\
Let $(M,C)$ and $(M,C_k)$ be  closed cone structures, where $C=\cap_n C_n$ and  for every $n$, $C_{n+1}\le C_n$, and let $h$ be a complete
Riemannian metric. Let $x_n\colon [0,a_n] \to M$ be a sequence  of $h$-arc length parametrized
continuous $C_n$-causal curves with endpoints $p_n\to p$, and $q_n\to q$.
Provided the curves $x_n$ do not contract to a point (which is the case if $p\ne q$) we can find either (i) a  continuous $C$-causal curve $x\colon [0,a] \to M$ to which a subsequence $x_k$, $a_k\to a$, converges uniformly on compact subsets,  or (ii) a future inextendible  parametrized continuous $C$-causal curve $x^p\colon [0,+\infty)  \to M $ starting from $p$, and a past inextendible  parametrized continuous $C$-causal curve $x^q\colon (-\infty,0]  \to M$ ending at $q$, to which some subsequence $x_k(t)$ (resp.\ $x_k(t+a_k)$) converges uniformly on compact subsets. Moreover, for every $p'\in x^p$ and $q'\in x^q$, $(p',q')\in \cap_n \bar J_n$.
\end{theorem}

\begin{proof}
The proof for the constant sequence case, $\forall n\ C_n=C$, coincides with that given in \cite{minguzzi07c} for a Lorentzian structure as the tools used there, such as the limit curve lemma, have been already generalized. The general case follows from the next argument. We apply the theorem of the constant sequence case to $(M,C_1)$ obtaining a subsequence $x^1_k$ which converges $h$-uniformly to some parametrized continuous $C_1$-casual curve $x^1$, then we apply it to $(M,C_2)$ obtaining a converging subsequence of $x^1_k$, denoted $x^2_k$, which converges $h$-uniformly to some continuous $C_2$-causal curve $x^2$, necessarily coincident with $x:=x^1$ by $h$-uniform convergence,  and so on. Finally, we take the diagonal subsequence $x^k_k$ converging $h$-uniformly to $x$. Since $x=x^k$  is a continuous $C_k$-causal curve for every $k$, it is also a continuous $C$-causal curve.
\end{proof}

The previous result together with Th.\ \ref{soa} implies (see \cite{chrusciel12} for the analogous Lorentzian statement)

\begin{remark}
The results of Lorentzian causality theory \cite{hawking73,beem96,minguzzi06c,chrusciel11} which do not explicitly
address normal neighborhoods or geodesics remain valid for locally Lipschitz cone structures.
\end{remark}

Actually, several results still make sense in the locally Lipschitz theory which do involve lightlike geodesics, as we shall see in the next section. In what follows we shall explore them and we shall investigate more closely  causality theory with the aim of understanding whether the locally Lipschitz condition can be weakened to an upper semi-continuity or a continuity condition.

\subsection{Peripheral properties and lightlike geodesics}


We need a generalization of  the notion of achronality.

\begin{definition}
Given a relation $R$ and a set $S$ we say that $S$ is $R${\em --arelated} if  no two points $p,q\in S$  are such that $(p,q)\in R$.
A set $S$ is {\em achronal} (resp.\ acausal) if no two points of $S$ are connected by a timelike curve (resp.\ continuous causal curve).
\end{definition}
Thus {\em achronal} stands for  $I$--arelated and {\em acausal} for $J\backslash \Delta$--arelated.
Since $I\subset \mathring{J}$, $\mathring{J}$--arelation is in general stronger than achronality when the latter  can be defined. They coincide in  locally Lipschitz proper cone structures.

On a  cone structure we can make sense of lightlike geodesics as follows. Notice that we do not include the property of inextendibility in the definition. Also most instances of {\em future} and {\em past} in the next definition refer to the relation direction not to the inextendibility of the domain.

\begin{definition}
 A {\em lightlike geodesic} is a continuous causal curve which is locally $\mathring{J}$-arelated. A {\em lightlike line} is an inextendible continuous causal curve which is $\mathring{J}$-arelated.
A {\em future lightlike geodesic}  is a continuous causal curve $\sigma$  such that every $r\in \sigma$ admits an open neighborhood $U$
 for which locally we cannot find two points in $\sigma$ such that $q\in \mathrm{Int} J^+(p,U)$.

We have also analogous past notions and global $\mathring{J}$--arelation notions in which {\em geodesic} is replaced by {\em line}.
A {\em  future lightlike ray} is a future inextendible lightlike geodesic which is $\mathring{J}$--arelated.
If in the second sentence of the previous paragraph {\em inextendibility} is replaced by {\em future inextendibility}, then {\em line} is replaced by {\em future ray}, and time dually. A future and past lightlike geodesic is a lightlike {\em bigeodesic}.
\end{definition}
A future or past lightlike geodesic is a lightlike geodesic (because $q\in E^+(p,U)$ implies $(p,q) \notin \mathring{J}(U)$), and the converse holds for locally Lipschitz proper cone structures. We defined lightlike geodesics using $\mathring{J}$--arelation in place of achronality because the natural generators of Cauchy horizons or horismos will be of this type in the future or past version. These lightlike geodesic concepts all coincide for locally Lipschitz proper cone structures.

\begin{remark}
In the Lorentzian case and under Lipschitz regularity one could write down the geodesic (spray) equation and, following Filippov, regularize the discontinuous  ($L^\infty_{loc}$) right-hand side into a multivalued map. Then one could show that the resulting differential inclusion admits $C^1$ solutions. This approach has been followed by Steinbauer in \cite{steinbauer14} but it has some limitations, for it seems difficult to prove the local achronality property of lightlike geodesics with such an approach. For this reason we use the local achronality (or  better said the $\mathring{J}$--arelation) property to introduce the very notion of lightlike geodesic. This definition is the best suited  in order to obtain non local results. Causal geodesics will be introduced with a similar idea.
\end{remark}

The neighborhood of the next result coincides with that constructed in the proof of Prop.\ \ref{iiu}. We recall that $E^+(p,U)= J^+(p, U)\backslash \mathrm{Int} J^+(p, U)$ and that a set $S$ is {\em causally convex} if $J^+(S)\cap J^{-}(S)\subset S$.

\begin{theorem} \label{dao}
Let $(M,C)$ be a closed cone structure. Every point in $M$ has an arbitrarily small coordinate neighborhood $U$ with the following property.
The relation $J(U)$ is closed and for every $p\in U$ and $q \in E^+(p, U)\backslash \{p\}$ there is a future lightlike geodesic joining $p$ and $q$ entirely contained in $E^+(p, U)$ (and time dually).
Moreover, if $(M,C)$ is locally Lipschitz every continuous causal curves connecting $p$ to $q$ is a lightlike geodesic contained in $E^+(p, U)$. Finally, if the closed cone structure $(M,C)$ admits arbitrarily small causally convex neighborhoods (strong causality) then $U$ can be chosen causally convex.
\end{theorem}

\begin{proof}
Let us prove the last statement.
Let $r\in M$ and let $V\ni r$ be an open set. Take $U\subset V$ constructed as  in the proof of Prop.\ \ref{iiu} or proceed as follows if $(M,C)$ is strongly causal.
 Let $\tilde U\ni r$ be a coordinate neighborhood contained in $V$ such that on $\tilde U$ we have a (flat) Minkowski metric $g$ wider than $C$. Let $U'\ni r$ be a $g$-chronological diamond contained in $\tilde U$, let $W\ni r$ be a $C$-causally convex neighborhood contained in $U'$, and let $U\ni r $ be  a smaller chronological $g$-diamond contained in $W$. Then $U$ is $C$-causally convex and of the same type as constructed in the proof of Prop.\ \ref{iiu}.

Let us prove that $J(U)$ is closed. Let $h$ be a Riemannian metric and let $\sigma_k$ be a sequence of continuous causal curves contained in $U$ connecting $p_k \to p$ to $q_k\to q$. By the limit  curve theorem either there is a continuous causal curve $\sigma$ connecting $p$ to $q$, necessarily contained in $U$ by the diamond shape of $U$, cf.\ the proof of Prop.\ \ref{iiu}, or  there is a past inextendible curve contained in $\bar U$ and ending at $q$. However, in the latter case the curve would have infinite $h$-arc length which is impossible by Prop.\ \ref{iiu}. Thus $J(U)$ is closed.

Let $q \in \p J^+(p, U)\backslash \{p\}$ and suppose that $(M,C)$ is locally Lipschitz. Since $J(U)$ is closed, there is a continuous causal curve from $p$ to $q$. If $\sigma$ is any  such curve,  no point of $\sigma\backslash q$ can belong to $I^+(p, U)$, otherwise $q\in I^+(p, U)$ by Th.\ \ref{soa}, a contradiction, thus $\sigma \subset E^+(p, U)$.

Let us prove the peripheral property under upper semi-continuity. Let $q \in E^+(p, U)\backslash \{p\}$. Consider the cone structure $F(x)=\{y\in C_x\cup\{0\}\colon \Vert y\Vert_h\le 1 \}$. Every continuous causal curve can be regarded as a solution of $\dot x(t)\in F(x(t))$ when parametrized with respect to $h$-arc length, and $F$ is compact and convex. By \cite[Th.\ 2.5]{smirnov02} (see also our Th.\ \ref{sqd}) it is possible to find a sequence of locally Lipschitz proper cone structures $(M,C_k)$ such that $C_{k+1}\le C_k\le C_g$, $C=\cap_k C_k$, hence $\mathrm{Int} J^+(p,U) \subset \mathrm{Int} J^+_k(p,U)=I^+_k(p,U)$. Suppose that we can find, passing to a subsequence if necessary, $q_k\in E^+_k(p,U)$ with $q_k\to q$, then there are   continuous $C_k$-causal curves $\sigma_k\subset E^+_k(p,U)$ connecting $p$ to $q_k$ (since $C_k\subset C_g$ by the previous argument $J_k(U)$ is closed). By the limit curve theorem, arguing as above, there is a continuous $C$-causal curve connecting $p$ to $q$, which does not have any point in  $\mathrm{Int} J^+(p,U)$ as none of $\sigma_k$ intersects it, so $\sigma \subset \p J^+(p,U)$ as desired. Suppose that we cannot find the sequence $q_k$ as above, then there is $\delta>0$ such that $B(q,\delta) \subset  J^+_k(p,U)$ for any sufficiently large $k$. For every $y\in B(q,\delta)$ by using again the limit curve theorem we get that $y\in J^+(p,U)$, thus $q\in \mathrm{Int} J^+(p,U)$, a contradiction. Finally, no two points of $p', q'\in \sigma$ can be such that $q'\in \mathrm{Int} J^+(p')$,  otherwise as $(p,p')\in J$, we  would have $q'\in \mathrm{Int} J^+(p)$, a contradiction which proves that $\sigma$ is a future lightlike geodesic.
 (this proof has some similarities with \cite[Th.\ 4.7]{smirnov02} but it is not quite the same, see the next Remark).
\end{proof}

\begin{remark} \label{ksp}
This peripheral type result should not be confused with the differential inclusion version of   Hukuhara's theorem which states   that the boundary points of the reachable set are peripherally attainable
\cite[Th.\ 7.3]{davy72} \cite[p.\ 110]{aubin84} \cite[Cor.\ 4.7]{smirnov02} \cite{kikuchi68} \cite[Th.\ 8]{kikuchi67}.

 Let $A_t(p)$ be  the set reachable in time $t$ by solutions to the differential inclusion $\dot x(t)\in F(x(t))$ where $F$ is a convex and compact multi-valued map.   That theorem states that it is possible to find $x\colon [0,T]\to U$, such that $x(t)\in \p A_t(p)$ for every $t$, if the endpoint $x(T)$ belongs to $\p A_T(p)$. However, in our framework (see the proof) this version would be of little use since $\p A_t(p) \cap \mathrm{Int}J^+(p, U)\ne \emptyset$ so the trajectory could enter $\mathrm{Int} J^+(p,U)$. Ultimately the usual differential inclusion version takes into account the parametrization which, instead, does not appear in our version.
\end{remark}


The next result is a simple consequence of Th.\ \ref{soa}.
\begin{theorem} \label{oop}
Let $(M,C)$ be a locally Lipschitz proper cone structure. A continuous causal curve connecting $p$ to $q$   is an achronal lightlike geodesic contained in $E^+(p)$ or there is a timelike curve connecting the same endpoints.
\end{theorem}

\begin{proposition} \label{ckj}
Let $(M,C)$ be a  cone structure and let $S$ be any set. Any continuous causal curve $\sigma$ contained in $E^+(S)$ is a future lightlike geodesic.
\end{proposition}

\begin{proof}
If not there would be $p,q\in E^+(S)\cap \sigma$ such that $q\in \mathrm{Int}  J^+(p)$. But  there exists $r\in S$, $(r,p)\in J$, and hence  $q\in \mathrm{Int} (J^+(r))\subset  \mathrm{Int} (J^+(S))$, a contradiction.
\end{proof}

The differentiability conditions in the next result will be improved in Th.\ \ref{sop}.

\begin{corollary} \label{chg}
Let $(M,C)$ be a locally Lipschitz proper cone structure. If $q\in E^+(S)\backslash S$ there is $p\in S$ and an achronal lightlike geodesic with endpoints $p$ and $q$ contained in $E^+(S)$.
\end{corollary}

\begin{proof}
By definition there is a continuous causal curve $\gamma$ from $S$ to $q$. Let   $p$ be starting point. From the definition there cannot be a timelike curve joining $p$ and $q$, thus from  Th.\ \ref{oop} $\gamma$ is an achronal lightlike geodesic. It cannot have any point contained in $I^+(S)$ otherwise $q\in I^+(S)$ by Th.\ \ref{soa}, thus $\gamma \subset E^+(S)$.
\end{proof}



The next theorem is somewhat similar to a result on differential inclusions by Kikuchi,  \cite{kikuchi68,davy72} \cite[Th.\ 12]{kikuchi67} but it is not quite the same due to the same reasons pointed out in Remark \ref{ksp}.

\begin{theorem} \label{aam}
Let $(M,C)$ be a closed cone structure.
Locally achronal continuous causal curves (e.g.\ lightlike geodesics) have lightlike tangents wherever they are differentiable (hence almost everywhere).
\end{theorem}

We stress that without a $C^0$ condition on the cone structure, a $C^1$ lightlike geodesic need not have tangents in $\p C_{x(t)}$ (although, by the theorem, they belong to $(\p C)_{x(t)}$), cf.\ Example \ref{mik2}, where the lightlike geodesic is $t \mapsto (0,t)$.

\begin{proof}
If not we can find a differentiability point $\dot x(t_0)\in  (\mathrm{Int} C)_{x(t_0)}$ for some $t_0$, so there is a coordinate neighborhood $U\ni x(t_0)$ and a round cone $R \subset (\mathrm{Int} C)_{x(t_0)}$, $\dot x(t_0)\in R$, such that using the identification of tangent spaces provided by the coordinates, $U\times R \subset \mathrm{Int} C$. As a consequence, we can choose the coordinates  $\{x^\alpha\}$ so that $\dot x(t_0)=\p_0$ at  $x(t_0)$ and  in such a way that its canonical  Minkowski metric  (hence flat) has  cone  $R$ at every point of the neighborhood. Since $x$ is differentiable at $t_0$, there is $a>0$ such that $x\vert_{(t_0,t_0+a)}\subset x(t_0)+R_{x_0}$ where the plus sign is understood using once again the affine structure of the coordinate neighborhood induced by the coordinate system. But every point in $x(t_0)+R_{x_0}$ is reachable from $x_0$ with an $R$-causal $C^1$ curve (a segment), which is $C$-timelike. Hence $x$ cannot be locally achronal.
%
%
\end{proof}


\begin{remark}
We did not prove, not even under a Lipschitz condition on $C$, that for sufficiently small $U\ni p$, there is just one lightlike geodesic connecting $p$ to $q\in E^+(p,U)\backslash \{p\}$. If the cones are not strictly convex it is easy to provide counterexamples to such a property, thus some stronger form of convexity on the cone is required. Moreover, we did not prove that lightlike geodesics can be made inextendible while remaining lightlike geodesics, nor that they cannot branch, namely that there cannot be two distinct lightlike geodesics sharing a segment. These properties hold in the regular theory, cf.\ Sec.\ \ref{xxo}.
\end{remark}

\subsection{Future sets and achronal boundaries}

We provide a generalization of the notion of contingent cone due to Severi and Bouligand to manifolds.

\begin{definition}
Let $M$ be a $C^{1}$ manifold and let $K\subset M$ be a subset. For every $p\in \bar K$ let us consider the affine structure induced by a local coordinate system in a neighborhood of $p$.
The {\em contingent cone} $T_K(p)$ of $K$ at $p\in \bar K$, is the set of all vectors $y\in T_pM\backslash 0$ for which we can can find $p_n\in K$, $p_n\to p$, and a sequence $\epsilon_n \to 0$ such that $y_n=\frac{p_n-p}{\epsilon_n}$ converges to $y$. The definition is independent of the affine structure used.
\end{definition}

\begin{proof} (Independence of the coordinate affine structure) Let us consider two coordinate systems $\{x^\alpha\}$ and $\{x'{}^\alpha\}$ in a neighborhood of $p$. Without loss of generality we can assume $ x^\alpha(p)=x'{}^\alpha(p)=0$.
Let $x^\alpha_n=x^\alpha(p_n)$, $x'{}^\alpha_n=x'{}^\alpha(p_n)$, the assumption is $y^\alpha_n=x^\alpha_n/\epsilon_n \to y^\alpha$, so $\vert x_n\vert \le M \epsilon_n$ for some $M>0$.
The coordinate change is $C^{1}$, thus $x'_n{}^\alpha=B^\alpha_{\beta} x^\beta_n+o(\vert x_n \vert)$ for some matrix $B$. We see that $x'_n{}^\alpha/\epsilon_n \to B^\alpha_{\beta} y^\beta$, namely $y$ would also have been counted in the contingent cone using the coordinate system $\{x'{}^\alpha\}$ in place of $\{x^\alpha\}$. By using the inverse coordinate transformation we obtain the independence of the definition.
\end{proof}

The result clarifies that in a local coordinate system this notion of contingent cone coincides with that used in the theory of differential inclusions in $\mathbb{R}^{n+1}$ and so we can import many results from this theory \cite{aubin84,aubin11}. For instance, the contingent cone is a closed cone and \cite[Lemma 11.2.2]{aubin11}
\begin{theorem}
Let $K\subset M$ and let $x\colon [0,1] \to K$ be differentiable at 0, then $\dot x(0) \in T_K(x(0))$.
\end{theorem}
\begin{proof}
In a local coordinate neighborhood of $x(0)$ we have for sufficiently small $t$,  $x^\alpha(t)=x^\alpha(0)+\dot x^\alpha(0) t+o(t)\in K$, thus $\dot x^\alpha(0)=\lim_{t\to 0} [ x^\alpha(t)-x^\alpha(0)]/t$, which implies $\dot x(0) \in T_K((x(0))$.
\end{proof}
In particular, from Theorems \ref{mmz} and  \ref{sbs}
\begin{corollary} \label{smy}
Let $(M,C)$ be a  $C^0$ closed cone structure. Let $K\subset M$ be any subset and let $p\in K$. If every  $C^1$ causal curve starting from $p$ remains in $K$ at least for a small domain interval, then $C_p \subset T_K(p)$.

Let $(M,C)$ be a proper cone structure. Let $K\subset M$ be any subset and let $p\in K$. If every  timelike curve starting from $p$ remains in $K$ at least for a small domain interval, then $(\mathrm{Int} C)_p\subset T_K(p)$.
\end{corollary}


 Since on a proper cone structure the chronological relation is open, for every subset $S$ we have  $I^+(S)=I^+(\bar S)$.

\begin{proposition} \label{por}
Let $(M,C)$ be a proper cone structure. For any subset $F\subset M$, if  $I^+(F)\subset F$ then $\p F$ is achronal.

Let $(M,C)$ be a locally Lipschitz proper cone structure.
 For any subset $F\subset M$, $I^+(F)\subset F$ implies $J^+(F) \subset \bar F$. If $F$ is open $F\subset I^+(F)$, moreover  $I^+(F)\subset F$ implies $J^+(F)\subset F$. Thus if $F$ is open or closed,  $I^+(F)\subset F$ implies $J^+(F)\subset F$.
\end{proposition}
Of course, there is a past version of this result. 

\begin{proof}
The first statement is an immediate consequence of the openness of $I$. Let us consider the second claim.
Suppose that there is a continuous causal curve which starts from $F$ and escapes $\bar F$, then there is a last point $p\in \bar F$, and we can find a continuous causal curve $x\colon [0,a)\to M$, such that $p=x(0)$, $x((0,a)) \cap \bar F=\emptyset$. But from Th.\ \ref{sbn} we can find an open neighborhood $U\ni p$ such that $J^+(p,U)\subset\overline{I^+(p,U)}\subset  \bar F$, a contradiction. If $F$ is open $F\subset I^+(F)$ because by Th.\ \ref{sbs} every point of $F$ is the ending point of a timelike curve contained in $F$. Thus $J^+(F)\subset J^+(I^+(F))\subset I^+(F)$ by Th.\ \ref{soa}.
\end{proof}

\begin{definition} Let $(M,C)$ be a proper cone structure.
A set $F$ such that $I^+(F)= F$ is called a {\em future} set. The boundary of a future set is an {\em achronal boundary}. A subset $S\subset M$ is a {\em local achronal boundary} if for every $p\in S$ we can find $U\ni p$ open  such that $S\cap U$ is an achronal boundary in $(U,C\vert_U)$.
\end{definition}
Of course, every (local) achronal boundary is (resp.\ locally) achronal.
For every $S\subset M$,
$\p I^+(S)$ is an  achronal boundary and every achronal boundary has this form for some $S$.

\begin{proposition} \label{mac}
Let $(M,C)$ be a proper cone structure and let $S$ be any set, then $\overline{\mathrm{Int} J^+(S)}= \overline{J^+(S)}$ and $\p{\mathrm{Int} J^+(S)}= \p{J^+(S)}$, and the latter set is an achronal boundary.
\end{proposition}

\begin{proof}
For every $p\in \p J^+(S)$, and $q\in I^+(p)$ we have $p\in I^{-}(q)$. Since $I^{-}(q)$ is open and contains some point in $J^+(S)$, we have $q\in J^+(S)$, and hence, by the arbitrariness of $q$,  $I^+(p) \subset J^+(S)$. Thus for every $q\in I^+(p)$, $q\in \mathrm{Int} J^+(S)$, and letting $q\to p$, $p\in \p \mathrm{Int} J^+(S)$. We just proved $\p J^+(S)\subset \p \mathrm{Int} J^+(S)$. Observe that $\mathrm{Int} J^+(S) \subset J^+(S)$ implies $\overline{\mathrm{Int} J^+(S)}\subset \overline{J^+(S)}$ and hence the converse inclusion $\p \textrm{Int} J^+(S)\subset \p J^+(S)$. Thus we have proved the identities and that $F=\mathrm{Int} J^+(S)$ is such that $\p F=\p{J^+(S)}$.
 Every point $q\in F$ is also the ending point of a timelike curve contained in $F$, thus $F\subset I^+(F)$. Moreover, $I^+(F)\subset  J^+(S)$, and as the former set is open, $I^+(F) \subset \mathrm{Int} J^+(S)=F$, thus  $I^+(F)= F$.
\end{proof}




\begin{theorem} \label{aoq}
Let $(M,C)$ be a proper cone structure.  A locally achronal boundary  is a topological hypersurface and, more precisely, a locally Lipschitz graph.
\end{theorem}

\begin{proof}
It is sufficient to give the proof for achronal boundaries.
Let $p\in A$ where $A=\p I^+(F)$, $F=I^+(F)$, is an achronal boundary. Let $\tilde C\le C$ be a $C^0$ proper cone structure contained in $C$. Since $\mathrm{Int} \tilde C_{p}$ is open there is a round cone $R_{p} \subset \mathrm{Int} \tilde C_{p}$. Thus we can find coordinates $\{ x^\alpha \}$, $ x^\alpha(p)=0$, in a neighborhood $U\ni p$ such that the cone $R_{p}$ is that of the Minkowski metric $g=-(\dd x^0)^2+\sum_i (\dd x^i)^2$, where  $\p_0 \in \mathrm{Int} \tilde C_{p}$ and $\dd x^0$ is positive on $C_{p}$. By continuity all these properties are preserved in a sufficiently small neighborhood $U\ni x$, in particular the timelike cones of the flat metric $g$ are contained in $\mathrm{Int} \tilde C$. Then the integral curve of $\p_0$ passing through $p$ is a $\tilde C$-timelike curve which belongs to $F$ just for $x^0>0$. The other integral curves of $\p_0$ intersect $F$ in a lower bounded open set since by achronality these intersections cannot enter $I^-_g(p, U)$. In other words, $\p F$ is locally the graph of a function $x^0(x^i)$. The graph is locally Lipschitz because by achronality for $q\in \p F$, $\p F$ cannot intersect $I^+_{g}(q, U)$.
\end{proof}

A trivial consequence of Cor.\ \ref{smy} is
\begin{theorem}
Let $(M,C)$ be a $C^0$ proper cone structure. If $F$ is a future set, then for every $p\in \p F$, $C_p \subset T_F(p)$.
\end{theorem}


We remark that the continuous causal curves in the next definition are not necessarily inextendible.
\begin{definition}
A continuous causal curve is {\em viable} in $K$ if it is contained in $K$. A subset $K$ is {\em viable} if for every $x_0\in K$ there is a continuous causal curve $x\colon [0,a)\to K$, with $x(0)=x_0$.
\end{definition}
Of course by a maximality argument from $x_0$  there emanates a future inextendible continuous causal curve entirely contained in $K$.
The next result is the manifold translation of Nagumo-Haddad's theorem \cite[p.\ 180]{aubin84}.

\begin{theorem} (Viability theorem)\\
Let $(M,C)$ be a closed cone structure and let $K\subset M$ be open or closed (more generally locally compact in the induced topology). Then $K$ is viable if and only if for every $x\in K$, $T_K(x)\cap C_x\ne \emptyset$.
\end{theorem}

\subsection{Imprisoned causal curves}

A cone structure is {\em causal } if it has no closed continuous causal curves.  A future inextendible continuous causal curve $\gamma$ which enters and remains in a compact set $K$ is said to be {\em imprisoned} in the compact set.
\begin{definition}
A cone structure $(M,C)$ is {\em non-imprisoning} if there is no future inextendible continuous causal curve contained in a compact set.
\end{definition}
We know that  closed cone structures are locally  non imprisoning, cf.\ Prop.\ \ref{iiu}. We want to investigate the non-local aspects related to imprisoned curves.
\begin{definition}
We say that a non-empty set $C$ is {\em biviable} if for each point $p\in C$ we can find an inextendible continuous causal curve passing through  $p$ which is contained in $C$.
\end{definition}

We recall that the future $\omega$-limit set of a future inextendible continuous causal curve $\gamma$  is $\Omega_f(\gamma)=\cap_{t \in \mathbb{R}} \overline{\gamma([t,+\infty))}$. The set $\Omega_p(\gamma)$ is defined similarly.

\begin{lemma} \label{pqg}  Let  $(M,C)$ be a closed cone structure.
For a future inextendible continuous causal curve $\gamma$ the set $\Omega_f(\gamma)$, if non-empty, is biviable, and analogously in the past case.
\end{lemma}

\begin{proof}
Let $p\in \Omega_f(\gamma)$ and let us parametrize $\gamma$ by $h$-arc length where $h$ is a complete Riemannian metric. Let us set $p_k=\gamma(2t_k)$, where the sequence $t_k\to +\infty$ is chosen so that $p_k\to p$. Applying the limit curve theorem  to $\gamma([t_k,3t_k])$ we obtain  an inextendible continuous causal curve $\sigma$ contained in $\Omega_f(\gamma)$ and passing through $p$.
\end{proof}

\begin{lemma} \label{pqf}  Let  $(M,C)$ be a closed cone structure. If the future inextendible continuous causal curve $\gamma$ is future imprisoned in a compact set then $\Omega_f(\gamma)\ne \emptyset$.
The condition with {\em past} replacing {\em future} in the definition of non-imprisoning gives the same property.
\end{lemma}

\begin{proof}
If $\gamma$ is imprisoned in $K$ then $\Omega_f(\gamma)\cap K=\cap_{t \in \mathbb{R}} [ \overline{\gamma([t,+\infty))} \cap K]$ which is non-empty by the finite intersection property.

Let $\sigma$ be a future inextendible continuous causal curve contained in a compact set $K$, then $\Omega_f(\sigma)\subset K$, and since $\Omega_f(\sigma)$ is biviable we can find inside it, and hence inside $K$, an inextendible continuous causal curve.
\end{proof}
 We are able  to generalize a result in \cite{minguzzi07f}. This new proof does not use the notion of lightlike geodesic.

\begin{theorem} \label{lof}
Let $(M,C)$ be a closed cone structure. Let $\gamma$ be a future inextendible continuous causal curve  imprisoned in a compact set $K$, then inside $\Omega_f(\gamma) \subset K$ there is a minimal biviable closed subset $B$. For every  inextendible continuous causal curve $\alpha\subset B$ we have $B=\overline\alpha=\Omega_f(\alpha)=\Omega_p(\alpha)$.
\end{theorem}

So the existence of an imprisoned continuous causal curve implies the existence of a continuous causal curve which accumulates on itself at every point.

\begin{proof}
The proof of Lemma \ref{pqf} shows that there is an inextendible continuous causal curve $\sigma$ contained in   $K$.
By the limit curve theorem $\bar \sigma$ is biviable. Let us consider the family $\mathcal{A}$ of all closed biviable subsets of $\bar \sigma$. This family is non-empty since it contains $\bar \sigma$. Let us order it through inclusion.  By
Hausdorff's maximum principle (equivalent to Zorn's lemma and the
axiom of choice) there is a maximal chain  of closed biviable sets  $\mathcal{C} \subset \mathcal{A}$. Since $M$ is second countable it is hereditary Lindel\"of, \cite[16E]{willard70} thus $\cap \mathcal{C}=\cap_k A_k$ where $\{A_k\}\subset \mathcal{C}$ is a countable subfamily. Notice that $\cap \mathcal{C}$ is non-empty
being the intersection of a nested family of non-empty compact sets (they have the finite intersection property). Every $p\in \cap \mathcal{C}$ belongs to $A_k$ so through it there passes an inextendible continuous causal curve $\eta_k$ contained in $A_k$. Since the $A_k$ are closed, by the limit curve theorem the limit curve $\eta$ passing through $p$ belongs to $A_k$ for every  $k$, and hence belongs to $\cap \mathcal{C}$. Thus $B:=\cap \mathcal{C}$ is a non-empty closed biviable set which must be minimal otherwise the chain $\mathcal{C}$ would not be maximal. If $p\in B$ through it there passes an inextendible continuous causal curve $\alpha$ contained in $B$, but since both $\Omega_f(\alpha)$ and $\Omega_p(\alpha)$ are biviable and contained in $B$, $\Omega_f(\alpha)=\Omega_p(\alpha)=B\supset \alpha$, which due to $\bar \alpha=\alpha \cup \Omega_f(\alpha) \cup \Omega_p(\alpha)$ implies $\bar\alpha =\Omega_f(\alpha)=\Omega_p(\alpha)$.
\end{proof}

\subsection{Stable causality}
In this section we investigate stable causality. The longest proofs connected to the main Theorem \ref{nio} are postponed to Sec.\ \ref{fir}-\ref{las}.


\begin{theorem} \label{ddo}
Let $C$ be a closed cone structure and let  $C'$  be a $C^0$ proper cone structure  such that $C<C'$. Then there is a locally Lipschitz proper cone structure $C''$ such that $C<C''<C'$.
\end{theorem}

For this type of result see also \cite[Th.\ 2.5]{smirnov02},
 \cite[Th.\ 1.13.1]{aubin84}.
\begin{proof} At every $p\in M$, we can find in a coordinate neighborhood of $p$ a locally Lipschitz proper cone structure $C''$ such that $C<C''<C'$. Indeed, take $C''_p$ a proper cone such that $C_p<C''_p<C'_p$ and extend $C''$ by translation using the affine structure induced by the coordinate neighborhood, so that $C''$ is locally Lipschitz. Then shrinking the neighborhood if necessary, we find using the upper semi-continuity of $C$ and the continuity of $C''$ and $C'$, $C<C''<C'$ in such neighborhood.

Let $\omega$ be a Lipschitz 1-form positive on $C'$ (Prop.\ \ref{ohg}). The result is globalized using a partition of unity $\{\varphi_i\}$ and Prop.\ \ref{doo} where $C_i$ is $C''$ for the coordinate neighborhood $U_i\supset \textrm{supp} \,\varphi_i$ and the plane $P_x\subset T_xM$ at $x\in M$ is  $\omega^{-1}(1)$. Then $C_{(P,\{\varphi_i\})}$ is a locally Lipschitz proper cone structure such that $C<C_{(P,\{\varphi_i\})}<C'$. \ \ \
%
%
\end{proof}

\begin{theorem} \label{dxp}
Let $C$ be a closed cone structure and let  $C'$ be a $C^0$ proper cone structure, with $C<C'$, then $\overline{J_{C}}\subset I_{C'}\cup \Delta$.
\end{theorem}

\begin{proof}
It is sufficient to prove ${J_{C}}\subset I_{C'}\cup \Delta$, in fact the general case follows from the limit curve theorem \ref{main} case (ii) with $J_n=J_C$ using the openness of $I_{C'}$.
Moreover, it is sufficient to prove the local version: every point has a coordinate neighborhood $U$ such that ${J_{C}(U)}\subset I_{C'}(U)\cup \Delta(U)$, since a continuous $C$-causal curve segment can be finitely covered by such neighborhoods. If $C'$ is a locally Lipschitz proper cone structure (with respect to a $C^1$-compatible smooth atlas)   the proof is as follows: let $(p,q)\in J_C(U)\backslash \Delta(U)$, the $h$-arc length parametrized continuous $C$-causal curve connecting $p$ to $q$ has almost everywhere derivative in $C$, thus in $\mathrm{Int} C'$, thus it is a continuous $C'$-causal curve. But it cannot be $C'$-achronal otherwise the derivative would stay almost everywhere in $\p \mathrm{Int} C'$ (Th.\ \ref{aam}), a contradiction, hence by Th.\ \ref{soa} $(p,q)\in I_{C'}(U)$.
If $C'$ is just continuous we can find $C''$ locally Lipschitz such that $C<C''<C'$ 
 then  arguing as above $(p,q)\in I_{C''}(U)\subset I_{C'}(U)$.
\end{proof}

\begin{definition}
Let $(M,C)$ be a closed cone structure. The Seifert relation is the intersection of the causal relations of all $C^0$ proper  cone structures having wider cones: $J_S=\cap_{C'>C} J_{C'}$.
\end{definition}
Theorem \ref{ddo} shows that in the definition of the Seifert relation the $C^0$ condition can be replaced by ``locally Lipschitz''.
Proposition \ref{ohg} shows that due to the sharpness of $C$ the family on the right-hand side is  non-empty.

The Seifert relation can be equivalently defined using the chronological relation of the enlarged cones or the closure of the causal relation. The proof is as in Lorentzian geometry and reported here for completeness \cite{minguzzi07}.

\begin{proposition} \label{paq}
Let $(M,C)$ be a closed cone structure. Then
\[
J_S=\Delta\cup \big[\cap_{C'>C} I_{C'}\big]=\cap_{C'>C} \bar J_{C'},
\]
where $C'$ is a $C^0$ proper  cone structure having wider cones, $C'>C$.
\end{proposition}

\begin{proof}
For the first equality we have only to show that  $\cap_{C'>C} J_{C'} \subset \Delta\cup \big[\cap_{C'>C} I_{C'}\big]$ the other inclusion being obvious. For every $C^0$ proper cone structure $\check C>C$, taking a $C^0$ proper cone structure $\tilde C$, such that $C<\tilde C<\check C$, it follows that $J_{\tilde C} \subset \Delta\cup I_{\check C}$ hence $\cap_{C'>C} J_{C'} \subset \Delta\cup I_{\check C}$. Since $\check C$ is arbitrary, the claim follows.

In order to prove the identity $J_S=\cap_{C'>C} \bar J_{C'}$ we need  just to prove the inclusion $\cap_{C'>C} \bar J_{C'}\subset \cap_{\check C>C} J_{\check C}$ which follows from $\cap_{C'>C} \bar J_{ C'}\subset J_{\check C}$ for $\check C>C$, which is immediate from Th.\ \ref{ddo} and Th.\ \ref{dxp}.
\end{proof}

\begin{theorem}
Let $(M,C)$ be a closed cone structure. Then $J_S$ is  closed, reflexive, transitive (a closed preorder) and contains $J$.
\end{theorem}

\begin{proof}
The property of being closed follows immediately from Prop.\ \ref{paq}, the other properties are clear.
\end{proof}

The next result is particularly useful in conjunction with the last statement of the limit  curve theorem \ref{main}.

\begin{theorem} \label{sqd}
Let $(M,C)$ be a closed cone structure. Then the family of   locally Lipschitz proper  cone structures $C'$ such that $C<C'$ is non-empty. Moreover, for every locally Lipschitz proper structure $\tilde C>C$  we can find a countable subfamily of locally Lipschitz proper cone structures $\{C_k\}$ such that $C<C_{k+1}< C_k<\tilde C$, $C=\cap_k C_k$ and $J_S=\cap_k \bar J_{k}=\Delta\cup \cap_k I_k=\cap_k J_k=\cap_k J_{k S}$.
\end{theorem}

\begin{proof}
The first statement is Prop.\ \ref{ohg}.
For every $v \notin C_p$ we can find a proper cone $C'_p>C_p$ such that $v\notin C'_p$. Using again the local affine structure induced by a coordinate neighborhood of $p$, and the upper semi-continuity of $C$, we can extend locally the inclusion, and then globalize with the same arguments  introduced in the previous proofs so as to find a locally Lipschitz proper cone structure $C'>C$ such that $v\notin C'$. This means that $C=\cap_{C'>C} C'$ where $C'$ runs over all locally Lipschitz proper cone structures.
Since $TM\backslash 0$ is second countable it is hereditary Lindel\"of \cite[16E]{willard70}, thus $[TM\backslash 0]\backslash C$ is Lindel\"of.
The intersection $\cap_{C'>C}  C'$  can be replaced by the intersection of a countable family $\{ C_i'\}$. At this point we define inductively the locally Lipschitz proper cone structure  $ C''_{i}$ is such a way that $ C< C''_{i+1}< C_{i+1}'\cap  C''_{i}$ starting from $ C''_0=TM\backslash 0$, so that  $C<C''_{k+1}< C''_k$, $C=\cap_k C''_k$.

Next we use the fact that $M\times M$ is second countable hence hereditary Lindel\"of. So $(M\times M)\backslash J_S$ is Lindel\"of, which implies that $J_S=\cap_s \bar J_{C_s'''}$ for some countable family of locally Lipschitz proper cone structures  $\{C_s'''\}$ with $C_s'''>C$. Now take inductively the locally Lipschitz proper cone structures  $C_k$ such that $ C< C_{k+1}< C_{k+1}'''\cap C_{k+1}''\cap  C_k$ starting from $C_0= TM\backslash 0$.

Finally, the $\bar J_k$ in the intersection can be replaced by $ \Delta \cup I_k$ or $J_k$ since by Th.\ \ref{dxp} for every $k$, $\bar J_{k+1}\subset  \Delta \cup I_k\subset J_k$.  Moreover, $J_k$ can be replaced by $J_{k  S}$ since $J_{k+1  S} \subset J_k$ as $C_{k+1}<C_k$.
\end{proof}

The relation $J_S$ was introduced by Seifert in 1971. It is stable as the next result proves, so the letter ``$S$'' can also be nicely read as  ``Stable'' and $J_S$ can be called the {\em stable (causal) relation} compatibly with the terminology recently adopted in \cite{bernard16}.

\begin{theorem}(stability of the Seifert relation)\\
Let $(M,C)$ be a closed cone structure. For every compact set $K\subset M$, and open neighborhood $V\supset J_S\cap K\times K$ in the topology of $K\times K$, we can find a locally Lipschitz proper cone structure $C'>C$ such that $J'_S\cap K\times K\subset V$ (thus the same holds for all cone structure narrower than $C'$).
\end{theorem}

\begin{proof}
Let $C_k$ be the sequence constructed in Th.\ \ref{sqd}. Since $J_S=\cap_k J_{kS}$ the open sets $\{K\times K \backslash J_{kS}\}$ (in the topology of $K\times K$) cover the compact set $K\times K\backslash V$, thus passing to a finite covering and taking the largest index value $i$ of the covering we get $J_{iS}\cap K\times K\subset V$.
\end{proof}

The next result is a simple consequence of the limit curve theorem \ref{main}, of the non-imprisoning property of the neighborhood constructed in Prop.\ \ref{iiu}, and of Prop.\ \ref{sqd}.

\begin{theorem} \label{xxl}
Let $(M,C)$ be a closed cone structure. Every point has a neighborhood $U$ such that $J_S(U)=J(U)$.
\end{theorem}

The neighborhood mentioned coincides with that  constructed  in Prop.\ \ref{iiu}.
\begin{proof}
Let $U$ be a neighborhood of the point in question constructed as in Prop.\ \ref{iiu} and let $C^g$ be the cone structure of the flat metric mentioned in that result.
Let $(p,q)\in J_S(U)$ then $(p,q)\in J_{k}(U)$ for every $k$, where $C_k$ is the sequence found in Th.\ \ref{sqd}, chosen so that $C_k<C^g$. Let $\sigma_k$ be a continuous $C_k$-causal curve connecting $p$ and $q$ contained in $U$. By the limit curve theorem \ref{main} and the non-imprisoning property of the neighborhood constructed in Prop.\ \ref{iiu}, there is a connecting curve $\sigma$ which is a continuous $C$-causal curve.
\end{proof}

\begin{definition}
Let $(M,C)$ be a closed cone structure. It is  {\em stably causal} if there is a $C^0$ proper cone structure $C'$ with $C<C'$ which is causal.
The {\em stable recurrent set} is the set of all those $p\in M$ such that  for  every $C^0$ proper cone structure $C'>C$ there is a closed continuous $C'$-causal curve passing through $p$.
\end{definition}

By Th.\ \ref{ddo} under stable causality $C'$ can be chosen locally Lipschitz.


We recall that $(M,C)$ is strongly causal if every point admits  arbitrarily small causally convex neighborhoods.

\begin{theorem} \label{mwh}
Any stably causal  closed cone structure $(M,C)$ is strongly causal.
\end{theorem}

\begin{proof}
If $(M,C)$ is not strongly causal at $x$ then  there is  a non-imprisoning neighborhood $U \ni x $ as in Prop.\ \ref{iiu} and a sequence of
continuous $C$-causal curves $\sigma_n$ of endpoints $x_n,z_n$, with $x_n\to
x$, $z_n \to x$, not entirely contained in $U$. Let $B$, $\bar B\subset U$ be a coordinate ball of $x$.
 Let $c_n\in \p {B}$ be the first
point at which $\sigma_n$ escapes $\bar B$. Since $\p {B}$ is compact
there is $c\in \p B$, and  a subsequence $\sigma_k$ such that
$c_k \to c$, thus $(x,c)\in \bar J$ and $(c,x)\in \bar J$.
By Th.\ \ref{dxp} for every $C^0$ proper cone structure  $C'$ with $C<C'$, we have $\overline{J_{C}}\subset I_{C'}\cup \Delta$, thus $(x,c)\in J_{C'}$ and  $(c,x)\in J_{C'}$, that is, $(M,C)$ is not stably causal.
\end{proof}

\begin{definition}
The relation $K$ is the smallest closed preorder containing $J$. We say that $K$-causality holds if $K$ is antisymmetric.
\end{definition}
 The next result is of central importance for causality theory. In the regular case it has been the focus of several investigations \cite{hawking68,seifert71,bernal04,minguzzi07b,minguzzi08b,minguzzi09c,chrusciel13}. Some of the proofs remain unaltered while others require substantial modifications. Under weak differentiability assumption the equivalence between ($i$) and ($vi$) has been previously obtained by Fathi and Siconolfi \cite{fathi12,fathi15} in the $C^0$ case, and by Bernard and Suhr \cite{bernard16} in the upper semi-continuous case (their result  contains also interesting information for the non stably causal case). Our proof of this equivalence is different and based on volume functions.



\begin{theorem} \label{nio}
Let $(M,C)$ be a closed cone structure. The following properties are equivalent:
\begin{itemize}
\item[(i)] Stable causality,
\item[(ii)] Antisymmetry of $J_S$,
\item[(iii)] Antisymmetry of $K$ ($K$-causality),
\item[(iv)]  Emptyness of the stable recurrent set,
\item[(v)]  Existence of a time function,
\item[(vi)] Existence of a smooth temporal function,
\end{itemize}
 Moreover, in this case $J_S=K=T_1=T_2$ where
 \begin{align*}
 T_1&=\{(p,q)\colon t(p)\le t(q)   \textrm{ for every time function } t \}, \\
 T_2&=\{(p,q)\colon t(p)\le t(q)  \textrm{ for every smooth temporal function } t \}.
 \end{align*}
\end{theorem}

The proof is given in Sec.\ \ref{fir}-\ref{las}.

Without the assumption of stable causality we might have $K\neq J_S$, a causal example is given in \cite[Ex.\ 5.2]{minguzzi07}.

\subsection{Reflectivity and distinction}
This section is devoted to the study of reflectivity and distinction, which taken jointly will provide the notion of causal continuity \cite{hawking74}. The next relational approach to  reflectivity was introduced in \cite{minguzzi07b}.

\begin{definition}
A closed cone structure $(M,C)$ is {\em  reflective} if  the  relations $D_f=\{(p,q)\colon q\in \overline{J^+(p)}\}$ and $D_p=\{(p,q)\colon p\in \overline{J^-(q)} \}$ coincide with $\bar J$.
\end{definition}


\begin{proposition} \label{diu}
For a proper cone structure reflectivity is equivalent to $D_f=D_p$, namely  $q\in \overline{J^+(p)} \Leftrightarrow p\in \overline{J^{-}(q)}$.
\end{proposition}

\begin{proof}
We prove the inclusion $\bar J \subset D_f$ assuming $q\in \overline{J^+(p)} \Leftarrow p\in \overline{J^{-}(q)}$ the other steps in the proof being similar or trivial. Let $(p_k,q_k)\to (p,q)$ and pick  a point $q' \in  I^+(q)$. Then for sufficiently large $k$, $q_k\in  I^{-}(q')$ so that $p_k\in J^{-}(q')$ and $p\in \overline{J^{-}(q')}$ which by the assumption gives $q'\in \overline{J^+(p)}$ and passing to the limit $q'\to q$, $q\in \overline{J^+(p)}$, that is $\bar J\subset D_f$.
\end{proof}

In Lorentzian geometry reflectivity guarantees the continuity of volume functions of the form $p \mapsto \mp \mu(I^\pm(p))$, where $\mu$ is a measure  absolutely continuous with respect to the local coordinate Lebesgue measures. This result does not seem to hold, not even under global hyperbolicity, unless one assumes local Lipschitzness of the cone structure, cf.\ the proof of Th.\ \ref{ger}. Nevertheless, the given notion of reflectivity and the related notion of causal continuity will be important and well behaved, as we shall see (Th.\ \ref{cau}).
One reason for paying attention to this concept lies in the possibility of generalizing  Clarke and Joshi's theorem \cite{clarke83}.

\begin{theorem}
Let $(M,C)$ be a proper cone structure. Let $\phi_t\colon M\to M$ be a complete $C^1$ flow, $t\in \mathbb{R}$, which preserves the cone structure, i.e.\ $(\phi_t)_*C=C$, and whose orbits are timelike curves,  $t \mapsto \phi_t(x)$. Then $(M,C)$ is reflective.
\end{theorem}

\begin{proof}
We prove $q\in \overline{J^{+}(p)} \Rightarrow p\in \overline{J^{-}(q)}$, the other direction being similar.  We can find  sequence $q_k=\phi_{\epsilon_k}(q)$ with $\epsilon_k\to 0$, so that $q\in I^{-}(q_k)$, and by the openness of $I$, $q_k\in J^+(p)$. As a consequence, defining $p_k=\phi_{-\epsilon_k}(q)$ we have by traslational invariance $p_k\in J^-(q)$, and since $p_k\to p$, $p\in \overline{J^-(q)}$.
\end{proof}

Case (b) in the next result appear to be new. It will be crucial for the inclusion of causal continuity into the causal ladder of spacetimes under weak differentiability conditions.

\begin{proposition} \label{jjw}
If the closed cone structure $(M,C)$ is  (a) proper and locally Lipschitz or (b)  reflective, then the relations $D_f$ and $D_p$ are transitive.
\end{proposition}

\begin{proof}
Let us prove the transitivity of $D_f$ under (a). The proof is identical to the Lorentzian one, cf.\ \cite{dowker00} \cite[Th.\ 3.3]{minguzzi07b}. Let $q\in \overline{J^+(p)}$ and $r\in \overline{J^+(q)}$. Let $r'\in I^+(r)$ then $r\in I^{-}(r')$ and by the openness of $I$ and Prop.\ \ref{soa}, $q \in I^{-}(r')$ and again by the openness of $I$,  $p \in I^{-}(r')$, that is $r'\in I^+(p)$ and taking the limit $r'\to r$, $r\in \overline{J^+(p)}$.

Let us prove the transitivity of $D_f$ under (b). Let $q\in \overline{J^+(p)}$ and $r\in \overline{J^+(q)}$. By  reflectivity $p\in \overline{J^-(q)}$, thus there is a sequence of continuous causal curves $\sigma_k$ with endpoints $p_k\to p$ and $r_k\to r$ passing through $q$, hence $(p,r)\in \bar J$, which using again  reflectivity gives $r\in \overline{J^+(p)}$.
\end{proof}

\begin{definition}
Let $(M,C)$ be a closed cone structure. We say that {\em distinction} holds, or $(M,C)$ is {\em distinguishing} if the next property holds true.
 Every point  admits arbitrarily small distinguishing open neighborhoods. Namely for every $p\in M$ and open set $U\ni p$, there is an open set $V$, $p\in V\subset U$, which distinguishes $p$, i.e.\ every continuous causal curve $x\colon I\to M$ passing through $p$ intersects $V$ in a connected subset of $I$.
\end{definition}
The next result under case (a) was observed in \cite{minguzzi07b} while case (b) is new and will prove important for the validity of the causal ladder.
\begin{proposition} \label{aah}
Let $(M,C)$ be a closed cone structure. The property `$D_f$ and $D_p$ are antisymmetric' implies distinction. The converse is true under any of the following assumptions: (a) $C$ is proper and locally Lipschitz; (b)  reflectivity.
\end{proposition}

\begin{proof}
Suppose that distinction is violated, then by Remark \ref{nff} we can find $p\in M$, $U\ni p$ and a sequence of continuous causal curves $\sigma_k$ escaping and reentering $U$ starting from $p$ and ending at $q_k\to p$ (or similarly in the past case, namely $\sigma_k$ start from $p_k\to p$ and end at $p$). By the limit curve theorem \ref{main} there is a continuous causal curve $\sigma^p$ (possibly closed) ending at $p$. Let $r\in \sigma^p\backslash\{p\}$, then $r$ is an accumulation point of $\sigma_k$, thus $p\in J^+(r)\subset \overline{J^+(r)}$ and $r\in \overline{J^+(p)}$, namely $(r,p)\in D_f$ and $(p,r)\in D_f$.

The converse under assumption (a).   Suppose that $(p,q)\in D_f$ and $(q,p)\in D_f$ with $p\ne q$. Let $q_k\in I^+(q)$ with $q_k\to q$. Since $I$ is open and $J\circ I\subset I$ we can find a timelike curve $\sigma_k$ starting from $q$ passing arbitrarily close to $p$ and ending at $q_k$. Thus distinction is violated at $q$.

The converse under assumption (b).  Suppose that $(p,q)\in D_f$ and $(q,p)\in D_f$ with $p\ne q$. By the limit curve theorem \ref{main} there is a  continuous causal curve $\sigma^q$ ending at $q$ (and possibly starting from $p$). Every $r\in \sigma^q\backslash \{q\}$, being an accumulation point of continuous causal curves starting from $p$, belongs to $\overline{J^+(p)}$, thus  $(q,p)\in D_f$ and $(p,r)\in D_f$. By Prop. \ref{jjw} $D_f$ is transitive, thus $(q,r)\in D_f$. Since $r$ can be chosen arbitrarily close to $q$, distinction is violated. In fact, there are continuous causal curves $\gamma_k$ with starting point $q$ and ending point $r_k\to r$. If there is a subsequence contained in the non-imprisoning neighborhood $U\ni q$ constructed in Prop.\ \ref{iiu}, by  the limit curve theorem there is a continuous causal curve connecting $q$ to $r$ and hence a closed continuous causal curve passing through $q$, violating  distinction (as it implies causality). On the other hand, if  there is a subsequence whose curves escape $U$ distinction is   again violated.
\end{proof}




\subsection{Domains of dependence and Cauchy horizons} \label{zxq}

\begin{definition}
Let $(M,C)$ be a closed cone structure.  The future {\em domain of dependence} or future {\em Cauchy development} $D^+(S)$ of a closed and achronal subset $S\subset M$, consists of those $p\in M$ such that every past inextendible continuous causal curve passing through $p$ intersect $S$. The future {\em Cauchy horizon} is
\[
H^+(S)=\overline{D^+(S)}\backslash I^{-}(D^+(S)).
\]
The {\em domain of dependence} or  {\em Cauchy development} $D(S)$ of a closed and achronal subset $S\subset M$, consists of those $p\in M$ such that every inextendible continuous causal curve passing through $p$ intersect $S$. (Clearly, $D(S)=D^+(S)\cup D^-(S)$.)  The  {\em Cauchy horizon} is $H(S):= H^+(S)\cup H^-(S) $.
\end{definition}
%

Chru\'sciel and Grant asked to clarify the phenomena of causal bubbling in connection with the Cauchy problem \cite{chrusciel12}. The next result goes in this direction by showing that the behavior they observed in a specific example is general.

\begin{theorem} (Cauchy horizons are generated by lightlike geodesics) \\ 
\label{juf}
Let $(M,C)$ be a proper cone structure and let $S$ be a closed and achronal subset. The set $H^+(S)\backslash S$ is  an achronal locally Lipschitz topological hypersurface and every $p\in H^+(S)$ is the future endpoint of a  lightlike geodesic contained in $H^+(S)$, either past inextendible or starting from some point in $S$. If $S$ is acausal no two points $p,q\in H^+(S)\backslash S$ can be such that $(p,q) \in \mathring{J}$, thus every continuous causal curve contained in $H^+(S)\backslash S$ is a  lightlike geodesic.

Moreover, if $(M,C)$  is locally Lipschitz then every past   inextendible continuous causal curve with future endpoint in $H^+(S)$ is such that the  maximal connected segment which does not intersect $S$ is contained in $H^+(S)$ (hence the segment is a lightlike geodesic).
\end{theorem}


Notice that we do not prove that the  generators cannot intersect other generators in their interior.

\begin{proof}
Let $p\in H^+(S)$, and let $(p_k,p)\in I$, $p_k \to p$, then $p_k\in  I^{-}(\overline{D^+(S)})$, thus $p\in  \p  I^{-}(\overline{D^+(S)})$, and  $H^+(S)$  being a subset of an achronal boundary $\p I^{-}(\overline{D^+(S)})$ it is itself achronal. If $p\in H^+(S)\backslash S$ then any past inextendible $C$-timelike curve reaching $p$ intersects $S$, thus $p\in I^+(S)$.  We have proved the inclusion $H^+(S)\backslash S \subset  I^+(S)\cap \p  I^{-}(\overline{D^+(S)})$


We want to prove the equality $H^+(S)\backslash S=  I^+(S)\cap \p  I^{-}(\overline{D^+(S)})$, so that $H^+(S)\backslash S$ is an open subset of an achronal boundary and hence a locally Lipschitz topological hypersurface itself. Let $p \in I^+(S)\cap   \p  I^{-}(\overline{D^+(S)})$, then there is a sequence of $C$-timelike curves $\sigma_k$ of starting points $p_k\to p$ and ending points $q_k\in D^+(S)$ (recall that $I$ is open).
For sufficiently large $k$ the $C$-timelike curve $\sigma_k$ cannot intersect $S$, otherwise, since $p_k\in  I^+(S)$, we would have a violation of the achronality of $S$.
Thus, we cannot have $p_k\in M\backslash D^+(S)$ otherwise  it would follow that $q_k\in  M\backslash D^+(S)$, a contradiction. Thus $p_k\in  D^+(S)$ and $p\in \overline{D^+(S)}$. But by assumption $p\notin  I^{-}(\overline{D^+(S)})$, thus $p  \in H^+(S)$ as we wished to prove.


If $q\in H^+(S)$ and $q'\gg q$, then we cannot have $q'\in H^+(S)$, as we would have $q\in I^{-}(D^+(S))$, a contradiction. Similarly, if $q\in H^+(S)\backslash S$ and  $r\in I^-(q,U)$ where $U$ is a sufficiently small neighborhood of $q$, $U\cap S=\emptyset$, then $r\in  I^-(q',U)$ for some $q'\in D^+(S)$ and we cannot have $r\in M\backslash  {D^+(S)}$, otherwise $q'\notin   D^+(S)$. Thus, for every $q \in H^+(S)\backslash S$ we have $I^-(q,U) \subset D^+(S)$.

Let $q\in H^+(S)\backslash S$ and let $q_k \to q$, with $q_k\gg q$. Since $q_k \notin D^+(S)$ it is the future endpoint of a past inextendible continuous causal curve $\sigma_k$, cf.\ Th.\ \ref{zar}, not intersecting $S$.  Necessarily it does not intersect $D^+(S)$ otherwise it would be forced to intersect $S$. By the limit curve theorem $q$ is the future endpoint of a past inextendible continuous causal curve $\sigma$ which does not intersect the open set $I^+(S)\cap I^{-}(D^+(S))\subset D^+(S)$ as none of the $\sigma_k$ does.

It remains to show that $\sigma \cap I^+(S)\subset H^+(S)$.
First we notice the equality $I^+(S)\cap H^+(S)=H^+(S)\backslash S$ which is due to the achronality of $S$. If the cone structure is locally Lipschitz by Th.\ \ref{sbn} $J^-(q,U)=\overline{I^-(q,U)} \subset \overline{D^+(S)}$, which proves the inclusion $\sigma \cap I^+(S)\subset H^+(S)$  in a neighborhood of $q$. At this point, since  $H^+(S)$ is closed, we can  extend the inclusion all over $\sigma \cap I^+(S)$ through a standard maximization argument. However, this result proves much more, namely that every continuous causal curve ending at $q$ and not intersecting $I^+(S)\cap I^{-}(D^+(S))$ is contained in $H^+(S)\backslash S$ as long as it does not intersect $S$. This is the last statement of the theorem.

In the upper semi-continuous case we wish to prove that there is at least one continuous causal curve with this property.
We proceed as follows. We introduce a complete Riemannian metric $h$, and  a timelike Lipschitz vector field $W$ (with respect to a $C^1$ compatible smooth atlas) so transverse to $H_0=H^+(S)\backslash S$, and its flow $\varphi_\tau$. The vector field is normalized so that $H_\tau:=\varphi_\tau(H_0)$ is well defined for $\vert \tau\vert\le 1$ and does not intersect $S$. A locally finite covering of $H_0$ with the neighborhoods used in the proof of Theorem \ref{aoq} might be used  to show that a vector field with these properties exists.

Having chosen $q\in H_0$ we are going to build a sequence of  continuous causal zig-zag curves ending at $q$. The ``zig'' is a segment of past inextendible continuous causal curve not intersecting $I^+(S)\cap I^{-}(D^+(S))$ while the zag is an integral curve of $W$ so that $\Delta \tau<1/N$ where  the total number of zig-zags  is unbounded unless one zig remains in the region $\tau< 1/N$ as long as it stays in $\varphi_{[0,1/N]} (H_0)$. The starting point of the zag is in $H^+(S)\backslash S$.  If the last piece is a zig  the curve is past inextendible.  For each $N$ we have a continuous causal curve ending at $q$ not intersecting $I^+(S)\cap I^{-}(D^+(S))$, and the sequence does not contract to a point so there must be a limit continuous causal curve $\sigma$ which does not intersect the same set. By construction its intersection with $I^+(S)$ is contained in $H^+(S)\backslash S$ and can only be past inextendible or converge to a point in $S$. Notice, that no two points $p,q\in H^+(S)\backslash S$ can be such that $(p,q)\in \mathring{J}$ since there would be $p'\in M\backslash D^+(M)\cap J^-(q')$ with $q'\in D^+(S)$, so it would be possible to find a past inextendible causal curve ending at $q'$ not intersecting $S$ (if $p'$ is chosen sufficiently close to $p$, $p'\in I^+(S)$, so the causal curve connecting $p'$ to $q'$ cannot intersect $S$ by acausality of $S$), a contradiction. Thus $\sigma$ is a lightlike geodesic.
\end{proof}
\begin{theorem} \label{jdp}
Let $(M,C)$ be a proper cone structure. Let $S$ be an acausal topological hypersurface. Then $S$ is locally Lipschitz and every point $p\in S$ admits an open neighborhood $U\ni p$ which is the disjoint union of $S\cap U$ and the open sets $J^+(S\cap U,U)\backslash S$, $J^-(S\cap U,U)\backslash S$. Moreover,  $J^+(S\cap U,U)\subset D^+(S)$.
Finally, the generators of $H^+(S)$ do not reach $S$.
\end{theorem}

\begin{proof}
By a proof analogous to that of Th.\ \ref{aoq}, $S$ is a local Lipschitz graph (just work with the continuous cone structure $\tilde C$ contained in $C$) and for every $p\in S$ and sufficiently small neighborhood $U\ni p$, $U$  is the disjoint union of $S\cap U$, $\tilde I^+(S\cap U,U)$, $\tilde I^-(S\cap U,U)$. By acausality they coinicide with  $S\cap U$, $J^+(S\cap U,U)\backslash S$, $J^-(S\cap U,U)\backslash S$.

Suppose that there is no open neighborhood $U$ such that $J^+(S\cap U,U)\subset D^+(S)$, then we can find a sequence $p_k\to p$, $p_k\in  J^+(S\cap U,U)$, consisting of endpoints of past inextendible continuous causal curves $\sigma_k$ not intersecting $S$ (hence $p_k\ne p$). By the limit curve theorem there is a past inextendible continuous causal curve $\sigma$ ending at $p$. But $p\in S$ and by acausality $\sigma\backslash\{p\} \subset J^{-}(S\cap U, U)\backslash S$. So $\sigma$ has to cross the local Lipschitz graph of $S$, and so must do $\sigma_k$ for sufficiently large $k$, a contradiction.
\end{proof}

\begin{theorem} \label{jdq}
Let $(M,C)$ be a proper cone structure. Let $S$ be an acausal topological hypersurface. Then $D^+(S)\backslash S$ is open and $S\cap H^+(S)=\emptyset$.
\end{theorem}

\begin{proof}
Suppose $D^+(S)\backslash S$ is not open then there is $q\in D^+(S)\backslash S$,  and  a sequence of past inextendible continuous causal curves $\sigma_k$ not intersecting $S$ of endpoints $q_k\to q$. Thus the limit past inextendible  continuous causal curve $\sigma$ ending at $q$ intersects $S$ at a point $p$ crossing the locally Lipschitz graph of $S$, and so must  $\sigma_k$ for sufficiently large $k$, a contradiction. The last equality follows from Th.\ \ref{jdp}.
\end{proof}

\begin{theorem} Let $(M,C)$ be a proper cone structure.
Let $S$ be closed and achronal, then we have the identity $\overline{D^+(S)}=\mathrm{Int} D^+(S) \cup H^+(S)\cup S$.
\end{theorem}

\begin{proof}
Clearly, $S$ belongs to $D^+(S)$ and by the achronality of $S$ no point of $I^-(S)$ belongs to $D^+(S)$, thus $S\subset \p  D^+(S) $. For every $q\in H^+(S)\backslash S$, $I^+(q)\cap D^+(S)=\emptyset$ otherwise $q\notin   H^+(S)$, a contradiction, thus $H^+(S)\backslash S \subset  \p  D^+(S)$. Letting $q\in \p  D^+(S)\backslash S$, every timelike curve $\sigma$ ending at $q$ must enter  $D^+(S)$ immediately, moving from $q$ in the past direction. Indeed, if not we can find $r\in \sigma \cap M\backslash D^+(S)$ and hence $q\in I^+(r)\subset M\backslash D^+(S)$, a contradiction with  $q\in \p  D^+(S)$. As a consequence, $\p  D^+(S)\backslash S=\p  D^+(S)\cap I^+(S)$.
Let $q\in \p  D^+(S)\backslash S$ then it cannot hold that $q\in I^-(D^+(S))$ otherwise there would be $q' \in [M\backslash D^+(S)]\cap I^-(D^+(S))\cap I^+(S)$, a contradiction. Thus $q \notin  I^-(D^+(S))=I^-(\overline{D^+(S)})$ and hence $q\in H^+(S)\backslash S$.
\end{proof}

\begin{theorem} \label{xiq}
Let $(M,C)$ be a proper cone structure.
Let $S$ be a closed and achronal set and let  $\tilde D^+(S)$ be the closed set of all points $p$ for which all timelike curves ending at $p$ intersect $S$. Then  $\tilde D^+(S)$ is closed, $\overline{D^+(S)} \subset  \tilde D^+(S)$, and if $(M,C)$ is also locally Lipschitz, then $\overline{D^+(S)}=\tilde D^+(S)$.
\end{theorem}

\begin{proof}
By the openness of $I$ the set $\tilde D^+(S)$ is closed. Clearly,  $D^+(S) \subset  \tilde D^+(S)$, thus $\overline{D^+(S)} \subset  \tilde D^+(S)$.

For the other direction,
suppose there were a point $p \in\tilde{D}^{+}(S)$ which had a
neighborhood $V$ which did not intersect $D^+(S)$. Choose a point $x
\in I^{-}(p,V)$. Since $x \notin D^{+}(S)$ we have  $p
\notin \tilde{D}^{+}(S)$, in fact using $I\circ J\subset I$ it is possible  to construct a past inextendible timelike curve ending at $p$ (avoidance Lemma \cite[p.416, lemma 30]{oneill83} \cite[Prop. 6.5.1]{hawking73}), a contradiction. Thus
$\tilde{D}^{+}(S)=\overline{D^{+}(S)}$.
\end{proof}

\subsection{Global hyperbolicity and its stability}
In this section we introduce the important notion of global hyperbolicity, which is the strongest among the causality conditions. We start with the following weaker notion \cite{bernal06b,minguzzi06c}.
\begin{definition}
A closed cone structure $(M,C)$ is {\em causally simple}  if $J$ is closed and antisymmetric.
\end{definition}

\begin{lemma} \label{mvh}
Let $(M,C)$ be a closed cone structure. If it is causally simple then it is strongly causal.
\end{lemma}

\begin{proof}
If $(M,C)$ is not strongly causal at $x$ then  there is  a non-imprisoning neighborhood $U \ni x $ as in Prop.\ \ref{iiu} and a sequence of
continuous $C$-causal curves $\sigma_n$ of endpoints $x_n,z_n$, with $x_n\to
x$, $z_n \to x$, not entirely contained in $U$. Let $B$, $\bar B\subset U$ be a coordinate ball of $x$.
 Let $c_n\in \p {B}$ be the first
point at which $\sigma_n$ escapes $\bar B$. Since $\p {B}$ is compact
there is $c\in \p B$, and  a subsequence $\sigma_k$ such that
$c_k \to c$. By the limit curve theorem and the non-imprisoning property of $U$, we have $(x,c)\in J$, while  $(c,x)\in \bar J= J$. Thus there is a closed continuous causal curve passing through $x$, a contradiction.
\end{proof}

For a proper cone structure we have the next equivalence.
\begin{proposition} \label{soy}
Let $(M,C)$ be a proper cone structure.  The property  $J=\bar J$ is equivalent to: for every $p\in M$, $J^+(p)$ and $J^{-}(p)$ are closed.
\end{proposition}

In the language of topological preordered spaces \cite{nachbin65} the proposition says that under the said assumption the topological  preordered space  $(M,\mathscr{T},J)$ is  $T_2$-preordered if and only if it is $T_1$-preordered. The next proof coincides with the usual one given in Lorentzian geometry.

\begin{proof} The direction which assumes $J$ closed is obvious.
Let  $(p,q)\in \bar J$, so that there are $(p_k,q_k) \to (p,q)$, $(p_k,q_k)\in J$.
Let $p'\in I^{-}(p, U)$ where $U\ni p$ is an open neighborhood. For sufficiently large $k$, $p_k \in I^{+}(p') \subset J^+( p')$, and $q_k\in {J^+(p')}$. Thus $q\in \overline{J^+(p')}=J^+(p')$, that is $p'\in J^-(q)$ and letting $p' \to p$, we get $p\in \overline{J^-(q)}= J^-(q)$ as we wished to prove.
\end{proof}



The next result which will mostly interest us for $R=J$ is \cite[Prop.\ 1.4]{nachbin65}. The short proof is included for completeness.
\begin{theorem} \label{xix}
Let $R\subset M\times M$ be a closed relation, then for every compact set $K$, $R^+(K)$ and $R^{-}(K)$ are closed.
\end{theorem}

\begin{proof}
Since $R$ is closed if $(p,q)\notin R$, there are open sets $U\ni p$, $V\ni q$, such that $(U\times V)\cap R=\emptyset$. Let $q\notin R^+(K)$, then for every $p\in K$ we can find $U_p\ni p$ and $V_p\ni q$, such that $(U_p\times V_p)\cap R=\emptyset$. Let $\{U_{p_i}\}$ be a finite covering for $K$ and $V=\cap_i V_{p_i}$, then not point in $V$ intersects $R^+(K)$, thus $R^+(K)$ is closed.
\end{proof}

In  \cite[Sec.\ 3]{minguzzi12d} we argued that the notions of causal simplicity and global hyperbolicity might be regarded as pertaining to the more abstract framework of  topological preordered spaces. In this theory a causally simple cone structure would be a {\em causally simple topological ordered space}, namely a topological preordered space $(M,\mathscr{T}, J)$ in which  the preorder $J$  is closed and antisymmetric  ($T_2$-ordered space). A {\em globally hyperbolic topological preordered space} would be just a $T_2$-ordered space for
 which additionally the causally convex hull of compact sets is compact. According to the results of \cite[Sec.\ 3]{minguzzi12d} this structure  has several interesting properties among them that of quasi-uniformizability.


This definition of global hyperbolicity, point ($\gamma$) below,  will be indeed that used in this work for closed cone structures but, in general, demanding directly causal simplicity does not seem to be the most useful way of introducing the concept. So we shall consider different characterizations.

\begin{definition}
A  {\em causal diamond} is a set of the form $J^+(p)\cap J^{-}(q)$ for $p,q\in M$. A {\em causal emerald} is a set of the form $J^+(K_1)\cap J^{-}(K_2)$, where $K_1$ and $K_2$ are compact subsets. A {\em Seifert diamond} is a set of the form $J_S^+(p)\cap J_S^{-}(q)$ for $p,q\in M$.
\end{definition}
The first definition is imported from mathematical relativity, while the second and third are new. We found the terminology  appropriate given the typical cuts of emeralds.

\begin{definition} \label{sot}
A closed cone structure $(M,C)$ is {\em globally hyperbolic} if the following equivalent conditions hold
\begin{itemize}
\item[($\alpha$)] Non-imprisonment and for every bounded set $B$ its causally convex hull $J^{-}(B)\cap J^+(B)$ is bounded.
\item[($\beta$)] Causality and causal emeralds are compact.
\item[($\gamma$)]  Causal simplicity and for every compact set $K$ its causally convex hull $J^{-}(K)\cap J^+(K)$ is compact.
\item[($\delta$)] Stable causality and the Seifert diamonds are compact.
\end{itemize}
\end{definition}

The definition ($\beta$) is that given recently in \cite{bernard16}. In what follows we shall prove the equivalences. We start with the
 next result which improves \cite[Lemma 38]{bernard16}.
\begin{theorem} \label{nnq}
Let $(M,C)$ be a closed cone structure. If causal emeralds are compact then $J$ is closed. Moreover, ($\beta$) and ($\gamma$) in Def. \ref{sot} coincide.
\end{theorem}

\begin{proof}
Preliminarly, let us prove that $J^+(x)$ is closed for every $x$ \cite[Lemma 38]{bernard16}. Let  $y\in \overline{J^+(x)}$, so there is a sequence $y_n\ge x$ such that $y_n\to y$.  The sets $K_1=\{x\}$ and  $K_2=\{y, y_1,y_2,\cdots\}$ are compact, thus $B=J^+(x)\cap J^{-}(K_2)$ is compact and hence closed. Since $y_n\in B$ we conclude $y\in B$, hence $y\in J^+(x)$. Similarly,  $J^-(x)$ is closed for every $x$.

Let $(p,q)\in \bar J$, we have to show that $(p,q)\in J$. If $p=q$ there is nothing to prove, so let $p\ne q$. There are sequences $p_k\to p$, $q_k\to q$, such that $(p_k,q_k)\in J$.  Let $K_p$ and $K_q$ be compact neighborhoods of $p$ and $q$ respectively. The set $J^+(K_p)\cap J^{-}(K_q)$ is compact and non-empty as it contains  $p_k$ and $q_k$ for sufficiently large $k$, moreover $A=\cap_{K_p,K_q} J^+(K_p)\cap J^{-}(K_q)$ is compact and non-empty, since the intersected family of compact sets satisfies the finite intersection property. Let $r\in A$, then for every $K_p$, $J^{-}(r)\cap K_p\ne \emptyset$, thus $p\in J^{-}(r)$ since $J^-(r)$ is closed. Similarly, $q\in J^+(r)$, and hence $(p,q)\in J$. As for the last statement, the direction ($\beta$) $\Rightarrow$ ($\gamma$) follows from the just proved result. For  the converse, it is well known and pretty easy to prove that for a closed relation, given a compact set $K$, $J^+(K)$ and $J^-(K)$ are closed (Th.\ \ref{xix}), thus $J^+(K_1)\cap J^{-}(K_2)$ is a closed subset of the compact set $J^+(K_1\cup K_2)\cap J^{-}(K_1\cup K_2)$, thus compact.
\end{proof}

%

Let us  prove the stability of global hyperbolicity \cite{minguzzi11e,fathi12,chrusciel12,samann16,bernard16}. With it we shall also end the proof of the equivalence of ($\alpha$), ($\beta$),  ($\gamma$) and ($\delta$).
The next result is quite general and its proof is short and particularly simple. In fact we get also the identity $J_S=J$ in globally hyperbolic cone structures, a result which in previous approaches required separate treatment. The next theorem is important in order to construct smooth time functions,  indeed by opening  slightly the cones one gets the `room' needed by the smoothing procedures based on convolution.


For shortness in this theorem and in its proof we write `globally hyperbolic' in place of `globally hyperbolic in the sense of ($\alpha$) in Def.\ \ref{sot}'.

\begin{theorem}  \label{mom} (Stability of global hyperbolicity)  \\
Let $(M,C)$ be a globally hyperbolic closed cone structure. Then $J_S=J$ and there is a globally hyperbolic locally Lipschitz proper cone structure  $(M,C')$, with  $C'>C$. Moreover, ($\alpha$), ($\beta$), ($\gamma$) and ($\delta$)  in Def.\ \ref{sot} coincide.
\end{theorem}

 Notice that a globally hyperbolic closed cone structure is causally simple due to characterization ($\gamma$) and strongly causal due to Lemma \ref{mvh}. Its stability implies that it is also stably causal.

\begin{proof}
Let $o\in M$ and let $h$ be a complete Riemannian metric. We can find compact sets $\{K_n\}$ such that $J^+(K_n)\cap J^{-}(K_n)  \subset \mathrm{Int} K_{n+1}$, $\cup_n K_n=M$, and $K_n$ contains the $h$-ball of radius $n$ centered at $o\in M$.
For every $n$ it is possible to find a closed cone structure $C_n\ge C$, which on $K_{n+2}$ satisfies $C_n>C$, on $M\backslash \mathrm{Int} K_{n+3}$ satisfies $C_n=C$, is a locally Lipschitz proper cone structure on $ K_{n+2}$, and is such that $J^+_{C_n}( K_n)\cap J^-_{C_n}( K_n)\subset \mathrm{Int} K_{n+1}$. Indeed suppose not, taking a sequence of closed cone structures $\{\tilde C_k\}$, $C\le \tilde C_{k+1} \le \tilde C_k$, $\cap_k \tilde C_k =C$, which on $ K_{n+2}$ satisfy $\tilde C_k>C$, on  $M\backslash \mathrm{Int} K_{n+3}$ satisfy $\tilde C_k=C$ and are  locally Lipschitz proper cone structures on $ K_{n+2}$,   we would have for every $k$ a continuous $\tilde C_k$-causal curve which starts from $K_n$ intersects $\p K_{n+1}$ and returns to $K_n$. The curve cannot enter $M\backslash K_{n+4} $ since it would be $C$-causal on $M\backslash \mathrm{Int} K_{n+3}$, thus contradicting the inclusion $J^+(K_{n+3})\cap J^{-}(K_{n+3})  \subset \mathrm{Int} K_{n+4}$.   By applying the limit curve theorem \ref{main} we would get a continuous $C$-causal curve joining two points in $K_n$ and passing through some point of $\p K_{n+1}$ (case (ii) of theorem \ref{main} does not apply since the inextendible limit curves would be imprisoned in $K_{n+4}$ contradicting the non-imprisonment property contained in the definition of global hyperbolicity), a contradiction with $J^+(K_n)\cap J^{-}(K_n)  \subset \mathrm{Int} K_{n+1}$.

Now let $C'>C$ be a locally Lipschitz proper cone structure such that for every $n$, $C'<C_n, C_{n-1}$ on $K_{n+1}\backslash \mathrm{Int} K_n$. Notice that $C_n>C$ on $K_{n+2}$ and $C_{n-1}>C$ on $K_{n+1}$, so $C'$ does exist. Let us consider a continuous $C'$-causal curve $\sigma$ which starts and ends in $K_n$ but is not entirely contained in $K_n$. Let $m$ be the maximum number such that $\sigma$ intersects $K_{m+2}\backslash \mathrm{Int} K_{m+1}$, thus $\sigma \subset K_{m+2}$. Evidently $m\ge n-1$ since it intersects $\p K_n$. Since $C'<C_m$ on $K_{m+2}\backslash \mathrm{Int} K_{m+1}$ and $C'<C_m$ on $K_{m+1}\backslash \mathrm{Int} K_m$ we have that a segment of $\sigma$ is a continuous $C_m$-causal curve which, unless $m=n-1$, starts and returns to $K_m$ intersecting $\p K_{m+1}$, a contradiction. Thus $J^-_{C'}(K_n)\cap J^+_{C'}(K_n) \subset K_{n+1}$. From this boundedness result, Th.\ \ref{main} and \ref{sqd} it is immediate that $J_S=J$, so $J_S$ is antisymmetric (non-imprisonment implies causality) which implies stable causality. As a consequence, $C'$ in the previous step can be chosen stably causal, hence non-imprisoning (it follows from Th.\ \ref{dxp} and \ref{lof}, or Th.\ \ref{mwh}). Since every bounded set
is contained in $K_n$ for some $n$, the $C'$-causally convex hull of every bounded set is bounded which proves that $(M,C')$ is globally hyperbolic. As for the last statement, $J_S=J$ implies that $J$ is closed. We already know that  ($\beta$) and ($\gamma$) coincide. Since under $(\alpha)$ the causal relation $J$ is closed and since non-imprisonment implies causality, we have causal simplicity. Moreover,
$J^+(K)$ and $J^{-}(K)$ are closed due to Th.\ \ref{xix}, thus the convex hull $J^+(K)\cap J^{-}(K)$ is  closed and bounded, hence compact. We conclude that $(\alpha) \Rightarrow (\gamma)$. As for $(\gamma) \Rightarrow (\alpha)$, causal simplicity implies strong causality (Lemma \ref{mwh}) which implies non-imprisonment. Moreover, if $B$ is bounded $\bar B$ is compact, thus as $J^+(B) \cap J^-(B)\subset J^+(\bar B) \cap J^-(\bar B)$ and the latter is compact, we have that $J^+(B) \cap J^-(B)$ is bounded. Let us prove $(\alpha) \Rightarrow (\delta)$. The first part of this proof proves that $(\alpha)$ implies stable causality and  $J=\bar J$ implies $J=K=J_S$ (Th.\ \ref{nio}), thus the Seifert diamonds coincide with causal diamonds which have been already shown to be compact by the equivalence between $(\alpha)$ and $(\beta)$. Finally, for $(\delta) \Rightarrow (\alpha)$, stable causality implies strong causality which implies non-imprisonment. Suppose by contradiction that there is a bounded set $B$ such that $J^+(B)\cap J^-(B)$ is not compact, then we can find $(p_n,q_n)\in B\times B$ such that $(p_n,q_n)\to (p,q)\in \bar B\times \bar B$ and continuous causal curves $\sigma_n$ such that $\sigma_n$ intersects  $\p B(o,n)$  where $ B(o,n)$ is the ball of radius $n$ centered at some chosen point $o\in M$. Let $p'<p$ and $q<q'$, then by Th.\ \ref{dxp} for $C'>C$, $(p',p)\in I_{C'}$, $(q,q')\in I_{C'}$. For any given $k$ we have for sufficiently large $n$ that $(p', p_n)\in I_{C'}$, $(q_n,q')\in I_{C'}$ and $k\le n$, thus for every $k$ and $C'>C$ there is a continuous $C'$-causal curve connecting $p'$ to $q'$ and intersecting $\p B(o,k)$. But the family of non-empty compact sets $\{\overline{J_{C'}^+(p')}\cap \overline{J_{C'}^-(q')}\cap \p B(o,k)\}_{C'}$ satisfies the finite intersection property, thus $\emptyset\ne \cap_{C'>C} \{ \overline{J_{C'}^+(p')}\cap \overline{ J_{C'}^-(q')}\cap \p B(o,k)\}\subset J^+_S(p')\cap J^{-}_S(q')\cap \p B(o,k)$, where we used Prop.\ \ref{paq}, so the arbitrariness of $k$ implies that a Seifert diamond is not compact, a contradiction.
\end{proof}%

In \cite{minguzzi08b} we introduced the transverse ladder; a useful structure which might be used to clarify the central position of stable causality. Remarkably, it holds true under much weaker differentiability assumptions.

\begin{theorem} \label{tra} (Transverse ladder)\\
Let $(M,C)$ be a closed cone structure. Compactness of causal emeralds
$\Rightarrow$  The causal relation is closed $\Rightarrow$  Reflectivity  $\Rightarrow$
Transitivity of $\bar J$.
\end{theorem}

\begin{proof}
The first implication is Th.\ \ref{nnq}.
 The causal relation is closed $\Rightarrow$  Reflectivity. It is clear that $D_p=D_f=J=\bar J$.
 Reflectivity  $\Rightarrow$
Transitivity of $\bar J$. Under reflectivity $D_p=D_f=\bar J$, but under reflectivity $D_f$ and $D_p$ are transitive by Prop.\ \ref{jjw}, thus $\bar J$ is transitive.
\end{proof}

For a proper cone structure  the first implication can be improved as follows.
We recall that a causal diamond is a set of the form $J^+(p)\cap J^{-}(q)$.

\begin{lemma} \label{kog}
Let $(M,C)$ be a proper cone structure.
Compactness of the causal diamonds
$\Rightarrow$   the causal relation is closed.
\end{lemma}

\begin{proof}
 By Prop.\ \ref{soy} we have to prove that $J^{-}(q)$ is closed for every $q$, the proof that $J^+(q)$ is closed being similar.
 Let $p_k\in J^{-}(q)$  be such that $p_k\to p$, and pick $p' \in I^{-}(p, U)$ where $U\ni p$ is an open neighborhood. Since $I$ is open we have  $p_k\in J^+(p')$ for sufficiently large $k$, moreover $J^+(p')\cap J^{-}(q)$ being  compact is closed, thus $p\in J^+(p')\cap J^{-}(q) \subset J^{-}(q)$.
\end{proof}

\begin{proposition} \label{sos}
A proper cone structure $(M,C)$ is (a) non-imprisoning with bounded causal diamonds  if and only if  it is  (b) causal with compact causal diamonds.
\end{proposition}

\begin{proof}
(a) $\Rightarrow$ (b). Let $(p,q)\in J$ and let $r_k\in J^+(p)\cap J^{-}(q)$, $r_k \to r\in M$.
The sequence of continuous causal curves connecting $p$ to $r_k$ and $r_k$ to $q$ must converge to a continuous causal curve connecting $p$ to $q$ passing through $r$, otherwise by the limit curve theorem \ref{main} there would be a future inextendible continuous causal curve starting from $p$ future imprisoned in $\overline{J^+(p)\cap J^{-}(q)}$, a contradiction.

(b) $\Rightarrow$ (a). We know from Th.\ \ref{kog} that the causal relation is closed. Suppose that there is a future imprisoned continuous causal curve, then by Th.\ \ref{lof} there exists a future imprisoned continuous causal curve $\alpha$ such that $\alpha\subset\Omega_f(\alpha)$. Pick a point $p\in \alpha$, and a point $q\in \alpha\backslash \{p\}\cap J^+(p)$ then $p\in \overline{J^+(q)}\cap J^{-}(q)=J^+(q)\cap J^{-}(q)$ so causality is violated, a contradiction.
\end{proof}


The next result is \cite[Prop.\ 1]{bernard16} and shows that in the proper case the definition of global hyperbolicity can be expressed with the compactness of causal diamonds. We include the proof for completeness.

\begin{proposition} \label{xxp}
Let $(M,C)$ be a proper cone structure such that causal diamonds are compact. Then causal emeralds are compact.
\end{proposition}

This result does not hold for closed cone structures, consider $M=\mathbb{R}^2\backslash (0,0)$ where the cone $C$ is generated by the vector field $\p_y$.

\begin{proof}
By theorem \ref{kog} $J$ is closed, and it is well known and pretty easy to prove that for a closed relation, given a compact set $K$, $J^+(K)$ and $J^-(K)$ are closed (Th.\ \ref{xix}). So it will be sufficient to prove that $J^+(K)\cap J^-(K)$ is bounded for every compact set $K$, since the fact that $J^+(K_1)$ and $J^-(K_2)$ are closed and the boundedness of $J^+(K_1\cup K_2)\cap J^-(K_1\cup K_2)$ are enough to prove the claim.
Let $\tilde K$ be a compact set such that $K\subset \mathrm{Int} \tilde K$. Using the (*) property for each $p\in K$ we can find $q,r\in \tilde K$ such that $p\in I^+(q)\cap I^-(r)$, so we can find a finite covering of $K$, given by sets of the form $I^+(q_i)\cap I^-(r_i)$, $i=1,\cdots,s$. Then $J^+(K)\cap J^-(K)\subset \cup_{i,j} J^+(q_i)\cap J^-(r_j)$, which being the union of compact sets is compact.
\end{proof}

\begin{corollary}
 In a proper cone structure we can just define global hyperbolicity with the equivalent conditions mentioned in Prop.\ \ref{sos}.
\end{corollary}

In the regular case the characterization of global hyperbolicity through the property \ref{sos}(b) was introduced in \cite{bernal06b,minguzzi06c} as an improvement over the classical definition \cite{hawking73} which assumed strong causality in place of causality. We introduced the characterization \ref{sos}(a) in \cite{minguzzi08e} and proved that it is particularly useful, for instance in the study of the stability of global hyperbolicity \cite{minguzzi11e}. The more general definitions for closed cone structures Def.\ \ref{sot} ($\alpha$) and ($\beta$) are  clearly inspired by those. Definition ($\alpha$) is quite robust, in fact it is that used to prove the stability of global hyperbolicity. Furthermore, it makes it clear that by narrowing the cones one does not spoil global hyperbolicity as both properties entering ($\alpha$) are preserved.  Charaterization ($\beta$) is also quite convenient as the property there mentioned enters nicely the transverse ladder. As for ($\delta$), stable causality is equivalent to the antisymmetry of the Seifert relation, thus global hyperbolicity can be expressed in a simple way using just the  Seifert relation, a result which is pretty satisfying given the importance of this relation for causality.

\begin{remark} \label{roc}
It is clear that the neighborhood constructed in Prop.\ \ref{iiu} or Th.\ \ref{dao} is globally hyperbolic as it is so for a Minkowski metric with wider cones.
\end{remark}

\begin{example}
A  closed cone structure which satisfies the properties of Prop.\ \ref{sos} need not be causally simple. Consider again the manifold $\mathbb{R}^2\backslash \{(0,0)\}$ of coordinates $(x,t)$, endowed with the stationary round cone structure $C$ determined by the vector field $\p_t$.
\end{example}

\begin{definition} Let $(M,C)$ be a closed cone structure.
 A {\em Cauchy hypersurface} is an acausal topological hypersurface $S$ such that  $D(S)=M$.  A {\em stable Cauchy hypersurface} is a Cauchy hypersurface for $(M,C')$ where $C'>C$ is a locally Lipschitz proper cone structure.
\end{definition}

\begin{proposition}
Let $(M,C)$ be a closed cone structure. Any two stable $C^k$, $0\le k\le \infty$,  Cauchy hypersurfaces are $C^k$ diffeomorphic. For a proper cone structure any two $C^k$, $0\le k\le \infty$, Cauchy hypersurfaces are $C^k$ diffeomorphic.
\end{proposition}

Here ``$C^0$ diffeomorphic" must be read as ``homeomorphic''.

\begin{proof}
Let $S_1$ and $S_2$ be Cauchy hypersurfaces for the locally Lipschitz proper cone structures  $C_1>C$, and $C_2>C$, respectively. Then we can find locally Lipschitz proper cone structure  $C_3>C$ such that, $C_3<C_1,C_2$, thus both $S_1$ and $S_2$ are Cauchy hypersurfaces for $(M,C_3)$. Let $V$ be a smooth $C_3$-timelike vector field. Its integral curves intersect $S_1$ and $S_2$ precisely once, so its flow can be used to establish a $C^k$ diffeomorphism between $S_1$ and $S_2$ in the usual way.
The argument for a proper cone structure is simpler, just let $V$ be a smooth timelike vector field for $(M,C)$ and argue as above.
\end{proof}

\begin{definition} Let $(M,C)$ be a closed cone structure.
A topological hypersurface $S$ is {\em stably acausal} if it is acausal with respect to $(M,C')$ where $C'>C$ is a locally Lipschitz proper cone structure.
\end{definition}

The notion of stable acausality is a kind of replacement for the `spacelikeness' notion in the smooth setting.


\begin{theorem} \label{sts}
Let $(M,C)$ be a closed cone structure. Stable Cauchy hypersurfaces and stably acausal Cauchy hypersurfaces coincide.
\end{theorem}


\begin{proof}
It is clear that every stable Cauchy hypersurface is stably acausal. For the converse, let $S$ be a stably acausal Cauchy hypersurface. There is a locally Lipschitz proper cone structure $C'>C$, such that $S$ is $C'$-acausal. Let $o\in M$, let $h$ be a complete Riemannian metric and let $K_n=\bar B(o,n)$ be a sequence of compact sets so that $K_n\subset \mathrm{Int} K_{n+1}$, $\cup_n K_n=M$. We can define inductively locally Lipschitz proper cone structures $C_n$, $C<C_n<C'$, $C_n<C_{n-1}$, in such a way that every inextendible continuous $C_n$-causal curve intersecting $K_n$ intersects $S$. In fact, if the inductive step were not allowed considering the limit $C_n \to C$ as in Th.\ \ref{sqd},  by the limit curve theorem \ref{main}  there would be an inextendible continuous $C$-causal curve intersecting $K_n$ but not $S$, a contradiction. Let $\tilde C$ be a locally Lipschitz proper cone structure chosen so that $C<\tilde C<C_n$ on $K_n\backslash \mathrm{Int} K_{n-1}$. Let $\sigma$ be an inextendible continuous $\tilde C$-causal curve, and let $k$ be the minimum number such that $\sigma \cap K_k\ne \emptyset$. Then $\sigma$ is $C_k$-causal and intersects $K_k$, thus it intersects $S$. There can only be one intersection since $S$ is $C'$-acausal and hence $\tilde C$-acausal. The arbitrariness of $\sigma$ proves that it is a Cauchy hypersurface of $\tilde C>C$.
\end{proof}

The next result is classical and in the present upper semi-continuous generalization can be found in \cite{bernard16}. Here we add the relationship with the notion of stable Cauchy hypersurface.  Our proof is closer in spirit to that of Lorentzian geometry  \cite{geroch70,hawking73} but the construction of Geroch's volume function is  really worked out on a wider cone structure.

\begin{theorem} \label{ger}
 Every globally hyperbolic closed cone structure $(M,C)$ admits a Cauchy time function (which is Geroch's time function $t$ for a wider locally Lipschitz proper cone structure).  So  every globally hyperbolic closed cone structure is a domain of  dependence, i.e.\ there is a stable Cauchy hypersurface $S$,  $M=D(S)$.
   Moreover, $M$ is topologically a product $\mathbb{R} \times S$ where  the first projection is $t$ and $S$ is smoothly diffeomorphic to the stable Cauchy hypersurface.
  For a proper cone structure  the fibers of the second projection can be chosen to be the integral timelike curves of a smooth timelike vector field.
 \end{theorem}

\begin{proof}
By the stability of global hyperbolicity we can find  a globally hyperbolic locally Lipschitz proper cone structure with wider cones, so it is sufficient to prove the theorem assuming $(M,C)$ to be a  locally Lipschitz proper cone structure, as any Cauchy time function for a cone structure is still a Cauchy time function for a narrower cone structure.


Let $\mu$ be a probability measure absolutely continuous with respect the Lebesgue measure of any local chart. Since $J^+(p)$ is a future set its boundary is a locally Lipschitz topological hypersurface so $\mu(\p J^+(p))=0$ and similarly in the past case (Th.\ \ref{aoq}). Let  $t^\pm (p)=\mp \mu (J^{\pm}(p))=\mp \mu (I^{\pm}(p))$ so that by strong causality both functions are strictly increasing over continuous causal curves. Let us  prove that $t^-$ is continuous, the proof for $t^+$ being analogous.

Let $\epsilon>0$, and let $K\subset I^{-}(p)$ be a compact set such that $\mu (I^{-}(p)\backslash K)<\epsilon$ (it exists by inner regularity of the measure). For every $q\in K$ we can find $r\in I^-(p)$ such that $q\in I^{-}(r)$, thus $K$ admits a finite covering of sets of the form $I^-(r_i)$, with $r_i\in I^{-}(p)$, then $O=\cap_i I^+(r_i)$, is such that for every $p'\in O$, for every $i$, $r_i\in I^{-}(p')$ and hence $K\subset I^-(p')$, thus $t^-(p')=\mu(I^-(p'))\ge \mu(K)\ge \mu (I^{-}(p))-\epsilon=t^-(p)-\epsilon$, which proves lower semi-continuity.

Let $\epsilon>0$, and let $K\subset M\backslash J^{-}(p)$ be a compact set such that $\mu (M\backslash (K\cup J^{-}(p)) )\le \epsilon$. Let $D$ be a compact neighborhood of $p$, then $J^-(D)\cap J^+(K)$ is compact. Thus there must be a neighborhood  $O\ni p$, such that $J^-(O)\cap K=\emptyset$, otherwise by the limit curve theorem \ref{main} we would get a continuous $C$-causal limit curve connecting $p$ to $K$, a contradiction. So for $p'\in O$ we have
\[
t^-(p')=\mu(J^-(p')) \le \mu(M\backslash K)= \mu(J^-(p))+ \mu (M\backslash (K\cup J^{-}(p)) )\le t^-(p)+\epsilon
\]
which proves upper semi-continuity. Moreover, given an inextendible causal curve $t \mapsto \sigma(t)$ we have  $t^-(\sigma(t))\to 0$ for $t\to -\infty$. In fact, let $r= \sigma(0)$, $\epsilon >0$ and  let $K$ be a compact set such that $\mu(M\backslash K)<\epsilon$. Then $J^-(r)\cap J^+(K)$ is compact and for sufficiently large $-t$, we must have $J^-(\sigma(t))\cap K=\emptyset$ otherwise, by the limit curve theorem, we would get a future inextendible causal curve starting from $K$ and contained in the compact set $J^-(r)\cap J^+(K)$, a contradiction. Similarly, $t^+(\sigma(t)) \to 0$ for $t\to +\infty$, thus the Geroch's time $\tau=\log \vert t^-/t^+\vert$ is continuous  and increasing with image $\mathbb{R}$ over every continuous causal curve.

Notice that  $S_0=t^{-1}(0)$ is a Cauchy hypersurface for the wider cone structure and so a stable Cauchy hypersurface for the original cone structure.

The last statement is a trivial consequence of the existence of a smooth  timelike vector field in proper cone structures and of  the quotient manifold theorem, cf.\ Th.\ 21.10 of \cite{lee13}.
\end{proof}

%
%
%
%


\begin{theorem} \label{mmm}
Let $(M,C)$ be a proper  cone structure and let $S$ be an  acausal topological hypersurface. Then $D(S)$ is open, causally convex and globally hyperbolic.
Let $(M,C)$ be a closed  cone structure which admits a Cauchy hypersurface, then $(M,C)$ is globally hyperbolic.
\end{theorem}

\begin{proof}
The openness of $D(S)$ follows  from Th.\ \ref{jdp} and \ref{jdq}.
Notice that given $p,q\in D(S)$ there cannot be $r\in [J^{+}(p)\cap
J^{-}(q)] \backslash D(S)$, in fact redefining $p$ and $q$ we can assume $p\in J^{-}(S)\backslash S$ and $q\in J^+(S)\backslash S$. Then by acausality of $S$ one of the curves connecting $p$ to $r$ or $r$ to $q$ does not intersect $S$. Thus there is a future inextendible continuous causal curve issued from $p$ not intersecting $S$, or a past inextendible continuous causal curve ending at $q$ not intersecting $S$. The contradiction proves that $D(S)$ is causally convex.

 Concerning causality, there cannot be continuous closed causal curves in $D(S)$ because they cannot intersect $S$ by its acausality, though they are inextendible and so must intersect it.

Let $K_1,K_2 \subset  D(S)$ be compact subsets and assume $\overline{J^{+}(K_1)\cap
J^{-}(K_2)}$ is not compact, then there are points $r_n \in
J^{+}(K_1)\cap J^{-}(K_2)$ with $r_n \to +\infty$ (i.e. escaping every
compact set contained in $D(S)$). Let $\sigma_n$ be a causal curve starting from $p_n\in K_1$ passing
through $r_n$ and ending at $q_n\in K_2$. By the limit curve theorem
\ref{main} there are a past inextendible causal curve
$\sigma^{q}$ ending at $q\in K_2$, and a future inextendible causal curve
$\sigma^p$ starting from $p\in K_1$, which are limits of $\sigma_n$. By
Th.\ \ref{jdp} there are $\tilde{p} \in \sigma^p \cap J^{+}(S)\backslash S$ (thus $\tilde p\in \mathrm{Int} J^+(S)$ because every past inextendible continuous causal curve future ending in a sufficiently small neighborhood of $\tilde{p}$ intersects $S$) and
$\tilde{q} \in \sigma^q \cap J^{-}(S)\backslash S$ (thus $\tilde{q}\in \mathrm{Int} J^{-}(S)$). However, the limit curve
theorem also states that $(\tilde{p},\tilde{q}) \in \bar{J}$, thus $S$ is not acausal, a
contradiction. Let us prove that $J^{+}(K_1)\cap J^{-}(K_2)$ is
closed, indeed, if it were not then there would be points  $r_n \in
J^{+}(K_1)\cap J^{-}(K_2)$ with $r_n \to r \in \overline{J^{+}(K_1)\cap
J^{-}(K_2)}\backslash [J^{+}(K_1)\cap J^{-}(K_2)]$. The argument goes as
above with the additional observation that at most one curve between
$\sigma^{q}$ and $\sigma^p$ can pass through $r$. Thus the  causal
emeralds $J^{+}(K_1)\cap J^{-}(K_2)$ are contained in $D(S)$ and compact.

The last statement follows from the last two paragraphs by replacing $D(S)$ with $M$.
%
%
%
\end{proof}
%
%
%



\begin{theorem} \label{jbu}
Let $(M,C)$ be a proper cone structure and let $S$ be an acausal topological hypersurface. Then for every compact subset $K\subset D^+(S)$, $J^-(K)\cap J^+(S)$ is compact.
\end{theorem}

\begin{proof}
Otherwise a limit curve argument would produce a past inextendible continuous causal curve $\sigma$ ending at $K$. Thus by acausality of $S$ it would cross $S$ entering, by Th.\ \ref{jdp}, $J^-(S)\backslash S\cap D^-(S)\subset I^-(S)$. But then the sequence $\sigma_k \to \sigma$ starting from $S$, would have to enter $I^-(S)$, contradicting the acausality of $S$.
\end{proof}




We end the section summarizing some other equivalent charaterizations of global hyperbolicity  which are familiar from Lorentzian geometry \cite{bernal03,chrusciel13}. Fathi and Siconolfi have first obtained a version for continuous cone structures using methods imported from weak KAM theory  \cite{fathi12}. These results have been generalized to upper semi-continuous cone structures by Bernard and Suhr who employed instead dynamical system methods based on Conley's theory \cite{bernard16}. They obtained the equivalence between points (i) and (iii) in Theorem \ref{xxy} below, and the relative splitting. We clarify the connection with the existence of non-smooth Cauchy time functions and Cauchy hypersurfaces. Our derivation is based on volume functions, and uses methods entirely developed in the field of mathematical relativity. Most of the  proof is given in  Sec.\ \ref{fir}-\ref{las}.

%

 \begin{theorem} \label{xxy}
Let $(M,C)$ be a closed cone structure and let $h$ be a complete Riemannian metric. Then the next conditions are equivalent:
\begin{itemize}
\item[(i)]  global hyperbolicity,
\item[(ii)] existence of a Cauchy time function,
\item[(iii)] existence of  a smooth  $h$-steep  Cauchy  temporal function,
\item[(iv)] existence of  a (stable) Cauchy hypersurface.
\end{itemize}
 Finally, under global hyperbolicity $M$ is smoothly diffeomorphic to a product $\mathbb{R}\times \mathbb{S}$ where the projection to $\mathbb{R}$ is a smooth  $h$-steep  Cauchy  temporal function (the fibers of the smooth projection to $S$ are not necessarily causal), and every stable Cauchy hypersurface is smoothly diffeomorphic to $S$.

Additionally, for a proper cone structure all Cauchy hypersurfaces are diffeomorphic to $S$ and the fibers of the smooth projection to $S$ are smooth timelike curves.
 \end{theorem}
The `stable' adjective in (iv) must be kept or dropped so as to get the strongest meaning of the implication considered.

\begin{proof}
(i) $\Rightarrow$ `(ii) and (iii) and (iv)' is proved in Th.\ \ref{xbh}. (ii) or (iii) or (iv) $\Rightarrow$ (i) is proved in Th.\ \ref{mmm}. The other statements are proved in Th.\ \ref{mom} and Cor.\ \ref{nin}.
\end{proof}

Finally, we mention the next stability result.

\begin{theorem} \label{rem}
Let $(M,C)$ be a closed cone structure. Every Cauchy temporal function $t$ is stable in the sense that we can find a locally Lipschitz proper cone structure $C'>C$ such that $t$ is Cauchy temporal for $C'$.
\end{theorem}

\begin{proof}
Since $t$ is a temporal function $(M,C)$ is stably causal, so all the cone structures that follow which are wider than $C$ will be taken stably causal (they will also be proper and locally Lipschitz).
Let $h$ be a complete Riemannian metric, and let $o\in M$. Let $B(o,r)$ be the ball of radius $r$ centered at $o$.
Let $K_1$ be a compact set containing $B(o,1)$. Let us redefine $t$ with an affine transformation so that $K_1\subset t^{-1}([-t_1,t_1])$, where $t_1=1$. There is $C_1>C$ such that $\dd t$ is positive on $C_1$ and all the inextendible continuous $C_1$-causal curves passing through $K_1$ intersect the level sets $t^{-1}(-(t_1+1))$ and $t^{-1}(t_1+1)$. In fact, this claim is proved using a sequence $\tilde C_k\to C$, $\tilde C_k>C$, as in Th.\ \ref{sqd}, by noticing that if it were not true then there would be a sequence $\tilde \sigma_k$ of continuous $\tilde C_k$-causal curves intersecting $K_1$ but not $t^{-1}(-(t_1+1))$ or $t^{-1}(t_1+1)$. The function $t$ would be bounded by $-(t_1+1)$ or $t_1+1$ on the limit inextendible continuous $C$-causal curve $\sigma$ in contradiction with the Cauchy property of $t$. By the   temporality of $t$, $t$ over  the inextendible continuous $C_1$-causal curves passing through $K_1$  has image strictly containing $[-(t_1+1), t_1+1]$

 By a similar limit curve argument we obtain that $C_1$ can be chosen so that there is a compact set $K_2$, $\mathrm{Int} K_2\supset K_1\cup B(o,2)$, such that
  the image of every continuous $C_1$-causal curve intersecting $K_1$ which stays in $t^{-1}([-(t_1+1),t_1+1])$ is contained in $\mathrm{Int} K_2$. Notice that there will be $t_2>t_1+1$ such that $K_2 \subset t^{-1}([-t_2,t_2])$. By proceeding in this way we obtain a sequence of compact sets $K_k$, $\mathrm{Int}  K_{k+1}\supset K_k\cup B(o,k+1) $, a sequence of times $t_{k}>0$, $t_{k+1} \ge t_k+1$, and a sequence $C_k>C$, $C_{k+1}<C_k$, of locally Lipschitz proper cone structures, such that every continuous $C_k$-causal curve intersecting $K_k\subset t^{-1}([-t_k,t_k])$ is bound to reach $t^{-1}(-(t_k+1))$ and $t^{-1}(t_k+1)$, and between such intersections to be contained in $K_{k+1}$. Let $C'>C$ be a locally Lipschitz proper cone structure such that $C'<C_k$ on $K_{k+1}\backslash \mathrm{Int} K_k$ for every $k$. Since $C'<C_1$, it is stably causal hence non-imprisoning. Let us consider an inextendible continuous $C'$-causal curve $\sigma$, then there is a minimum value of $k$ such that $\sigma\cap K_k\ne \emptyset$. Taking a point $q_k\in \sigma$ in this set (thus $t(q_k)\le t_k$) and following $\sigma$ in the future direction $t$ increases over $\sigma$ because $\dd t$ is positive on $C_1$ and hence on $C'$. Moreover, by the non-imprisoning property $\sigma$ escapes $K_k$, and so becomes $C_k$-causal on $K_{k+1}$, thus it reaches $q_{k+1}\in t^{-1}(t_k+1)$ where we can repeat the argument since $q_{k+1}\in K_{k+1}$. So the argument shows that $t$ goes to infinity in both directions of $\sigma$, that is $t$ is Cauchy temporal for $(M,C')$.
\end{proof}

\begin{corollary}
Let $(M,C)$ be a closed cone structure. The level sets of a Cauchy temporal function are stable Cauchy hypersurfaces.
\end{corollary}



\subsection{The causal ladder}


Many standard results of Lorentzian causality theory under a $C^2$ assumption on the metric are obtained with an application of a limit curve argument. However, many  other arguments use the composition rule $I\circ J\cup J\circ I\subset I$ so it should come as a surprise that the locally Lipschitz condition will not appear in this section. In fact, even more strikingly  neither a continuity assumption on $C$ will appear and furthermore, the chronological relation will not need to  be defined.


%
We have already met some of the next concepts.
\begin{definition}
A  closed cone structure $(M,C)$ is
\begin{itemize}
\item {\em Causal}. If there is no closed continuous causal curve,
\item {\em Non-total imprisoning}. If there is no future inextendible continuous causal curve contained in a compact set.
\item {\em Distinguishing}. Every point  admits arbitrarily small distinguishing open neighborhoods. Namely for every $p\in M$ and open set $U\ni p$, there is an open set $V$, $p\in V\subset U$, which distinguishes $p$, namely every continuous causal curve $x\colon I\to M$ passing through $p$ intersects $V$ in a connected subset of $I$. (One can give less restrictive future and past versions.)
\item {\em Strongly causal}. Every point  admits arbitrarily small causally convex open neighborhoods. Namely for every $p\in M$ and open set $U\ni p$, there is an open set $V$, $p\in V\subset U$, which is causally convex, namely every continuous causal curve $x\colon I\to M$, with endpoints in $V$ is entirely contained in $V$.
\item {\em Stably causal}. There is $C'>C$ such that $(M,C')$ is a  causal $C^0$ cone structure.
\item {\em Causally easy}. Strongly causal and $\bar J$ is transitive.
\item {\em Causally continuous}. Distinguishing and reflective.
\item {\em Causally simple}.  Causal and $J=\bar J$.
\item {\em Globally hyperbolic}. Causally simple and the causally convex hull of compact sets is compact.
\end{itemize}
Moreover, for a proper cone structure we say that  $(M,C)$ is chronological if there is no closed timelike curve.
\end{definition}



\begin{figure}[h!]
\begin{center}
 \includegraphics[width=12cm]{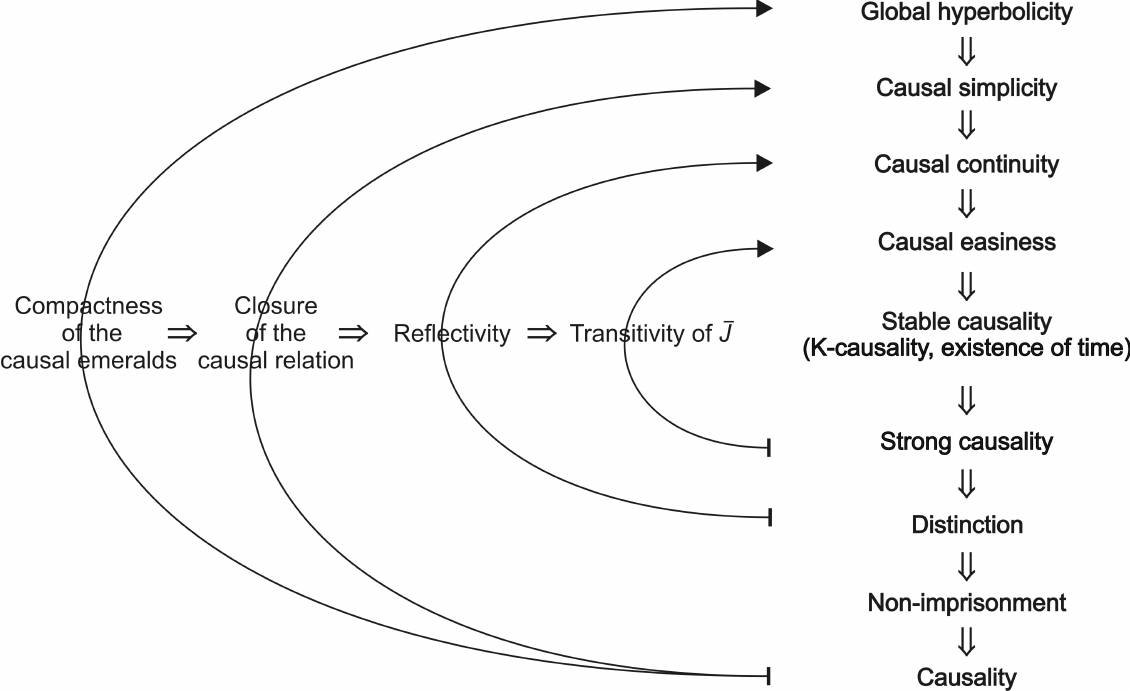}
\end{center}
\caption{The causal ladder and the transverse ladder for closed cone structures. The arrows crossing a property use it in the implication.} \label{lad}
\end{figure}

The proof of the next theorem uses the equivalence between stable causality and $K$-causality which is proved in Sec.\ \ref{fir}-\ref{las}.

\begin{theorem} \label{cau}(Causal ladder, see Fig.\ \ref{lad}) \\
  If $(M,C)$ is a closed cone structure:
\begin{quote}
 Globally hyperbolic $\Rightarrow$  Causally simple $\Rightarrow$ Causally continuous $\Rightarrow$ Causally easy $\Rightarrow$ Stably causal  $\Rightarrow$  Strongly causal $\Rightarrow$ Distinguishing
  $\Rightarrow$ Non-total imprisoning $\Rightarrow$  Causal.
   \end{quote}
 Moreover, causality implies chronology.
\end{theorem}
We have separated the last implication from the rest since the chronological relation is not particularly interesting unless we are in a proper cone structure.
\begin{proof}
Stably causal $\Rightarrow$  Strongly causal: this is Th.\ \ref{mwh}.
Strongly causal  $\Rightarrow$ Distinguishing: trivial.
Distinguishing $\Rightarrow$ Non-total imprisoning: if there is an imprisoned curve then by  Theorem \ref{lof} the starting point of $\alpha$ cannot be distinguished by arbitrarily small neighborhoods since $\alpha$ escapes and reenters them.
Non-total imprisoning $\Rightarrow$ causal: trivial.

Globally hyperbolic $\Rightarrow$ causally simple: This is  Th.\ \ref{nnq}.
Causally simple $\Rightarrow$  Causally continuous: Since $\bar J=J=D_p=D_f$ reflectivity holds true and by Lemma \ref{mvh} strong causality holds which implies distinction. Causally continuous  $\Rightarrow$  Causally easy: By reflectivity $D_p=D_f=\bar J$ and by Prop.\ \ref{jjw} $D_p$ is transitive, thus $\bar J$ is transitive and hence $K=\bar J$. Moreover, We have distinction which by Prop.\ \ref{aah} implies the antisymmetry of $D_f$ and hence that of $K$. Thus by Th.\ \ref{nio} stable causality holds which as mentioned implies strong causality.
Causally easy $\Rightarrow$ Stably causal: observe that $\bar J$ is antisymmetric,  indeed, suppose $(p,q) \in \bar J$ and $(q,p)\in \bar J$, with $p\ne q$, then by the limit curve theorem for every neighborhood $U\ni p$ we can find $p'\in J^+(p)\backslash \{p\}\cap U$ such that $(p',q)\in \bar J$, thus by the transitivity of $\bar J$, $(p',p)\in \bar J$, in contradiction with the strong causality at $p$. But we have $K=\bar J$, thus $K$-casuality holds which implies stable causality (Th.\ \ref{nio}).
Since every timelike curve is a continuous causal curve the last implication is clear.
%
%
%
\end{proof}

\subsection{Fermat's principle} \label{fer}

The next result improves the differentiability conditions in Prop.\ \ref{chg} by strengthening the other assumptions, including the causality condition, see Th.\ \ref{nig} for another version.

\begin{theorem} \label{sop}
Let $(M,C)$ be a globally hyperbolic closed cone structure and let  $S$ be a compact set. If $q\in E^+(S)\backslash S$ there is $p\in S$ and a  future lightlike geodesic with endpoints $p$ and $q$ contained in $E^+(S)$.
\end{theorem}
Notice that since $J$ is closed and $S$ is compact,  $E^+(S)$ is closed (again by Th.\ \ref{xix}).
\begin{proof}
Let $C_k$, $C_{k+1}\le C_k$, $C=\cap_k C_k$, be a sequence of locally Lipschitz proper cone structures as in Th.\ \ref{sqd} where $C_1$, and hence every $C_k$, is globally hyperbolic. In particular, $J_k$ is closed for every $k$. Suppose that we can find, passing to a subsequence if necessary, $q_k\in  E_k^+(S)\backslash S$, with $q_k\to q$. Then by Th.\ \ref{chg} there is a continuous $C_k$-causal curve $\sigma_k$ of starting point $p_k\in S$ and ending point $q_k$ entirely contained in $E_k^+(S)$, so not intersecting $\mathrm{Int} J^+_k(S)\supset \mathrm{Int} J^+(S)$. Let $V$ be a compact neighborhood of $q$, then $J^+_1(S)\cap J^-_1(V)$ is a  compact set which contains all $\sigma_k$ for sufficiently large $k$. By the limit curve theorem there is a subsequence denoted in the same way such that $\sigma_k$ converges uniformly to a continuous $C$-causal curve $\sigma$ connecting $p$ to $q$. It does not intersect the open set $\mathrm{Int} J^+(S)$ since none of the $\sigma_k$ does. It remains to prove that the sequence $q_k$ exists. Suppose that we cannot find $q_k$ as above, then there is $\epsilon>0$ such that $B(q,\epsilon) \subset  J^+_k(S)$ for sufficiently large $k$. For every $y\in B(q,\epsilon)$ using again the limit curve theorem we get that $y\in J^+(S)$, thus $q\in \mathrm{Int} J^+(S)$, a contradiction. The fact that $\sigma$ is a future lightlike geodesic is immediate, since if $p',q'\in\sigma$, with $q'\in \mathrm{Int}J^+(p')$ then $q'\in \mathrm{Int}J^+(p) \subset \mathrm{Int} J^+(S)$, a contradiction.
\end{proof}

The next corollary to be used in this section is the global version of Th.\ \ref{dao}.

\begin{corollary} \label{gtr}
Let $(M,C)$ be a globally hyperbolic closed cone structure. If $q\in E^+(p)\backslash\{p\}$ then there is a future lightlike geodesic $\sigma$ connecting $p$ to $q$ entirely contained in $E^+(p)$.
\end{corollary}

On the general relativistic spacetime the metric induces a distribution of cones  $C'$ referred as {\em light cones}. It turns out, however, that they should be more properly called {\em gravity cones} since light in presence of matter propagates at a smaller speed.
The distribution of cones $C'$ represents more properly the speed of gravitational waves rather than light. Still in presence of media with different refractive indices we have a distribution of (true) light cones which can be modeled with a cone structure $C$. If the gravity cone $C'\ge C$ defines a globally hyperbolic cone structure, as it is often assumed, then the same will hold for the light cone structure. Now, in presence of media with different refractive indices the cone distribution $C$ will be discontinuous at the interface of the media. Thus discontinuous cone structures are pretty natural in the context of light propagation though for mathematical convenience they have been mostly disregarded.

One difficulty is that even by knowing the refractive indices in each medium, one would still be faced with the problem of assigning a refractive index value at their interface.
Here the mathematical theory comes into help since it shows that the theory is particularly satisfying if we take the lower value, in fact the speed of light is $c/n$ so this choice guarantees that the cone distribution will be upper semi-continuous, thus enjoying all the properties that we obtained in the previous sections. 
Similar considerations for what concerns Fermat's principle on Euclidean space can be found in Cellina \cite{cellina05}.

Given this preliminary discussion the next result expresses the Fermat principle in curved spacetime continua admitting discontinuous refractive indices. The principle states that
among the many virtual  paths that connect the source to the observer, there is one which minimizes the time of sight by the observer. The whole principle
relies on the existence of a minimum time of sight and of some special trajectory connecting the events of emission and reception.

\begin{theorem} (Fermat's principle)\\
Let $(M,C)$ be a globally hyperbolic closed cone structure (representing light propagation), let $t\mapsto \sigma(t)$ be an inextendibile causal curve (the observer) and let $p$ be a point source. If $J^+(p)\cap \sigma \ne \emptyset$ (at least some virtual light path reaches the observer) and $p\notin \sigma$ (source and observer do not coincide), then there is a minimum value $t_0$ such that $\sigma(t_0) \in J^+(p)$ and a future lightlike geodesic connecting $p$ to $\sigma(t_0)$ entirely contained in $\p J^+(p)$.
\end{theorem}
\begin{proof}
Since the spacetime is globally hyperbolic it admits a Cauchy time function $\tau$. Thus $\tau$ is lower bounded by $\tau(p)$ on $J^+(p)$ and hence on $J^+(p)\cap \sigma$. So we can find a smallest value $t_0$ such that $\sigma(t_0)\in \p J^+(p)$. The desired conclusion is now a consequence of  $J$ being closed and of Cor.\ \ref{gtr}.
\end{proof}

The next version will be useful in the proof of Penrose's theorem.
\begin{theorem} \label{nig}
Let $(M,C)$ be a non-imprisoning closed cone structure and let  $S$ be a compact set such that $E^+(S)$ is bounded.  If $q\in E^+(p)\backslash\{p\}$ then there is a future lightlike geodesic $\sigma$ connecting $p$ to $q$ entirely contained in $E^+(p)$.
\end{theorem}

\begin{proof}
Let $C_k$, $C_{k+1}\le C_k$, $C=\cap_k C_k$, be a sequence of locally Lipschitz proper cone structures as in Th.\ \ref{sqd}. Suppose that we can find, passing to a subsequence if necessary, $q_k\in  E_k^+(S)\backslash S$, with $q_k\to q$. Then by Th.\ \ref{chg} there is a continuous $C_k$-causal curve $\sigma_k$ of starting point $p_k\in S$ and ending point $q_k$ entirely contained in $E_k^+(S)$, so not intersecting $\mathrm{Int} J^+_k(S)\supset \mathrm{Int} J^+(S)$. By the limit curve theorem there is a subsequence denoted in the same way such that $\sigma_k$ converges uniformly to a  continuous $C$-causal curve $\sigma$ which is either future inextendible and starting from  some $p\in S$ or connecting some $p\in S$ to $q$. It does not intersects the open set $\mathrm{Int} J^+(S)$ since none of the $\sigma_k$ does, so $\sigma$ is contained in $E^+(S)$. But the former case  cannot apply due to  the non-imprisonment condition, thus $\sigma$ connects $p$ to $q$.

It remains to prove that the sequence $q_k$ exists. Suppose that we cannot find $q_k$ as above, then there is $\epsilon>0$ such that $B(q,\epsilon) \subset  J^+_k(S)$ for sufficiently large $k$. For every $y\in B(q,\epsilon)$ using again the limit curve theorem and reasoning as above we get either a future inextendible continuous causal curve imprisoned in $E^+(S)$ starting from some  $p\in S$, which is impossible, or that $y\in J^+(S)$, thus as the conclusion holds for every $y$, $q\in \mathrm{Int} J^+(S)$, a contradiction. The fact that $\sigma$ is a future lightlike geodesic is immediate, since if $p',q'\in\sigma$, with $q'\in \mathrm{Int}J^+(p')$ then $q'\in \mathrm{Int}J^+(p) \subset \mathrm{Int} J^+(S)$, a contradiction.
\end{proof}

\subsection{Lorentz-Finsler space} \label{ngd}
Let $(M,C)$ be a closed cone structure, and let $\mathscr{F}\colon C\to [0,+\infty)$ be a concave positive homogeneous function (that is, these properties are to hold when restricting to each individual tangent space). Notice that we do not demand $\mathscr{F}(\p C)=0$.
Let us introduce the cone structure on  $M^\times=M\times\mathbb{R}$ defined at $P=(p,r)$ by
\begin{equation} \label{nhz}
C^\times_{P}=\{(y,z) \colon y\in C_p, \ \vert z\vert \le \mathscr{F}(y) \}.
\end{equation}
It is indeed easy to check that this is a non-empty convex sharp cone in the tangent space $T_PM^\times$. We shall also say that the cone structure $(M^\times, C^\times)$ is a {\em Lorentz-Finsler space} $(M,\mathscr{F})$. A Lorentz-Finsler space might also be called a {\em spacetime}. Similarly we define the concept of {\em Lorentz-Minkowski space} which is simply the model geometry to the tangent space of a Lorentz-Finsler space: in other words it is given by a vector space $V$,  a cone $C\subset V$ which is non-empty convex and  sharp, and a concave positive homogeneous function on $C$.

The connection with the usual more regular notion of Lorentz-Finsler space is explored in Sec.\ \ref{zzp} and \ref{xxo}.
The fact that  a Lorentz-Finsler space can be seen as a cone structure in a space with one additional dimension will be a central idea of this work.

\begin{definition}
 $(M,\mathscr{F})$ is a closed (proper) Lorentz-Finsler space iff $(M^\times, C^\times)$ is a closed (resp.\ proper) cone structure.  We say that $(M,\mathscr{F})$ is locally Lipschitz (or $C^0$) if $C^\times$ is locally Lipschitz (resp. $C^0$).
\end{definition}

The next result follows easily from the definitions.
\begin{proposition} \label{hll}
$(M,\mathscr{F})$ is a closed Lorentz-Finsler space iff $C$ and $\mathscr{F}$ are upper semi-continuous.
$(M,\mathscr{F})$ is a  $C^0$  proper Lorentz-Finsler space iff $(M,C)$ is a $C^0$  proper cone structure and  $\mathscr{F}$ is continuous and not identically zero on any fiber.  $(M,\mathscr{F})$ is a    proper Lorentz-Finsler space if there are $\tilde C\le C$ and $\tilde{\mathscr{F}}\colon \tilde C\to [0,+\infty)$, with $\tilde{\mathscr{F}}\le \mathscr{F}$ such that $(M,\tilde{ \mathscr{F}})$ is a  $C^0$ proper cone structure.
\end{proposition}

Here the upper semi-continuity of $\mathscr{F}$ is understood as follows: for $y_n \in C$, $y_n\to y\in C$, $\limsup_{y_n\to y} \mathscr{F}(y_n)\le \mathscr{F}(y)$. Continuity is understood similarly, where the latter equation is replaced by $\lim_{y_n\to y} \mathscr{F}(y_n)= \mathscr{F}(y)$.

\begin{remark}
If $(M,\mathscr{F})$ is a  locally Lipschitz   proper Lorentz-Finsler space then $(M,C)$ is a locally Lipschitz   proper cone structure
but the base dependence of $\mathscr{F}$ need not be locally Lipschitz.
Let us consider the metric $g=-(\dd x^0)+(x^1)^2 (\dd x^1)^2$ on $\mathbb{R}\times \mathbb{R}+$, then $\mathscr{F}=\sqrt{(y^0)^2-(x^1)^2(y^1)^2}$. For $x^1\le 1$ the vector $y=(y^0,y^1)=(1,1)$ is causal and $\mathscr{F}(y)=\sqrt{1-(x^1)^2}$ which clearly is not locally Lipschitz in a neighborhood of $(x^0,x^1)=(0,1)$ although $C^\times$ is locally Lipschitz.
\end{remark}

The next result proves that our approach to the regularity of Lorentz-Finsler spaces  is compatible with the natural definitions coming from Lorentzian geometry.
\begin{theorem}
For a time oriented Lorentzian manifold $(M,g)$ the metric $g$ is continuous (locally Lipschitz) iff $C^\times$ is continuous (resp. locally Lipschitz).
\end{theorem}

Due to the metric signature, the equivalence does not hold for upper/lower semi-continuity, for the cone can get narrower or wider depending on the discontinuous metric coefficient. In this case the most useful concept of upper/lower semi-continuity is the new one derived from $C^\times$.

\begin{proof}
Only if direction: the cone distribution $C^\times$ is the bundle of causal vectors for the Lorentzian metric $g^\times=g+\dd (y^{n+1})^2$, where $z=y^{n+1}$ is the extra tangent space coordinate of $T(M\times \mathbb{R})$. The metric coefficients of $g^\times$ are continuous (resp. locally Lipschitz) because those of $g$ are, so by Prop.\ \ref{jss} $C^\times$ is continuous (resp.\ locally Lipschitz).

If direction: The continuity of $C^\times$ implies the continuity of the function $\sqrt{\max\{-g(x)(y,y),0\}}$ and  hence that of $g$ (by the arbitrariness of $y$ and polarization formulas). Suppose that $C^\times$ is locally Lipschitz. Let $U$ be a coordinate neighborhood of $\bar x\in M$, so that $TU$ trivializes as  $U\times \mathbb{R}^{n+1}$ with $(x^\alpha, y^\alpha)$ local coordinates. Suppose also that the coordinates have been chosen so that $g_{\alpha \beta}({\bar x})=\eta_{\alpha \beta}$, i.e.\ the Minkowski metric, with $\p/\p x^0$ future directed.
Since $g$ does not depend on the extra coordinate $x^{n+2}$ it is sufficient to consider the dependence of $C^\times$ on the coordinates $\{x^\alpha\}$ of $M$ by fixing $x^{n+2}=0$. Thus $C^\times_{(x,0)}$ is a cone of $\mathbb{R}^{n+2}$ which for $x$ close to $\bar x$,  intersects the locus $\{y^0=1, y^{n+2}\ge 0\}$
on a half ellipsoid of $\mathbb{R}^{n+1}$ of equation $y^{n+2}=\sqrt{-g_{\alpha \beta}(x) y^\alpha y^\beta}$ which intersects orthogonally $\{ y^{n+2}=0\}$. By continuity for every round cone $A\subset \mathrm{Int} C_x$, we have for $x\in \bar V$, $V$ sufficiently small neighborhood of $\bar x$,  $\bar V\subset U$, $A\subset \mathrm{Int} C_x\subset \mathbb{R}^{n+1}$. The idea is to show that $g_{\alpha \beta}(x)y^\alpha y^\beta$ is locally Lipschitz in $x$ for any chosen $\bar y^\alpha \in A$ and hence, by the arbitrariness of $\bar y$ an polarization formulas, that all the coefficients of the metric are locally Lipschitz.
By positive homogeneity we can restrict $\bar y^\alpha \in \tilde A=A\cap \{y^0=1\}$, then for $(x^\alpha, y^\alpha)\in  \bar V\times \tilde A$ we have that the function $\sqrt{-g_{\alpha \beta}(x) y^\alpha y^\beta}$ is well defined, bounded from below by a positive constant and with differential bounded from above.  We suppose to have chosen $U$ so small that  there is a constant $K>0$ such that for $x_1,x_2\in U$,
\[
D_{12}:=D\big(C^\times_{(x_1,0)}\cap \{y^0\!=\!1, y^{n+2}\!\ge 0\} , C^\times_{(x_2,0)}\cap \{y^0\!=\!1, y^{n+2}\!\ge 0\}\big)\le K \Vert x_1-x_2\Vert ,
\]
 where $D$ is the Hausdorff distance on $\mathbb{R}^{n+1}$, and $\Vert \Vert$ is the Euclidean norm on $\mathbb{R}^{n+1}$. That is, the distance among the half ellipsoids has Lipschitz regularity.
Let us consider two points $x_1,x_2\in V$ where the label $2$ is chosen so that $f(x_2, \bar y)\le f(x_1,\bar y)$, with $f(x,y)=\sqrt{-g_{\alpha \beta}(x)  y^\alpha  y^\beta}$ and $y\in \tilde A$. The distance on $\mathbb{R}^{n+1}$ of the point $p=(\bar y^1,\ldots, \bar y^{n+1}, f(x_2))$ on the ellipsoid 2 to the ellipsoid 1 of graph $y^{n+2}=\sqrt{-g_{\alpha \beta}(x_1) y^\alpha y^\beta}$ is smaller that the Hausdorff distance $D_{12}$ between the ellipsoids. This minimum distance is realized by a segment between the point $p$ and a point $q$ on the ellipsoid 1 ($q$ has  projection $\tilde y$ in general different from $\bar y$). Notice that the tangent plane to the ellipsoid 1 at $q$ is orthogonal to $pq$ and intersects the fiber of $\bar y$ at a point $r$ of extra coordinate larger than $f(x_1,\bar y)$ due to the convexity of the ellipsoid 1. Moreover, $\overline{pr}\le \overline{pq}/\cos \theta$ where $\tan \theta=\Vert \nabla_y f(x_1,\tilde y)\Vert$ is the slope of the mentioned tangent plane, thus since this derivative is bounded on $\bar V\times \tilde A$, $1/\cos \theta$ is bounded and we  can find a constant $L>0$ independent of $\bar y$ such that for $x_1,x_2\in \bar V$, $f(x_1,\bar y) -f(x_2,\bar y) \le \overline{pr}\le L D_{12}$. Since $f$ is bounded from above by a constant $R>0$ on $\bar V\times \tilde A$, we have  $0<f(x_1,\bar y) +f(x_2,\bar y)\le R$, and so multiplying the two inequalities $0\le [g_{\alpha \beta}(x_2 )- g_{\alpha \beta}(x_1 )] \bar y^\alpha  \bar y^\beta\le RLD_{12}\le  KRL \Vert x_1-x_2\Vert$.
\end{proof}

An interesting  large class of closed Lorentz-Finsler spaces is selected by the next theorem (see also Remark \ref{eqq}).

\begin{theorem}
Let $C\subset TM\backslash 0$ be a proper cone structure and let $\mathscr{F}\colon C\to [0,+\infty)$ be a positive  homogeneous  $C^0$ function, such that $\mathscr{F}^{-1}(0)=\p C$. Let $f\colon \mathbb{R}_+\to \mathbb{R}$  (for instance $f(x)=x^a/a$, $a > 1$, the typical case being $a=2$) be $C^1([0,+\infty))\cap C^2(\mathbb{R}_+)$ and such that $f'(x)>0, f''(x)>0$ for $x>0$. Suppose that $\mathscr{L}=-f(\mathscr{F})$ is $C^1(C)\cap C^2(\mathrm{Int} C)$, that it has  Lorentzian vertical Hessian $\dd^2_y \mathscr{L}$, and that $\dd_y \mathscr{L}\ne 0$ on $\p C$. Then $\mathscr{F}$ is concave, and $(M,\mathscr{F})$ is a locally Lipschitz  proper Lorentz-Finsler space (hence both $C$ and $C^\times$ are locally Lipschitz).
\end{theorem}


\begin{proof}The proof that $\mathscr{F}$ is concave and that $(M,\mathscr{F})$ is a $C^0$ proper Lorentz-Finsler space follows from the continuity of $\mathscr{F}$ and from the results on Lorentz-Minkowski spaces of Sec.\ \ref{zzp}, particularly Remark \ref{eqq}, so we need only to prove that $C$ and $C^\times$ are locally Lipschitz. Let us prove that $C$ is locally Lipschitz.
Let $\bar x\in M$, and let $U$ be a  coordinate neighborhood of $\bar x$. Let us consider the trivialization of the bundle $T U$, as induced by the coordinates. We are going to focus on the subbundle of $TU$ of vectors that in coordinates read as follows $(x^\alpha, y^\alpha)$ where $y^0=1$, i.e. we are going to work on $U\times \mathbb{R}^n$ (the function $\mathscr{L}$ will be thought as restricted to this set though we shall keep the original notation).
It will be sufficient to prove the locally Lipschitz property for this distribution of sliced cones. Let $\Vert \cdot \Vert$ be the Euclidean norm on $\mathbb{R}^n$.
Since the cone distribution over the sliced subbundle has compact fibers, we can find $U$ sufficiently small that there is a constant $K>0$, such that $ \Vert \nabla_x \mathscr{L}\Vert/\Vert \nabla_y \mathscr{L}\Vert<K$ for all lightlike vectors on the sliced subbundle (here the labels $x,y$ refer to base or  vertical gradients).

Let us consider two sliced cones relative to the points $x_1$ and $x_2$. Let $y_1$ and $y_2$ be two points that realize the Hausdorff distance $D(x_1,x_2)$ between the sliced cone boundaries, i.e. $D(x_1,x_2)= \Vert \delta y \Vert$, $\delta y=y_1-y_2$, where the vector $\delta y= y_1-y_2$ can be identified with a vector of $\mathbb{R}^n$ since its 0-th component vanishes. The definition of Hausdorff distance easily implies that $\delta y$ is orthogonal to one
of the sliced cone boundaries. Let it be that of $x_2$, otherwise switch the labels 1 and 2. So we have $\delta y\propto \nabla_y \mathscr{L}(x_2,y_2)$, and hence $\Vert\nabla_y \mathscr{L} \cdot \delta y\Vert=\Vert\nabla_y \mathscr{L} \Vert \Vert \delta y\Vert$. Let $\delta x=x_1-x_2$, $\delta \mathscr{L}=\mathscr{L}(x_1,y_1)-\mathscr{L}(x_2,y_2)=0$. By continuity as $\delta x \to 0$, we have $\delta y\to 0$, up to higher order terms the Taylor expansion at $(x_2,y_2)$ gives
$0=\nabla_y \mathscr{L} \cdot \delta y+\nabla_x \mathscr{L} \cdot \delta x$, and hence for sufficiently small $\delta x$,  $D(x_1,x_2)=\Vert \delta y\Vert\le K\Vert \delta x\Vert$, which proves the locally Lipschitz property. It can be observed that this argument used only one property of $\mathscr{L}$, namely that of being $C^1$  on $\p C$ with non-vanishing vertical gradient. This property  has been used for deducing the existence of $K$. Now, for the local Lipschitzness of $C^\times$ we need only to show that there is a $C^1$ function up to $\p C^\times$ with non-vanishing gradient on $\p C^\times\cap \{z\ge 0\}$, where $z$ is the extra tangent space coordinate. Evidently, the function $\mathscr{L}(x,y)+f(z)$, has the desired property.
\end{proof}

Let $(M,\mathscr{F})$ be a closed Lorentz-Finsler space. Over every relatively compact coordinate neighborhood  $U$ we can find a constant $a>0$ such that for every $x\in U$, $y\in T_xM$, $\mathscr{F}( y) \le a \sum_\mu \vert y^\mu\vert$. In fact, this is a consequence of the upper semi-continuity and positive homogeneity of $\mathscr{F}$. On a parametrized continuous causal curve $t \mapsto x(t)$, as each component $x^\mu(t)$ is absolutely continuous, each derivative $\dot x^\alpha$ is integrable and so $\mathscr{F}( \dot x)$ is integrable.

\begin{definition}
Let $(M,\mathscr{F})$ be a closed Lorentz-Finsler space. The (Lorentz-Finsler) length of a continuous causal curve $x\colon [0,1]\to M$, is $\ell(x)=\int_0^1 \mathscr{F}(\dot x) \dd t$ (it is independent of the parametrization).   The (Lorentz-Finsler) distance is defined by: for $(p,q) \notin J$, we set $d(p,q)=0$, while for $(p,q)\in J$
\begin{equation}
d(p,q)=\textrm{sup}_x \ell(x) ,
\end{equation}
where $x$ runs over the continuous causal curves which connect $p$ to $q$.
\end{definition}
Clearly, the reverse triangle inequality holds true: if $(p,q)\in J$ and $(q,r)\in J$ then
\begin{equation} \label{lfr}
d(p,r)\ge d(p,q)+d(q,r).
\end{equation}

\begin{theorem} \label{ddb}
Let $(M,\mathscr{F})$ be a locally Lipschitz proper Lorentz-Finsler space such that $\mathscr{F}(\p C)=0$. Then $d$ is lower semi-continuous.
\end{theorem}

\begin{proof}
Let $p,q \in M$. If $d(p,q)=0$, $d$ is lower semi-continuous at $(p,q)$. Thus let us assume $d(p,q)>0$. Let $\epsilon>0$, $\epsilon< d(p,q)$, and let $x\colon [0,1]\to M$, $x(0)=p$, $x(1)=q$, be a continuous causal curve such that $\ell(x)\ge d(p,q)-\epsilon/3>0$. and let $a,b\in (0,1)$, $a<b$, be such that $0<\int_0^a \mathscr{F}(\dot x) \dd t <\epsilon/3$, $0<\int_b^1 \mathscr{F}(\dot x) \dd t <\epsilon/3$. The subset of $[0,a]$ at which $x$ is differentiable with differential not lightlike is non-empty since $\int_0^a \mathscr{F}(\dot x) \dd t >0$.
If $a'<a$ is one such differentiability time then $\dot x(a')$ is timelike and an argument similar to that employed in Th.\ \ref{aam} shows that $x$ is chronal in any neighborhood of $x(a')$. As a consequence, $p\in I^-(x(a))$, $q\in I^+(x(b))$, so for every $p'\in I^-(x(a))$ and $q'\in I^+(x(b))$, $d(p',q')\ge \ell(x_{[a,b]})\ge \ell(x)-2 \epsilon/3\ge d(p,q)-\epsilon$.
\end{proof}

The next result is an improvement over \cite[Th.\ 4.24]{beem96}.

\begin{proposition} \label{imp}
Let $(M,\mathscr{F})$ be a proper Lorentz-Finsler space. If $d$ is upper semi-continuous then $(M,C)$ is reflective.
\end{proposition}
The proper condition is really necessary, for the condition on $d$ would be empty with $\mathscr{F}=0$.
\begin{proof}
Otherwise we can find $p,q$ such that $q \in \overline{J^+(p)}$ but $p \notin \overline{J^{-}(q)}$ (or dually), cf.\ Prop.\ \ref{diu}. Let $\gamma_n$ be causal curves starting from $p$ with endpoint $q_n \to q$. Taking $r \ll p$ so that $r \notin  \overline{J^{-}(q)}$, we have $d(r,q)=0$ but if $\sigma$ is a timelike curve connecting $r$ to $p$ we have $d(r,q_n)\ge l(\sigma)>0$ (the second inequality is due to the proper condition on the Lorentz-Finsler space which implies that $\mathscr{F}$ is positive on $\mathrm{Int} C$), so $d$ is not upper semi-continuous as can be seen taking the limit $(r, q_n)\to (r,q)$.
\end{proof}

The study of  $(M,\mathscr{F})$ passes through the study of the causality of $(M^\times, C^\times)$.
\begin{proposition} \label{maq}
Let $(M,\mathscr{F})$ be a closed Lorentz-Finsler space.
If $x\colon I\to M$, $t\mapsto x(t)$, is a continuous causal curve then for every $r\in \mathbb{R}$
\[
t \mapsto (x(t), r\pm \ell(x\vert_{[0,t]}))
\]
is a lightlike continuous causal curve on $(M^\times,C^\times)$ with starting point $(x(0), r)$. Moreover, every parametrized continuous causal curve on $(M^\times,C^\times)$ reads $X(t)=(x(t),r(t))$ where $x$ is a parametrized continuous causal curve on $(M,C)$ and $\vert r(t)-r(0)\vert\le \ell(x\vert_{[0,t]})$ for every $t$.
The causal future of $(M^\times,C^\times)$ satisfies
\begin{equation} \label{kki}
J^\times\subset\{((p,r),(p',r'))\colon (p,p')\in J\  \textrm{ and } \vert r'-r\vert\le d(p,p') \}.
\end{equation}
If for every $(p,p')\in J$ there is a continuous causal curve $x$ which maximizes $\ell$, i.e.\ $\ell(x)=d(p,p')$,  the inclusion (\ref{kki}) is actually an equality.
\end{proposition}

\begin{proof}
The derivative of the curve $X$ in display is $(\dot x, \mathscr{F}(\dot x))$ a.e., which is $C^\times$-lightlike, thus $X$ is a continuous causal curve.

Let $X(t)=(x(t),r(t))$ be a parametrized continuous causal curve on $(M^\times, C^\times)$ then its projection $x$ to $M$ is also a parametrized continuous causal curve. In fact the projection to $M$ is locally Lipschitz and the composition $g\circ f$, with $f$ AC and $g$ locally Lipschitz, is AC. As a consequence, $x(t)$ is absolutely continuous and by the definition of $C^\times$, $\dot x\in C$ and  $\vert \dot r\vert \le \mathscr{F}(\dot x)$. Thus $x$ is a continuous causal curve and $\vert r(t)-r(0)\vert \le \ell(x\vert_{[0,t]})$. The inclusion $\subset$ in the last statement follows easily from the previous results. Let us prove the equality statement. Let $(p',r')$ be such that $(p,p')\in J$ and $\vert r'-r\vert\le d(p,p')$. 
Suppose that for every $(p,p')\in J$ there is a continuous causal curve which maximizes $\ell$, and let  $x:[0,1]\to M$ be a parametrized continuous causal curve connecting $p$ to $p'$ such that $\ell(x)=d(p,p')$. Suppose without loss of generality that $r'\ge r$ the other case being analogous.
Then $X(t)=\big(x(t), r(t)\big)$, with $r(t)= r+\frac{r'-r}{d(p,p')}\ell(x\vert_{[0,t]})$, is a continuous causal curve (because $0\le \dot r = \frac{r'-r}{d(p,p')}\mathscr{F}(\dot x)\le \mathscr{F}(\dot x)$ a.e.\ ) on $(M^\times,C^\times)$ which connects $(p,r)$ to $(p',r')$. The fraction must be replaced by 0 if $r'=r$ or $d(p,p')=0$.
\end{proof}

The uniform convergence in the next proposition might be defined with respect to an auxiliary Riemannian metric $h$ on $M$, but it is really  independent of it (this is the same convergence appearing in the limit curve theorem \ref{main}). Usually  the next result is applied with $\mathscr{F}_n=\mathscr{F}$, $C_n=C$.

\begin{theorem} \label{upp} (Upper semi-continuity of the length functional)\\
Let $(M,\mathscr{F})$ and $(M,\mathscr{F}_n)$ be closed Lorentz-Finsler spaces. Let $x_n\colon [a_n,b_n]\to M$, be continuous $C_n$-causal curves, parametrized with respect to $h$-arc length, which converge uniformly on compact subsets to $x\colon [a,b]\to M$, $a_n\to a$, $b_n\to b$, where for every $n$, $C_{n+1}\le C_n$, $\mathscr{F}_{n+1} \le \mathscr{F}_{n}\vert_{C_{n+1}}$, and $C=\cap_n C_n$, $\lim_{n\to \infty} \mathscr{F}_n=\mathscr{F}$. Then $x$ is a continuous $C$-causal curve and
\[
\limsup_n \ell_n(x_n)\le \ell(x).
\]
\end{theorem}

\begin{proof}
Let us pass to a subsequence, denoted in the same way, such that \[\limsup_n \ell_n(x_n)=\lim_n \ell_n(x_n).\]
The curves $X_n(t)=(x_n(t), \int_{a_n}^t \mathscr{F}_n(\dot x_n(s)) \dd s)$ are continuous $C_n^\times$-causal curves on $(M^\times, C_n^\times)$ (notice that $\mathscr{F}_n$ enters in the definition of $C_n^\times$). The assumptions imply that $C_{n+1}^\times\subset C_n^\times$ and $C^\times=\cap_n C_n^\times$, thus by applying the limit curve theorem \ref{main}  in $M^\times$ we get that there is a limit continuous $C^\times$-causal curve $X(t)=(\tilde x(t),r(t))$, $r(0)=0$, to which a subsequence of $X_n$, denoted in the same way, converges uniformly. Consequently, $\tilde x=x$ and $\lim_n \ell_n (x_n)=r(b)$. But the $C^\times$-causality condition for $X$ reads $\vert \dot r\vert \le \mathscr{F}(\dot x)$ a.e.,  thus integrating $r(b)\le \ell(x)$, which gives the desired result.
\end{proof}

\begin{proposition} \label{ppp}
Let $(M,\mathscr{F})$ be a globally hyperbolic closed Lorentz-Finsler space, then  $(M^\times,C^\times)$ is a globally hyperbolic closed cone structure and $d$ is finite  and bounded on compact subsets of $M\times M$.
\end{proposition}

\begin{proof}
In the proof of Prop.\ \ref{maq} we have shown that the projection to $M$ of a continuous causal curve on $M^\times$ is itself a continuous causal curve. Thus $(M^\times, C^\times)$ must be non-imprisoning since the projection of a continuous causal curve imprisoned in a compact set would give a continuous causal curve imprisoned in the projection of the compact set.

Let us prove that $d$ is bounded on compact subsets of $M\times M$, from which it follows that it is finite. Let $K_1, K_2\subset M$ be compact sets, and let us consider the causally convex compact set $K=J^+(K_1)\cap J^-(K_2)$.
Let $h$ be a Riemannian metric and let $S$ be the unit $h$-sphere bundle over $K$. Since $\mathscr{F}$ is upper semi-continuous it reaches a maximum over $S$. By rescaling $h$ if necessary we can let this maximum be less than 1. Thus for every $v\in TK$, $\mathscr{F}(v)\le \Vert v\Vert_h$. The $h$-arc length of the continuous causal curves connecting $K_1$ to $K_2$ is bounded. This fact follows from strong causality and from the fact that $K$ can be finitely covered by  the local non-imprisoning causally convex neighborhoods constructed in Prop.\ \ref{iiu} and Th.\ \ref{dao}. Since the length of the causal curves connecting $K_1$ to $K_2$ is  bounded, we have that $d(K_1,K_2)$ is finite.

Now due to Eq.\ (\ref{kki}),
\begin{align*}
(J^\times)^+((p,a)) \cap J^-((q,b))&\subset \big\{(p',r')\colon p'\in J^+(p)\cap J^-(q) \textrm{ and } \vert r'-a\vert \le d(p,p') \\
&\qquad \textrm{ and } \vert r'-b\vert \le d(p',q) \big\} .
\end{align*}
Clearly we have $\textrm{max}\{ d(p,p'),d(p',q) \}\le d(p,q)$, thus given a   compact set $K$, and a compact subset of the real line $I$ there is a compact subset of $M\times M$ which contains the set in display for every  $p,q \in K$  and $a, b \in I$. We conclude that  $(M^\times,C^\times)$ is globally hyperbolic.
\end{proof}

The next result generalizes previous improvements \cite{samann16,galloway17} of the classical Avez-Seifert theorem \cite{hawking73} in that it does require neither the roundness of the cones nor the continuity of the cone distribution. In fact, even $\mathscr{F}$  need not be continuous.
As in Tonelli's theorem \cite[Th.\ 3.7]{buttazzo98}   $-\mathscr{F}$ is convex in the fiber variables and  lower semi-continuous, however it is not superlinear and its domain is a cone distribution.

\begin{theorem} \label{sww} (Generalization of the Avez-Seifert theorem)\\
Let $(M,\mathscr{F})$ be a globally hyperbolic closed Lorentz-Finsler space, then $\ell$ is maximized, namely for every $(p,q)\in J$ we can find a continuous causal curve $x\colon [0,1]\to M$, $p=x(0)$, $q=x(1)$, such that $\ell(x)=d(p,q)$.
\end{theorem}

So under these assumptions  equality holds in (\ref{kki}).
\begin{proof}
We know that $d(p,q)$ is finite. Let $x_n$ be a sequence of continuous causal curves connecting $p$ to $q$ such that $\ell(x_n)\ge d(p,q)-\epsilon_n$, with $\epsilon_n \to 0^+$. By the limit curve theorem a subsequence, which we  denote in the same way, converges uniformly to a continuous causal curve $x$, and by Th.\  \ref{upp} $d(p,q)= \limsup_n \ell(x_n)\le \ell(x)$, thus $\ell(x)=d(p,q)$.
\end{proof}

\begin{theorem} \label{ddc}
Let $(M,\mathscr{F})$ be a locally Lipschitz proper Lorentz-Finsler space such that $\mathscr{F}(\p C)=0$. Let $(p,q)\in J$ be such that $d(p,q)>0$, then for every  $R$ such that $0<R<d(p,q)$ there is a timelike curve $ x$ with the same endpoints such that $\ell( x)>R$.
\end{theorem}

\begin{proof}
By definition of Lorentz-Finsler distance we can find a continuous causal curve $\check x$ connecting $p$ to $q$,  such that $0<R<\ell(\check x)\le d(p,q)$. The continuous $C^\times$-causal curve given by $\check X(t)=(\check x(t), \check r(t))$ with $\check r(t)=  \frac{R}{\ell(\check x)}\ell(\check x\vert_{[0,t]})$, connects $P=(p,0)$ with $Q=(q,R)$ and has tangent $V=(\dot{\check{x}}, \frac{R}{\ell(\check x)} \mathscr{F}(\dot{\check{x}}))$ almost everywhere. Notice that $\mathscr{F}(\dot{\check{x}})>0$ for some $t$ because $\ell(\check x)>0$, thus by the proper condition $\dot{\check{x}}\in \mathrm{Int} C$ at that $t$ which implies that the tangent $V$ is timelike and hence
 (Th.\ \ref{aam}) that $\check X$ is chronal so the endpoints $P$ and $Q$ are connected by a $C^\times$-timelike curve $ X=( x, r)$ (Th.\ \ref{soa}). The $C^\times$-timelike condition implies $\dot r<\mathscr{F}(\dot x)$, thus integrating $R<\ell(x)$. Hence the projection $ x$ is the searched timelike curve.
\end{proof}

\begin{theorem} (Local causal geodesic connectedness)\\
Let $(M,\mathscr{F})$ be a closed Lorentz-Finsler space. Every point admits an arbitrarily small globally hyperbolic neighborhood $U$ such that if $(p,q)\in J(U)$ then there is a continuous causal curve $x$ contained in $U$ connecting $p$ to $q$ such that $\ell(x)=d_U(p,q)$ where $d_U$ is the Lorentz-Finsler distance of $(U,\mathscr{F}\vert_U)$.
\end{theorem}

The neighborhood $U$ coincides with that constructed in Prop.\ \ref{iiu}. If $(M,C)$ is strongly causal then the neighborhood can be chosen causally convex, in which case $d_U=d\vert_{U\times U}$.

\begin{proof}
Let $U$ be the neighborhood constructed in Prop.\ \ref{iiu} (see also Th.\ \ref{dao}). It is clearly globally hyperbolic, thus the theorem follows from Th.\ \ref{sww}.
\end{proof}

\begin{definition}
A continuous causal curve $x\colon I\to M$ for which $d(x(a),x(b))=\ell(x\vert_{[a,b]})$, for every $a<b$, is said to be {\em maximizing}.
\end{definition}

Due to the reverse triangle inequality if $I=[c,d]$ then it is sufficient to check  $d(x(c),x(d))=\ell(x\vert_{[c,d]})$.

For locally Lipschitz Lorentzian metrics Graf and Ling \cite{graf18} have recently proved that maximizing continuous causal curves are either almost everywhere timelike or almost everywhere lightlike (in the latter case the tangent cannot be  timelike at any point due to Th.\  \ref{aam}). In a different work we shall show that this result can be suitably generalized to the Lorentz-Finsler case.

The next result can be useful in order to frame closed Lorentz-Finsler spaces into the general theory of (Lorentzian) length spaces.

\begin{theorem}
Let $(M,\mathscr{F})$ be a strongly causal closed Lorentz-Finsler space and let $x\colon [a,b] \to M$ be a continuous causal curve. Then
\begin{equation} \label{leg}
l(x)=\inf\sum_i d(x(t_i), x(t_{i+1})),
\end{equation}
where the infimum is over all the partitions $a=t_0<t_1<\cdots<t_n=b$.
\end{theorem}

For the $C^2$ theory this result can be found in \cite{kunzinger18,minguzzi18b} where convex normal neighborhoods are used. In the $C^0$ Lorentzian metric theory Kunzinger and S\"amann \cite[Lemma 5.10]{kunzinger18} prove that the ``sup inf'' operation on the Lorentzian distance is involutive as required for length spaces, yet they do not prove that the length defined through the right-hand side of Eq.\ (\ref{leg}) corresponds to  the usual integral definition. This  is a piece of additional information given by this theorem and follows from the fact that we were able to prove the upper semi-continuity of the length functional without using convex neighborhoods.
Notice that without additional conditions $(M,\mathscr{F})$  would not be a Lorentzian length space according to their definition since $d$ might not be lower semi-continuous.

\begin{proof}
The inequality $\le$ is clear, so we have only to prove the other direction.
The image of the curve can be covered by a finite number of causally convex globally hyperbolic neighborhoods $\{C_i\}$ with the properties of Prop.\ \ref{iiu}. For some partition the consecutive points $\{x(t_i),x(t_{i+1})\}$ belong to $C_{j(i)}$ and so can be joined by a continuous causal curve  $\sigma_i$ included in $C_j$ such that $l(\sigma_i)=d_{C_j}(x(t_i),x(t_{i+1})=d(x(t_i),x(t_{i+1})$, where the last equality is due to causal convexity. The infimum in Eq.\ (\ref{leg}) can be restricted to piecewise maximizing continuous causal curves for which the consecutive corners belong to some $C_k$, for by increasing the number of corners to a piecewise continuous causal curves  one can only decrease the Lorentzian length (due to the reverse triangle inequality). By the same observation we can restrict the right-hand side to partitions that share a particular subpartition such that the consecutive points of the subpartition belong to some $C_k$. Thus we can really work out the proof in just one causally convex globally hyperbolic neighborhood $V$ chosen as in Prop.\ \ref{iiu} where by strong causality $d=d_V$, and where $x^0$ provides a Lipschitz parametrization for all continuous causal curves with image in $V$. Let $h$ be a complete Riemannian metric and let $d^h$ be its distance. Let $\tilde x$ denote $x$ reparametrized with $x^0$. For every neighborhood of the image of $x$ of $d^h$-radius $\epsilon$ we can find a piecewise maximizing continuous causal curve in the neighborhood (because the image of $x$ can be covered by arbitrarily small globally hyperbolic causally convex neighborhoods as in Prop.\ \ref{iiu}). That means that we can find a sequence of $x^0$-parametrized piecewise maximizing continuous causal curve $x_n$ that converges uniformly on compact subsets to $\tilde x$. By the upper semi-continuity of the length functional (Th.\ \ref{upp}) $\inf\sum_i d(x(t_i), x(t_{i+1}))\le \limsup_n \ell(x_n)\le \ell(x)$.
\end{proof}

\begin{definition}
A {\em (future) causal geodesic} on $(M,C)$, is a continuous causal curve which is the projection of a (resp.\ future) lightlike geodesic  on $(M^\times, C^\times)$.
Similarly, a {\em  causal bigeodesic} is the projection of a lightlike bigeodesic.
\end{definition}
Again a future or past causal geodesic is a causal geodesic. The converse holds for locally Lipschitz proper Lorentz-Finsler spaces. Notice that every future lightlike geodesic $x(t)$ is a future causal geodesic (and similarly for the notion of lightlike bigeodesic), just consider $(x(t),0)$. However, its length can be different from zero (unless the cone structure is $C^0$ and $\mathscr{F}=0$ on $\p C$, cf.\ Th. \ref{aam}) and even in that case it might not maximize the Lorentz-Finsler distance between any pair of its points.

\begin{proposition}
Let $x\colon [0,1]\to M$ be a continuous causal curve such that $d(x(0),x(1))=\ell(x)$, then $x$ is a causal bigeodesic.
\end{proposition}

Actually the assumption can be weakened to its local version.

\begin{proof}
Let us consider the continuous $C^\times$-causal curve $X(t)=\big(x(t),\ell(x\vert_{[0,t]})\big)$. By Prop.\ \ref{maq} the point $\big(x(1),d(x(0),x(1))+\epsilon\big)$, $\epsilon>0$, cannot be reached by a continuous $C^\times$-causal curve starting from $X(0)$. Thus $X$ is a future lightlike geodesic.  Similarly, $\big(x(0),-\epsilon\big)$ cannot be the starting point of a continuous $C^\times$-causal curve which reaches  $X(1)=\big(x(1),d(x(0),x(1))\big)$, thus $X$ is a past lightlike geodesic.
\end{proof}

\begin{theorem} \label{xux}
Let $(M,\mathscr{F})$ be a proper Lorentz-Finsler space and let $S$ be an acausal topological hypersurface. Then for every $q\in  D^+(S)$, there is a causal bigeodesic $x$ connecting some point $p\in S$ to $q$ such that $d(p,q)=d(S,q)=\ell(x)$.
\end{theorem}

\begin{proof}
By Th.\ \ref{mmm} $D(S)$ is open, causally convex and globally hyperbolic.
By Th.\ \ref{jbu} $J^-(q)\cap J^+(S)$ is a compact subset of $D(S)$, thus $d$ restricted to $D(S)\times D(S)$ is finite. Let $x_n$ be a sequence of continuous causal curves connecting $p_n\in S$ to $q$, such that $\ell(x_n)\to  d(S,q)$, then up to extracting a subsequence, it converges uniformly to some continuous causal curve $x$ of starting point $p\in S$ and $d(S,q)=\limsup_n \ell(x_n)\le \ell(x)$, by Th.\ \ref{upp}, thus $\ell(x)=d(S,q)$.
\end{proof}

\subsection{Stable distance and stable spacetimes} \label{mvb}

In Sec.\ \ref{ngd} we have defined the notion of Lorentz-Finsler space $(M,\mathscr{F})$.
We shall write $\mathscr{F}'>\mathscr{F}$, with no mention of the cone domains, if
$(M,\mathscr{F})$ is a closed Lorentz-Finsler space, $(M,\mathscr{F}')$ is a proper Lorentz-Finsler space, and $C'{}^\times > C^\times$, which implies
$C'>C$ and $\mathscr{F}'>\mathscr{F}$ on $C$.
%
%
%
%
%

The next result follows from a construction similar to that used in Prop.\ \ref{ohg} but framed in $M^\times$.
\begin{proposition} \label{cso}
Given a closed Lorentz-Finsler space $(M,\mathscr{F})$ there is a locally Lipschitz proper Lorentz-Finsler space $(M,\mathscr{F}')$, $\mathscr{F}'>\mathscr{F}$.
Given a closed Lorentz-Finsler space $(M,\mathscr{F})$ and a locally Lipschitz proper cone structure $C'>C$, there is a locally Lipschitz proper Lorentz-Finsler space $(M,\mathscr{F}')$, $\mathscr{F}'>\mathscr{F}$. Given a closed Lorentz-Finsler space $(M,\mathscr{F})$, and a $C^0$ proper Lorentz-Finsler space $\check { \mathscr{F}}>\mathscr{F}$,  there is a locally Lipschitz proper Lorentz-Finsler space $\mathscr{F}'$, $\mathscr{F}<\mathscr{F}'<\check{ \mathscr{F}}$.
\end{proposition}

In the next proofs given the closed Lorentz-Finsler space $(M,\mathscr{F})$ all the other introduced proper Lorentz-Finsler spaces $(M,\mathscr{F}')$, $\mathscr{F}'>\mathscr{F}$, will be  locally Lipschitz.
Next we define a novel distance.

\begin{definition}
Let $(M,\mathscr{F})$  be a closed Lorentz-Finsler space. We define the {\em stable distance} $D\colon M\times M\to [0,+\infty]$ as follows. For $p,q\in M$
\begin{equation}
D(p,q)=\mathrm{inf}_{\mathscr{F}'>\mathscr{F}} d'(p,q),
\end{equation}
where $d'$ is the Lorentz-Finsler distance for the locally Lipschitz proper  Lorentz-Finsler space $(M,\mathscr{F}')$.
\end{definition}
Observe that for $\mathscr{F}_1>\mathscr{F}_2$ we have $d_1\ge d_2$, and the set $\{\mathscr{F}':\mathscr{F}'>\mathscr{F}\}$ is directed in the sense that if $\mathscr{F}_1>\mathscr{F}$ and $\mathscr{F}_2>\mathscr{F}$ there is  $\mathscr{F}_3>\mathscr{F}$ such that $\mathscr{F}_3<\mathscr{F}_1,\mathscr{F}_2$.


\begin{theorem} \label{upq}
Let $(M,\mathscr{F})$ be a closed Lorentz-Finsler space. The following properties hold true:
\begin{itemize}
\item[(a)] If $(p,q)\notin J_S$, then $D(p,q)=0$,
\item[(b)] Suppose that $(M,\mathscr{F})$ is a proper Lorentz-Finsler space. If $(p,q) \in \mathrm{Int} J_S$ or $q\in \mathrm{Int} J_S^+(p)$ or $p\in \mathrm{Int} J_S^{-}(q)$, then $D(p,q)>0$,
\item[(c)]    If $(p,q)\in J_S$ and $(q,r)\in J_S$, then $(p,r)\in J_S$ and (see also Fig.\ \ref{best})
\begin{align*}
 D(p,q)+D(q,r)\le D(p,r),
\end{align*}
\item[(d)] D is upper semi-continuous,
\item[(e)]
Suppose that  $(M,\mathscr{F})$ is a  proper Lorentz-Finsler space. If $D=d$ then the spacetime is reflective (so causally continuous if distinguishing),
\item[(f)] $d \le D$,
\item[(g)] If $(M,C)$ is globally hyperbolic then $D=d$,
\item[(h)] If $(M,C)$ is stably causal then for every $p\in M$ there is a globally hyperbolic $J_S$-causally convex neighborhood $U$, such that $D\vert_{U\times U}=d\vert_{U\times U}=d^U=D^U$, where $d^U$ is the Lorentz-Finsler distance of the spacetime $U$ and similarly for $D^U$. In particular $D(p,p)=0$.
\end{itemize}
\end{theorem}

\begin{proof}
(a). By  definition of Seifert relation, if $(p,q)\notin J_S$ then there is $\hat{\mathscr{F}}>\mathscr{F}$ such that $\hat C>C$ and $(p,q) \notin \hat J$. Then for every $\mathscr{F}'$ such that $\mathscr{F}<\mathscr{F}'<\hat{\mathscr{F}}$, $d'(p,q)=0$, hence the claim.

(b). If $(p,q) \in \mathrm{Int} J_S$ pick $(p',q')$ sufficiently close to $(p,q)$ and such that $p\ll p'$, $q'\ll p$, then $(p',q')\in J_S$. For every $\mathscr{F}'>\mathscr{F}$, $d'(p,q)$ is larger than the sum of the $\mathscr{F}$-Lorentz-Finsler lengths of the $C$-timelike curves connecting $p$ to $p'$ and $q'$ to $q$, which are positive and independent of $\mathscr{F}'$, thus the claim. The proofs with the assumptions  $q\in \mathrm{Int} J_S^+(p)$ or $p\in \mathrm{Int} J_S^{-}(q)$ are analogous, but there is only one timelike curve.

(c). Let $\mathscr{F}'>\mathscr{F}$, then   $(p,q)\in J'$, $(q,r)\in J'$ and
\begin{align*}
D(p,q)+D(q,r)\le d'(p,q)+d'(q,r)\le d'(p,r),
\end{align*}
where we used the Lorentz-Finsler reverse triangle inequality, cf.\ Eq.\ (\ref{lfr}).
Since the equation in display holds for every  $\mathscr{F}'>\mathscr{F}$, taking the infimum we obtain the desired result.

(d).  We can assume that $D(p,q)$ is finite. Suppose $D$ is not upper semi-continuous at $(p,q)$, then there is $\epsilon>0$ and a sequence $(p_n,q_n)\to (p,q)$ such that $D(p_n,q_n) \ge D(p,q)+4\epsilon$.
 By definition of $D(p,q)$ we can find $\mathscr{F}'>\mathscr{F}$ such that for every continuous $C'$-causal curve $\gamma$ connecting $p$ to $q$, $\ell'(\gamma)\le d'(p,q)\le  D(p,q)+\epsilon$.
 Let $\mathscr{F}_n\to \mathscr{F}$, be a sequence such that  $\mathscr{F} <\mathscr{F}_{n+1}<\mathscr{F}_n<\mathscr{F}'$. Since $d_n(p_n, q_n) \ge D(p_n,q_n)$, for every $n$ we can find $\gamma_n$ continuous $C_n$-causal curve connecting $p_n$ to $q_n$ such that $\ell_n(\gamma_n)\ge D(p_n,q_n)-\epsilon $. Applying the limit curve theorem  \ref{main}  to $\{\gamma_n\}$ we obtain the existence of two continuous $C$-causal limit curves $\sigma^q$ ending at $q$ and $\sigma^p$ starting at $p$ (possibly inextendible in the other direction) to which a subsequence (denoted in the same way) $\gamma_n$ converges uniformly over compact subsets. Let $p'\in \sigma^p$ be chosen so close to $p$ that,  with the obvious meaning of the notation, $\sigma^p_{p\to p'}$   belongs to a neighborhood $U_p$ such that the
 $\mathscr{F}'$-length (and hence the $\mathscr{F}_n$-length)  of any $C'$-causal curve contained in $U_p$ is less than $\epsilon/2$ (recall the local non-imprisoning result Prop.\ \ref{iiu} and the fact that a Riemannian metric can be found such that $\mathscr{F}'(y)\le \Vert y\Vert_h$ for every $y\in TU_p$).
 Similarly choose $q' \in U_q$ with the analogous criteria. Let $p'_n\in \gamma_n$ be such that $p'_n\to p'$ and similarly for $q'_n\to q$.  Since the limit curves are $C$-causal $(p,p')$ and $(q',q)$ belong to $I_{C'}(U_p)$, moreover, as the chronology relation is open, $(p,p'_n)\in I_{C'}(U_p)$ and $(q'_n,q)\in I_{C'}(U_q)$, thus we can go  from $p$ to $p'_n$ follow $\gamma_n$ to $q'_n$ and then go from $q'_n$ to $q$, all along a continuous $C'$-causal curve $\eta_n$. Notice that
  $\ell'(\eta_n)\ge \ell'(\eta_n\vert_{p'_n \to q'_n}) \ge \ell_n(\gamma_n\vert_{p'_n \to q'_n})\ge \ell_n(\gamma_n) -\epsilon$. Putting everything together
 \[
 \ell'(\eta_n)\ge \ell_n(\gamma_n)-\epsilon\ge D(p_n,q_n)-2\epsilon \ge  D(p,q)+2\epsilon,
 \]
 which gives a contradiction, since we know that  for every continuous $C'$-causal curve $\gamma$ connecting $p$ to $q$, $\ell'(\gamma)\le D(p,q)+\epsilon$.

(e). Since $D$ is upper semi-continuous,  $d$ is upper semi-continuous which implies reflectivity (Th.\ \ref{imp}).

(f). We can assume $d(p,q)>0$, the other case being trivial. Whenever $\mathscr{F}'>\mathscr{F}$, we have $\ell'(\gamma)\ge \ell(\gamma)$ for every $C$-causal curve, so the statement follows.

(g). We know that $J_S=J$ (Th.\ \ref{mom}), so we have only to show that for  $(p,q)\in J$, $D(p,q)=d(p,q)$.  Let $\{C_k\}$ be a sequence as in Prop.\ \ref{sqd}, $C<C_{k+1}<C_k$, $C=\cap_k C_k$, where by the stability of global hyperbolicity we can assume that $C_1$ and hence every $C_k$ is globally hyperbolic. For every $n>0$ we can find $k_n$ such that $D(p,q)\le d_{k_n}(p,q)\le D(p,q)+1/n$. Indeed, if not  then $d_k(p,q)>D(p,q)+1/n$ for every $k$, thus we can find a sequence of continuous $C_k$-causal curves $\sigma_k$ such that $\ell_k(p,q)>D(p,q)+1/n$. By the global hyperbolicity of $C_1$ it converges to a continuous $C$-causal curve $\sigma$, and by the upper semi-continuity of the length functional, $\ell(x) \ge D(p,q)+1/n\ge d(p,q)+1/n$, a contradiction. Thus we can find a continuous $C_{k_n}$-causal curve $x_n$ connecting $p$ to $q$ such that $d_{k_n}(p,q)-1/n \le \ell_{k_n}(x_n)\le d_{k_n}(p,q)$, thus $\vert \ell_{k_n}(x_n)- D(p,q)\vert \le 1/n$. By the limit curve theorem \ref{main} and the non-imprisoning property of global hyperbolicity there is a limit continuous $C$-causal curve $x$ connecting $p$ to $q$.
By the upper semi-continuity of the length functional, Th.\ \ref{upp},  $D(p,q)=\lim_n \ell_{k_n}(x_n)\le \ell(x)\le d(p,q)\le D(p,q)$, thus $D(p,q)=d(p,q)$.

(h).  Let $(V,C\vert_V)$ be a globally hyperbolic neighborhood of $p$. Since $(M,C)$ is stably causal there is $C'>C$ such that $(M,C')$ is stably causal and hence strongly causal. Let $U\subset V$, $p\in U$, be a relatively compact $J'$-causally convex set, then since $J\subset J_S\subset J'$, it is also $J_S$-causally convex and  $J$-causally convex. Thus $(U, C\vert_U)$ is globally hyperbolic and $d^U=d\vert_{U\times U}$. But since $U$ is $J'$-causally convex $D^U=D\vert_{U\times U}$. Finally, the equality $D^U=d^U$ follows from (g).
\end{proof}

\begin{figure}[ht]
\begin{center}
 \includegraphics[width=3.5cm]{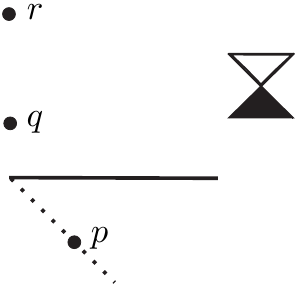}
\end{center}
\caption{Minkowski's 1+1-dimensional spacetime with a spacelike segment removed. The usual reverse triangle inequality does not hold $d(p,q)=d(p,r)=0$, $d(q,r)>0$, hence $d(p,q)+d(q,r) \nleq d(p,r)$, because $ (p,q)\notin J$. However, the reverse triangle inequality applies for $D$ because $(p,q),(q,r)\in J_S$. So the good properties of $D$ find wider applicability than those of $d$.} \label{best}
\end{figure}

\begin{definition}
A {\em stable} closed Lorentz-Finsler space  $(M,\mathscr{F})$ is a  stably causal closed Lorentz-Finsler space such that $D<\infty$.
\end{definition}

This terminology is well posed due to the following result.

\begin{theorem} \label{con} (Stable  spacetimes are stable)\\
Let $(M,\mathscr{F})$ be a stable closed Lorentz-Finsler space, then there is $\bar{\mathscr{F}}>\mathscr{F}$ such that $(M,\bar{\mathscr{F}})$ is a stable locally Lipschitz proper Lorentz-Finsler space. In particular, $d\le D\le \bar d\le \bar D<+\infty$. Finally, let $\epsilon>0$ and let $K$ be a compact set, then $\bar{\mathscr{F}}$ can be chosen so that $\bar D\vert_{K\times K}-D\vert_{K\times K}\le \epsilon$.
\end{theorem}

Thus the first sentence implies that the finiteness of $D$ is stable, while the second sentence states that  $D$ itself is stable.

\begin{proof}
By stable causality we can find $\hat{\mathscr{F}}>\mathscr{F}$ such that  $(M,\hat{\mathscr{F}})$ is a stably causal locally Lipschitz proper Lorentz-Finsler space.
We need some preliminary results.

Result 1:  Under the theorem's assumptions, given a compact set $K$ we can find $\check{\mathscr{F}}$, $\mathscr{F}<\check{\mathscr{F}}<\hat{\mathscr{F}}$, such that $\check{d}\vert_{K\times K}<R$ for some $R(K)>0$.

Proof of result 1. Let $\hat K$ be a compact set which contains $K$ in its interior. Since $D$ is finite and upper semi-continuous  there is $R>0$ such that $D\vert_{\hat K\times \hat K}<R$.
Let $p,q\in \hat K$. Since $D(p,q)<R<+\infty$, by definition of $D$ we know that there is $\mathscr{F}_{pq}$, $\mathscr{F}<\mathscr{F}_{pq}<\hat{\mathscr{F}}$, such that $d_{pq}(p,q)<R$.
Consider the open sets $I^+_{pq}(p)\times I^-_{pq}(q)$ for $p,q\in \hat K$. They cover $K\times K$ in fact given $(p,q)\in K\times K$, we can always find $p'\in J^-(p)\backslash\{p\}\cap \hat K$ and $q'\in J^+(q)\backslash\{q\}\cap \hat K$ (by Th.\ \ref{zar}) so by Th.\ \ref{dxp} $p\in I^+_{p'q'}(p')$, $q\in I^-_{p'q'}(q')$. By compactness of $K\times K$ we can find $(p_i,q_i)\in \hat K\times \hat K$, $1\le i\le s$,  such that writing  $I_{i}$ in place of $I_{p_iq_i}$ (and $d_i$, $\mathscr{F}_i$ in place of $d_{p_iq_i}$, $\mathscr{F}_{p_iq_i}$), $I^+_{i}(p_i)\times I^-_{i}(q_i)$ cover $K\times K$.
 Let $\check{\mathscr{F}}$ be such that $\mathscr{F}<\check{\mathscr{F}}< \mathscr{F}_{i}$ for every $1\le i\le s$. Every $(p,q)\in K\times K$ belongs to some element of the covering $(p,q)\in I^+_{i}(p_i)\times I^-_{i}(q_i)$.
  But now $\check d(p,q)<R$ otherwise
 \[
 d_{i}(p_i,q_i)\ge  d_i(p,q) \ge \check d (p,q)\ge R,
 \]
a contradiction. Result 1 is proved.

 %
%

Result 2: Under the theorem's assumptions, given a compact set $K$ we can find  $\check{\mathscr{F}}$, $\mathscr{F}<\check{\mathscr{F}}<\hat{\mathscr{F}}$,  such that every $\mathscr{F}'$, $\mathscr{F}<{\mathscr{F}}'<\hat{\mathscr{F}}$, such that  $\mathscr{F}<\mathscr{F}'<\check{\mathscr{F}}$ on $M\backslash \mathrm{Int} K$, has bounded distance $d'\vert_{K\times K}$.

Observe that the definition of the restricted distance $d'\vert_{K\times K}$ involves continuous  $C'$-causal curves that might escape $K$. Also observe that $\mathscr{F}'$ is bounded by $\hat{\mathscr{F}}$ on $K$, however  $\hat{\mathscr{F}}$ is independent of $K$.

Proof of result 2.
By Result 1 there is $\check{\mathscr{F}}$, $\mathscr{F}<\check{\mathscr{F}}<\hat{\mathscr{F}}$, such that $\check{d}\vert_{\tilde K\times \tilde K}<R$, where $\tilde K$ is a compact set which contains $K$ in its interior. But we can enlarge the cones in $\mathrm{Int} K$, and alter $\check{\mathscr{F}}$ to $\check{\mathscr{F}}'<\hat{\mathscr{F}}$, preserving the boundedness of $\check{d}'$ in $K\times K$. In order to prove this point, observe that by  the stable causality of $\hat{C}$, $K$ can be covered by a finite number $N$ of $\hat{C}$-causally convex neighborhoods contained in $\mathrm{Int} \tilde K$. Every continuous $\check{C}$-causal curve starting and ending in $K$ can escape $\tilde K$ but at most $N$ times since each time it reenters a different causally convex neighborhood. The $\check{\mathscr{F}}'$-length of continuous $\check{C}'$-causal curves contained in $\tilde K$ is bounded by a constant $B>0$ (use the existence of a Riemannian metric $h$ on $M$ such  that $\hat{\mathscr{F}}(y)\le \Vert y \Vert_h$ on $T\tilde K$, and  cover $\tilde K$ with a finite number of $\hat C$-causally convex non-imprisoning neighborhoods), thus such an alteration of $\check{\mathscr{F}}$ in $\mathrm{Int} K$ would nevertheless keep $\check{d}'\vert_{K\times K}$ bounded by $(B+R)N$. Result 2 is proved.

Let $h$ be an auxiliary complete Riemannian metric on $M$ and let $o\in M$. Let $K_n=\bar{B}(o,n)$ and let $\mathscr{F}_n$, $\mathscr{F}<\mathscr{F}_n<\hat{\mathscr{F}}$,  be the function $\check{\mathscr{F}}$ appearing in result 2 for the choice $K=K_n$. The sequence can be chosen so that $ \mathscr{F}_{n+1}<\mathscr{F}_n$. Let $\mathscr{F}'>\mathscr{F}$ be such that for every $n$, $\mathscr{F}'\vert_{K_n\backslash \mathrm{Int}  K_{n-1}}< \mathscr{F}_n\vert_{K_n\backslash \mathrm{Int} K_{n-1}}$.
 Let $p,q\in M$, then there is some $m$ such that $p,q\in \mathrm{Int} K_m$.  By $\mathscr{F}'<\mathscr{F}_{m+1}<\mathscr{F}_{m}$ on $M\backslash \mathrm{Int} K_m$ and the Result 2 we have $d'(p,q)<+\infty$. By the arbitrariness of $p$ and $q$, $d'$ is finite. Taking $\bar{\mathscr{F}}$ such that $ \mathscr{F}<\bar{\mathscr{F}}< \mathscr{F}'$ gives a finite $\bar{D}$.

For the last statement, for every $\mathscr{F}'>\mathscr{F}$, $D'\vert_{K\times K}$ is an upper semi-continuous finite function over a compact set. Its subgraph $S'$ is closed and hence compact and contains the compact set $S$, the subgraph of $D\vert_{K\times K}$. Notice that $\cap_{\mathscr{F}'>\mathscr{F}} S'=S$ because $D=\inf_{\mathscr{F}'>\mathscr{F}} D'$. Consider a compact neighborhood $E$ of the graph of $(D+\epsilon)\vert_{K\times K}$ which does not intersect $S$. Then $\cap_{\mathscr{F}'>\mathscr{F}} (S'\cap E)=\emptyset$, thus $\{(S'\cap E)^C\}$ form an open covering of $E$, thus there is a finite covering $\{(S_i\cap E)^C\}$ and taking $\mathscr{F}'$ so that  $\mathscr{F}<\mathscr{F}'<\mathscr{F}_i$ for every $i$, $S'\cap E=\emptyset$ which implies $D'\vert_{K\times K}<(D+\epsilon)\vert_{K\times K}$. \ \ \,
\end{proof}

\begin{theorem} \label{maa} (stable representatives in the stably causal conformal class) \\
Let $(M,\mathscr{F})$ be a stably causal closed Lorentz-Finsler space, then there is a smooth function $\alpha\colon M\to \mathbb{R}_+$ such that $(M,\alpha \mathscr{F})$ is stable.
\end{theorem}

Notice that given a stably causal closed cone structure, one can take $\mathscr{F}=0$ to get a stably causal closed Lorentz-Finsler space which will also be a stable  closed Lorentz-Finsler space.

\begin{proof}
By assumption there is $\hat C>C$ stably causal and hence strongly causal (Th.\ \ref{mwh}), moreover by Prop.\ \ref{cso} we  can actually find $(M,\hat{\mathscr{F}})$, $\hat{\mathscr{F}}>{\mathscr{F}}$, such that  $(M,\hat C)$ is a stably causal locally Lipschitz proper cone structure. Regard $M$ as the union of a countable number of compact shells  $K_n \backslash \mathrm{Int} K_{m-1}$, with $K_m$ a closed ball of radius $m$ and center $o\in M$ with respect to a complete Riemannian metric. The $m$-th shell is covered by a finite number $N_m$ of $\hat C$-causally convex non-imprisoning neighborhoods (Th.\ \ref{dao}), and any continuous $\hat C$-causal curve connecting the boundary of the $m$-th shell to itself and entirely contained in the $m$-th shell
 has bounded $\hat{\mathscr{F}}$-length $L_m$ since it is bound to escape every neighborhood of the covering (notice that there is a Riemannian metric $h$ such that $\hat{\mathscr{F}}(y)\le \Vert y\Vert_h$ for every $y\in \hat C$, and that the $h$-arc length is bounded on every neighborhood cf.\ Prop.\ \ref{iiu}). We can find $\alpha$ so that it is smaller than  $\frac{1}{N_m L_m 2^m}$ on the $m$-th shell. Every continuous $\hat C$-causal curve intersects the $m$-th shell in at most $N_m$ segments since each of them intersect at least one $\hat C$-causally convex neighborhood of the covering. From here it follows that the $\alpha \hat{\mathscr{F}}$-length of any continuous $\hat C$-causal curve is bounded by $2$, thus since $\alpha\hat{\mathscr{F}}>\alpha{\mathscr{F}}$, $(M,\alpha \mathscr{F})$ is stable.
\end{proof}

\begin{theorem} \label{mab} (globally hyperbolic spacetimes are stable)\\
Let $(M,\mathscr{F})$ be a globally hyperbolic closed Lorentz-Finsler space, then it is stable regardless the choice of $\mathscr{F}$.
\end{theorem}

\begin{proof}
By Th.\ \ref{mom}  there is a globally hyperbolic locally Lipschitz proper cone structure  $(M,C')$,   $C'>C$. So we can find a closed Lorentz-Finsler space  $(M,\mathscr{F}')$, $\mathscr{F}'>\mathscr{F}$,  such that $(M,C')$ is a globally hyperbolic locally Lipschitz proper cone structure, $C'>C$. By the
improved Avez-Seifert theorem (Th.\ \ref{sww})
the distance $d'(p,q)$ is attained by the $\mathscr{F}'$-length of some continuous $C'$-causal curve connecting $p$ to $q$. In particular, it is finite, which proves that $D<+\infty$.
\end{proof}

\subsection{Singularity theorems} \label{sin}
In this section we show that the causal content of some singularity theorems is preserved in the upper semi-continuous regularity framework. In the classical $C^2$ theory a singularity is just an incomplete geodesic. In the present context we do not have a notion of affine parameter at our disposal, however, the non-causal ingredients in the classical singularity theorems, such as affine parameter, energy and genericity conditions or divergence conditions,  might be seen as means to produce sets on the manifold with specific causality properties. Each singularity theorem has a core which  relates such causality concepts, and it is this type of result which is  preserved. Our generalization is therefore of a different nature with respect to that found in \cite{kunzinger15,kunzinger15b,graf17} where the authors assume the stronger $C^{1,1}$ regularity but  make sense of some other analytical objects entering the classical theorems. Also, we shall
 not recall the classical versions of the singularity theorems, nor  shall we explain in detail why our theorems provide the causality content of such statements. The reader might easily verify the correspondence by checking some classical references \cite{hawking73}.

We recall that a lightlike line is an inextendible causal curve for which no two points are $\mathring{J}$-related. For locally Lipschitz proper cone structures the condition is equivalent to achronality.

\begin{lemma} \label{rom}
Let $(M,C)$ be a causal closed cone structure. If there are no lightlike lines
then $(M,C)$ is strongly causal.
\end{lemma}

\begin{proof}
If $(M,C)$ is not strongly causal at $x$ then  there is  a non-imprisoning neighborhood $U \ni x $ as in Prop.\ \ref{iiu} and a sequence of
continuous causal curves $\sigma_n$ with endpoints $x_n,z_n$, with $x_n\to
x$, $z_n \to x$, not entirely contained in $U$. Let $B$, $\bar B\subset U$ be a coordinate ball of $x$. By the limit curve theorem there  are a future inextendible continuous causal curve $\sigma^x$ starting from $x$ and a past inextendible continuous causal curve $\sigma^z$ ending at $x$ such that for every $\tilde x\in \sigma^x$ and $\tilde z\in \sigma^z$, $(\tilde z,\tilde x) \in \bar J$. But $\sigma^x\circ \sigma^z$ is not a lightlike line so $\tilde x$ and $\tilde z$ can be chosen so that $(\tilde x,\tilde z) \in \mathring{J}$, thus there is a closed causal curve, a contradiction.
\end{proof}


Let us give a version of the theorem  we proved  in \cite{minguzzi07d}.

\begin{theorem} \label{mbx}
Let $(M,C)$ be a locally Lipschitz proper cone structure. If there are no lightlike lines
then $(M,C)$ is causally easy, thus
there is a time function.
\end{theorem}

\begin{proof}
Let $(p,q) \in \bar J$ and $(q,r)\in \bar J$. The are continuous causal curves $\sigma_k$ with endpoints $p_k\to p$, $q_k\to q$, and continuous causal curves $\gamma_s$ with endpoints $q'_s\to q$, $r_s\to r$. By the limit curve theorem either there is a limit continuous causal curve $\sigma$ connecting $p$ to $q$ or a past inextendible continuous causal curve $\sigma^q$ ending at $q$, such that for every $\tilde q\in \sigma^q$, $(p,\tilde q)\in \bar J$. Similarly, there is a  limit continuous causal curve $\gamma$ connecting $q$ to $r$ or a future inextendible continuous causal curve $\gamma^q$ starting from $q$, such that for every $\check q \in \gamma^q$, $(\check q,r)\in \bar J$. If $(p,q) \in  J$, taking $p'\ll p$, we have $p'\ll q$ due to
$I\circ J\cup J\circ I\subset I$ and hence $p'\ll r_n$, which implies $(p,r)\in \bar J$. Similarly, the assumption $(q,r)\in J$ gives $(p,r)\in \bar J$. It remains to consider the case $(p,q) \notin  J$ and $(q,r)\notin  J$.
Let $\eta$ be the inextendible causal curve obtained joining $\sigma^q$ and $\gamma^q$, as there are two points in $\eta$ such that $(\tilde q,\check q)\in I$,  we have that $(p,r)\in \bar J$. Thus $\bar{J}$ is transitive. By the lemma $(M,C)$ is strongly causal, thus $(M,C)$ is causally easy.
\end{proof}

The next stability result is interesting though it will not be used.

\begin{theorem}
Let $(M,C)$ be a closed cone structure which  does not have lightlike lines. There is a locally Lipschitz proper cone structure $\tilde C>C$, such that for every locally Lipschitz proper cone structure $C'$, $C<C'<\tilde C$, $(M,C')$  does not have lightlike lines.
\end{theorem}

\begin{proof}
Suppose not; there are two cases:  either (i) there is a compact set $K$ and a sequence $C_k$ with the properties of Th.\ \ref{sqd}  such that there is a $C_k$-lightlike line $\sigma_k$  intersecting $K$, or (ii) let  $C_k$ have the properties of Th.\ \ref{sqd},  for every compact set $K$ we can find $n(K)$ such that for all $C'$, $C<C'\le C_{n(K)}$, all $C'$-lightlike lines do not intersect $K$.

In case (i) by the limit curve theorem there is a subsequence, denoted in the same way, which converges to an inextendible continuous causal curve  $\sigma$, thus by the absence of lightlike lines there are $p,q\in \sigma$ such that $(p,q)\in \mathring{J}\subset \mathring{J}_k$. Let $p_k,q_k\in \sigma_k$, $p_k\to p$ and $q_k\to q$. For  sufficiently large $k$ we have $(p_k,q_k) \in \mathring{J}_k$. This is a contradiction because  $\sigma_k$ is a lightlike line.

In case (ii), let  $h$ be a  complete Riemannian metric on $M$, and $o\in M$, and let $n(m)\ge n(\bar B(o,m))$ be an increasing function. Let $\tilde C>C$ be such that for every $m$, $C_{n(m+1)}\le \tilde C\le  C_{n(m)}$ on  $\bar B(o,m) \backslash B(o,m-1)$ in such a way that $\tilde C=C_{n(m)}$ on $\p B(o,m-1)$  (here $\tilde C$ is built using  convex combinations as in Prop.\ \ref{doo}).
 Assume there
is $C'$, $C<C'<\tilde{C}$, such that $(M,C')$ has a lightlike
line $\gamma$. There is a minimum value of $m$ such that $\gamma
\cap \bar B(o,m)\ne \emptyset$. Let $\hat C$, $C<\hat C\le C'$, be coincident with $C'$  outside  $\bar B(o,m-1)$ and such that $\hat C\le C_{n(m)}$ on $\bar B(o,m-1)$.
The curve $\gamma$ is a $\hat C $-lightlike line  since it is $\hat C$-causal and $\hat J\subset J'$. But $\hat C\le C_{n(m)}$ on $M$, thus $\gamma$ cannot be a $\hat C $-lightlike line for it intersects $\bar B(o,m)$, a contradiction.
\end{proof}

%
%
%
%

Let us come to  Penrose's 1965 singularity theorem \cite{penrose65}. It was generalized to the round cone (Lorentzian) $C^{1,1}$ case  \cite{kunzinger15}, and to the non-round (Lorentz-Finsler) $C^2$ case  \cite{minguzzi15}.
We need a definition.

\begin{definition}
A future trapped set is a non-empty set $S$ such that $E^+(S)$ is compact.
\end{definition}

The next result, which does not seem to have been previously noticed, not even in the $C^2$ Lorentzian theory, will be very important as it will allow us to improve the differentiability assumption on the cone structure from `locally Lipschitz and proper' to `upper semi-continuous'.

\begin{theorem} \label{sto} (Stability of compact trapped sets)\\
Let $(M,C)$ be a non-imprisoning closed cone structure. Let $S$ be a compact set such that $\overline{E^+(S)}$ is compact. Then $E^+(S)$ is closed and there is a locally Lipschitz proper  cone structure $\tilde C>C$ such that for every locally Lipschitz proper  cone structure $C<\hat C <\tilde C$,  $\overline{\hat E^+(S)}$ is compact.

If $(M,C)$ is stably causal we can take $\tilde C$ stably causal and also for every $C<\hat C <\tilde C$, we have that  ${\hat E^+(S)}$ is compact.
\end{theorem}

\begin{proof}
Let $q\in \overline{E^+(S)}$ so there are $q_n \in E^+(S)$, $q_n\to q$. We cannot have $q\in I^+(S)$ otherwise for sufficiently large $n$, $q_n\in I^+(S)$, which is impossible. Thus we have only to prove that $q\in J^+(S)$. If $q\in S$ there is nothing to prove, so let us suppose $q\notin S$. Let $\sigma_n$ be a causal curve connecting $S$ to $q_n$, necessarily contained $E^+(S)$. Since $q\notin S$ the curves $\sigma_n$ do not contract to a point and by the limit curve theorem either there is a continuous causal curve connecting $S$ to $q$, and we are finished, or there is a future inextendible continuous causal curve $\sigma^q\subset \overline{E^+(S)}$. By the compactness of the last we have a contradiction with the non-imprisoning property.

As for the next statement, suppose that it does not hold then for every $\tilde C>C$ we can find some $C<\hat C <\tilde C$ such that $\overline{\hat E^+(S)}$ is non-compact.
 For every sequence of locally Lipschitz proper cone structures as in Prop.\ \ref{sqd} $\{C_k\}$, $C<C_{k+1}<C_k<\tilde C$, $\cap_k C_k=C$ we can pass to another sequence $\{\hat C_k\}$, $C<\hat C_{k+1}<\hat C_k<C_k$, such that $\overline{ \hat E_k^+(S)}$ is non-compact for every $k$. We rename $\hat C_k\to C_k$.
  Let $h$ be a complete Riemannian metric, let $o\in M$ and let $B_n=B(o,n)$.
Let $q_k\in  E_k^+(S)\backslash S$ be a sequence of points such that $q_k \notin B_k$. By Cor.\ \ref{chg} there is a continuous causal curve $\sigma_k$ connecting some $p_k\in S$ to $q_k$, entirely contained in $E_k^+(S)$. Notice that $\mathrm{Int} J^+(S)\subset \mathrm{Int} J_k^+(S)$, thus these curves do not intersect $\mathrm{Int} J^+(S)$. Applying  the limit curve theorem  to $\sigma_k$ we obtain that there is a future inextendible continuous $C$-causal curve $\sigma$ starting from $S$ and not intersecting  $\mathrm{Int} J^+(S)$. Thus this curve belongs to $E^+(S)$ which contradicts the non-imprisonment property.

As for the last statement, it is clear that we can take $\tilde C$ stably causal. We have to show that for $C<\hat C <\tilde C$,  ${\hat E^+(S)}$ is closed. By contradiction, let $q\in \overline{ \hat E^+(S)}\backslash  \hat E^+(S)$ then there is a sequence $q_k\to q$, $q_k\in \hat E^+(S)$.  There are continuous $\hat  C$-causal curves $\sigma_k$ connecting $S$ to $q_k$ entirely contained in $\hat E^+(S)$ hence not intersecting $\mathrm{Int} \hat J^+(S)$. Since $\hat C$ is stably causal it is non-imprisoning. Using again the limit curve theorem we get a future inextendible continuous $\hat C$-causal curve starting from $S$ and not intersecting $\mathrm{Int} \hat J^+(S)$, hence contained in ${\hat E^+(S)}$ again a contradiction with the non-imprisonment property. Applying the same argument  to $E^+(S)$ shows that this set is closed (remember Th.\ \ref{nig}).
\end{proof}

We recall that every globally hyperbolic closed cone structure admits a stable (equiv.\ stably acausal, cf.\ Th.\ \ref{sts}) Cauchy hypersurface, and that any two stable Cauchy hypersurfaces are smoothly diffeomorphic (Th.\ \ref{ger}).

\begin{theorem} \label{pen} (Improved Penrose's  singularity theorem)\\
Let $(M,C)$ be a globally hyperbolic closed cone structure admitting a non-compact stable Cauchy hypersurface. Then there are no compact future trapped sets and if $S$ is non-empty and compact there is a future inextendible future lightlike geodesic entirely contained in $E^+(S)$.
\end{theorem}

The argument of Penrose's theorem in the $C^2$ Lorentzian case really might be continued as follows: one assumes the existence of a trapped surface $S$, which is a codimension  two closed spacelike manifold whose local orthogonal null fields are converging. Then under the null energy condition the lightlike geodesics starting with those tangents would be refocusing if complete, which implies that $S$ is a trapped set, a contradiction.

\begin{proof}
Suppose there is a non-empty compact set $S$ such that $\overline{E^+(S)}$ is compact. Then by Th.\ \ref{sto}  there is globally hyperbolic locally Lipschitz proper cone structure $\tilde C>C$, such that $\tilde E^+(S)$ is compact. Since $\tilde J$ is closed, by Th.\ \ref{xix} $\tilde J^+(S)$ is closed so $\tilde E^+(S)=\p \tilde J^+(S)$, where $\tilde E^+(S)$ is a $\tilde C$-achronal boundary (Prop.\ \ref{mac}) hence a compact locally Lipschitz  hypersurface. If $V$ is a smooth $\tilde C$-timelike vector field its flow can be used to project $\tilde E^+(S)$ to the Cauchy hypersurface $Q$ (which is a stable Cauchy hypersurface for $(M,C)$). As the projection is compact its boundary as a subset of $Q$ is non-empty. But the integral lines passing through the boundary of the image cannot be transverse to $\tilde E^+(S)$ though they are, a contradiction. Thus for every non-empty compact set $S$, $E^+(S)$ is non-bounded, so we can find $q_n\in E^+(S)\backslash S$ escaping every  compact set, so by Th.\ \ref{sop} there are continuous causal curves $\sigma_n$ connecting $p_n\in S$ to $q_n$ entirely contained in $E^+(S)$. An application of the limit curve theorem gives a future inextendible continuous causal curve $\sigma$ starting from some $p\in S$, entirely contained in $E^+(S)$. By shortening it if necessary, we can assume that it intersects $S$ just in $p$. Since $\sigma$ does not intersect $\mathrm{Int} J^+(p)\subset \mathrm{Int} J^+(S)$ it is a future lightlike geodesic.
\end{proof}

The next result is a low differentiability version of Hawking's 1966 singularity theorem \cite{hawking66}.  In its first version Hawking's theorem included a global hyperbolicity assumption which was removed in  \cite{hawking73}. A $C^{1,1}$ Lorentzian version can be found in \cite{kunzinger15b} and a non-round (Lorentz-Finsler) $C^2$ version can be found in  \cite{minguzzi15}.

\begin{theorem}(Improved Hawking's singularity theorem)\\
Let $(M,\mathscr{F})$ be a non-imprisoning proper  Lorentz-Finsler space and let $S$ be a compact acausal topological hypersurface.
There  is a  future inextendible future causal  geodesic $x\colon [0,+\infty) \to M$ issued from $ x(0)\in S$ and contained in  $\overline{D^+(S)}$ such that for every $t>0$,
\begin{equation} \label{did}
\liminf_{q\to x(t)} d(S,q) \le \ell(x\vert_{[0,t]}).
\end{equation}
Suppose that $(M,C)$ is globally hyperbolic or that it is $C^0$ and such that $\mathscr{F}(\p C)=0$. Then, either $\ell(x)$ is bounded (geodesic singularity), or for every constant $R>0$ we can find a future causal  geodesic $\tilde x\colon [0,1]\to M$  issued from $\tilde x(0)\in S$ and contained in ${D^+(S)}$ such that
\begin{equation} \label{diq}
\ell(\tilde x)=d(S,\tilde x(1))>R.
\end{equation}
If additionally  $(M,\mathscr{F})$ is  locally Lipschitz then $d$ is lower semi-continuous (Th.\ \ref{ddb}) and inequality (\ref{did}) can be replaced by
$d(S,x(t))=\ell(x\vert_{[0,t]})$.
\end{theorem}%

The argument of Hawking's theorem in the $C^2$ Lorentzian case looks for a contradiction in the timelike completeness assumption. It really uses only Eq.\ (\ref{diq}) and  goes as follows: since the boundedness of $\ell(x)$  is excluded, an assumption on the  convergence of the vector field orthogonal to $S$ jointly with the strong energy condition leads to the refocusing of the geodesics starting orthogonally to $S$ within a length $\tau$. By taking $R>\tau$, and noticing that $\tilde x$ being maximizing is orthogonal to $S$ one gets a contradiction since $\tilde x$ cannot have focusing points in its interior.

\begin{proof}
The set $S^\times =S\times \{0\}$ is a compact subset of $M^\times$.  The cone structure $(M^\times, C^\times)$ is non-imprisoning, for if there were a future inextendible continuous $C^\times$-causal curve imprisoned in a compact set the same would be true for the projection of the curve in the projection of the compact set, a contradiction with the non-imprisoning property of $(M,C)$.

By definition of proper Lorentz-Finsler space there is  a smooth $C$-timelike vector field on $M$ such that $V^\times=V\oplus 0$ is a $C^\times$-timelike vector field on $M^\times$. Notice that the flow of $V^\times$ preserves the second coordinate.
Let us consider the set $(E^\times)^+(S^\times)=(J^\times)^+(S^\times)\backslash \mathrm{Int} (J^\times)^+(S^\times)$ on $(M^\times, C^\times)$. No points of $[\overline{D^+(S)}\backslash S]\times \{0\}$ can be contained in this set, since they are connected to $S^\times$ by an integral curve of $V^\times$, and so belong to $\mathrm{Int} (J^\times)^+(S^\times)$ by the openness of the chronological relation. By Th.\ \ref{xux}   for every $q\in D^+(S)$ there is a maximizing bigeodesic $x^q$ connecting $S$ to $q$, $\ell(x^q)=d(S,q)$, so by  Eq.\ (\ref{kki}), $X^q(t)=(x(t), \ell(x\vert_{[0,t]}))$ is a continuous $C^\times$-causal curve contained in $(E^\times)^+(S^\times)$. As a consequence, $D^+(S)\subset \pi_1((E^\times)^+(S^\times))$. By Th.\ \ref{jdq} $S\cap H^+(S)=\emptyset$ so $D^+(S)$ cannot be bounded otherwise $H^+(S)$ would be compact and its generators would be imprisoned in a compact set, a contradiction. Thus both $D^+(S)$ and $(E^\times)^+(S^\times)$ are unbounded (under a global hyperbolicity assumption one could obtain the latter result with Penrose's theorem framed in $M^\times$. Notice that under non-imprisonment we cannot claim that every point of $(E^\times)^+(S^\times)\backslash S^\times$ is connected to $S^\times$ by a $C^\times$-causal curve contained in $(E^\times)^+(S^\times)$, but this fact will not be used). Let $q_n\in D^+(S)$ be an unbounded sequence, then $Q_n=(q_n, d(S,q_n))$ is an unbounded sequence in $(E^\times)^+(S^\times)$. By Th.\ \ref{xux}   for every $q_n$ there is a maximizing bigeodesic $x_n \subset D^+(S)$ connecting $S$ to $q_n$, $\ell(x_n)=d(S,q_n)$, so by Eq.\ (\ref{kki}), $X_n(t)=(x_n(t), \ell(x_n\vert_{[0,t]}))$ is a continuous $C^\times$-causal curve contained in $(E^\times)^+(S^\times)$ and connecting $S^\times$ to $Q_n$.
By the limit curve theorem we find a future inextendible future $C^\times$-lightlike geodesic $\check X(t)=(x(t), \check r(t))$,  entirely contained in $(E^\times)^+(S^\times)$, with $x$ entirely contained in $\overline{D^+(S)}$.  No points of $ [\overline{D^+(S)}\backslash S]\times \{0\}$ can be contained in $(E^\times)^+(S^\times)$, since its points are connected to $S^\times$ by an integral curve of $V^\times$, and so belong to $\mathrm{Int} (J^\times)^+(S^\times)$ by the openness of the chronological relation. This means that $\check X$ does not intersect $[\overline{D^+(S)}\backslash S]\times \{0\}$ and so, reflecting it with respect to the $r=0$ section if necessary, we might assume that $\check X$ belongs to the region $r\ge 0$.

The curve $ X(t)=(x(t), r(t) )$, with $r(t)=\ell(x\vert_{[0,t]})$ is also a future inextendible continuous causal curve entirely contained in $(E^\times)^+(S^\times)$, hence a future $C^\times$-lightlike geodesic. In order to prove this fact, notice that by causality of $\check X$, $\vert \dot{\check{r}}\vert \le \mathscr{F}(\dot x)$, thus for $t>0$, $0\le \check r(t) \le r(t)$. If there were $\bar t$ such that $X(\bar t) \in \mathrm{Int} (J^\times)^+(S^\times)$ then the same would be true for $\check X(\bar t)$, which gives a contradiction. In fact, there would be a product neighborhood of $U\times (a,b)\ni X(\bar t)$ reached at time $\bar t$ by continuous $C^\times$-causal curves issued from $S^\times$, e.g.\ $ Y(t)=(y(t), s(t) )$, but then $\tilde Y=(y(t), \frac{\check r(\bar t)}{r(\bar t)} s(t) )$ would also be a continuous $C^\times$-causal curve, that is there would be a product neighborhood $U\times \frac{\check r(\bar t)}{r(\bar t)}  (a,b)\ni \check X(\bar t)$ in $(J^\times)^+(S^\times)$ as we claimed.

For every $t$, and for every product neighborhood $U\times (r(t)-\epsilon,r(t)+\epsilon)\ni X(t)=(x(t),r(t))$ we can find some point in the product neighborhood which is not reached by $C^\times$-continuous causal curves starting from $S^\times$. Since Eq.\ (\ref{kki}) holds with the equality sign, this means that for every $\epsilon$ and for every $U\ni x(t)$ we can find $q\in U$ such that $d(S,q)\le r(t)+\epsilon$, that is $\liminf_{q\to x(t)}d(S,q) \le r(t)$.

Finally, observe that for $R>0$, either $(E^\times)^+(S^\times)\cap \pi_1^{-1}({D^+(S)})$ is all contained in the regions $r\le R$, which implies that $x$ has length no larger than $R$ on $D^+(S)$ (in the globally hyperbolic case $x$ cannot continue on $H^+(S)$ since this set is empty; in the other case  $x$ can continue on $H^+(S)$ but its length there is zero because of the $C^0$ assumption cf.\ Th.\ \ref{aam}),
or not, which implies that we can find $q\in D^+(S)$ such that $d(S,q)>R$. By  Th.\ \ref{xux} we can find a causal bigeodesic $\tilde x\colon [0,1]\to M$, $\tilde x(0)\in S$, $\tilde x(1)=q$ such that Eq.\ (\ref{diq}) holds true.
\end{proof}

Below we give a causal version of Hawking and Penrose's 1970 singularity theorem \cite{penrose65}. It has been recently generalized to $C^{1,1}$ regularity \cite{graf17} and to the non-round (Lorentz-Finsler) $C^2$ case  \cite{minguzzi15}.
As for Penrose's theorem the key for the generalization to the closed cone structure case stays in the stability of compact trapped sets.
%
\begin{definition}
The cone structure $(M,C)$  is {\em causally disconnected by a compact set} $K$ if
there are sequences $p_n$ and $q_n$, $p_n<q_n$, going to infinity
(i.e. escaping every compact set) such that for each $n$ {\em every}
continuous causal curve connecting $p_n$ to $q_n$ intersects $K$. It is {\em causally connected} if there is no compact set which  causally disconnects it.
\end{definition}

\begin{lemma} \label{bqx} Let $(M,C)$ be a locally Lipschitz proper cone structure.
If $S$ is a closed and achronal set and  $H^{+}(\overline{E^{+}(S)})$ is compact non-empty then strong causality is violated on every neighborhood of it.
\end{lemma}

\begin{proof}
Assume  $H^{+}(\overline{E^{+}(S)})$ is compact, non-empty and strong causality holds in a neighborhood of it. Let
$U$ be a relatively compact neighborhood of
$H^{+}(\overline{E^{+}(S)})$. We want to show that $U$ cannot be covered by causally convex sets. Cover $H^{+}(\overline{E^{+}(S)})$
with a finite number of globally hyperbolic open neighborhoods $U_i,
i=1\ldots n$ whose closures are respectively contained in globally
hyperbolic open neighborhoods $V_i, i=1\ldots n$, i.e. $\bar{U}_i
\subset V_i$, which in turn are contained in $U$. Take a point
$p_1\in H^{+}(\overline{E^{+}(S)})$, then $p_1 \in U_{i_1}$ for
some $0\le i_1\le n$. Let $q_1 \in I^{+}(p_1)\cap U_{i_1}$. As $H^{+}(\overline{E^{+}(S)})\subset \overline{E^{+}(S)}\cup
I^{+}(\overline{E^{+}(S)})\subset \overline{I^{+}(S)}$ we have
$q_1\in I^{+}(S)$.  Clearly
$q_1 \notin \tilde{D}^{+}(\overline{E^{+}(S)})$ otherwise $p_1\in
I^{-}(\tilde{D}^{+}(\overline{E^{+}(S)}))$, a contradiction (recall Th.\ \ref{xiq}).

As a consequence, $q_1 \notin \tilde{D}^{+}(\p {I}^{+}(S))$. In fact, suppose not,  $q_1\in \tilde{D}^{+}(\p {I}^{+}(S))$, then not all timelike curves ending at $q_1$ can intersect $\overline{E^{+}(S)}$ otherwise $q_1 \in \tilde{D}^{+}(\overline{E^{+}(S)})$. Thus there is one timelike curve $\sigma$ which intersects $\p {I}^{+}(S)\backslash \overline{E^{+}(S)}$ at a point $r$. But there is a past inextendible causal curve entirely contained in $\p {I}^{+}(S)$ with future endpoint $r$, thus if $r\ne q_1$ using Th.\ \ref{soa} we can modify $\sigma$ getting a past inextendible timelike curve ending at $q_1$ and not intersecting $\p {I}^{+}(S)$, a contradiction.  The possibility $r=q$ is excluded since $q_1\in I^{+}(S)$.

Since $q_1 \notin \tilde{D}^{+}(\p {I}^{+}(S))$
there is a past inextendible timelike curve $\gamma_1$ that does not
intersect $\p {I}^{+}(S)$ (and hence
$\tilde{D}^{+}(\p {I}^{+}(S))$), and thus it is entirely contained
in $I^{+}(S)$. This curve cannot be totally imprisoned in $U_{i_1}$
otherwise strong causality is violated in $\overline{U}_{i_1}$ a
contradiction with the global hyperbolicity of $V_{i_1}$. Thus there
is a point $q_1'\in \gamma_1\cap I^{+}(S)\cap U_{i_1}^{C}\cap
\tilde{D}^{+}(\p {I}^{+}(S))^{C}$. The timelike curve $\mu_1$
joining $S$ to $q_1'$ leaves the closed set
$\tilde{D}^{+}(\overline{E^{+}(S)})\subset
\tilde{D}^{+}(\p {I}^{+}(S))$ at a last point $p_2 \in
\p {\tilde{D}}^{+}(\overline{E^{+}(S)})=H^{+}(\overline{E^{+}(S)})\cup
\overline{E^{+}(S)}$.

If we had $p_2 \in
\overline{E^{+}(S)}\backslash H^{+}(\overline{E^{+}(S)})$ then  $p_2\in I^{-}(\tilde{D}^{+}(\overline{E^{+}(S)}))$, thus moving forward along $\mu$ starting from $p_2$ we would still be in $I^{-}(\tilde{D}^{+}(\overline{E^{+}(S)}))\cap I^+(S)\subset \tilde{D}^{+}(\overline{E^{+}(S)})$ at least for a small segment, a contradiction since $p_2$ was the last point in this set.
 Thus
$p_2 \in H^{+}(\overline{E^{+}(S)})$, and there is some $i_2$ such
that $p_2 \in U_{i_2}$ (here we do not claim that $i_2\ne i_1$, the
important fact is that $q_1'\notin U_{i_1}$). Following $\mu_1$
after $p_2$ we can find a point $q_2 \in I^{+}(p_2)\cap U_{i_2}$
before $q_1'$. Repeating the arguments given above and continuing in
this way we get a timelike curve $\eta$ which joins (past direction)
$q_1$ to $q_1'$ (along $\gamma_1$), $q_1'$ to $q_2$ (along $\mu_1$),
$q_2$ to $q_2'$ (along $\gamma_2$), and so on with $q_n \in
U_{i_n}$.  As $\eta$ is past inextendible  $U$ cannot be covered by causally convex sets.
\end{proof}

\begin{corollary} \label{vkj} Let $(M,C)$ be a locally Lipschitz proper cone structure.
Let $S$ be a closed and achronal set. Suppose that $E^{+}(S)\ne \emptyset$,
$\overline{E^{+}(S)}$ is compact and the strong causality
condition holds on $\overline{J^{+}(S)}$, then there is a future
inextendible timelike curve issued from $S$ and contained in
$D^{+}(\overline{E^{+}(S)})$.
\end{corollary}

\begin{proof}
Let $V$  be a $C^1$ complete timelike vector
field. If $H^{+}(\overline{E^{+}(S)})$ is empty the desired result is
trivial, just follow an integral line starting from $S$. If not the
integral lines of the field ending at $H^{+}(\overline{E^{+}(S)})$
must intersect $\overline{E^{+}(S)}$ as
$H^{+}(\overline{E^{+}(S)})\subset
\tilde{D}^{+}(\overline{E^{+}(S)})$. This continuous map sends
$H^{+}(\overline{E^{+}(S)})$ to $\overline{E^{+}(S)}$ and has a
continuous inverse defined on its image. Thus if it is surjective
there is a homeomorphism between $H^{+}(\overline{E^{+}(S)})$ and
$\overline{E^{+}(S)}$ with the induced topologies. However, this is
impossible because the former is non-compact while the latter is
compact. Thus there is a future inextendible integral line issued
from $\overline{E^{+}(S)}$ which does not intersect
$H^{+}(\overline{E^{+}(S)})$. By achronality it cannot intersect
$\overline{E^{+}(S)}$ thus it is contained in
$D^{+}(\overline{E^{+}(S)})$.
\end{proof}

\begin{lemma} \label{bha}
Let $(M,C)$ be a locally Lipschitz proper cone structure.
For every set $S$, $E^{+}(S)\subset E^{+}(\bar{S})$. If $S$ is
achronal and $E^{+}(S)$ is closed then $E^{+}(\bar{S})=E^{+}(S)$.
\end{lemma}

\begin{proof}
Since $I$ is open, $I^{+}(\bar{S})=I^{+}(S)$, and since
$J^{+}(S)\subset J^{+}(\bar{S})$, we get $E^{+}(S)\subset
E^{+}(\bar{S})$.  Suppose $S$ is
achronal and $E^{+}(S)$ is closed. Since $S \subset E^{+}(S)$, and the latter set is
closed, $\bar{S} \subset E^{+}(S)$. If $q \in
E^{+}(\bar{S})$ then there is $p\in \bar{S}$ such
that $q \in E^{+}(p)$. But $p\in \bar{S}\subset E^{+}(S)$, thus there
is $r \in S$ such that $p \in E^{+}(r)$. Thus $q \in J^{+}(r)$ that
is, $q \in J^{+}(S)$, and using $I^{+}(S)=I^{+}(\bar{S})$ it follows that
$q \in E^{+}(S)$.
\end{proof}

\begin{proposition} \label{cod}
Let $(M,C)$ be a locally Lipschitz proper cone structure.
Let $S$ be a non-empty compact set, then $E^{+}(S)\cap S \ne
\emptyset$ or $S$ intersects the chronology violating set of
$(M,C)$. In the former case, defining $A= E^{+}(S)\cap S$, $A$ is
non-empty, closed achronal and we have $I^{+}(A)\subset I^{+}(S)$,
$J^{+}(A)\subset J^{+}(S)$ and $E^{+}(S)\subset E^{+}(A)$. Moreover,
if strong causality holds on $S$ the converse inclusions hold.
\end{proposition}

\begin{proof}
Clearly $S\subset J^{+}(S)$, thus the only way in which it could be
$E^{+}(S)\cap S=\emptyset$ is that $S\subset I^{+}(S)$. Consider the
set of open sets $\mathcal{A}=\{I^{+}(p), p \in S\}$, it provides a covering
of the compact $S$, thus there  are a finite number of points
$p_1,\ldots p_n \in S$ and a finite subcovering
$\{I^{+}(p_1),\ldots, I^{+}(p_n)\}$. Each $p_i$ belongs to the
future of some $p_j$, and going backwards, since there are only finitely many elements, one finally finds twice the same $p_k$, thus $p_k\ll p_k$.

Let us consider the case $E^{+}(S)\cap S \ne \emptyset$ and let us
define $A= E^{+}(S)\cap S$. Since $A\subset S$, we have
$I^{+}(A)\subset I^{+}(S)$, $J^{+}(A)\subset J^{+}(S)$. Let $q\in
E^{+}(S)$, then there is a point $p \in S$, such that $p \le q$. It
cannot be $p \in I^{+}(S)$ otherwise $q\in I^{+}(S)$, thus $p \in
S\backslash I^{+}(S)=A$. As a consequence $q \in J^{+}(A)$.
Moreover, $q \notin I^{+}(A)$ otherwise $q \in I^{+}(S)$. We
conclude $q \in E^{+}(A)$, and hence $E^{+}(S)\subset E^{+}(A)$.


For the reverse inclusions assume strong causality holds at $S$. Suppose by contradiction that
$q \in I^{+}(S)\backslash I^+(A)$ (or $q \in J^{+}(S)\backslash J^+(A)$) then there is some $p_1 \in S$,
$p_1 \ll q$ (resp. $p_1 \le q$). We cannot have $p_1\in A$, thus $p_1\in I^{+}(S)$ and there is $p_2\in S$ such that $p_2 \ll
p$. Again necessarily $p_2\notin A$ otherwise $q\in I^+(S)$, so $p_2 \in I^+(S)$. We want to formalize what it means to ``continue in this way''. Let $h$ be a complete Riemannian metric, and let $l_1$ be the $h$-arc length of a timelike curve $\gamma_1$ connecting $p_2$ to $p_1$.
The point $p_2$  and the timelike curve $\gamma_1$ might be chosen in many ways. It is chosen so that $l_1\ge \textrm{min}(d_1/2, 1)$ where $d_1$ is the supremum of $l_1$ for all the possible choices (possibly $d_1=\infty$). By imposing the same criterion  for each step we obtain a sequence of timelike curves which can be joined to form a curve $\gamma$.

Let us show that it cannot hold that $0<a=\sum_i l_i<+\infty$. The convergence of the series implies that
$p_k$ is a Cauchy sequence, thus converging to some point $r\in S$. Then the $h$-arc length parametrized continuous causal curve $\gamma\colon (-a,0] \to M$, $\gamma(0)=p_1$, becomes a continuous causal curve $\gamma\colon [-a,0] \to M$ by setting $\gamma(-a)=r$ (i.e.\ continuous and almost everywhere differentiable with causal tangent). Moreover, for some $\delta>0$, $\gamma(-a)\le \gamma(-a+\delta)\ll p_1\le q$, thus $\gamma(-a)\in I^{-}(q)$ so $\gamma(-a)\notin A$ and hence $\gamma(-a)\in I^+(S)$. Thus $\gamma$ could be extended to an $h$-arc length parametrized continuous causal curve $\tilde \gamma\colon [-a-\epsilon,0]\to M$ with $\gamma(-a-\epsilon)\in S$. For sufficiently large $i$, $l_i<\epsilon/2$ which contradicts the definition of $l_i$.

The possibility $a=+\infty$ would imply that $\gamma$ is past inextendible  and partially imprisoned in $S$ which is impossible because by strong causality on $S$, $S$ is covered by a finite number of causally convex neighborhoods, each of them being intersected only once by $\gamma$.

We conclude that $I^+(S)=I^+(A)$ and $J^+(S)=J^+(A)$ and thus also $E^{+}(S)=E^{+}(A)$.
\end{proof}

The next version of Hawking and Penrose's theorem is completely causal. Unfortunately, the causality condition has to be strengthened since we have proved Th.\ \ref{mbx} only under a locally Lipschitz and proper condition. Under such  condition (i) can be replaced by causality.

\begin{theorem} \label{bae} (Improved  Hawking and Penrose's singularity theorem)\\
Let $(M,C)$ be a closed cone structure. The following conditions cannot all hold:
\begin{itemize}
\item[(i)] $(M,C)$ is stably causal,
\item[(ii)] $(M,C)$ has no lightlike line and it is causally connected,
\item[(iii)] there is a compact future (or past) trapped set $S$.
\end{itemize}
Moreover, (iii) can be weakened to
\begin{itemize}
\item[(iii')] there is a non-empty compact set $S$ such that $E^+(S)$ or $E^-(S)$ are bounded.
\end{itemize}
\end{theorem}

\begin{proof}
Assume (i), (ii), (iii')  hold true.
 By Th.\ \ref{sto} we can find a locally Lipschitz proper cone structure $\tilde C>C$ which is stably causal and such that $\tilde E^+(S)$ is compact (or analogously in the past case).
 By Prop. \ref{cod} the set
$A=S\backslash \tilde I^{+}(S)$ is non-empty, compact and $\tilde C$-achronal, and
moreover, $\tilde E^{+}(A)=\tilde E^{+}(S)$ is compact, thus $A$ is a compact
$\tilde C$-achronal trapped set for $(M,\tilde C)$.
By  $\tilde C$-achronality $A \subset \tilde E^{+}(A)$,
and hence $\tilde E^{+}(A)\ne \emptyset$. By corollary \ref{vkj} there is a
future inextendible $\tilde C$-timelike curve issued from $A$ and contained in
$\tilde D^{+}(\tilde E^{+}(A))$. Extend it to the past to obtain an inextendible
$\tilde C$-timelike curve $\gamma\colon \mathbb{R} \to M$. This curve intersects
$\tilde E^{+}(A)$ only once because of the $\tilde C$-achronality of this set. Let
$p_n=\gamma(t_n)$ with $t_n \to -\infty$, and let $q_n=\gamma(t'_n)$
with $t'_n \to +\infty$.  We have for all $n$, $q_n \in
\tilde D^{+}({\tilde E^{+}(A)})\cap \tilde I^{+}(A)$ and $p_n\in
\tilde I^{-}(\tilde E^{+}(A))$. Let us prove that the compact set $\tilde E^{+}(A)$
disconnects $(M,\tilde C)$. We have only to show that every continuous $\tilde C$-causal curve
$\sigma_n$ connecting $p_n$ to $q_n$ intersects $\tilde E^{+}(A)$. Continue
$\sigma_n$ below $p_n$ along $\gamma$ to obtain a past inextendible continuous
$\tilde C$-causal curve. Since $q_n \in \tilde{D}^{+}({\tilde E^{+}(A)})$, this
curve intersects $\tilde E^{+}(A)$ and the intersection point cannot be in the past of
$p_n\in
\tilde I^{-}(\tilde E^{+}(A))$, since this would violate the achronality of $\tilde E^{+}(A)$. Thus the intersection point is in $\sigma_n$ as required. But if $\tilde E^{+}(A)$
disconnects $(M,\tilde C)$ then it disconnects   $(M, C)$, a contradiction with (ii).
\end{proof}

\begin{theorem} \label{hws}
Let $(M,C)$ be a locally Lipschitz non-imprisoning proper Lorentz-Finsler space such that $\mathscr{F}(\p C)=0$. If $(M,C)$ is
causally disconnected by a compact set $K$ then there is a maximizing inextendible causal
geodesic which intersects $K$.
\end{theorem}

\begin{proof}
By assumption there are sequences $p_k$ and $q_k$, $p_k<q_k$, going to infinity
(i.e. escaping every compact set) such that for each $k$ {\em every}
continuous causal curve connecting $p_k$ to $q_k$ intersects $K$. Let $h$ be a complete Riemannian metric, $o\in M$, and let $C_k$ be a sequence of compact sets such that $B(o,k)\cup K\subset C_k$ and there is  at least one  continuous causal curve connecting
$p_k$ to $q_k$ which intersects $K$
 contained in $C_k$.
%

Denote by $d_{k}(x,z)=\sup_{\eta \subset C_k} l(\eta)$ the
Lorentz-Finsler distance on $C_k$ obtained considering just the continuous
causal curves contained in $C_k$ and intersecting $K$.
We have $d_{k}(x,z)<+\infty$
otherwise there would be a sequence of continuous causal curves contained in
$C_k$ whose Lorentz-Finsler length goes to infinity and by the compactness of $C_k$, there would be a future inextendible continuous
causal curve totally imprisoned in $C_k$ which is impossible (see also Prop.\ \ref{iiu}).

For each $k$ let $\gamma_m^{(k)}$ be a sequence of continuous causal curves
such that $l(\gamma_m^{(k)})\to d_{k}(p_k,q_k)$.
By the limit curve theorem there is a  continuous causal limit curve,
$\gamma_k\colon [a_k, b_k] \to M$ (which we parametrize with respect to
$h$-length), $\gamma_k \subset C_k$ which connects
$p_k=\gamma_k(a_k)$ to $q_k=\gamma_k(b_k)$ and intersects $K$ (the
other possibility involves  a past inextendible causal curve ending
at $q_k$ totally imprisoned in $C_k$, which is impossible). Since
the length functional is upper semi-continuous
$d_{k}(p_k,q_k)=\limsup_{m \to +\infty} l(\gamma^{(k)}_m)\le
l(\gamma_k)\le d_{k}(p_k,q_k)$, thus $d_{k}(p_k,q_k)=l(\gamma_k)$,
i.e. the curve $\gamma_k$ maximizes the Lorentzian length on the
chosen curve set.

Again by the limit curve theorem a subsequence of $\gamma_k\colon [a_k,
b_k] \to M$ converges $h$-uniformly on compact subsets to an
inextendible limit curve $\gamma\colon\mathbb{R} \to M$ intersecting $K$
(we can assume that the subsequence coincides with $\gamma_k$, and
that $\gamma_k(0)\in K$). In particular $-a_k, b_k \to +\infty$. Let
us prove that $\gamma$ is a line.
Let
$a,b \in \mathbb{R}$, $a<b$, for sufficiently large $k$, $a_k<a$,
$b_k>b$ and $\gamma_k(a) \to \gamma(a)$ and $\gamma_k(b) \to
\gamma(b)$. For sufficiently large $s$, $\gamma([a,b])\subset
\mathrm{Int} C_s$. By $h$-uniform convergence  there is $n(s)$ such that for $k>n(s)$,
 $\gamma_k([a,b])\subset
\mathrm{Int} C_s \subset C_k$. But $\gamma_k\vert_{[a,b]}$ is a
restriction of $\gamma_k$ and thus it is also distance maximizing on
$\mathrm{Int} C_s $,  that is $d_{\mathrm{Int} C_s
}(\gamma_k(a),\gamma_k(b))= l(\gamma_k\vert_{[a,b]})$.

Using the lower semi-continuity of the distance $d_{\mathrm{Int}
C_s}$ on  $(\mathrm{Int} C_s, \mathscr{F}\vert_{T\mathrm{Int}
C_s})$ (Th.\ \ref{ddb}) and the upper semi-continuity of the length functional
\begin{align*}
 d_{\mathrm{Int} C_s}(\gamma(a),\gamma(b)) & \le \liminf d_{\mathrm{Int}
 C_s}(\gamma_k(a),\gamma_k(b))  \le \limsup
l(\gamma_k\vert_{[a,b]})  \\
& \le l(\gamma \vert_{[a,b]}) \le d_{\mathrm{Int}
C_s}(\gamma(a),\gamma(b))
\end{align*}
hence $d_{\mathrm{Int} C_s}(\gamma(a),\gamma(b))=l(\gamma
\vert_{[a,b]})$. Since every causal curve connecting $\gamma(a)$ to
$\gamma(b)$ belongs to some $\mathrm{Int} C_s$, we have
$d(\gamma(a),\gamma(b))=l(\gamma \vert_{[a,b]})$ that is $\gamma$ is
a maximizing causal geodesic.
\end{proof}


The next version can be easily compared with the original one, cf.\ \cite{hawking73} Remark on p.\ 267, however it uses a locally Lipschitz and proper assumption.

\begin{theorem}  (Improved Hawking and Penrose's singularity theorem II)\\
Let $(M,\mathscr{F})$ be a locally Lipschitz proper Lorentz-Finsler space such that $\mathscr{F}(\p C)=0$. The following conditions cannot all hold:
\begin{itemize}
\item[(i)] $(M,C)$ is chronological,
\item[(ii)] there are no maximizing inextendible causal geodesics,
\item[(iii)] there is an achronal or compact future (or past) trapped set $S$.
\end{itemize}
\end{theorem}

\begin{proof}
Assume they all hold true. If there were a closed causal curve then it would be achronal  by (i) hence a lightlike line (recall that $C$ is locally Lipschitz), a case which is excluded by (ii). Thus $(M,C)$ is causal.  By Th.\ \ref{mbx} $(M,C)$ is stably causal hence non-imprisoning.
Assume $S$ is achronal and let us prove
that is can be assumed closed and achronal. Indeed, if it is not
closed then $\bar{S}$ is closed and achronal, moreover by Lemma
\ref{bha}, $E^{+}(\bar{S})=E^{+}(S)$ is compact. If there is no maximizing inextendible causal geodesics then by Th.\ \ref{hws} $(M,C)$ is causally connected. Thus we are back to the conditions (i)-(iii) of the first version, Th.\ \ref{bae}, and so we get a contradiction.
\end{proof}


\section{Special topics} \label{ca3}
This section is devoted to the development of some special topics which could be skipped on first reading. Sections \ref{fir}-\ref{las} should be read in this order and provide the proofs to some results already presented in the previous sections. 

\subsection{Proper Lorentz-Minkowski spaces and Legendre transform} \label{zzp}
We have already introduced the notion of  Lorentz-Minkowski space in  Sec.\ \ref{ngd}.
Here we develop the theory of proper Lorentz-Minkowski spaces, so in this section all cones will be proper (sharp, convex, closed and with non-empty interior). This study will motivate some of our terminology connected to Lorentz-Finsler spaces as it shows that some inequalities which are met in the $C^2$ Lorentz-Finsler theory \cite{minguzzi13c} really hold under much weaker assumptions. We shall also prove that the Legendre duality between Lorentz-Finsler Lagrangian and Hamiltonian does not require a $C^2$ assumption.

The polar cone of a proper cone is
\[
C^o=\{p\in V^*\backslash 0\colon \langle p, y\rangle\le 0, \textrm{for every } y\in C \},
\]
and it has the same properties as $C$, namely, it is a proper cone.
The polar of the polar is the original cone $(C^o)^o=C$.
If $D\subset C$ is a another proper cone then $D^o\supset C^o$.

\begin{remark} The polar of a round cone (ellipsoidal section)  is round. This fact can be easily understood with the concept of ice-cream cone which is a cone of angular aperture of $\pi/2$ with respect to some chosen scalar product. Notice that given a round cone $C$
  we can always find a scalar product and associated Cartesian coordinates such that $C$ is an ice-cream cone $0< (\sum^n_i (y^i)^2)^{1/2}\le y^0$. Then $C^o$ becomes the ice-cream cone $0<( \sum^n_i (p_i)^2)^{1/2}\le -p_0$ with respect to the dual coordinates, hence in arbitrary coordinates on the vector space they are round.
\end{remark}
 We have
\begin{align}
\mathrm{Int} (C^o)&=\{p\in V^*\backslash 0\colon \langle p, y\rangle<0, \textrm{for every } y\in C \}, \label{haa} \\
\mathrm{Int} C\,&=\{y\in V\backslash 0\colon \langle p, y\rangle<0, \textrm{for every } p\in C^o \}. \label{hab}
\end{align}

\begin{remark} \label{oqf}
As a consequence, for every $p \in \p C^o$ there is some $y\in \p C$ such that $\langle p, y\rangle=0$ and for every $y\in \p C$ there is some  $p \in \p C^o$ such that $\langle p, y\rangle=0$. Any pair $y\in C$, $p\in C^o$, such that $\langle p, y\rangle=0$ is said to be a polarly related pair. The polar relation is denoted $R$ and is positive homogeneous: $s>0$, $(p,y)\in R \Rightarrow$ $(sp,y)\in R$ and $(p,s y) \in R$. Up to  constants a polarly related pair represents, geometrically, a vector on the boundary of $\p C$ and a hyperplane tangent (supporting) $C$ at $y$.

Suppose that $D\subset C$ is a proper cone such that $\p D\cap \p C\ne \emptyset$, then $\p D^o\cap \p C^o\ne \emptyset$. In fact, by Eq.\ (\ref{hab}) for $y\in \p D\cap \p C$ there is $p\in C^o\subset D^o$ such that $ \langle p, y\rangle=0$, then, by Eq.\ (\ref{haa}) applied to the cone $D$, $p\in \p D^o$. This observation is particularly useful when $D$ is a round cone.
\end{remark}

\begin{remark}
Throughout this section we might equivalently use the notion of dual cone $C^*:=-C^o$ provided the convention  $(+,-,\cdots,-)$ is chosen for the Lorentzian signature. We shall need to work with the polar because we use the Lorentzian signature $(-,+,\cdots,+)$ which at present is the most used in mathematical relativity. Admittedly several formulas would look simpler using the other convention.
\end{remark}

The convexity of $C$  implies the  twice differentiability almost everywhere of the boundary $\p C$. In what follows we might need to consider cones with better regularity properties. Given the nice polar relationship, it will be convenient to focus on those properties which have a nice polar formulation or which are polar invariant. These properties are familiar from the study of the Legendre transform, in fact the polarly related cones can be described, near a polarly related pair, by suitable graphing functions which are Legendre dual to each other.

Let us introduce coordinates $\{y^\alpha\}$ on $V$, and dual coordinates $\{p_\beta\}$ on $V^*$. Let $e_\alpha =\p/\p y^\alpha$ and $e^\alpha =\p/\p p_\alpha$.  The coordinates are chosen in such a way that $e_0\in \mathrm{Int} C$, and $\{y^0=0\}\cap C=\emptyset$, so that $y^0$ is positive over $C$. As consequence,  $-e^0\in \mathrm{Int} C^o$, $\{p_0=0\}\cap C^o=\emptyset$, and $p_0$ is negative over $C^o$. We are interested in the description of the cones near a polarly related pair $\langle\bar p,\bar y\rangle=0$. Let us orient  $e_n$ in such a way that $\bar y\in \textrm{Span}(e_0,e_n)$, and $e_1,\ldots, e_{n-1}$ in such a way that they are annihilated by $\bar p$. Then we have dually that $\bar p\in \textrm{Span}(e^0,e^n)$ and $e^1,\ldots, e^{n-1}$ are annihilated by $\bar y$.
We can also redefine  $e_n\to -e_n$ if necessary, in  such a way that $\bar y^n<0$, so that the portion of boundary $\p C$ on which we are interested is on the region $y^n<0$.  If $y$ is a point in this region then  $y'=y+\epsilon e_n\in \mathrm{Int} C$ for sufficiently small $\epsilon>0$.
 Any polarly related value  $p$, $\langle p,y\rangle=0$ is such that $p_n<0$ (e.g.\ $\bar p_n<0$).  Indeed, for sufficiently small $\epsilon>0$,  $y'\in \mathrm{Int} C$, where $y'$ has the same coordinates of $y$ saved for $y'{}^n$. Since $0> \langle p,y'\rangle= \langle p,y'-y\rangle =p_n \epsilon$, we get $p_n< 0$. Dually, if $\langle p,y\rangle=0$ is a polarly related pair with $p$ close to $\bar p$, then $y$ is such that $y^n<0$.

Now, the section $\p C\cap \{y^0=1, y^n< 0\}$ near the suitably rescaled $\bar y$ is locally described by a convex negative function $y^n=u(y^{A})$, while the section $\p C^o\cap \{p_n=-1\}$ is locally the  graph of $p_0=-u^*(p_A)$. This fact is easily inferred from the polarity condition $\langle p,y\rangle \le 0$, which reads for $p\in \p C^o$, $y\in \p C$, in the image of the local graphs, $p_0-u(y^A)+p_A y^A\le 0$, equality holding at a polarly related pair. Thus $\sup_{y^A} [p_A y^A -u(y^A)]=-p_0(p_A)$, which proves the claim. This result clarifies that the regularity properties that are invariant under Legendre duality, once applied to cones, are invariant under polarity.

For instance, $C$ is $C^1$ iff $C^o$ is strictly convex\footnote{When speaking of convexity properties of a proper cone we really refer to such properties for its compact sections obtained through the intersection with a hyperplane.} (and analogously with $C$ and $C^o$ exchanged) \cite[Th.\ 26.3]{rockafellar70}\cite[Chap.\ 4]{hiriart93}. As another example: $C$ is strongly convex iff $C^o$ is strongly smooth (and analogously with $C$ and $C^o$ exchanged) \cite{goebel08,kakade09}. Geometrically, the condition of strong convexity means the following: we can find  a scalar product on $V$ and a corresponding dual scalar product on $V^*$, such that every point $y\in \p C$ admits an ice-cream cone ($\pi/2$-aperture) $R\supset C$ such that $\p R\cap \p C=y$.  The property of strong smoothness for $C^o$ is similar, but this time the ice-cream cone is contained in $C^o$. The property of strong smoothness is equivalent to the $C^{1,1}$ regularity of the boundary of the cone $C$ once expressed as a local graph. This result is obtained in the mentioned references. Alternatively, it can be derived from the mentioned geometrical interpretation  and from the fact that a function $h$ which is semiconvex with its negative is really $C^{1,1}$ \cite[Cor.\ 3.3.8]{cannarsa04} (it is also useful to recall that the $C^{1,1}$ functions are semiconvex \cite{vial83}). As a final example: $C$ is $C^2$ and strongly convex iff $C^o$ is $C^2$ and strongly convex.

Lorentz-Minkowski spaces are the models to the tangent spaces of Lorentz-Finsler spaces cf.\ Sec.\ \ref{ngd}.
We are looking for a  notion of proper Lorentz-Minkowski space $(V,\mathscr{F})$, where $\mathscr{F}\colon C\to [0,+\infty)$ has as  domain a proper cone $C$.
The main idea here is that of regarding $\mathscr{F}$ as defining a cone on the vector space $V\oplus \mathbb{R}$, through
\begin{equation} \label{sss}
C^\times=\{(y,z)\colon \ \vert z\vert \le \mathscr{F}(y), \ y \in C \}.
\end{equation}
The cone has to have the same properties as $C$, so $C^\times$ has to be a proper cone (so closed in the topology of $V\oplus \mathbb{R}\backslash 0$; again   $0\notin \p C^\times$). The cone property demands that $\mathscr{F}$ be positive homogeneous, the convexity property that $\mathscr{F}$ be concave, the sharpness property that $\mathscr{F}$ be finite, as it is by definition, and the non-empty interior condition, that $\mathscr{F}$ be not identically zero. These conditions are also sufficient to get a cone $C^\times$ with the desired  properties. Thus

\begin{definition}
A proper Lorentz-Minkowski space is a pair $(V,\mathscr{F})$, where $\mathscr{F}\colon C\to [0,+\infty)$,  $C$ is  a proper cone, $\mathscr{F}$ is positive homogeneous, concave (hence locally Lipschitz on $\mathrm{Int} C$), and not identically zero. Equivalently, it is a pair $(V,\mathscr{F})$ such that $C^\times$ is a proper cone in $V\oplus \mathbb{R}$.
\end{definition}

The conditions imply that $\mathscr{F}$ is positive on $\mathrm{Int} C$, and so $\mathrm{Int} C \times \{0\}\subset \mathrm{Int} (C^\times)$, however, we do not impose that $\mathscr{F}$ vanishes on $\p C$, namely the indicatrix $\mathscr{I}:=\mathscr{F}^{-1}(1)$ might intersect $\p C$. The conditions on $\mathscr{F}$ in the previous definition are equivalent to: the indicatrix intersects every half-line in $\mathrm{Int} C$ issued from 0 and it is convex.

This definition is well behaved under duality (polarity). Observe that the polar $(C^\times)^o$ will be a proper cone in $V^*\oplus \mathbb{R}$.
The polar cone is symmetric with respect to the $V^*\times \{0\}$ plane because $C^\times$ is symmetric with respect to $V\times \{0\}$, so there is a function $\mathscr{F}^o$ such that
\begin{equation} \label{ssz}
(C^\times)^o=\{(p,z^o)\colon \ \vert z^o\vert \le \mathscr{F}^o(p), \ p \in C^o \}.
\end{equation}
 The function $\mathscr{F}^o$ is finite because $(C^\times)^o$ is sharp, thus $\mathscr{F}^o$ shares all the properties of $\mathscr{F}$. Furthermore, since the polar of the polar is the original cone, $(\mathscr{F}^o)^o=\mathscr{F}$.


By concavity and positive homogeneity of $\mathscr{F}$, the reverse triangle inequality holds true.
\begin{proposition}[Reverse triangle inequality] \label{nuy}
For every $y,w\in C$ we have
\begin{equation}
\mathscr{F}(y+w)\ge \mathscr{F}(y)+\mathscr{F}(w).
\end{equation}
\end{proposition}

\begin{proof}
For every $y,w\in C$, we can find constants $a,b> 0$ such that $y'=y/a$ and $w'=w/b$ belong to a  section to which $y+w$ belongs, so $a+b=1$ and
\[
\mathscr{F}(y+w)=\mathscr{F}(a y'+bw')\ge a \mathscr{F}(y')+ b \mathscr{F}(w')= \mathscr{F}(y)+\mathscr{F}(w).
\]
\end{proof}
Of course, in the equality case the proportionality of $y$ and $w$ can be inferred only under the 
the strict convexity of $C^\times$. We shall have an analogous reverse triangle inequality on $V^*$.


\begin{proposition}[Reverse Cauchy-Schwarz inequality] \label{nut}
For every $y\in C$, $p\in C^o$ we have
\begin{equation} \label{rcs}
-\langle p, y\rangle\ge \mathscr{F}^o(p) \mathscr{F}(y).
\end{equation}
\end{proposition}
\begin{proof}
For every $y\in C$, $p\in C^o$ we have $(y, \mathscr{F}(y))\in C^\times$, $(p,   \mathscr{F}^o(p))\in (C^\times)^o$ and the definition of  polarity reads
\begin{equation} \label{erc}
0\ge \langle (p,   \mathscr{F}^o(p)), (y, \mathscr{F}(y)) \rangle = \langle p, y\rangle+ \mathscr{F}^o(p) \mathscr{F}(y) ,
\end{equation}
which is the desired inequality.
\end{proof}

In the next equation it is understood that the ratio on the right-hand side equals $+\infty$ for $ \langle p, y\rangle<0$ and $\mathscr{F}(y)=0$.

\begin{corollary}
The polar Finsler function satisfies: for $p\in \mathrm{Int} C^o$
\[
\mathscr{F}^o(p) = \underset{y\in   C}{\mathrm{inf}} \left( \frac{-\langle p, y\rangle}{\mathscr{F}(y)}\right) .
\]
The infimum is attained on at least an half-line.
\end{corollary}
The previous equation could have been used as a definition of $\mathscr{F}^o$ on $\mathrm{Int} C^o$. Nevertheless, such an approach would  hide the geometrical interpretation in terms of the polarity relation of the proper cones $C^\times$ and $(C^\times)^o$.

\begin{proof}
For every $y\in C$, $p\in C^o$, by polarity of $C^\times$ and $(C^\times)^o$ inequality (\ref{erc}) holds true.
Since $p\in  \mathrm{Int} C^o$, we have by the proper condition $\mathscr{F}^o(p)>0$,
 and for every $y\in C$, $\langle p, y\rangle<0$.
 For every $y \in \mathrm{Int}  C$ we have  by the proper condition, $\mathscr{F}(y)>0$, thus  $\mathscr{F}^o(p) \le {-\langle p, y\rangle}/{\mathscr{F}(y)}$. Let us prove that the equality is attained. Indeed,  $(p,   \mathscr{F}^o(p))\in \p (C^\times)^o$ thus, by Remark \ref{oqf} (applied to $C^\times$ instead of $C$), there is an element $(y,b)\in \p C^\times$ such that $0=\langle (p,   \mathscr{F}^o(p)), (y, b) \rangle = \langle p, y\rangle+ \mathscr{F}^o(p) b$. If $y\in \mathrm{Int} C$ we have concluded since necessarily $b=\mathscr{F}(y)$ where the plus sign follows from  $\langle p, y\rangle<0$.
 If $y \in \p C$ we still have  $\langle p, y\rangle<0$ so $0<b\le\mathscr{F}(y)$,
 but $b$ cannot be strictly less than $\mathscr{F}(y)$ otherwise
  $\langle (p,   \mathscr{F}^o(p)), (y, \mathscr{F}(y)) \rangle=\mathscr{F}^o(p) [\mathscr{F}(y)-b]>0$ a contradiction with polarity since $(y, \mathscr{F}(y) \in C^\times$.
\end{proof}

Pairs $(p,y)\in C^o\times C$ for which the equality holds in (\ref{rcs}), or equivalently in Eq.\ (\ref{erc}), form a relation which might be called {\em $\times$-polar relation} $R^\times$. It is invariant under positive homogeneity: $(p,y) \in R^\times \Rightarrow$  $\ (s p, y)\in R^\times$ and $(p, sy)\in R^\times$, for any $s>0$.  This $\times$-polar relation is a function from $C/\mathbb{R}_+$ to $C^o/\mathbb{R}_+$ (resp.\ opposite direction)  iff $C^\times$ is  $C^1$ (resp.\ strictly convex). It is a bijection iff $C^\times$ is $C^1$ and strictly convex.

In the latter case the bijection from $C/\mathbb{R}_+$ to $C^o/\mathbb{R}_+$  might be used to get a bijection $\ell\colon \mathrm{Int} C\to \mathrm{Int} (C^o)$, provided we stipulate that  the indicatrix $\mathscr{I}$ is sent  to the polar indicatrix $\mathscr{I}^o$, i.e.\ $\mathscr{F}^o(\ell(y))=1$ whenever $\mathscr{F}(y)=1$. The extension is accomplished as follows.
Let $f,f^o\colon \mathbb{R}_+\to \mathbb{R}$ be Legendre dual $C^1$ functions such that $f',f^o{}' \colon \mathbb{R}_+\to \mathbb{R}_+$ are positive strictly monotone bijections (thus $f,f^o$ are either both strictly convex or strictly concave). We recall that  $f'$ and $f^o{}'$ are functional inverses of each other: $f'(f^o{}'(x))=x$, $f^o{}'(f'(x))=x$. Thus we can redefine $f$ by rescaling it by a positive constant in such a way that $f'(1)=1$ and hence $f^o{}'(1)=1$.  The extension will be dependent on the chosen  pair $(f,f^o)$ with the mentioned normalization, in fact we impose $\mathscr{F}^o(\ell(y))=f'(\mathscr{F}(y))$. Since the direction of $\ell(y)$ is determined by the $\times$-polar relation, this equation by fixing its length determines $\ell(y)$ completely. It can be rewritten in the equivalent dual form $f^o{}'(\mathscr{F}^o(p))=\mathscr{F}(\ell^{-1}(p))$ where $\ell^{-1}$ is the functional inverse of $\ell$.

\begin{example}
Particularly interesting will be the next choice of functions which up to an additive constant are the only ones for which, $\ell$ and $\ell^{-1}$ are positive homogeneous (of degree $a-1$ and $b-1$, respectively). Let $a,b \in \mathbb{R}\backslash \{1\}$, be conjugate exponents $\frac{1}{a}+\frac{1}{b}=1$, where the pair $(a,b)=(0,0)$ is allowed and understood for shortness as a limiting case (hence $a/b=b/a=-1$). The Legendre dual functions are
\begin{align} \label{kql}
&f(x)=\tfrac{1}{2} + \tfrac{1}{a} [x^a-1],
&f^o(x)=\tfrac{1}{2} + \tfrac{1}{b} [x^b-1],
\end{align}
which for $a=b=0$ stand respectively for $f=f^o=\frac{1}{2}+\log x$.
The symmetric subcases $a=b=0$, just mentioned, and $a=b=2$ seem the most interesting. The former case corresponds to some choices for the function $\ell$ (homogeneity of degree -1) familiar from  the theory of homogeneous cones \cite{vinberg60}. The latter case  gives $f=f^o=x^2/2$, with $\ell$ positive homogeneous of degree 1, and corresponds to the standard formalism of Lorentz-Finsler theory. If one is not interested in recovering the  case $a=b=0$ as a limit (notice that  $\lim_{\epsilon \to 0} \frac{1}{\epsilon} [x^\epsilon-1]=\log x$, for $x>0$) the additive constants in the definitions of $f$ and $f^o$ can be dropped, i.e.\ $f=\tfrac{1}{a} x^a$, $f^o=\tfrac{1}{b} x^b$, preserving Legendre duality and gaining positive homogeneity.
\end{example}


\begin{theorem} \label{dpq}
On a proper Lorentz-Minkowski space the next conditions are equivalent and invariant under polarity
\begin{itemize}
\item[(i)]     $C^\times$ is  $C^1$ and strictly convex.

\item[(ii)] $C$ is $C^1$ and strictly convex, $\mathscr{F}^{-1}(0)=\p C$,   $\mathscr{F}\in C^1(\mathrm{Int} C)\cap C^0(C)$, $\mathscr{F}$ is strictly concave on one (and hence every) relatively compact section of $\mathrm{Int} C$, and $\dd \mathscr{F} \to \infty$ for $y \to \p C$.

\item[(iii)] $C$ is $C^1$ and strictly convex, the indicatrix does not intersect $\p C$ but intersects every non-compact section of $C\cup\{0\}$ containing the origin,  and the indicatrix is $C^1$ and strictly convex,
%
\end{itemize}
If they hold the equality case in the reverse triangle inequality (Prop.\ \ref{nuy}) holds iff $y$ and $w$ are proportional. Similarly,  the equality case in the reverse Cauchy-Schwarz inequality (Prop.\ \ref{nut}) holds iff
$(p,y)\in R^\times$.
\end{theorem}

The section in $(iii)$ is obtained  by means of a hyperplane not necessarily passing through the origin. Actually, the proof shows that on $(iii)$ one can replace such general sections with those determined by hyperplanes parallel to $\ker p$ for  $p\in \p C^o$.

The proof is really more specific, and provides some equivalences not apparent from the statement. For instance, `$\mathscr{F}^{-1}(0)=\p C$, and $\dd \mathscr{F} \to \infty$ for $y \to \p C$' implies that  the indicatrix does not intersect $\p C$ but intersects every non-compact section of $C$ containing the origin.

\begin{proof}
The statement concerning the invariance under polarity is obvious from $(i)$ and standard results of convexity theory on the duality between differentiability and strict convexity \cite[Th.\ 26.3]{rockafellar70}. So we need only to prove the equivalences.

 $(i) \Leftrightarrow (ii)$.
The strict convexity of $C^\times$ can be expressed with the property that for  $Z\in \p C^\times$  and for every supporting hyperplane $P\ni Z$ of $C^\times$, the intersection $P\cap C^\times$ is one-dimensional. For $Z \in \p C\times\{0\}$ this property  is equivalent to $\mathscr{F}^{-1}(0)=\p C$. For $Z\in [\p C^\times]\backslash [\p C\times\{0\}]$ it is equivalent to the  strict concavity of $\mathscr{F}$ on one (and hence every) relatively compact section of $\mathrm{Int} C$.
The $C^1$ differentiability of  $C^\times$ is equivalent, at $Z\in [\p C^\times]\backslash [\p C\times\{0\}]$ to $\mathscr{F}\in C^1(\mathrm{Int} C)$, and at  $Z\in [\p C\times\{0\}] \subset \p C^\times$ to $\dd \mathscr{F} \to \infty$ for $y \to \p C$.

$(i)$ and $(ii)$ $\Rightarrow (iii)$. We know that $\mathscr{F}$ is positive on $\mathrm{Int} C$ so by positive homogeneity  $\mathscr{F}^{-1}(0)=\p C$ is equivalent to the condition that the indicatrix does not intersect $\p C$. The strict concavity of $\mathscr{F}$ on a relatively compact section of $\mathrm{Int} C$ implies that $\mathscr{F}$ satisfies the reverse triangle inequality on $\mathrm{Int} C$ with the usual equality case, so from positive homogeneity the indicatrix is strictly convex. Conversely, the strict convexity of the indicatrix is, by positive homogeneity, equivalent to the strict concavity of $\mathscr{F}$ on one (and hence every)  relatively compact section of $\mathrm{Int} C$.
The $C^1$ differentiability of $\mathscr{F}$ on $\mathrm{Int} C$ is equivalent to the $C^1$ differentiability of the indicatrix.

Let us prove that $\dd \mathscr{F} \to \infty$ for $y \to \p C$, implies that the indicatrix  intersects every non-compact section of $C\cup\{0\}$ containing the origin.
Suppose, by contradiction, that there is a hyperplane on $V$ cutting $C\cup\{0\}$ on two sectors, one of which, call it $D$, is non-compact, contains the origin and does not intersect $\mathscr{I}$. On $V\oplus \mathbb{R}$, consider the hyperplane passing through $[\p D\cap \mathrm{Int} C]\times \{1\}$ and the origin. One half-space determined by it contains $C^\times$, and its intersection with $V\times \{0\}$  is parallel to $[\p D\cap \mathrm{Int} C]$, and one of its half-spaces contains $C$. Thus $[\p D\cap \mathrm{Int} C]$ can only be parallel to the kernel of an element on $\p C^o$, call it $p$. So the indicatrix is contained in $\{y\colon \langle p, y\rangle< d<0\}$, for some $d$, while $C^\times$ is contained in $\{(y,z)\colon \langle p, y\rangle\le d z \}$, which contains some point $w\in \p C$ (polarly related with $p$), thus $C^\times$ is not $C^1$ at $w$, which in view of the already proved equivalence is the desired contradiction.

Conversely, suppose that the indicatrix  intersects every non-compact section of $C \cup\{0\}$ containing the origin, and let $p\in \p C^o\backslash 0$. Let $(p,z) \in V^*\oplus \mathbb{R}$, then for $z>0$, $\langle (p,z), (y, 1)\rangle=\langle p, y\rangle+z$ for every $y\in \mathscr{I}$ is neither non-negative nor non-positive, which implies that  $\langle (p,z), (y, \mathscr{F}(y))\rangle$ is not everywhere non-positive for $y\in \mathrm{Int} C$, and hence that $(p,z) \notin (C^\times)^o$, which implies that $(C^\times)^o$ is strictly convex on $\p C^o \times \{0\}$, and hence $C^\times$ is differentiable on $\p C \times \{0\}$.
\end{proof}

\begin{definition}
A  proper Lorentz-Minkowski space is said to be a $C^1$ strictly convex Lorentz-Minkowski space if the previous equivalent conditions hold true. More generally, it is said to be a [property] Lorentz-Minkowski space if $C^\times$ satisfies [property].
\end{definition}
The theory developed so far clarifies our philosophy in dealing with regularity conditions for the Lorentz-Minkowski structure: the additional properties on $(V,\mathscr{F})$  should be imposed taking into account their geometrical content in terms of the cone $C^\times$. Once again, particularly interesting are those conditions which are invariant under polarity or which have a clear polar counterpart. For instance, the strong convexity of $C^\times$ is equivalent to the strong smoothness of $(C^\times)^o$ (and conversely). Any of the cones is strongly convex and $C^2$ iff the other is.

\subsubsection{Legendre transform}

Let $(V,\mathscr{F})$ be a $C^1$ strictly convex Lorentz-Minkowski space.
Let $(p,y)$ be a $\times$-polar pair with $y\in \mathrm{Int} C$ and $p=\ell(y)$ and $y(s)=y+ s w$, where $w\in V\backslash 0$, so that for sufficiently small $\vert s\vert$,   $y(s)\in \mathrm{Int} C$. Since by polarity  \[\langle (p,   \mathscr{F}^o(p)), (y(s), \mathscr{F}(y(s))) \rangle\le 0,\] the function on the left-hand side is $C^1$ and reaches a maximum (zero) at $s=0$, thus  differentiating with respect  to $s$ and setting $s=0$    we obtain
\begin{equation} \label{jjh}
0=\langle p, w\rangle+ \mathscr{F}^o(p) \p_w \mathscr{F}= \langle p, w\rangle+ f'(\mathscr{F}(y)) \p_w \mathscr{F} ,
\end{equation}
 thus $p=\ell(y)=\dd  \mathscr{L}\vert_y$, where $\mathscr{L}\colon \mathrm{Int} C\to \mathbb{R}$,
\begin{equation}
\mathscr{L}(z)=-f(\mathscr{F}(z))
\end{equation}

Let us calculate the Legendre transform $\mathscr{H}\colon \mathrm{Int} C^o\to \mathbb{R}$ of  $\mathscr{L}$.
\begin{align*}
\mathscr{H}(\ell (y))&=y^\mu \frac{\p\mathscr{L}}{\p y^\mu}(y)-\mathscr{L}(y)=-f'(\mathscr{F}) y^\mu \frac{\p\mathscr{F}}{\p y^\mu}+f(\mathscr{F})=-f'(\mathscr{F}) \mathscr{F}+f(\mathscr{F})\\
&=-(x f'(x)-f(x))\vert_{x=\mathscr{F}}=-f^o(f'(\mathscr{F}))=-f^o(\mathscr{F}^o(\ell(y))),
\end{align*}
where we used the positive homogeneity of degree one of $\mathscr{F}$. We conclude that the Legendre transform of $\mathscr{L}$ is
\begin{equation}
\mathscr{H}(p)=-f^o(\mathscr{F}^o(p)).
\end{equation}
In most cases it will be possible to extend by continuity $\mathscr{L}$ to $\p C$, e.g.\ when $\p C=\mathscr{F}^{-1}(0)$ and $f$ extends continuously to the origin by setting $f(0)=0$.  An analogous observation holds for $\mathscr{H}$.

Suppose now that the Lorentz-Finsler space is $C^2$ and strongly convex.
 and that $f,f^o$ are $C^2$ and either both strongly convex or strongly concave.

Let us denote with $\dd^2$ the Hessian operator. We have shown in \cite{minguzzi15e} (see also \cite{beem70,laugwitz11}) that  $h\colon \mathrm{Int} C\to V^*\otimes V^*$
\begin{equation} \label{med}
h=-\frac{1}{\mathscr{F}}\dd^2 \mathscr{F}
\end{equation}
 is a metric of signature $(0,+,\ldots,+)$ which pulled back to  $\mathscr{I}$ provides the affine  metric of the indicatrix. The definition of $\mathscr{L}$ gives
\begin{equation} \label{men}
\dd^2 \mathscr{L}=f'(\mathscr{F})\mathscr{F} h- f''(\mathscr{F}) \dd \mathscr{F}\otimes \dd \mathscr{F}
\end{equation}
We know that  $\mathrm{sgn}f'=1$, so let $s=\mathrm{sgn} f''$. The metric $ \dd^2 \mathscr{L}$  is Riemannian for $s<0$ and Lorentzian for $s>0$. Since $\mathscr{H}$ is the Legendre dual of $\mathscr{L}$, the metric $\dd^2 \mathscr{H}$ is the inverse of $\dd^2\mathscr{L}$.  For instance, for $f$ as in Eq.\ (\ref{kql}), $\dd^2 \mathscr{L}$ is Riemannian for  $a<1$ (hence also in the logarithmic case $a=0$) and Lorentzian for $a>1$ (hence also in the standard Lorentz-Finsler case $a=2$).

\begin{remark}[Equivalence of some Finslerian relativistic theories] \label{eqq} $\empty$\\
Let $f\colon \mathbb{R}_+\to \mathbb{R}$  be $C^2$ and such that $f'(1)=1$, $f'>0, s:=\mathrm{sgn}(f'')\ne 0$ (e.g.\ $f=x^a/a$ with $a>1$, s=+1, or $f=\frac{1}{2} + \log x$, $s=-1$).
Let $\mathscr{F}$ be positive homogeneous and defined over a proper cone, and let $\mathscr{L}=-f(\mathscr{F})$. Suppose that $\dd^2\mathscr{L}$ has signature $(-s,+,\ldots, +)$. By Eq.\ (\ref{men}) the metric $h$ in (\ref{med}) is positive definite and so  $\mathscr{I}$ is strongly convex which implies that we are in the framework of the Lorentz-Minkowski spaces of this work.  Again by (\ref{men}) the Hessian $\dd^2\tilde{\mathscr{L}}$ where $\tilde{\mathscr{L}}=-\tilde f(\mathscr{F})$, has signature $(-\tilde s,+,\ldots, +)$ for any other choice of $\tilde f$ with the same properties. In particular, it is Lorentzian
for $f=x^2/2$ which is the usual Lorentz-Finsler choice (e.g.\ Beem \cite{beem70}).

This result clarifies that the kinematics of physical theories based on the function $\mathscr{L}$, even when defined with different choices of $f$, is essentially the same. The next result shows that the dynamics is also largely independent of $f$. In general, the choice of $f$ is related to the regularity of the theory at $\p C$, namely to the extendibility of the map $\ell$ to $\p C$, which, physically, it is connected to the correspondence velocity-momenta for lightlike particles, see Remark \ref{ren}.
\end{remark}

We summarize the previous results as follows

\begin{theorem}
 Let $f,f^o\colon \mathbb{R}_+\to \mathbb{R}$ be Legendre dual $C^1$ functions such that $f',f^o{}' \colon \mathbb{R}_+\to \mathbb{R}_+$ are strictly monotone bijections (thus $f,f^o$ are either both strictly convex or strictly concave), normalized so that $f'(1)=1=f^o{}'(1)$. On a $C^1$  strictly convex  Lorentz-Minkowski space the continuous maps
\begin{align*}
&\ell\colon \mathrm{Int} C \to \mathrm{Int} (C^o),\qquad  \ell(y)=\dd \mathscr{L}\vert_y,   &&\mathscr{L}:=-f(\mathscr{F}), \\
&\ell^o\colon \mathrm{Int} (C^o) \to
\mathrm{Int} C, \quad \ \ell^o(p)=\dd \mathscr{H}\vert_p,   &&\mathscr{H}:=-f^o(\mathscr{F}^o),
\end{align*}
are bijections and inverse of each other; they send the indicatrix $\mathscr{I}$ to $\mathscr{I}^o$, and conversely. The function $\mathscr{H}$ is the Legendre transform of $\mathscr{L}$, and conversely.

For a $C^2$ strongly convex Lorentz-Minkowski space, with $C^2$ strongly convex or concave  functions $f,f^o$, we have that $\dd^2 \mathscr{H}(\ell(y))$ is the inverse of $\dd ^2\mathscr{L}(y)$ and they are either both Riemannian or Lorentzian depending only on the sign $s=\mathrm{sgn} f''$ either negative or positive, respectively. Finally, working on $TM$, namely including an $x$-dependance of $\mathscr{L}$,  the spray
\begin{equation}
G^\alpha(x,y)=\frac{1}{2} \, g^{\alpha \beta} \Big( \frac{\p^2 \mathscr{L}}{\p x^\gamma \p y^\beta} y^\gamma-\frac{\p \mathscr{L}}{\p x^\beta} \Big)
\end{equation}
does not depend on $f$.


 For the choice given by Eq.\ (\ref{kql}) we have the identity \begin{equation} \label{mof}
b[\mathscr{H}(\ell(y))+\tfrac{1}{2}]=a[\mathscr{L}(y)+\tfrac{1}{2}],
 \end{equation}
 for every $y\in \mathrm{Int} C$.   Finally, for $a,b\ne 0$, dropping the additive constants in Eq.\ (\ref{kql}),  $\mathscr{L}$ is positive homogeneous of degree $a$,  $\mathscr{H}$ is positive homogeneous of degree $b$ and the identity (\ref{mof}) is replaced by  $ b\mathscr{H}(\ell(y))=a\mathscr{L}(y)$.
\end{theorem}

Particularly interesting is the fact that the spray, and hence  the non-linear connection, does not depend on the pair $(f,f^o)$. It shows that the dynamics of Lorentz-Finsler gravitational theories  is largely independent of such choice. In fact it must be recalled that the non-linear connection of (Lorentz-)Finsler geometry follows from the spray, and the non-linear curvature follows from the non-linear connection. Many dynamical equations proposed in the literature are formulated in terms of the non-linear curvature.

\begin{proof}
Let us prove the identities in the last paragraph.  We have already calculated $\mathscr{H}(\ell(y))=-f'(\mathscr{F}) \mathscr{F}+f(\mathscr{F})$, thus if $f=x^a/a$  we have $\mathscr{H}(\ell(y))=(1-a)f(\mathscr{F}(y))=(a-1)\mathscr{L}(y)=\frac{a}{b} \mathscr{L}(y)$. For $f$  given by  Eq.\ (\ref{kql}) it is sufficient to redefine $\mathscr{L}$ and $\mathscr{H}$ in the previous expression by adding to them a suitable constant.

Let us prove the spray identity.
Let $x(t)$ be a geodesic for the spray $G$ induced by $\mathscr{L}=-\mathscr{F}^2/2$, with initial conditions $(x(0),\dot x(0))=(x_0,y_0)$, hence a solution of $\ddot x^\mu+G^\mu(x,\dot x)=0$ with such initial conditions. Let us check whether it is a stationary point for the action functional $\int \tilde{\mathscr{L}}$, where $\tilde {\mathscr{L}}=h(\mathscr{L})$.
The Euler-Lagrange equations are
\begin{align*}
\frac{\dd }{\dd t}\frac{\p \tilde {\mathscr{L}}}{\p \dot x^\mu}-\frac{\p \tilde {\mathscr{L}}}{\p x^\mu} =\frac{\dd h' }{\dd t} \frac{\p \mathscr{L}}{\p \dot x^\mu}+h[\frac{\dd }{\dd t}\frac{\p  \mathscr{L}}{\p \dot x^\mu}-\frac{\p  \mathscr{L}}{\p x^\mu} ]=\frac{\dd h' }{\dd t} \frac{\p \mathscr{L}}{\p \dot x^\mu}
\end{align*}
but $\mathscr{L}$ is constant over $x(t)$, and so is $h'(\mathscr{L})$. We conclude that  $x(t)$ is a solution of
\[
\frac{\dd }{\dd t}\frac{\p \tilde {\mathscr{L}}}{\p \dot x^\mu}-\frac{\p \tilde {\mathscr{L}}}{\p x^\mu}=0
\]
so provided $\tilde g=\dd^2 \tilde{\mathscr{L}}$ is non-degenerate, it solves the spray equation $\ddot x+2\tilde G(x,\dot x)=0$. Now, at $t=0$ we have
\[
G(x_0,y_0)=-\ddot x(0)/2=\tilde G(x_0,y_0)
\]
and since the initial conditions are arbitrary we conclude that the sprays coincide. Now, let $h(x)= -\tilde f(\sqrt{-2 x})$, where $\tilde s\ne 0$, so that $h(\mathscr{L})=-\tilde f(\mathscr{F})$. We have already shown that $\dd^2 \tilde f(\mathscr{F})$ is non-degenerate, thus the desired result follows.
\end{proof}

%


Notice that we have recovered a Legendre duality without using a convexity assumption on $\mathscr{L}$ (the Legendre-Fenchel generalization of the Legendre transform is not viable since $\mathscr{L}$ is not convex) or a $C^2$ assumption.



\begin{remark}[Differentiability at the boundary and  exponents choice] $\empty$  \label{ren}\\
Observe that on a $C^1$  strictly convex  Lorentz-Minkowski space,  $\mathscr{L}$ need not be differentiable in $\p C$. Let us consider the homogeneous case $f=x^a/a$, $f^o=x^b/b$. If $a,b>1$ and $\mathscr{L}$ is $C^1$ on $C$, and $\dd \mathscr{L}\ne 0$ everywhere on $\p C$, then the condition $\dd \mathscr{F} \to \infty$ for $y \to \p C$ easily follows from $\mathscr{F}^{-1}(0)=\mathscr{L}^{-1}(0)=\p C$. Moreover, $\ell$ can be extended by continuity to $C$.


The most convenient choice of exponents $a,b>1$ can sometimes be inferred from the differentiability of $\mathscr{L}$ at $\p C$. Suppose there is one choice such that the differential of $\dd \mathscr{L}$ is continuous on $C$ and  $-a\dd \mathscr{L}=\dd \mathscr{F}^{a} \ne 0$ at every point of $\p C$. This means that the map $\ell$ does not send any point of $\p C$ to zero. No larger or smaller exponent $a'$ could be used to the same effect, for if $a'>a$, setting $\alpha=\frac{a'}{a}>1$ we have $ \dd \mathscr{F}^{a'}= \dd (\mathscr{F}^{a})^\alpha= \alpha (\mathscr{F}^{a})^{\alpha-1} \dd \mathscr{F}^{a}$ which implies $ \dd \mathscr{F}^{a'} =0$ at $\p C$. Similarly, for $1<a'<a$  we have $\alpha<1$ and $ \dd \mathscr{F}^{a'} \to \infty$ at $\p C$. In other words, the existence of conjugate exponents with the mentioned property implies their uniqueness. Once the most convenient $a$-positive homogeneous function $\mathscr{L}$ has been identified one can introduce the dependence on the base manifold through  a variable $x$ and impose suitable dynamical equations.


In any case one can also argue \cite{minguzzi15e} that, physically speaking,  not having an extended bijection $\ell$ among the {\em closed} cones could really be an interesting feature, since such bijection does not seem to be observable. All boils down to the fact that in relativity physics massive particles have a natural affine parameter, the proper time, in fact structured ones can decay with characteristic half-times, while lightlike particles do not carry a clock and so do not need to be associated to an  affine parameter, just to a lightlike direction.
\end{remark}

\subsection{Stable recurrent set} \label{fir}
In this section we establish the equivalence between stable causality and the antisymmetry of Seifert's relation $J_S$. The proofs coincide with those given in \cite{minguzzi07} save for some  modifications required by the generalization from Lorentzian cones to general cones. Only minimal changes are required: notice that convex neighborhoods in \cite{minguzzi07} are not really used, the local non-imprisoning property pointed out in Prop.\ \ref{iiu} is sufficient.
We provide the proofs for completeness.

\begin{definition}
The Seifert violating  set $vJ_S\subset M$ is given by those $p\in M$ for which there is $q\ne p$ such that $(p,q)\in J_S$ and $(q,p)\in J_S$.
\end{definition}
Clearly, $vJ_S=\emptyset$ if and only if $J_S$ is antisymmetric. We recall that $p$ belongs to the stable  recurrent set if  for  every $C^0$ (equiv.\ locally Lipschitz) proper cone structure $C'>C$ there is a closed continuous $C'$-causal curve passing through $p$.

\begin{theorem} \label{ssp}
Let $(M,C)$ be a closed cone structure. The set $vJ_S$ is closed and coincides with the stable recurrent set.
\end{theorem}

\begin{proof}
Suppose that $p\in vJ_S$ and let $q\ne p$ be such that  $(p,q)\in J_S$ and $(q,p)\in J_S$, then by Prop.\ \ref{paq}  $(p,q)\in I_{C'}$ and $(q,p)\in I_{C'}$, thus there is a closed continuous $C'$-causal curve passing through $p$.

For the converse, Let $U\ni p$ be the non-imprisoning neighborhood constructed in Prop.\ \ref{iiu} and let $B\ni p$ be an open coordinate ball whose closure is contained in $U$. Let $C_k>C$ be the sequence of proper cone structures of Prop.\ \ref{sqd}  and let $\sigma_k$ be a closed continuous $C_k$-causal curve passing through $p$. Starting from $p$ let $q_k\in \sigma_k\cap \p B$ be the first escaping point from $B$. By the limit curve theorem \ref{main}  there is $q\in B$ and a continuous $C$-causal curve contained in $U$ joining $p$ to $q$, i.e. $(p,q)\in J\subset J_S$. Moreover, still by the limit curve theorem and Theorem \ref{sqd} $(q,p)\in J_S$.

Let $p_k\to p$ where $p_k\in vJ_S$. Let $C_k>C$ be the sequence of proper cone structure of Prop.\ \ref{sqd} and let $\sigma_k$ be closed $C_k$-timelike curves passing through $p_k$.  Let $U\ni p$ be the non-imprisoning neighborhood constructed in Prop.\ \ref{iiu} and let $B\ni p$ be an open coordinate ball whose closure is contained in $U$. Starting from $p_k$ let $q_k\in \sigma^s_k\cap \p B$ be the first escaping point from $B$. Up to subsequences $q_k\to q\in \p B$. By the same argument used above $(p,q)\in J_S$ and $(q,p)\in J_S$, thus $p\in vJ_S$.
\end{proof}

In the next proofs all the cone structures wider than $C$ are locally Lipschitz proper cone structures. A cone structure is strongly causal at $p$ if $p$ admits arbitrarily small causally convex open neighborhoods. It is strongly causal if it is strongly causal everywhere.

\begin{lemma} \label{spm}
Let $(M,C)$ be a closed cone structure. The set at which the cone structure is strongly causal is open.
\end{lemma}

\begin{proof}
Let $(M,C)$ be strongly causal at $p$ and let $U\ni p$ be the non-imprisoning neighborhood constructed in Prop.\ \ref{iiu} (recall that $U$ is obtained from the chronological diamond of a local (flat) Minkowski metric $g$). We know that there is a causally convex open neighborhood $V\ni p$, $\bar V\subset U$, so if $q\in V$, and $U'$ is an open neighborhood of $q$ we can find a $g$-chronological diamond $Q$, such that $\bar Q\subset U'\cap V$. Then $Q$ is a causally convex neighborhood for $q$, in fact every continuous $C$-causal curve $\gamma$ starting and ending in $Q$ and leaving $Q$ cannot be entirely contained in $U$, for otherwise it would be a continuous $g$-causal curve, thus violating the $g$-causal convexity of $Q$ in $U$. But then $\gamma$ would escape and reenter $U$ and hence $V$ in contradiction to the $C$-causal convexity of $V$. Since $q\in V$ is arbitrary $(M,C)$ is strongly causal at every point of $V$ which finishes the proof.
\end{proof}

\begin{lemma} \label{paf}
Let $(M,C)$ be a closed cone structure.
If the locally Lipschitz proper cone structure $(M, \check C)$, $\check C> C$, is causal at $x$ then for every locally Lipschitz proper cone structure $C'$, $C<C'<\check C$, $(M,C')$
is strongly causal at $x$.
\end{lemma}

\begin{proof}
If $(M,C')$ is not strongly causal at $x$ then  there is  a non-imprisoning neighborhood $U \ni x $ as in Prop.\ \ref{iiu} and a sequence of
continuous $C'$-causal curves $\sigma_n$ of endpoints $x_n,z_n$, with $x_n\to
x$, $z_n \to x$, not entirely contained in $U$. Let $B$, $\bar B\subset U$ be a coordinate ball of $x$.
 Let $c_n\in \p {B}$ be the first
point at which $\sigma_n$ escapes $\bar B$, and let $d_n$ be the last
point at which $\sigma_n$ reenters $\bar B$. Since $\p {B}$ is compact
there are $c,d \in \p B$, and  a subsequence $\sigma_k$ such that
$c_k \to c$, $d_k \to d$. By the limit curve theorem $(x,c), (d,x)\in J_{C'}$, while $(c,d)\in \bar J_{C'}$. By Th.\ \ref{dxp} there is a closed $\check C$-timelike curve passing through $x$, a contradiction.
\end{proof}

We obtain another proof of the strong causality of stably causal closed cone structures (f.\ Th.\ \ref{mwh}).

\begin{corollary} \label{jjb}
Any stably causal closed cone structure $(M,C)$ is  strongly causal.
\end{corollary}

\begin{proof}
By the assumption there is $\check  C>C$ causal, and we can find $C<C'<\check C$, which is strongly causal by Lemma \ref{paf}, thus $(M,C)$ is itself strongly causal.
\end{proof}

\begin{theorem} \label{nkk}
Let $(M,C)$ be a closed cone structure. $(M,C)$ is stably causal if and only if  $J_S$ is antisymmetric (i.e.\ $vJ_S=\emptyset$).
\end{theorem}

\begin{proof}
Suppose that $(M,C)$ is stably causal then there is $C'>C$ such that $(M,C')$ is causal, thus $J_{C'}$ is antisymmetric and since $J_S\subset J_{C'}$, $J_S$ is antisymmetric.

For the converse suppose that $J_S$ is antisymmetric. In the course of the proof we shall have to take convex combinations of cones. By Prop.\ \ref{ohg} there is a Lipschitz 1-form $\omega$ such that the distribution of hyperplanes $P=\omega^{-1}(1)$ cuts $C$ as well as some $C'>C$ in compacts sets. The convex combinations of cones should be understood relative to $P$, as explained in the paragraph before Prop.\ \ref{doo}.

 By Th.\ \ref{ssp} for every $x\in M$ there is a ($x$ dependent) $\check C_x>C$ such that
$(M,\check C_x)$ is causal at $x$. By Lemma \ref{paf},
taking $C_x$ such that $C< C_x<\check{C}_x$, $(M,C_x)$ is strongly
causal at $x$  and hence it is strongly causal in an open
neighborhood $U_x$ of $x$ (Lemma \ref{spm}).

Let $K$ be a compact set. From the open covering $\{U_y, y \in K\}$, a
finite covering can be extracted $\{U_{y_1}, U_{y_2}, \dots,
U_{y_k}\}$. A cone structure $C_K>C$, on $M$ can be found such that for $
i=1,\ldots k$, $C_K<C_{y_i}$ on $M$. Thus $(M,C_K)$ is strongly
causal, and hence causal, on a open set $A=\cup_i U_{y_i}
\supset K$. Let $(C_n,K_n, A_n)$ be a sequence of cone fields $C_n>C$,
$C_{n+1}<C_{n}$, and strictly increasing compact sets and open sets
$K_n\subset A_n \subset K_{n+1}$, such that $(M,C_n)$ is
causal on $A_n$, and $\cup_n K_n=M$ (for instance introduce a
complete Riemannian metric and let $K_n$ contain the balls $\bar B(o,n)$ of radius
$n$ centered at $o \in M$). Let $\chi_n:M \to
[0,1]$ be locally Lipschitz functions such that $\chi_n=1$ on $K_n$, and
$\chi_n=0$ outside an open set $B_n$ such that
\[\cdots \subset K_n\subset B_n \subset \bar{B}_n\subset  A_n \subset
K_{n+1}\subset B_{n+1} \subset \bar{B}_{n+1} \subset  A_{n+1}
\subset \cdots.\] We construct a cone field $C'>C$ on $M$ as follows.
The cone structure $C'$ on $K_{n+1}\backslash B_n$ coincides with $C_{n+1}$, and on $B_{n}\backslash K_{n}$ its intersection with $P$ is given by $\chi_n
\tilde C_n+(1-\chi_n)  \tilde C_{n+1}$.

The spacetime $(M,C')$ is causal otherwise there would be a
closed continuous $C'$-causal curve $\gamma$. Let $i$ be the minimum
integer such that $\bar{B}_i \cap \gamma \ne \emptyset$, and let $p
\in \bar{B}_i \cap \gamma$. Then $\gamma$ is also a closed
continuous $C_i$-causal curve in $(M,C_i)$, thus $C_i$-causality is violated at
$p \in \bar{B}_i \subset A_i$ a contradiction.
Thus $(M,C')$ is causal so $(M,C)$ is stably causal.
\end{proof}

\subsection{Hawking's averaging for closed cone structures} \label{aho}
In 1968 Stephen Hawking showed how to construct a time function by taking a suitable average of volume functions relative to wider (round) cone structures \cite{hawking68,hawking73}. In this section we wish to show that the method still works for general closed cone structures.
This result might not be immediately obvious since the original proofs used the  existence of convex neighborhoods.  Once again we show that the result  depends only on the local non-imprisoning properties of spacetime.

Actually, in the mentioned works  Hawking did not provide the details of the proof of the lower semi-continuity of his function; a proof which is not so trivial after all. So this section is likely to be useful to any researcher looking for an introduction  to this technique.

Unfortunately, whereas in the regular $C^2$ theory we have convex neighborhoods at our disposal and Hawking's time function can be shown to be locally anti-Lipschitz and hence smoothable \cite{chrusciel13}, here we lack convex neighborhoods and so we are unable to prove such a property. Still by using a product trick we shall be able to construct  locally anti-Lipschitz functions by using an averaging argument similar to Hawking's, cf.\ Sec.\ \ref{mis}. Therefore, this section will be useful since the average idea will return in some key arguments of this work.

Let $(M,C)$ be a stably causal closed cone structure. By stable causality and  Prop.\ \ref{sqd} there is    a  locally Lipschitz proper cone structure $C_3>C$ which is itself stably causal. By Prop.\ \ref{sqd} we can find a locally Lipschitz proper cone structure $C<C_0<C_3$.
The convex combination $C_a=(1-\frac{a}{3}) C_0+\frac{a}{3} C_3$ (with respect to some smooth hyperplane section $P$) will also be a locally Lipschitz proper cone structure for $a\in [0,3]$. Let $\mu$ be a unit measure on $M$, on every chart absolutely continuous with respect the Lebesgue measure induced by the chart. The Hawking's time function is defined with the next expression. For $p\in M$ let
\[
 t(p)=\int_1^2\theta(p,a) \dd a, \ \textrm{ where } \ \theta(p,a)=\mu(I^{-}_{C_a}(p)).
\]

\begin{theorem} \label{has}
Let $(M,C)$ be a stably causal closed cone structure, then $t$ is a time  function for $(M,C')$ where $C'$ is a  locally Lipschitz proper cone structure such that $C'>C$ .
\end{theorem}

Actually, the proof of continuity does not use the assumption of stable causality.

\begin{proof}
We start with a claim.

{\em Claim 1}. Let $\epsilon>0$ and $N\ge 2/\epsilon$. Let $B$ be a coordinate ball centered at $p$, of volume $\mu(B)\le \epsilon/2$, whose  closure is contained in the non-imprisoning neighborhood $U$ for the cone structure $C_3$ constructed in Prop.\ \ref{iiu}.  There is a neighborhood $G\ni p$, $G\subset U$, such that for $ q\in G$
\[
 I_{(U, C_{a})}^-(p) \cap \p B \subset I^-_{(U, C_{a+{1}/{N}})}(q) \cap \p B , \ \textrm {for every } \ a\in[1,2].
\]

{\em Proof of Claim 1}. Let $a_1,a_2\in [0,3]$, $a_1<a_2$, clearly by Th.\ \ref{dxp}
\[ J_{(U, C_{a_1})}^-(p) \cap \p B \subset I_{(U, C_{a_2})}^-(p) \cap \p B .
\]
The set on the left-hand side is compact by Th.\ \ref{xxl} and due to $J_S$ being closed. Let us denote it $K$. It is covered by sets of the form $I_{(U, C_{a_2})}^-(r)$ with $r\in I_{(U, C_{a_2})}^-(p)$, thus there is a finite covering $\{I_{(U, C_{a_2})}^-(r_i)\}$. The set  $\mathcal{G}(a_1,a_2):= \cap_i I_{(U, C_{a_2})}^+(r_i)$ is such that for every $q\in \mathcal{G}(a_1,a_2)$
\[
I_{(U, C_{a_1})}^-(p) \cap \p B \subset K\subset I_{(U, C_{a_2})}^-(q) \cap \p B .
\]
Let us regard $[1,2+\frac{1}{N}]$ as the union of intervals of length $1/(2N)$, $\mathcal{I}_k=[1+\frac{k}{2N}, 1+\frac{k+1}{2N}]$, $k=0,\cdots, 2N+1$, in such a way that inside every interval $[a,a+\frac{1}{N}]$ for $a\in [1,2]$ we can find an $\mathcal{I}_{\bar k}$ interval for some $\bar k$. Let
\[G=\cap_{k} \mathcal{G}(1+\tfrac{k}{2N}, 1+\tfrac{k+1}{2N}), \]
For $a\in [1,2]$ and $q\in G$
\begin{align*}
I_{(U, C_{a})}^-(p) \cap \p B &\subset I_{(U, C_{1+{\bar k}/(2N)})}^-(p) \cap \p B \\ &\subset I_{(U, C_{1+({\bar k}+1)/(2N)})}^-(q) \cap \p B \subset I_{(U, C_{a+1/N})}^-(q) \cap \p B .
\end{align*}

{\em Lower semi-continuity}. Let $\epsilon >0$, $N>2/\epsilon$, $B$ and $G$ as in Claim 1. For $a\in [1,2]$ and $q\in G$,
\[
I_{(U, C_{a})}^-(p) \cap \p B \subset I^-_{(U, C_{a+{1}/{N}})}(q) \cap \p B ,
\]
thus
$I^{-}_{C_a}(p) \backslash \bar B\subset I^{-}_{C_{a+1/N}}(q) \backslash \bar B$ which implies $I^{-}_{C_a}(p)\subset I^{-}_{C_{a+1/N}}(q) \cup \bar B$, which taking the volume gives
\[
\theta(p,a)\le \theta(q,a+1/N)+ \mu(B)\le  \theta(q,a+\tfrac{\epsilon}{2})+ \tfrac{\epsilon}{2}
\]
Integrating
\[
t(p)=\int_1^2 \theta(p,a) \dd a\le \int_1^2\theta(q,a) \dd a+\int_2^{2+\tfrac{\epsilon}{2}}\theta(q,a) \dd a+\tfrac{\epsilon}{2}\le t(q)+\epsilon.
\]

{\em Claim 2}. Let $\epsilon>0$ and $N\ge 2/\epsilon$. Let $B$ be a coordinate ball centered at $p$, of volume $\mu(B)\le \epsilon/2$, whose  closure is contained in the non-imprisoning neighborhood $U$ for the cone structure $C_3$ constructed in Prop.\ \ref{iiu}.  There is a neighborhood $G\ni p$, $G\subset U$, such that
\[
I^-_{(U, C_a)}(G) \cap \p B \subset  I_{(U, C_{a+1/N})}^-(p) \cap \p B , \ \textrm {for every } \ a\in[1,2].
\]

{\em Proof of Claim 2}. Let $a_1,a_2\in [1,3]$, $a_1<a_2$, clearly by Th.\ \ref{dxp}
\[
J_{(U, C_{a_1})}^-(p) \cap \p B\subset I_{(U, C_{a_2})}^-(p) \cap \p B.
\]
By the limit curve theorem and the non-imprisoning property of $U$ there must be a neighborhood  $\mathcal{G}(a_1,a_2)$ such that
\[
J_{(U, C_{a_1})}^-(\mathcal{G}(a_1,a_2)) \cap \p B\subset I_{(U, C_{a_2})}^-(p) \cap \p B.
\]
Let us regard $[1,2+\frac{1}{N}]$ as the union of intervals of length $1/(2N)$, $\mathcal{I}_k=[1+\frac{k}{2N}, 1+\frac{k+1}{2N}]$, $k=0,\cdots, 2N+1$, in such a way that inside every interval $[a,a+\frac{1}{N}]$ for $a\in [1,2]$ there is an interval $\mathcal{I}_{\bar k}$ for some $\bar k$. Let
\[G=\cap_{k} \mathcal{G}(1+\tfrac{k}{2N}, 1+\tfrac{k+1}{2N}), \]
then
\begin{align*}
I^-_{(U, C_a)}(G) \cap \p B &  \subset  I^-_{(U, C_a)}(\mathcal{G}(1+\tfrac{\bar k}{2N}, 1+\tfrac{\bar k+1}{2N})) \cap \p B  \\
& \subset I^-_{(U, C_{1+{\bar k}/(2N)})}(\mathcal{G}(1+\tfrac{\bar k}{2N}, 1+\tfrac{\bar k+1}{2N})) \cap \p B  \subset I^-_{(U, C_{1+(\bar k+1)/(2N)})}(p)\cap \p B  \\
& \subset I^-_{(U, C_{a+{1}/{N}})}(p)\cap \p B
\end{align*}

{\em Upper semi-continuity}. Let $\epsilon >0$, $N>2/\epsilon$, $B$ and $G$ as in Claim 2. For $a\in [1,2]$ and $q\in G$,
\[
I^-_{(U, C_a)}(q) \cap \p B \subset I^-_{(U, C_{a+{1}/{N}})}(p)\cap \p B,
\]
thus $I^-_{ C_a}(q) \backslash \bar B \subset I^-_{C_{a+1/N}}(p)\backslash \bar B $ which implies $I^-_{ C_a}(q) \subset I^-_{C_{a+1/N}}(p)\cup \bar B$, which taking the volume gives
\[
\theta(q,a)\le \theta(p,a+1/N)+ \mu(B)\le  \theta(p,a+\tfrac{\epsilon}{2})+ \tfrac{\epsilon}{2}
\]
Integrating
\[
t(q)=\int_1^2 \theta(q,a) \dd a\le \int_1^2\theta(p,a) \dd a+\int_2^{2+\tfrac{\epsilon}{2}}\theta(p,a) \dd a+\tfrac{\epsilon}{2}\le t(p)+\epsilon.
\]

{\em Time function}. Since $C_3$ is causal, $C_2$ is strongly causal (Lemma \ref{paf}), so if $p\in M$ and $U$ is a $C_2$-causally convex neighborhood of $p$, given any continuous $C_{1/2}$-causal curve contained in $U$ and starting from $p$, the endpoint $q$ must have a larger value of $t$, $t(q)>t(p)$, as for every $a\in [1,2]$, $\mu(I^{-}_{C_a}(q))>\mu( I^{-}_{C_a}(p))$.
%
%
%
%
%
\end{proof}

\subsection{Anti-Lipschitzness and the product trick} \label{mis}

In Sec.\ \ref{ngd} we have defined the notion of Lorentz-Finsler space $(M,\mathscr{F})$, and in Sec.\ \ref{mvb} the notions of stable distance and stable Lorentz-Finsler space. In this section
we write $\mathscr{F}'>\mathscr{F}$, with no mention to the cone domains, if
$(M,\mathscr{F})$ is a closed Lorentz-Finsler space, $(M,\mathscr{F}')$ is a locally Lipschitz proper Lorentz-Finsler space, and $C'{}^\times > C^\times$, which is equivalent to
$C'>C$ and $\mathscr{F}'>\mathscr{F}$ on $C$.

Let us introduce a different cone structure in $M^\times=M\times\mathbb{R}$ defined at $P=(p,r)$ by
\begin{equation} \label{xyb}
C{}^\downarrow_{P}=\{(y,z) \colon y\in C_p\cup \{0\}, \ z \le \mathscr{F}(y) \}\backslash \{(0,0)\}.
\end{equation}

\begin{theorem} \label{soq}
Let  $(M,\mathscr{F})$ be a stably causal closed Lorentz-Finsler space, then $(M^\times, C{}^\downarrow)$ is strongly causal.
\end{theorem}

\begin{proof}
It is sufficient to prove strong causality  at $P=(p,0)\in M^\times$, $p\in M$. We already know that stable causality implies strong causality, cf.\ Th.\ \ref{mwh}. Let $V$ be a  $C$-causally convex open set which is also globally hyperbolic. One can easily construct a time function $\tau$ on $V\times \mathbb{R}$, for sufficiently small $V$, e.g.\ one whose level sets are obtained by vertically translating a local $C{}^\downarrow$-spacelike $C^1$ hypersurface passing through $P$ which intersects $V\times \mathbb{R}$ on a relatively compact set (notice that $C{}^\downarrow$ is sharp). Let $\tau$ be such that $\tau(P)=0$.

Let $U=\tau^{-1}(-\delta,\delta)\subset V\times\mathbb{R}$ be an open neighborhood of $P$, and let us consider a parametrized $C{}^\downarrow$-causal curve $\Gamma$ which starts from some point of $U$. We have to show that it cannot reenter $U$ once it escapes $U$. First we show that it cannot escape $V\times\mathbb{R}$.
The curve $\Gamma$ is a absolutely continuous  with derivative in $C{}^\downarrow$ a.e. 
so its projection $\gamma$ (the projection to $M$ is Lipschitz, and the composition $g\circ f$, with $f$ AC and $g$ Lipschitz, is AC) is absolutely continuous with derivative in $C'\cup\{0\}$ a.e., and so $\gamma$ is a continuous causal curve.
As a consequence, if $\Gamma$ escapes  $V\times\mathbb{R}$, then $\gamma$ escapes $V$ and so it cannot reenter it.

But if $\Gamma$ remains in $V\times\mathbb{R}$, $\tau$ is increasing over it, so once it escapes $U$ it cannot reenter it.
\end{proof}%

\begin{theorem} \label{ssh}
Let  $(M,\mathscr{F})$ be a closed Lorentz-Finsler space, then $(M,C)$ is stably causal if and only if $(M^\times, C{}^\downarrow)$ is stably causal.
\end{theorem}

The proof is really the first step in the proof of the next theorem so it introduces a few more structures than strictly required.
\begin{proof}
It is clear that the stable causality of $(M^\times, C{}^\downarrow)$ implies the stable causality of $(M,C)$ so we shall be concerned with the other direction.

Let $\mu$ be a strictly positive unit measure on $M^\times$, absolutely continuous with respect to the Lebesgue measure of any chart. Let  $(M,\mathscr{F})$ be a stable closed Lorentz-Finsler space and let $\mathscr{F}'>\mathscr{F}$ be such that $(M,\mathscr{F}')$ is a stably causal   locally Lipschitz proper Lorentz-Finsler space. Let
$\mathscr{F}_3>\mathscr{F}_0>\mathscr{F}'$  be other stably causal   locally Lipschitz proper Lorentz-Finsler spaces.


 By Prop.\ \ref{ohg} we can find a Lipschitz 1-form $\omega$ such that $P=\omega^{-1}(1)$ is a distribution of planes cutting  $C_3$ over compact subsets. In particular $P^\times=P\times \mathbb{R}$ cuts $C_3^\times$ over compact subsets. The next convex combinations of cones are defined with respect to $P^\times$.
 Let  $C_a$, and $\mathscr{F}_a\colon C_a \to [0,+\infty)$ be defined through the  convex combination $C^\times_a=(1-\frac{a}{3})C^\times_0+\frac{a}{3} C_3^\times$, $a\in[0,3]$. We have ${\mathscr{F}}_a<{\mathscr{F}}_{a'}$ for $a<a'$.

Let $C_a{}^\downarrow_{P}=\{(y,z) \colon y\in (C_a)_p\cup\{0\}, \ z \le \mathscr{F}_a(y) \}\backslash \{(0,0)\}$ where $P=(p,r)$.
Observe that it is not true that $C^\downarrow_a<C^\downarrow_{a'}$, for $a<a'$, since both share the downward vertical vectors $(0,z)$, $z< 0$. Still, we are going to construct a time function on $(M^\times, C_a^\downarrow)$ by using an averaging procedure analogous to that employed by Hawking \cite{hawking68,hawking73} in which, however, cones do not open in the fiber direction.

By strong causality (distinction suffices) of $(M^\times, C^\downarrow_a)$ the function  $t^\downarrow_a(P)=-\mu(I^+_{C^\downarrow_a }(P))$, is increasing over every $C^\downarrow_a$-causal curve and $t^\downarrow_a<t^\downarrow_{a'}$ for $a<a'$. However, it is not necessarily continuous, so the idea is to take the average
\[
t^\downarrow(P)=\int_1^2 t^\downarrow_a(P) \dd a = -\int_1^2 \mu(I^+_{C^\downarrow_a }(P)) \dd a .
\]
It suffices to prove continuity at $P=(p,0)$, $p\in M$. Let $\epsilon >0$,
 let $V$ be a  $C_3$-causally convex (hence $C_a$-causally convex)  open neighborhood of $p$, constructed as in Prop.\ \ref{iiu} to get a bounded  $h$-arc length of $C_3$-causal curves contained in $V$ where $h$ is a  Riemannian metric such that  $\mathscr{F}_3(\cdot)\le \Vert \cdot \Vert_h$ on $TM$.
From the proof of Theorem \ref{soq} we know that if $V$ is sufficiently small, every continuous $C_3^\downarrow$-causal curve escaping $W:=V\times \mathbb{R}$ cannot reenter it so it  intersects $\p W$ only once. We take $V$ so small that $\mu (W) <\epsilon/2$. Now, observe that $Q\in (J^\downarrow_a)^+(P)$, if $Q=(q,r)$, where there is a  continuous $C_a$-causal curve $\gamma$ connecting $p$ to $q$ and $r \le \ell_a(\gamma)$. So for $P_r=(p,r)$ we also have $Q\in (J^\times_a)^+(P_r)$, in other words $(J^\downarrow_a)^+(P)=\cup_{r\le 0} (J^\times_a)^+(P_r)$.

For $a,a'\in [0,3]$, $a<a'$, we have by Th.\ \ref{dxp}
\[
(J^\times_a)^+(P)\cap \p W\subset (I^\times_{a'})^+(P)\cap \p W.
\]
For sufficiently large $\delta$ both sides of this inclusion are contained in the compact boundary of the $C^\downarrow_3$-causally convex set $D= \tau^{-1}(-\delta,\delta)\subset V\times\mathbb{R}$  constructed in the proof of the previous theorem (because there is a Riemannian metric $h$ such that $\mathscr{F}_3(y)\le \Vert y\Vert_h$ on $TV$, and the $h$-arc length of $C_3$-causal curves is bounded on $V$). (Notice that both sides in the previous inclusion could be written with respect to the relations $J^\times_a(O)$ or $I^\times_{a'}(O)$, where $O$ is a relatively compact neighborhood of $D$ since $D$ is $C^\times_3$-causally convex convex.)

By the same limit curve argument presented in the proofs of the claims in Th.\ \ref{has},  there is an open neighborhood $\mathscr{A}(a,a')\ni P$ such that
\[
(J^\times_a)^+(Q)\cap \p W\subset (I^\times_{a'})^+(P)\cap \p W, \quad \forall Q \in \mathcal{A}(a,a').
\]
and an open neighborhood $\mathscr{B}(a,a')\ni P$ such that
\[
(J^\times_a)^+(P)\cap \p W\subset (I^\times_{a'})^+(Q)\cap \p W, \quad \forall Q \in \mathcal{B}(a,a').
\]
By translational invariance similar inclusions hold for $P_r$ in place of $P$, where the novel sets $\mathcal{A}_r$, $\mathcal{B}_r$, are the translates of $\mathcal{A}$ and $\mathcal{B}$. Thus
\begin{align*}
(J^\downarrow_a)^+(Q)\cap \p W &\subset (I^\downarrow_{a'})^+(P)\cap \p W, \quad \forall Q \in \mathcal{A}(a,a'), \\
(J^\downarrow_a)^+(P)\cap \p W&\subset (I^\downarrow_{a'})^+(Q)\cap \p W \quad \forall Q \in \mathcal{B}(a,a').
\end{align*}
Let $N$ be an integer such that $N>2/\epsilon$, and let us regard $[1,2+\frac{1}{N}]$ as the union of intervals $\mathcal{I}_k=[1+\frac{k}{2N}, 1+\frac{k+1}{2N}]$, $k=0,\cdots, 2N+1$, in such a way that inside every interval $[a,a+\frac{1}{N}]$ for $a\in [1,2]$ there is an interval $\mathcal{I}_{\bar k}$  for some $\bar k$. Let \[A=\cap_{k} \mathcal{A}(1+\tfrac{k}{2N}, 1+\tfrac{k+1}{2N}), \qquad B=\cap_{k} \mathcal{B}(1+\tfrac{k}{2N}, 1+\tfrac{k+1}{2N}).\]

{\em Lower semi-continuity}. Let $Q\in A$ and $a\in [1,2]$; choosing $\mathcal{I}_k\subset [a,a+\frac{1}{N}]$, we have $Q\in \mathcal{A}(1+\tfrac{k}{2N},1+ \tfrac{k+1}{2N})$ and
\begin{align*}
(J^\downarrow_a)^+(Q)\cap \p W &\subset (J^\downarrow_{1+k/(2N)})^+(Q)\cap \p W  \subset (I^\downarrow_{1+(k+1)/(2N)})^+(P)\cap \p W \\& \subset(I^\downarrow_{a+{1}/{N}})^+(P)\cap \p W.
\end{align*}
Thus $(I^\downarrow_a)^+(Q) \backslash W\subset (I^\downarrow_{a+{1}/{N}})^+(P)\backslash W$ hence
\[
\mu((I^\downarrow_a)^+(Q)) \le \mu((I^\downarrow_{a+{1}/{N}})^+(P))+\mu(W)\le \mu((I^\downarrow_{a+\frac{\epsilon}{2}})^+(P))+\tfrac{\epsilon}{2}.
\]
That is, for every $Q\in A$ and $a \in [1,2]$
\[
-t^\downarrow_a(Q) \le - t^\downarrow_{a+\frac{\epsilon}{2}}(P)+\tfrac{\epsilon}{2} ,
\]
 and averaging (notice that $-1\le t^\downarrow_s\le 0$)
\begin{align*}
-t^\downarrow(Q)=- \int_1^2 t^\downarrow_a(Q) \dd a\le -t^\downarrow(P)-\int_2^{2+\frac{\epsilon}{2}}  t^\downarrow_{s} (P)\dd s +\frac{\epsilon}{2}\le  -t^\downarrow(P)+\epsilon,
\end{align*}
which proves the lower semi-continuity.

{\em Upper semi-continuity}. Let $Q\in B$ and let $a\in [1,2]$; choosing $\mathcal{I}_k\subset [a,a+\frac{1}{N}]$, we have $Q\in \mathcal{B}(1+\tfrac{k}{2N}, 1+\tfrac{k+1}{2N})$ and
\begin{align*}
(I^\downarrow_a)^+(P)\cap \p W&\subset  (I^\downarrow_{1+k/(2N)})^+(P)\cap \p W \subset (I^\downarrow_{1+(k+1)(2N)})^+(Q)\cap \p W \\
&\subset (I^\downarrow_{a+1/N})^+(Q) \cap \p W.
\end{align*}
Thus $(I^\downarrow_a)^+(P)\backslash W \subset (I^\downarrow_{a+{1}/{N}})^+(Q)\backslash W$ hence
\[
\mu((I^\downarrow_a)^+(P)) \le \mu((I^\downarrow_{a+{1}/{N}})^+(Q))+\mu(W)\le \mu((I^\downarrow_{a+\frac{\epsilon}{2}})^+(Q))+\tfrac{\epsilon}{2}
\]
That is, for every $Q\in B$ and $a \in [1,2]$
\[
-t^\downarrow_a(P)\le  - t^\downarrow_{a+\frac{\epsilon}{2}}(Q)+\tfrac{\epsilon}{2}
\]
and averaging (notice that $-1\le t^\downarrow_s\le 0$)
\begin{align*}
-t^\downarrow(P)\le -t^\downarrow(Q)-\int_2^{2+\frac{\epsilon}{2}}  t^\downarrow_{s} (Q)\dd s +\frac{\epsilon}{2}\le  -t^\downarrow(Q)+\epsilon,
\end{align*}
which proves the upper semi-continuity.
Thus $t^\downarrow$ is a time function for $(M^\times, C_1^\downarrow)$ which, therefore, is stably causal.
\end{proof}

\begin{theorem}(existence of anti-Lipschitz functions in stable spacetimes) \label{mih} \\
Let $(M,\mathscr{F})$ be a stable closed Lorentz-Finsler space, and let $(M,\mathscr{F}')$ be a stable locally Lipschitz proper Lorentz-Finsler space such that $\mathscr{F}'>\mathscr{F}$ (it exists as shown in Prop.\ \ref{con}). Then there is a continuous function  $t\colon M\to \mathbb{R}$ which is strictly $\mathscr{F}'$-anti-Lipschitz, namely such that for every continuous $C'$-causal curve $\sigma\colon [0,1] \to M$
\[
t(\sigma(1))-t(\sigma(0))>\int_\sigma \mathscr{F}'(\dot \sigma) \dd t.
\]
Moreover,  (a) given two points such that $(p,q)\notin J_S$ we can find $t$ so that $t(p)>t(q)$, and (b)  given $p\in M$ and an open neighborhood $O\ni p$ we can find $\check t$ and $\hat t$ continuous strictly $\mathscr{F}'$-anti-Lipschitz functions such that  $p \in [\{q\colon \check t(q)<0\}\cap \{q\colon \hat t(q)>0\}] \subset O$.
\end{theorem}

The idea is to show that there is a time function on $(M^\times, C'{}^\downarrow)$, such that its zero level set $S_0$ intersects exactly once every $\mathbb{R}$-fiber of $M^\times$. This set $S_0$ regarded as a graph over $M$ provides the anti-Lipschitz  time function.

\begin{remark} \label{sof}
The function $t$ constructed in this theorem is really stably locally anti-Lipschitz (Sec.\ \ref{nug}). Indeed let $\check C$ be a locally Lipschitz proper cone structure such that $C<\check C< C'$, then the indicatrix $\mathscr{F}'{}^{-1}(1)$ intersects $\check C$ in a compact set. Let $h$ be a Riemannian metric whose unit balls contain such intersection then for every $\check C$-causal vector $y$, $\mathscr{F}'(y)\ge \Vert y\Vert_h$, thus if $\sigma\colon [0,1] \to M$ is a continuous  $\check C$-causal curve, $t(\sigma(1))-t(\sigma(0))>\ell'(\sigma)\ge \ell^h(\sigma)$.
\end{remark}

\begin{proof}
This proof is the continuation of the previous one. The only difference is that in the first step we let $\mathscr{F}',\mathscr{F}_3,\mathscr{F}_0$, be   stable   locally Lipschitz proper Lorentz-Finsler spaces, which exist by Th.\ \ref{con} (hence $D',D_3,D_0$ are finite). In the previous proof we constructed
 $t^\downarrow$,  a time function for $(M^\times, C_1^\downarrow)$. Observe that the particular shape of the cone $C_1^\downarrow$, that is the fact that it contains a vertical half-line, implies that the level sets of $t^\downarrow$ can intersect the fiber at most once, in fact the fibers are $C_1^\downarrow$-causal and so $t^\downarrow$ strictly increases over every fiber. Unfortunately, the level sets of $t^\downarrow$ might `go to infinity' before crossing some fibers. This circumstance is cured as follows.

Let  $t^\uparrow$ be the  time function that one would obtain taking the opposite cones  on $M^\times$. Both are time function on $(M^\times, C_1^\downarrow)$, where $t^\downarrow$ uses the measure of the chronological futures to build the time function, while $t^\uparrow$ uses the measure of the chronological pasts. The important point is that over a given fiber $(p,r)$, $r\in \mathbb{R}$, $t^\downarrow \to 0$ for $r \to -\infty$, and  $t^\uparrow \to 0$ for $r \to +\infty$. Let us prove this claim for $t^\downarrow$, the other claim being proved dually. Let $\epsilon>0$, and let $K\times [-G,G]$ be a compact set such that $p\in K$ and $\mu(M^\times\backslash K\times [-G,G])<\epsilon$. Since ${D}_3$ is upper semi-continuous, ${D}_3(p, \cdot)$ has an upper bound $R$ on $K$. Hence for every $a\in [1,2]$, $d_a(p, \cdot)< R$ on $K$.  As a consequence $(I^\downarrow_a)^+((p,-R-G))\cap \{K\times [-G,G]\}=\emptyset$, for every $a\in [1,2]$, which implies $\vert t ^\downarrow((p,-R-G))\vert<\epsilon$. The $C_1^\downarrow$-time  function $\tau=\log \vert t^\uparrow/t^\downarrow\vert $ is   continuous and strictly monotone over every $\mathbb{R}$-fiber with image $(-\infty,+\infty)$ as it goes to $\pm \infty$ for $r\to \mp \infty$ (the future direction for $C^\downarrow_a$ over the fiber corresponds to decreasing $r$). The level set $S_0=\tau^{-1}(0)$ being $C_1^\downarrow$-acausal  provides the graph of the searched function $t$.  In fact, let $\sigma$ be a continuous $C_1$-causal curve $\sigma\colon [0,1]\to M$, then $(\sigma(t),\ell_1(\sigma\vert_{[0,t)}))$ is a continuous $C_1$-causal curve. By definition of $t$, $(\sigma(0),t(\sigma(0))) \in S_0$. Function $\tau$ increases  over $\sigma$, thus $t(\sigma(1))>t(\sigma(0))+\ell_1(\sigma))$. Since $\mathscr{F}_1>\mathscr{F}'$ the first statement is proved.


Let us prove (a). Suppose to have been given $(p,q)\notin J_S$ then we can choose $C_3$ in the above construction in such a way that $(p,q)\notin \bar J_3$. Moreover, in the definition of $t^\downarrow$ and $t^\uparrow$ we are free to use different measures $\mu^\downarrow$ and $\mu^\uparrow$. We are going to alter $\mu^\downarrow$ by dispacing it over $M\times \mathbb{R}$ while keeping the extra coordinate invariant.
 Since $O:=I_3^+(q)\backslash \overline{ J_3^+(p)}\ne \emptyset$ we can move most of the measure (not all since its density with respect to Lebesgue has to be positive) $\mu^\downarrow(\overline{ J_3^+(p)}\times \mathbb{R})$ to the fiber of the open set $O$ so non-decreasing $\vert t^\downarrow\vert $ over the fiber of $q$ while decreasing as much as desired $\vert t^\downarrow\vert$ over the fiber of $p$.
 Defining $\tau=\log \vert t^\uparrow/t^\downarrow\vert $ the operation is used to alter the graphing function $t$ of the level set $S_0=\tau^{-1}(0)$ over $p$ and $q$, in such a way that $t(p)>t(q)$.

Let us prove (b). In the proof of the first statement we have found $\mathscr{F}'':=\mathscr{F}_1>\mathscr{F}'$, and a function $t$ such that for every continuous $C''$-causal curve $\sigma\colon [0,1]\to M$, $t(\sigma(1))-t(\sigma(0))>\ell''(\sigma)$, in particular, for every $C'$-causal curve $\sigma\colon [0,1]\to M$, $t(\sigma(1))-t(\sigma(0))>\ell'(\sigma)$. We introduce locally Lipschitz proper cone structures $C_0$ and $C_3$, not to be confused with those appearing in the previous steps of this proof (which we do not use anymore), such that $C'<C_0<C_3<C''$. Let $p\in M$, and $O\ni p$. Without loss of generality we  can assume $t(p)=0$. From $C_0$ and $C_3$ we define $C_a$, $a\in [0,3]$, introduce a measure $\mu$ and  build a  $C'$-time function $\tau$ a la Hawking as done in Sec.\ \ref{aho}, then $t+\tau$ is also $C'$-anti-Lipschitz.

Since $C_2$ is stably causal it is strongly causal and so there is a $C_2$-causally convex open neighborhood $U\subset O$, $p\in O$, such that $U$ is $C_2$-non-imprisoning and $J_2(U)$ is closed, cf.\ Prop.\ \ref{iiu} and \ref{dao}.
Let $B\ni p$ be a compact neighborhood, $B\subset U$, then since  $J_2(U)$ is closed, $J^+_2(p,U)\cap \p B$ is  a compact set. As $t$ increases over every $C_2$-causal curve (because $C_2<C''$), it is positive at every point of $J^+_2(p,U)\cap \p B$ and hence there is $\epsilon>0$ such that $J_2^+(p,U)\cap \p B$ stays in  the region $t>\epsilon$. By the limit curve theorem there is a neighborhood $U'\subset B$, $p\in U'$, such that $J^+_2(U',U)\cap \p B$ stays in  the region $t>\epsilon$. An analogous argument in the past case leads to the definition of the set $U''$. Thus let $V\subset U'\cap U''$, $p\in V$, be a $C_2$-causally convex compact neighborhood,   $J^+_2(V,U)\cap \p B$ stays in  the region $t>\epsilon$ (and $J^-_2(V,U)\cap \p B$ stays in  the region $t<-\epsilon$), thus
$J_2^+(V)\subset O\cup t^{-1}((\epsilon,+\infty))$ since $t$ is a $C_2$-time function, and similarly in the past case.

Let $\mu$ be supported in  $I^-_1(p) \cap V$, then recalling that  $\tau(r)=\int_1^2\mu(I^-_a(r))\dd a$ we have $\tau=0$ outside $J_2^+(V)$. By construction $\tau\ge 0$, thus $\{q\colon \tau(q)>0\}\subset O \cup \{q\colon t(q)>0\}$. Defining $\hat t=t+\tau$, we have $\{q\colon \hat t(q)>0\}\subset O \cup \{q\colon t(q)>0\}$.

A similar construction with $\mu$ supported in  $I^+_1(p) \cap V$ but constructing Hawking's function with the opposite cones $\tau(r)=-\int_1^2\mu(I^+_a(r))\dd a$, gives a function $\check t=t+\tau$, $\tau(p)<0$, such that $\{q\colon \check t(q)<0\}\subset O \cup \{q\colon t(q)<0\}$. Thus $p \in [\{q\colon \check t(q)<0\}\cap \{q\colon \hat t(q)>0\}] \subset O$.

\end{proof}

\subsection{Smoothing anti-Lipschitz functions}

For the next theorem and  corollary J.\ Grant, P.\ Chrusciel and the author should be credited, since it is really a polished and improved version of our theorem \cite[Th.\ 4.8]{chrusciel13}. I didn't change the original wording where it wasn't necessary. The new proof makes  manifest an important feature hidden in the original proof, namely the possibility of bounding the derivative of the smoothing function. Furthermore, it holds for general cone structures. Other techniques useful for smoothing increasing functions can be found in \cite{camilli09,bernard16}.

\begin{theorem} \label{moz}
  Let $({ M},C)$ be a closed cone structure and
 let $\tau\colon M\to \mathbb{R}$ be a continuous function. Suppose that there is a $C^0$ proper cone structure $\hat C>C$ and continuous functions  homogeneous of degree one on the fiber $\underline F, \overline F\colon \hat C\to \mathbb{R}$ such that for every  $\hat C$-timelike  curve $x\colon [0,1]\to M$
 \begin{equation} \label{mos}
 \int_x \underline F(\dot x) \dd t\le \tau(x(1))-\tau(x(0))\le \int_x \overline F(\dot x) \dd t.
 \end{equation}
 Let $h$ be an arbitrary Riemannian metric, then for every function $\alpha\colon { M} \to (0,+\infty)$ there exists
a smooth  function $\hat{\tau}$ such that $\vert \hat\tau-\tau\vert <\alpha$ and for every $v\in C$
\begin{equation} \label{kid}
\underline F(v)- \Vert v\Vert_h \le \dd \hat \tau(v) \le \overline F(v)+ \Vert v\Vert_h .
\end{equation}
Similar versions, in which some of the functions $\underline F, \overline F$ do not exist hold true. One has just to drop the corresponding inequalities in (\ref{kid}).
\end{theorem}

Since $h$ is arbitrary the last inequality can be made as stringent as desired, e.g.\ redefining the metric through multiplication by a small conformal factor.

%
\begin{proof}
Let $p\in { M}$  and let $\{x^\mu\}$ be local coordinates in a neighborhood of $p$. We rescale the coordinates in such a way that in a relatively compact neighborhood $V$ of $p$, the coordinate ball of $T_q V$, for every $q\in V$, contains the $h$-unit balls.

Let $B_p(3\epsilon(p))\subset V$ be a coordinate ball.
The coordinates split $TM$ over  $\overline{B_p(3\epsilon(p))}$ as $\overline{B_p(3\epsilon(p))}\times \mathbb{R}^{n+1}$ (which admits the coordinate sphere subbundle $\overline{B_p(3\epsilon(p))}\times \mathbb{S}^{n}$). The second projection, as induced by the local coordinate system, provides an identification of the fibers. Notice that if $(q,v)\in \overline{B_p(3\epsilon(p))}\times\mathbb{S}^{n}$ then $ \Vert v\Vert_h(q)\ge 1$.

 At $p$ we can find $\check C_p$ such that $C_p<\check C_p<\hat C_p$. By upper semi-continuity of $C$ and continuity of $\hat C$ the constant $\epsilon$ can be chosen so small that if we define $\check C=\overline{B_p(3\epsilon(p))}\times \check C_p$, then we still have $C<\check C<\hat C$ over the neighborhood. Since $\check C$ is translationally invariant if $v$ is $C$-causal at $q\in B_p(3\epsilon(p))$ then  the tangent vector to the curve $q'(s)=q_0'+vs$, $q_0'\in B_p(3\epsilon(p))$  is $\check C$-causal and hence $\hat C$-timelike as long as $q'(s)$ stays in the neighborhood.


Finally, $\underline F$  is a   continuous function, positive homogeneous of degree one, determined by its value on the compact set $\{\overline{B_p(3\epsilon(p))}\times \mathbb{S}^{n}\}\cap \hat C$, where it is uniformly continuous, so we can find $\epsilon$ so small and $\delta >0$ such that for every
$(q,v), (q',v')\in \{\overline{B_p(3\epsilon(p))}\times \mathbb{S}^{n}\}\cap \hat C$ with $d_{\mathbb{S}^{n}}(v,v')<\delta$, we have
\[
\vert\underline F(q',v')-\underline F(q,v)\vert < 1/2\le \Vert v\Vert_h(q)/2,
\]
and similarly for $\overline F$. In particular, if $v'=v$
\[
-\Vert v\Vert_h(q)/2<\underline F(q',v)-\underline F(q,v) < \Vert v\Vert_h(q)/2,
\]
%
which must also hold for $v$ not necessarily coordinate normalized since all functions appearing in this expression are positive homogeneous of degree one.

By $\sigma$-compactness there is  a locally finite
covering of ${ M}$ consisting of coordinate balls $\{{\mathscr O}_i:={B_{p_i}(\epsilon_i)}\}$, where $\epsilon_i$ is as above.
Let $\varphi_i$ be a smooth partition of unity subordinate
to the cover $\{{\mathscr O}_i\}$. Choose some
$0<\eta_j<\epsilon_j$. In local coordinates on ${\mathscr O}_j$ let $\tau_j$
be defined by convolution with an even
non-negative smooth function $\chi$,
supported in the coordinate ball of radius one, with integral one:
\[
 \tau_j (x) =\left\{
               \begin{array}{ll}
 \frac 1 {\eta_j^{n+1}} \int_{B_{p_j}(3\epsilon_j)} \chi\left(\frac{y-x}{\eta_j}\right) \tau (y)\, d^{n+1}y, & \hbox{$x\in B_{p_j}(2\epsilon_j)$;} \\
                 0, & \hbox{otherwise.}
               \end{array}
             \right.
\]
We define the smooth function
\[
 \hat \tau:= \sum_j \varphi_j \tau_j
 \;.
\]
The non-vanishing terms at each point are finite in number, and
$\hat\tau$ converges pointwisely to $\tau$ as we let the constants
$\eta_j$ converge to zero. The idea is to control the constants
$\eta_j$ to get the desired properties for $\hat\tau$.

Let $x\in M$ and $v\in C_x$, where $\Vert v\Vert_h=1$.
There is $j$ such that $x\in {\mathscr O}_j=B_{p_j}(\epsilon_j)$. In local coordinates  $v$ reads  $v = v^\mu
\partial_\mu$,  so that the $C^1$ curve $x^\mu (s)= x^\mu + v^\mu s$ is
$\hat C$-timelike
 as long as it stays within $B_{p_j}(3\epsilon_j)$. We observe that $s$ is not the $h$-arc length parametrization of the curve, however it will be sufficient to observe that $\Vert \frac{\dd }{\dd s}\Vert_h=\Vert v\Vert_h=1$ at $x$.


We write:
\begin{align*}
 \hat \tau(x(s))-\hat \tau(x)&= \underbrace{\sum_j \big(\varphi_j (x(s))-\varphi_j (x)\big)\tau_j(x(s))}_{=:I(s)} +
  \underbrace{
  \sum_j \varphi_j (x )\big(\tau_j(x(s))-\tau_j(x)\big)
   }_{=:I\!I(s)}
   \;.
\end{align*}
We have at $x\in {\mathscr O}_j$, (here we use Eq.\ (\ref{mos}))
\begin{align*}
 \lim_{s\to 0} \frac{I\!I (s) }{s}
     &=
  \lim_{s\to 0} \frac 1 s
  \sum_k \varphi_k (x )\big(\tau_k(x+ v s)-\tau_k(x))
\\
& =
  \lim_{s\to 0} \frac 1 s
  \sum_k \frac {\varphi_k (x )}{\eta_k ^{n+1}}\int_{B_{0}(\epsilon_k)} \chi\left(\frac{z}{\eta_k}\right) \big( \underbrace{\tau (x+vs +z)-\tau (x+z)}_{\ge \int_0^s \underline F(x+z+tv,v)  \dd t  }\big)
   \, d^{n+1}z
\\
&\ge
  \sum_k \frac {\varphi_k (x )}{\eta_k ^{n+1}}\int_{B_{0}(\epsilon_k)} \chi\left(\frac{z}{\eta_k}\right) \underline F(x+z,v)
   \, d^{n+1}z.
\end{align*}
Thus
\begin{align*}
 \lim_{s\to 0} \frac{I\!I (s) }{s}-\underline F(x,v)
&\ge
  \sum_k \frac {\varphi_k (x )}{\eta_k ^{n+1}}\int_{B_{0}(\epsilon_k)} \chi\left(\frac{z}{\eta_k}\right) [\underline F(x+z,v) - \underline F(x,v)]
   \, d^{n+1}z\\
   &\ge  \sum_k \frac {\varphi_k (x )}{\eta_k ^{n+1}}\int_{B_{0}(\epsilon_k)} \chi\left(\frac{z}{\eta_k}\right) (- \Vert v\Vert_h(x)/2)
   \, d^{n+1}z\ge -\Vert v\Vert_h(x) /2
\end{align*}
So we arrive at
\[
\underline F(x,v)-\Vert v\Vert_h(x)/2\le \lim_{s\to 0} \frac{I\!I (s) }{s} \le  \overline F(x,v)+\Vert v\Vert_h(x)/2
\]
where the second inequality is obtained following analogous calculations and using the second inequality in (\ref{mos}).
%

%

For every $j$ let
\[
R_j:=\sup_{k\,:\,{\mathscr O}_k\cap {\mathscr O}_j\ne \emptyset}\sup_{x\in
\overline{{\mathscr O}_j}} \Vert \nabla^h \varphi_k(x)\Vert_h\; ,
\]
let $N_j$ be the number of distinct sets ${\mathscr O}_k$ which have
non-empty intersection with ${\mathscr O}_j$, and let us choose  $\eta_j$ so
small that
\[
\sup_{x\in \overline{{\mathscr O}_j}} |\tau(x)-\tau_j(x)|<
\min_{\ell:{\mathscr O}_\ell\cap {\mathscr O}_j\ne \emptyset} \{\,\frac{1}{N_\ell}
\inf_{\overline{{\mathscr O}_\ell}}\alpha, \,\frac{1}{2N_\ell
R_\ell}\}\; .
\]
Let $\chi_k$ be the characteristic function of ${\mathscr O}_k$, so that
$\varphi_k\le \chi_k$.  The sets ${\mathscr O}_j$ and
$\overline{{\mathscr O}_j}$ intersect the same sets of the covering
$\{{\mathscr O}_i\}$, which are $N_j$ in number, thus
\begin{align*}
 \sup_{x\in \overline{{\mathscr O}_j}}  \sum_{k:{\mathscr O}_k\cap
{\mathscr O}_j\ne \emptyset} \!\!\![\chi_k(x) |\tau(x)-\tau_k(x)|]&\le \! \!\!
\sum_{k:{\mathscr O}_k\cap {\mathscr O}_j\ne \emptyset}\, \sup_{x\in
\overline{{\mathscr O}_k}}|\tau(x)-\tau_k(x)|\\
&\le\!\!\!
\sum_{k\,:\,{\mathscr O}_k\cap {\mathscr O}_j\ne \emptyset} \frac{1}{2R_j N_j}=
\frac{1}{2R_j}
 \;.
\end{align*}


Then at $x\in {\mathscr O}_j$, (recall that $\Vert v\Vert_h=1$ at $x$)
\begin{align*}
\bigg |
\lim_{s\to 0} \frac{I (s) }{s} \bigg|
     &= \bigg |
  \lim_{s\to 0} \sum_k \frac{\varphi_k (x(s))-\varphi_k (x)} s\, \tau_k(x(s))
 \bigg|
\\
& = \bigg |
  \sum_k v\big(\varphi_k(x)\big) \tau_k (x )\bigg |
= \bigg |\sum_k v\big(\varphi_k(x)\big) \big[\tau(x)-\big(\tau(x)
-\tau_k (x )\big)\big]
 \bigg|
\\
 &=  \bigg |\underbrace{v\bigg(\sum_k \varphi_k(x) \bigg)}_{=  v(1)  = 0}\tau (x )
- \sum_k v\big(\varphi_k(x)\big) (\tau(x) -\tau_k (x ) )
 \bigg|
\\
 & \le  \sum_k |v\big(\varphi_k(x)\big)| \,|\tau(x) -\tau_k (x )|=\sum_{k\,:\,{\mathscr O}_k\cap {\mathscr O}_j\ne \emptyset} |v\big(\varphi_k(x)\big)|\, |\tau(x) -\tau_k (x )|
\\
 &\le  R_j \sum_{k:{\mathscr O}_k\cap {\mathscr O}_j\ne \emptyset}  \chi_k(x) |\tau(x) -\tau_k (x )|\le \frac{1}{2}=\frac{\Vert v\Vert_h(x)}{2}
  \;.
\end{align*}
Hence, for   every $x\in M$ and every $C$-causal vector $v\in
T_x{ M}$ of $h$-length one,
 we have
\begin{equation}
\underline F(x,v)-\Vert v\Vert_h(x)\le v(\hat\tau)  \le  \overline F(x,v)+\Vert v\Vert_h(x).
\end{equation}
By positive homogeneity we can drop the condition  $\Vert v\Vert_h(x)=1$ and so this equation holds for every $C$-causal vector $v$.

Finally, for every $x\in {M}$, there is some $j$ such
that $x\in {\mathscr O}_j$, hence
\begin{align*}
\vert \tau(x)-\hat\tau(x)\vert&=\vert \sum_k \varphi_k(x)[\tau(x)-\tau_k(x)] \vert\le \sum_{k:{\mathscr O}_k\cap {\mathscr O}_j\ne \emptyset} \sup_{x\in \overline{{\mathscr O}_k}} |\tau(x)-\tau_k(x)|\\
&\le \sum_{k:{\mathscr O}_k\cap {\mathscr O}_j\ne \emptyset} \frac{1}{N_j}
\inf_{\overline{{\mathscr O}_j}} \alpha\le \alpha(x) \sum_{k:{\mathscr O}_k\cap
{\mathscr O}_j\ne \emptyset} \frac{1}{N_j}= \alpha(x) \; .
\end{align*}
%


\end{proof}

Note that the differentiability degree  of $\hat\tau$ depends only upon the
differentiability degree of ${ M}$, regardless of the regularity of $C$.


\begin{theorem} \label{vkf}
 Every  stable closed Lorentz-Finsler space $(M,\mathscr{F})$  admits a smooth strictly $\mathscr{F}$-steep function $t$. Moreover, if $(p,q)\notin J_S$ we can find $t$ such that $ t(p)>t(q)$.
 Finally, for every $p\in M$ and  every open neighborhood $O\ni p$ we can find smooth strictly $\mathscr{F}$-steep  functions $\check t, \hat t$  such that $p\in [\{q\colon \check t(q)<0\}\cap \{q\colon \hat t(q)>0\}] \subset O$.
\end{theorem}

We recall that the strictly $\mathscr{F}$-steep functions are temporal.

\begin{proof}
By Th.\ \ref{mih} there is  $(M,\mathscr{F}')$ a stable locally Lipschitz proper Lorentz-Finsler space such that $\mathscr{F}'>\mathscr{F}$ and a continuous function  $\tilde t\colon M\to \mathbb{R}$ which is strictly $\mathscr{F}'$-anti-Lipschitz, namely such that for every continuous $C'$-causal curve $\sigma\colon [0,1] \to M$
\[
\tilde t(\sigma(1))-\tilde t(\sigma(0))>\int_\sigma \mathscr{F}'(\dot \sigma) \dd t.
\]
Let $\gamma$ be a Riemannian metric whose balls contain $\mathscr{I}'\cap C$ where $\mathscr{I}'=\mathscr{F}'^{-1}(1)$ is the indicatrix of $(M, \mathscr{F}')$. Moreover, let us choose the unit balls of $\gamma$ so large, or equivalently $\gamma$ so small, that $\mathscr{F}'(v) - \frac{1}{2}\Vert v \Vert_\gamma> \mathscr{F}(v)$ for $v \in C$ (on $T_pM$ it holds in a compact transverse section of $C_p$ and hence everywhere on $C_p$ by positive homogeneity).
Let $h=\gamma/4$, by Th.\ \ref{moz} we can find a smooth function $t$ such that $\mathscr{F}'(v) - \Vert v \Vert_h \le \dd {t} (v)$ for every $v\in C$,
thus $\mathscr{F}(v)<\dd {t} (v)$, which means that $t$ is strictly $\mathscr{F}$-steep.

The penultimate statement follows from the penultimate statement of Th.\ \ref{mih}, which guarantees that $\tilde t$ above can be chosen so that $\tilde t(p)-\tilde t(q)>3\epsilon >0$. Then by Th.\ \ref{moz} and the previous point we can find a smooth strictly $\mathscr{F}$-steep  function $t$ such that $\vert t-\tilde t\vert<\epsilon$, so that $ t(p)- t(q)>\epsilon >0$.

The final statement follows from the final statement of Th.\ \ref{mih}, which guarantees that  we can find continuous strictly $\mathscr{F}'$-anti-Lipschitz functions $\check{\tilde{ t}}$, $\hat{\tilde{ t}}$ such that $p \in [\{q\colon \check{ \tilde{t}}(q)<0\}\cap \{q\colon \hat{\tilde{t}}(q)>0\}] \subset O$.  Then by Th.\ \ref{moz} and the first paragraph of this proof we can find  smooth strictly $\mathscr{F}$-steep  functions $\check t$ and $\hat t$ which approximate $\check{\tilde{ t}}+\vert\check{\tilde{ t}}(p)\vert /2$ and $\hat{\tilde{ t}}-\hat{\tilde{ t}}(p)/2$ respectively, with an error at most $\vert \check{\tilde{ t}}(p)\vert/2$ (resp.\ $\hat{\tilde{ t}}(p)/2$) so that $\check{\tilde{ t}} \le \check t$ and $\hat t\le \hat{\tilde{ t}}$. Thus $p \in [\{q\colon \check{{t}}(q)<0\}\cap \{q\colon \hat{{t}}(q)>0\}] \subset [\{q\colon \check{ \tilde{t}}(q)<0\}\cap \{q\colon \hat{\tilde{t}}(q)>0\}] \subset O$
\end{proof}

\begin{theorem} \label{vkg}
 Every  stably causal closed cone structure $(M,C)$ admits a smooth temporal function $t$. Moreover, if $(p,q)\notin J_S$ we can find $t$ such that $ t(p)>t(q)$.
 Finally, for every $p\in M$ and  every open neighborhood $O\ni p$ we can find smooth temporal functions $\check t, \hat t$  such that $[\{q\colon \check t(q)<0\}\cap \{q\colon \hat t(q)>0\}] \subset O$.
\end{theorem}

\begin{proof}
Set $\mathscr{F}=0$, then by Th.\ \ref{maa} $(M,\mathscr{F})$ is stable and the result follows from Th.\ \ref{vkf} by replacing ``strictly $\mathscr{F}$-steep'' with ``temporal'' as they are equivalent for $\mathscr{F}=0$.
\end{proof}

We have a similar result for globally hyperbolic spacetimes.

\begin{theorem} \label{xbh}
Let $(M,\mathscr{F})$ be a globally hyperbolic closed Lorentz-Finsler space and let $h$ be a complete Riemannian metric on $M$. Then
there is a smooth Cauchy $h$-steep strictly $\mathscr{F}$-steep (hence temporal) function $t$. Moreover, if $(p,q)\notin J$ we can find $t$ such that $ t(p)>t(q)$.
 Finally, for every $p\in M$ and  every open neighborhood $O\ni p$ we can find smooth Cauchy $h$-steep strictly $\mathscr{F}$-steep  functions  $\check t, \hat t$  such that $[\{q\colon \check t(q)<0\}\cap \{q\colon \hat t(q)>0\}] \subset O$.
\end{theorem}
The proof clarifies that in general one can control the lower bound on the steepness of the temporal function.

\begin{proof}
Let $(M,\mathscr{F}')$ be any locally Lipschitz proper Lorentz-Finsler space such that $C'>C$ which is globally hyperbolic. Here $\mathscr{F}'$ is chosen sufficiently large so that $\mathscr{F}'> 2 \mathscr{F}$ on $C$, and $\mathscr{I}'\cap C$,  with $\mathscr{I}'=\mathscr{F}'{}^{-1}(1)$, is contained in the unit ball of $4h$. As a consequence, for every $v\in C$, $2 \Vert v \Vert_h \le  \mathscr{F}'(v)$. By Theorem \ref{mab} $(M,\mathscr{F}')$ is stable, and by Theorem \ref{mih} there is a continuous function  $ \tilde t\colon M\to \mathbb{R}$ which is strictly $\mathscr{F}'$-anti-Lipschitz, namely such that for every continuous $C'$-causal curve $\sigma\colon [0,1] \to M$
\[
\tilde t(\sigma(1))-\tilde t(\sigma(0))>\int_\sigma \mathscr{F}'(\dot \sigma) \dd t.
\]
By Theorem \ref{moz} we can find a smooth function $t$ such that for every $v\in C$
\begin{equation}
\mathscr{F}'(v)- \Vert v\Vert_h \le \dd  t(v)
\end{equation}
but $\mathscr{F}'(v)- \Vert v\Vert_h \ge \mathscr{F}'(v)/2 > \mathscr{F}(v)$ and $\mathscr{F}'(v)- \Vert v\Vert_h \ge \mathscr{F}'(v)/2 \ge \Vert v \Vert_h$. The Cauchy property follows from the last inequality.

By Th.\ \ref{mom} in a globally hyperbolic closed cone structure $J_S=J$. The penultimate statement follows from the penultimate statement of Th.\ \ref{mih}, which guarantees that $\tilde t$ above can be chosen so that $\tilde t(p)-\tilde t(q)>3\epsilon >0$. Then by Th.\ \ref{moz} and the previous point we can find a smooth Cauchy $h$-steep strictly $\mathscr{F}$-steep  function $t$ such that $\vert t-\tilde t\vert<\epsilon$, so that $ t(p)- t(q)>\epsilon >0$.

The final statement follows from the final statement of Th.\ \ref{mih}, which guarantees that  we can find continuous strictly $\mathscr{F}'$-anti-Lipschitz functions $\check{\tilde{ t}}$, $\hat{\tilde{ t}}$ such that $p \in [\{q\colon \check{ \tilde{t}}(q)<0\}\cap \{q\colon \hat{\tilde{t}}(q)>0\}] \subset O$.  Then by Th.\ \ref{moz} and the first paragraph of this proof we can find  smooth Cauchy $h$-steep strictly $\mathscr{F}$-steep  functions $\check t$ and $\hat t$ which approximate $\check{\tilde{ t}}+\vert\check{\tilde{ t}}(p)\vert /2$ and $\hat{\tilde{ t}}-\hat{\tilde{ t}}(p)/2$ respectively, with an error at most $\vert \check{\tilde{ t}}(p)\vert/2$ (resp.\ $\hat{\tilde{ t}}(p)/2$) so that $\check{\tilde{ t}} \le \check t$ and $\hat t\le \hat{\tilde{ t}}$. Thus $p \in [\{q\colon \check{{t}}(q)<0\}\cap \{q\colon \hat{{t}}(q)>0\}] \subset [\{q\colon \check{ \tilde{t}}(q)<0\}\cap \{q\colon \hat{\tilde{t}}(q)>0\}] \subset O$.
\end{proof}


\begin{corollary} \label{nin}
Let $(M,C)$ be a globally hyperbolic proper cone structure. Then $M$ is smoothly diffeomorphic to $S\times \mathbb{R}$, where  $S$ is smoothly diffeomorphic to any Cauchy hypersurface, the projection to the first factor has smooth timelike curves as fibers, and the projection to the second factor is function $t$ of Th.\ \ref{xbh}.
\end{corollary}

\begin{proof}
Let $t$ be the function constructed in the previous theorem and let $S_0=t^{-1}(0)$.  Since $t$ is smooth and temporal, $S_0$ can be endowed with a smooth structure which makes the immersion smooth. Let $V$ be a smooth timelike vector field. The integral curves of $V$ intersect $S_0$ only once and provide a smooth projection $\pi\colon M\to S_0$. Let $\varphi_a\colon M\to M$ be the 1-parameter family of smooth diffeomorphisms generated by $V$, then the map $ S_0\times \mathbb{R} \to M$, given by $(s,t)\mapsto \varphi_t(s)$ is a smooth diffeomorphism.
\end{proof}

\subsection{Equivalence between $K$-causality and stable causality} \label{las}

Sorkin and Woolgar \cite{sorkin96} introduced the $K$ relation as the smallest closed and transitive relation containing the causal relation $J$. Similar concepts had been introduced in dynamical system theory where one spoke of {\em Auslander's prolongations} \cite{auslander64}.
The antisymmetry of $K$ is called $K$-{\em causality} in analogy with stable causality which corresponds to the antisymmetry of the Seifert relation $J_S$, see Sec.\ \ref{fir}. The relations $K$ and $J_S$ are both closed and transitive so it is natural to ask if they coincide \cite{sorkin96}.
The conjecture due to R. Low was indeed proved in Lorentzian geometry where we gave two proofs, an entirely topological one \cite{minguzzi08b}, and a much simpler one  \cite{minguzzi09c} which made use of smoothability results for time functions and Auslander-Levin's theorem \cite{auslander64,levin83}. In this section we give a new proof which does not use smoothability results for time functions   but contains some elements of the proof in \cite{minguzzi09c}. As a result the novel proof applies also to closed cone structures.

To start with, we recall that a utility function $f$ is an isotone function such that $x\le y$ and $y\nleq x$ implies $f(x)<f(y)$, then the  Auslander-Levin's theorem is
 \begin{theorem} \label{lev} (Auslander-Levin)
Let $X$ be a second countable locally compact Hausdorff space, and
$R$ a closed preorder on $X$, then there exists a continuous utility
function. Moreover, denoting with $\mathcal{U}$  the set of
continuous utilities
 we have that the
preorder $R$ can be recovered from the continuous utility functions,
namely there is a {\em multi-utility representation}
\begin{equation}
(x,y) \in R \Leftrightarrow \forall u \in \mathcal{U}, \ u(x)\le
u(y).
\end{equation}
\end{theorem}

As a first step we prove that the continuous $J_S$-utilities are precisely the time functions.

\begin{lemma} \label{pag}
Let $(M,C)$ be a closed cone structure.
If there is a time function then $(M,C)$
is strongly causal.
\end{lemma}

\begin{proof}
If $(M,C)$ is not strongly causal at $x$ then  there is  a non-imprisoning neighborhood $U \ni x $ as in Prop.\ \ref{iiu} and a sequence of
continuous $C'$-causal curves $\sigma_n$ of endpoints $x_n,z_n$, with $x_n\to
x$, $z_n \to x$, not entirely contained in $U$. Let $B$, $\bar B\subset U$ be a coordinate ball of $x$.
 Let $c_n\in \p {B}$ be the first
point at which $\sigma_n$ escapes $\bar B$, and let $d_n$ be the last
point at which $\sigma_n$ reenters $\bar B$. Since $\p {B}$ is compact
there are $c,d \in \p B$, and  a subsequence $\sigma_k$ such that
$c_k \to c$, $d_k \to d$. By the limit curve theorem $(x,c), (d,x)\in J$, while $(c,d)\in \bar J$, which is impossible since $t(d)<t(c)$ but the sequence $\sigma_k$ has starting points  close to $c$ and reaches points close to $d$.
\end{proof}

\begin{lemma} \label{soh}
Let $(M,C)$ be a closed cone structure which admits a time function $t$.
 Let $h$ be a complete Riemannian metric and let $\epsilon >0$. Let $R_\epsilon=\{(p,q)\colon d^h(p,q)<\epsilon\}$, let $U$ be an open relatively compact set, then there is a locally Lipschitz  proper cone structure $C'>C$ such that $\overline{J_{C'}( U)}\backslash R_\epsilon\subset \{(p,q)\in \bar U\times \bar U\colon t(p)< t(q)\}$.
\end{lemma}

\begin{proof}
We know that $(M,C)$ is strongly causal so it is non-imprisoning.
Suppose the inclusion does not hold, then for every locally Lipschitz  proper cone structure $C'>C$,  $ \{(p,q)\in \bar U\times \bar U\colon t(p)\ge t(q)\}\cap [\overline{J_{C'}( U)}\backslash R_\epsilon]\ne \emptyset$. But this set is compact being a closed subset of $\bar U\times \bar U$, and the family obtained for $C'>C$ has the finite intersection property, thus
\begin{align*}
\emptyset &\ne \cap_{C'>C}\{(p,q)\in \bar U\times \bar U\colon t(p)\ge t(q)\}\cap [\overline{J_{C'}( U)}\backslash R_\epsilon]\\
&=\{(p,q)\in \bar U\times \bar U\colon t(p)\ge t(q)\}\cap [\cap_{C'>C}\overline{J_{C'}( U)}\backslash R_\epsilon]\\
& \subset \{(p,q)\in \bar U\times \bar U\colon t(p)\ge t(q)\}\cap [J(\bar U)\backslash R_\epsilon]
\end{align*}
since by Prop.\ \ref{paq} $J_S(U)=\cap_{C'>C}\overline{J_{C'}(U)}$ and $J_S( U)\subset J(\bar{U})\cap ( U\times U)$, by Th.\ \ref{sqd}, the limit curve theorem \ref{main} and the non-imprisoning property of $\bar U$. Thus there are $p,q\in \bar U$ such that $d^h(p,q)\ge \epsilon$ connected by a continuous causal curve entirely contained in $\bar{U}$ with starting point $p$ and ending point $q$, and moreover $t(p)\ge t(q)$. But $t$ is a time function and $p\ne q$, thus $t(p)<t(q)$, a contradiction.
\end{proof}

The next important lemma states that provided we do not consider  points that are too close, a time function $t$ preserves its increasing property for a wider cone structure.

\begin{lemma} \label{lmn}
Let $(M,C)$ be a closed cone structure which admits a time function $t$.
 Let $h$ be a complete Riemannian metric and let $\epsilon>0$. Let $R_\epsilon=\{(p,q)\colon d^h(p,q)<\epsilon\}$,  then there is a locally Lipschitz proper cone structure $C'>C$ such that ${J_{C'}}\backslash R_\epsilon\subset \{(p,q)\colon t(p)< t(q)\}$.
\end{lemma}

\begin{proof}
Let $o\in M$ and Let $U_i=B(o,(i+7)\epsilon )\backslash \bar B (o,i \epsilon)$ be a sequence of open relatively compact sets constructed with $h$-balls centered at $o$. Notice that $\cup_i U_i=M$.
By Lemma \ref{soh} we can find a locally Lipschitz  proper cone structure $C_i>C$ such that $J_{C_i}(U_i)\backslash R_\epsilon\subset  \{(p,q)\in \bar U_i\times \bar U_i\colon t(p)< t(q)\}$. Let $C'>C$ be chosen so that $C'<C_i$ on $\bar U_i$. On every compact set this is a finite number of conditions so $C'$ exists.  Let us consider a $C'$-causal curve  $\sigma$ connecting $p$ to $q$.
Let $p_1=p$, then $p_1\in  U_{i_1}$ where $i_1$ is chosen so that $p_1$ is at distance at least $3\epsilon$ from $\p U_{i_1}$. Let $p_2$ be the first escaping point from $\bar B(o,(i_1+6)\epsilon )\backslash  B (o,(i_1+1) \epsilon)$ and choose $i_2$ so that $p_2\in U_{i_2}$ is at distance at least $3\epsilon$ from $\p U_{i_2}$, following $\sigma$ from $p_2$, let $p_3$ be the first escaping point from $\bar B(o,(i_2+6)\epsilon )\backslash  B (o,(i_2+1) \epsilon)$, and so on.  The succession of points $\{p_k\}$ over $\sigma$ are such that $d^h(p_k, p_{k+1})\ge \epsilon$, so since the $h$-length of $\sigma$ is bounded, the segments are finite in number. If $p_m\in U_{i_m}$ is  the last point of the sequence then $q\in \bar B(o,(i_m+6)\epsilon )\backslash  B (o,(i_m+1) \epsilon)\subset U_{i_m}$. If $d(p_m, q)\ge \epsilon$ we set $p_{m+1}=q$, otherwise we redefine $p_m=q$, so that $p_{m-1}, q\in U_{i_{m-1}}$. Then for every $k$, $d^h(p_k, p_{k+1})\ge \epsilon$, with $(p_k,p_{k+1})\in J({U}_{i_k})$.
By Lemma \ref{soh} $ t(p_k)<t(p_{k+1})$ which implies $t(p)<t(q)$.
\end{proof}

\begin{theorem} \label{aob}
Let $(M,C)$ be a closed cone structure which admits a time function $t$.
Then $J_S\backslash \Delta\subset \{(p,q)\colon t(p)< t(q)\}$, so $(M,C)$ is stably causal and the continuous $J_S$-utilities are precisely the time functions.
\end{theorem}

\begin{proof}
Let $(p,q)\in J_S\backslash \Delta$ and let $\epsilon>0$ be chosen so that $\epsilon\le  d^h(p,q)$. By Lemma \ref{lmn}  there is $C'>C$ such that $J_{C'}\backslash R_\epsilon \subset \{(p,q)\in\colon t(p)< t(q)\}$. But $(p,q)\notin R_\epsilon$ and $(p,q)\in J_{C'}$, thus $t(p)<t(q)$. The proved inclusion $J_S\backslash \Delta\subset \{(p,q)\colon t(p)< t(q)\}$ implies the antisymmetry of $J_S$, hence stable causality. It also means that every time function is a $J_S$-utility. For the converse, under stable causality a continuous $J_S$-utility $t$ is a $J_S$-isotone function (hence $J$-isotone) such that if $(p,q)\in J_S\backslash \Delta$ then $t(p)< t(q)$, but $J\subset J_S$, thus this property implies  $(p,q)\in J\backslash \Delta \Rightarrow t(p)< t(q)$, that is $t$ is a time function.
\end{proof}

As a second step we prove that the continuous $K$-utilities are precisely the time functions.
The next lemmas appeared in \cite{minguzzi09c}.

\begin{lemma} \label{pkw}
Let $(M,C)$ be a non-imprisoning closed cone structure. Let $(p,q)\in K$ then
either $(p,q)\in J$ or for every relatively compact open
set $B\ni p$ there is $r \in \p {B}$ such that
$p<r$ and $(r,q)\in K$.
\end{lemma}

\begin{proof}
Consider the relation
\begin{align*}
R=\{(p,q)\in K\colon  &\ (p,q)\in J \textrm{ or for every
relatively compact open set } B\ni p \\
&  \textrm{ there is } r \in \p {B} \textrm{ such that } p<r
\textrm{ and } (r,q)\in K \}.
\end{align*}
It is easy to check that $J\subset R\subset K$. We are
going to prove that $R$ is closed and transitive. From that and
from the minimality of $K$ it follows $R=K$ and hence
the desired result.

Transitivity: assume $(p,q) \in R$ and $(q,s) \in R$. If
$(p,s) \in J$ there is nothing to prove.  Otherwise we have
$(p,q)\notin J$ or $(q,s) \notin J$.

If $(p,q)\notin J$ for every $B\ni p$  open relatively compact set there is $r \in \p {B}$ such that $p<r$ and
$(r,q)\in K$, thus $(r,s) \in K$ and hence $(p,s) \in
R$.

It remains to consider the case $(p,q)\in J$ and $(q,s) \notin
J^{+}$. If $p=q$ then $(p,s)=(q,s) \in R$. Otherwise, $p<q$ and for every $B\ni p$  open relatively compact set we have two possibilities, whether $q\notin B$ or $q\in B$.
If $q \notin B$ the causal curve $\gamma$ joining $p$ to $q$
intersects $\p {B}$ at a point $r \in \p {B}$ (possibly coincident
with $q$ but different from $p$). Thus $p<r$, $(r,q) \in J$,
hence $p<r$ and $(r,s) \in K$. If instead $q \in B$, since $(q,s)\in
R\backslash J$, there is $r \in \p {B}$ such that $q<r$ and
$(r,s) \in K$, moreover, since $p\le q$, we have $p<r$. Since the
searched conclusion ``$p<r$, $r\in \p B$ and $(r,s)\in K$'' holds in both cases, we conclude $(p,s) \in R$.

Relation $R$ is closed: let $(p_n,q_n) \to (p,q)$, $(p_n,q_n) \in R$. Assume,
by contradiction, that $(p,q) \notin R$, then $p\ne q$ as
$J\subset R$. Without loss of generality we can assume two
cases: (a) $(p_n,q_n) \in J$ for all $n$; (b) $(p_n,q_n) \notin
J$  for all $n$.

(a) Let $B\ni p$ be an open relatively compact set. For sufficiently
large $n$, $p_n \ne q_n$ and $p_n \in B$. By the limit curve theorem
 either there is a limit continuous causal curve
joining $p$ to $q$, and thus $p<q$ (a contradiction), or there is a
future inextendible continuous causal curve $\sigma^p$ starting from
$p$ such that for every $p'\in \sigma^p$, $(p',q) \in
\overline{J}$. Since $(M,C)$ is non-imprisoning,
$\sigma^p$ intersects $\p {B}$ at some point $r$. Thus $p<r$ and
since $(r,q)\in \overline{J}\subset K$ we have $(p,q)\in
R$, a contradiction.

(b) Let $B\ni p$ be an open relatively compact set. For sufficiently
large $n$, $p_n \ne q_n$ and $p_n \in B$. Since $(p_n,q_n) \in
R\backslash J$ there is $r_n \in \p {B}$, $p_n<r_n$, and $(r_n,q_n)\in
K$. Without loss of generality we can assume $r_n \to r \in
\p {B}$, so that $(r,q) \in K$. Arguing as in (a) either $p<r$
(and $(r,q)\in K$) or there is $r'\in \p {B}$ such that $p<r'$
and $(r',r) \in \overline{J}\subset K$, from which it
follows that $(r',q) \in K$. Because of the arbitrariness of $B$,
$(p,q)\in R$, a contradiction. \end{proof}

\begin{lemma} \label{jba}
Let $(M,C)$ be a closed cone structure.
\begin{description}
\item[(a)] Let $\tilde{t}$ be a continuous function such that $x\le y \Rightarrow
\tilde{t}(x)\le \tilde{t}(y)$. If $(p,q) \in K$ then
$\tilde{t}(p)\le \tilde{t}(q)$.
\item[(b)] Let $t$ be a time function on $(M,C)$. If
$(p,q) \in K$ then  $p=q$ or $t(p)<t(q)$.
\end{description}
\end{lemma}

\begin{proof}

Proof of (a). Consider the relation
\[
\tilde{R}=\{(p,q)\in K\colon  \ \tilde{t}(p)\le \tilde{t}(q) \}.
\]
Clearly $J\subset \tilde{R}\subset K$ and $\tilde{R}$ is
transitive.

Let us prove that $\tilde{R}$ is closed. If $(x_n,z_n) \in
\tilde{R}$ is a sequence such that $(x_n,z_n) \to (x,z)$, then
passing to the limit $\tilde{t}(x_n)\le \tilde{t}(z_n)$ and using
the continuity of $\tilde{t}$ we get $\tilde{t}(x)\le \tilde{t}(z)$,
moreover since $K$ is closed, $(x,z) \in K$, which implies
$(x,z) \in \tilde{R}$, that is $\tilde{R}$ is closed.

Since $J\subset \tilde{R}\subset K$, and $\tilde{R}$
is closed and transitive, by using the minimality of $K$ it
follows that $\tilde{R}=K$. As a consequence, if $(p,q)\in
K$ then $\tilde{t}(p)\le \tilde{t}(q)$.

Proof of (b).  By lemma \ref{pag}, since $t$ is a time function
$(M,C)$ is strongly causal and thus non-imprisoning. Consider
the relation
\[
R=\{(p,q)\in K\colon \ p=q \textrm{ or } t(p)<t(q) \} .
\]
Clearly $J\subset {R}\subset K$ and ${R}$ is transitive.
Let us prove that $R$ is closed by keeping in mind  the result given by (a) that we just
obtained.  Let $(p_n,q_n) \in {R}\subset
K$ be a sequence such that $(p_n,q_n) \to (p,q)$. As $K$ is
closed, $(p,q) \in K$. If, by contradiction, $(p,q) \notin
R$ then $(p,q) \notin J$,  thus by lemma \ref{pkw}, chosen
an open relatively compact set $B\ni p$ there is $r\in \p {B}$,
with $p<r$, $(r,q)\in K$, thus $t(p) <t(r)\le t(q)$ and hence
$(p,q)\in R$, a contradiction.

Since $J\subset {R}\subset K$, and ${R}$ is closed
and transitive, by using the minimality of $K$ it follows that
${R}=K$. As a consequence, if $(p,q)\in K$ then either
$p=q$ or $t(p)< t(q)$.
\end{proof}

\begin{theorem} \label{ndc}
Let $(M,C)$ be a closed cone structure which is  $K$-causal, then the continuous $K$-utilities are precisely the continuous $J$-utilities (namely,
 the time functions). Similarly, the continuous $K$-isotone function are precisely the continuous $J$-isotone functions.
\end{theorem}

\begin{proof}
A $K$-utility is a function $u$ which satisfies (i) $(x,y)\in
K\Rightarrow u(x)\le u(y)$ and (ii) $(x,y)\in K$ and $(y,x)
\notin K\Rightarrow u(x)<u(y)$. Since the spacetime is
$K$-causal this condition is equivalent to $(x,y)\in K
\Rightarrow x=y$ or $u(x)<u(y)$. Thus by Lemma \ref{jba} point (b),
every time function is a continuous $K$-utility. Conversely, in
a $K$-causal spacetime a continuous $K$-utility satisfies $x<y$
$\Rightarrow (x,y)\in K\backslash\Delta \Rightarrow u(x)<u(y)$
and hence it is a time function. The last statement is just Lemma \ref{jba} (a).
\end{proof}

Finally we are able to prove the next important result.

\begin{theorem} \label{bhd}
Let $(M,C)$ be a closed cone structure. The following properties are equivalent:
\begin{itemize}
\item[(i)] Stable causality,
\item[(ii)] Antisymmetry of $J_S$,
\item[(iii)] Antisymmetry of $K$ ($K$-causality),
\item[(iv)]  Emptyness of the stable recurrent set,
\item[(v)]  Existence of a time function,
\item[(vi)] Existence of a smooth temporal function,
\end{itemize}
 Moreover, in this case $J_S=K=T_1=T_2$ where
 \begin{align*}
 T_1&=\{(p,q)\colon t(p)\le t(q) \textrm{ with } t \ \textrm{time function} \}, \\
 T_2&=\{(p,q)\colon t(p)\le t(q) \textrm{ with } t \ \textrm{smooth temporal function} \}.
 \end{align*}
\end{theorem}

\begin{proof}
$(i) \Leftrightarrow (ii)$ is Th. \ref{nkk}.  $(ii) \Leftrightarrow (iv)$ is Th.\ \ref{ssp}.
By Th.\ \ref{vkg} $(i) \Rightarrow (vi)$, and clearly $(vi)\Rightarrow (v)$. By Lemma \ref{jba} (b), $(v)\Rightarrow (iii)$. Finally, by Auslander-Levin's theorem $(iii)$ implies the existence of a continuous $K$-utility, namely a time function (Th.\ \ref{ndc}) so by Th.\ \ref{aob} $(M,C)$ is stably causal, hence $(i)$.
The last statement follows from the multi-utility representation in Auslander-Levin's theorem. Applying it  to the closed order $J_S$ jointly with Th.\ \ref{aob} gives $J_S=T_1$. Applying it to the closed order $K$ jointly with Th.\ \ref{ndc} gives $K=T_1$.
It remains to prove $J_S=T_2$. Clearly, $T_1\subset T_2$, thus $J_S\subset T_2$.
Let $(p,q)\notin J_S$, by the last statement of Th.\ \ref{vkg} we can find a smooth temporal function such that $t(p)>t(q)$, that is $(p,q)\notin T_2$, which concludes the proof.
\end{proof}
%

%
%
We recall that causal simplicity means: $J$ closed and antisymmetric.
\begin{theorem} \label{mbg}
Let $(M,C)$ be a  causally simple closed cone structure, then  $(M,C)$ is stably causal and $J=J_S=K=T_1=T_2$.
\end{theorem}

\begin{proof}
Under the assumption $J$ is  closed, thus it is the smallest closed transitive relation containing $J$, hence $K=J$. But under the assumption $J$ is antisymmetric hence $K$ is antisymmetric (stable causality).
The equality of relations follows from the previous theorem.
\end{proof}


%
%
%
%


\subsubsection{From time to temporal functions}


This section can be completely skipped. Its main purpose is to present a different method for constructing a temporal function in cone structures admitting a time function, provided one has proved with a different method that globally hyperbolic proper cone structures admit Cauchy temporal functions. The method only works for proper cone structures so the result is really weaker than that obtained in Th.\ \ref{bhd}. Nevertheless, the proof is instructive and passes through the notion of domain of dependence.
\begin{theorem} \label{mbn}
Suppose that every globally hyperbolic proper cone structure admits a Cauchy temporal function.
Let $(M,C)$ be a proper cone structure which admits a time function $t$, then it is stably causal hence it admits a temporal function (Th.\ \ref{vkg}). Moreover, $T=\{(p,q)\colon t(p)\le t(q) \textrm{ for every } t \in \mathcal{T}\}$
is independent of whether $\mathcal{T}$ represents the set of time or temporal functions.  
\end{theorem}

\begin{proof}
Let $S_a=t^{-1}(a)$, with $a\in t(M)$, then $S_a$ is acausal and $D(S_a)$  is an open set (Prop.\ \ref{maq}, Th.\ \ref{aoq} and \ref{mmm}). Hence $D(S_a)$ endowed with the induced cone distribution  is a globally hyperbolic proper cone structure. By assumption there is a smooth Cauchy time function $t_a$ on it.
 Let $\tau_a\colon D(S_a)\to \mathbb{R}$, be given by $\tau_a=\varphi\circ t_a$, where $\varphi(x)$ is a smooth function such that $\varphi =1$ for $x\ge 1$, $\varphi=-1$ for $x\le -1$, and $\varphi'>0$ on $(-1,1)$. It is smooth and strictly increasing over causal curves in a neighborhood of $t^{-1}_a(0)$.

 The function $\tau_a$ can be extended to a smooth function all over $M$ by setting $\tau_a=1$ on $\{p\colon t_a(p)\ge 1\}\backslash D(S_a)$ and $\tau_a=-1$ on $\{p\colon t_a(p)\le  -1\}\backslash D(S_a)$. In fact, given  $q\in H^+(S)$ we have just to show that $\tau_a=1$ in a neighborhood of $q$ (and similarly for $q\in H^-(S)$). We know that $H^+(S)$ is a locally Lipschitz hypersurface, thus we can find a coordinate neighborhood of  $q$, diffeomorphic to a  cylinder $C=A\times [0,1]\subset M$ whose $[0,1]$ fiber are timelike as generated by a timelike Lipschitz vector field $V$. On the portion of integral line of $V$ passing through $q$ which is contained in $D(S_a)$, $t_a$ goes to infinity as the evaluation point approaches $q$. Thus taking the base $A$ sufficiently small, by continuity of $t_a$ we can choose the cylinder so that $t _a>1$ on $A\times \{0\}$, and hence everywhere on $C$.  This fact proves that $\tau_a$ is indeed  extended to a smooth function.

Let $T_a=\tau_a^{-1}((-1,1))$, then $M$ is covered by such open strips, and since it is Lindel\"of it admits a countable subcovering $\{T_{a_i}\}$.
Let $o\in M$ and let $h$ be a complete Riemannian metric. Let us set $C_i=1+\sup_{\bar B(o, i)} \Vert D \tau_{a_i}\Vert_h $ and define
\[
\tau=\sum_{i} \frac{1}{2^i C_i} \tau_{a_i}.
\]
Since $\vert\tau_{a_i}\vert \le 1$ the right-hand side converges uniformly so $\tau$ is continuous.
Let $K$ be a compact set then there is some $n$ such that $K\subset B(o,n)$ and we have
\begin{align*}
\sup_{\bar B(o,n)}\sum_{i=1}^\infty \frac{1}{2^i C_i} \Vert D \tau_{a_i} \Vert_h &\le \underbrace{\sup_{\bar B(o,n)}\sum_{i=1}^n \frac{1}{2^i C_i} \Vert D \tau_{a_i} \Vert_h}_{<\infty} +\sum_{i=n+1}^\infty \frac{1}{2^i C_i}  \underbrace{ \sup_{ \bar B(o,n)} \Vert D \tau_{a_i} \Vert_h}_{\le C_i} \\
&< \infty .
\end{align*}
This inequality shows that the series defining $\tau$ converges in $C^1$ norm on
every compact set, resulting in a $C^1$ function. Since every point belongs to some $T_{a_i}$ and  on every $T_{a_i}$, $\dd \tau_{a_i}=\phi'(t_a) \dd t_a $ which is positive on $C$, $\tau$ is temporal.

Next, since every temporal function is a time function we need only to prove one inclusion, which follows from the next result: suppose that there are $p,q\in M$ such that there exist a time function $t$ such that $t(p)>t(q)$ then we can find a temporal function $\tau$ such that $\tau(p)>\tau(q)$.
In fact let us use $t$ as initial time function in the above construction. Choosing $a$ in such a way that $t(q)<a<t(p)$ and including $T_{a_1}$ in the covering, with $a_1=a$ we have that $\tau_{a_1}$ is a smooth isotone function such that $\tau_{a_1}(q)<0<\tau_{a_1}(p)$, with non-negative $\dd \tau_{a_1}$ on $C$. Now given any temporal function $g$, and a constant $b>0$, $\tau=  g+b\tau_{a_1}$ is temporal  and for sufficient large $b$,   $\tau(q)<\tau(p)$.
\end{proof}

\subsection{The regular ($C^{1,1}$) theory} \label{xxo}
This is section is meant to summarize what is known under stronger regularity conditions on the cone distribution.

Let us consider a $C^0$ distribution of proper cones and assume that the differentiability degree of the local immersion defining the hypersurface $\p C_x\subset T_xM\backslash 0$, is  $C^{3,1}$ at every $x$, while the immersion defining the hypersurface $\p C \subset TM$ is $C^{1,1}$. The pair $(M,C)$ is also called {\em cone structure}, and for brevity we say that it is {\em regular} or $C^{1,1}$.

A regular {\em Lorentz-Finsler space} (space) is a pair $(M,\mathscr{F})$ where $\mathscr{F}\colon C\to [0,+\infty)$ is positive homogeneous of degree one, $\p C=\mathscr{F}^{-1}(0)$, $(M,C)$ is a regular cone structure and $\mathscr{L}:=-\frac{1}{2} \mathscr{F}^2$ has $C^{1,1}$ Lorentzian vertical Hessian on $C$. Here $\mathscr{F}$ is the {\em Finsler function}, while $\mathscr{L}$ is the {\em Finsler Lagrangian}.

Under these differentiability conditions on the boundary of the cone, $\mathscr{L}$ can be extended all over $TM\backslash 0$ preserving the differentiability degree of the Hessian and its Lorentzian nature. In this way the manifold is converted into a Lorentz-Finsler space in Beem's sense \cite{minguzzi14h}. Physically, what really matters is the function $\mathscr{F}$ on the cone, nevertheless, some results such as the existence of convex neighborhoods are better expressed using a Lorentz-Finsler space in Beem's sense \cite{minguzzi13d}, so it is useful to know that one can pass to the latter type of space.

Every regular cone structure comes from a regular Lorentz-Finsler space. The argument has been used a  few times, the details on how to construct the Finsler function can be found in  \cite[Prop.\ 13]{minguzzi15e}. Of course, the Finsler function compatible with the cone structure is not unique.

Geodesics are the stationary points for the Lagrangian $\mathscr{L}$. Reference \cite{minguzzi15e} points out that the notion of unparametrized lightlike geodesic does not depend on the Finsler Lagrangian, just on $(M,C)$ and that such geodesics are locally achronal. Thus the notions of causal curve, timelike curve, lightlike curve, unparametrized lightlike geodesics, really refer to the cone structure $(M,C)$.

The results for Lorentz-Finsler theory developed in \cite{minguzzi13d} \cite[Sec.\ 6]{minguzzi15}  show that most of causality theory, particulary those portions that do not use the notion of curvature but just maximization and limit curve arguments \cite[Remark 2]{minguzzi14h}, pass through to the Lorentz-Finsler case (in fact, many results which involve curvature also do \cite{minguzzi15} but they require more care, see also \cite{kunzinger15,kunzinger15b,graf17} for singularity theorems in the Lorentzian case). This result is due to the existence of convex neighborhoods in this theory and to the fact that the exponential map is well defined and provides a local lipeomorphism \cite{minguzzi13d,kunzinger13,kunzinger13b}.
Since every cone structure comes from a Lorentz-Finsler space we have
\begin{remark}
All the results on causality theory for regular Lorentz-Finsler spaces developed in  \cite{minguzzi13d} \cite[Sec.\ 6]{minguzzi15} which do not involve in their statement the Finsler function $\mathscr{F}$ but just the future causal  cone pass through to the cone structure case.
\end{remark}
For instance, the geodesic equation has a unique solution for every initial condition, lightlike geodesics do not branch,  convex neighborhoods do exist, Cauchy horizons are generated by lightlike geodesics,  lightlike geodesics are locally achronal, the chronological relation is open, and so on.

Some of the properties hold under fairly weaker assumptions as we have shown in this work.


\section{Applications}
In this section we apply the previous results to two important problems.

\subsection{Functional representations and the distance formula}

On a topological ordered space $(M,\mathscr{T}, R)$ we say that a family of continuous $R$-isotone functions $\mathcal{F}$  {\em represents the order} $R$ (a reflexive, transitive and antisymmetric relation)  on $M$ if
\begin{itemize}
\item[(i)] $(p,q)\in R$ if and only if  $f(p)\le f(q)$ for all $f\in \mathcal{F}$,
\end{itemize}
and that it {\em represents the topology} if
\begin{itemize}
\item[(ii)] for every $p\in M$ and open neighborhood $O\ni p$ we can find $\check f, \hat f \in \mathcal{F}$ such that $[\{q\colon \check f(q)<0\}\cap \{q\colon \hat f(q)>0\}] \subset O$.
\end{itemize}
Notice that if the representation holds for a family, it also holds for a larger family.
If both topology and order are represented by the continuous isotone functions the topological ordered space is called a {\em completely regularly ordered space} \cite{nachbin65}. These topological ordered spaces are important since they are equivalent to the quasi-uniformizable spaces, namely to  those spaces which  can be Nachbin compactified.

A rewording of Theorem \ref{vkg} is
\begin{theorem} \label{swq}
In a stably causal closed cone structure $(M,C)$ the Seifert order $J_S$ and the manifold topology are represented by the smooth temporal functions. As a consequence, $(M,\mathscr{T}, J_S)$ is a completely regularly ordered space.
\end{theorem}

The quasi-uniformizability of stably causal spacetimes was proved as a corollary of a more general result in \cite{minguzzi12d} (where the proof did not depend on the roundness of the cone). Here we have restricted the family of representing functions, thereby getting a stronger result. For what concerns the representation of just the order other related works are \cite{minguzzi09c,bernard16}.

A rewording of Theorem \ref{xbh} is
\begin{theorem}
Let $(M,\mathscr{F})$ be a globally hyperbolic closed Lorentz-Finsler space and let $h$ be a complete Riemannian metric on $M$. Then
the causal order $J$  and the manifold topology are represented by the smooth Cauchy $h$-steep strictly $\mathscr{F}$-steep temporal functions.
\end{theorem}

Let $r^+=\textrm{max}\{r,0\}$. On a closed Lorentz-Finsler space $(M,\mathscr{F})$ we say that {\em the stable distance is represented by} $\mathcal{F}$ if for every $p,q\in M$
\begin{equation} \label{djj}
 D(p,q)=\mathrm{inf} \big\{[f(q)-f(p)]^+\colon \ f \in \mathscr{\mathcal{F}}\big\}.
\end{equation}
If $D$ is represented and $\mathcal{F}$ consists of continuous and isotone functions we have for $f\in \mathcal{F}$, and $(p,q)\in J$, $f(q)-f(p)\ge d(p,q)\ge \int_x \mathscr{F}(\dot x)\dd t$, where $x$ is any $C$-causal connecting curve. This fact suggests to try to represent the stable distance using $\mathscr{F}$-steep temporal functions.

Notice that in Eq.\ (\ref{djj}) the right-hand side is upper semi-continuous so the distance that has chances to be representable is $D$ rather than $d$.


\begin{proposition} \label{ihg}
Let $(M,\mathscr{F})$ be a proper Lorentz-Finsler space. If $\tau\colon M\to \mathbb{R}$ is $\mathscr{F}$-steep and temporal, then $\mathscr{F}^o(-\dd \tau)\ge 1$.
\end{proposition}

\begin{proof}
 Since $\tau$ is temporal  $-\dd \tau \in \mathrm{Int} (C^o)$. By definition of $\mathscr{F}^o$ we  have \[(-\dd \tau, \mathscr{F}^o(-\dd \tau)) \in \p (C^\times)^o.\] By applying Remark \ref{oqf}  to the pair of proper cones $C^\times$ and $(C^\times)^o$, we get that there is $Y\in \p C^\times$ such that $\langle (-\dd \tau, \mathscr{F}^o(-\dd \tau)) , Y \rangle=0$, and $Y$ can be chosen of the form $Y=(y,\mathscr{F}(y))$ with $y\in C$, thus $0<\dd \tau(y)=\mathscr{F}^o(-\dd \tau)\mathscr{F}(y)$, where the first inequality follows from $-\dd \tau \in \mathrm{Int} (C^o)$, hence $\mathscr{F}(y)\ne 0$. But $\tau$ is also $\mathscr{F}$-steep, thus $\mathscr{F}(y)\le \mathscr{F}^o(-\dd \tau)\mathscr{F}(y)$ which implies $1\le \mathscr{F}^o(-\dd \tau)$.
\end{proof}

\begin{theorem} \label{xzo}
Let $(M,\mathscr{F})$ be a closed Lorentz-Finsler space.
If there is a  $\mathscr{F}$-steep temporal function $f$, then $(M,\mathscr{F})$ is stable and  for every $p,q\in M$
\[
 D(p,q) \le [f(q)-f(p)]^+.
\]

\end{theorem}
By Th.\ \ref{aob}  the right-hand side is positive if $(p,q)\in J_S\backslash \Delta$.

\begin{proof}
The existence of $f$ implies that $(M,C)$ is stably causal.
If $(p,q)\notin J_S$ or $p=q$ by Th.\ \ref{upq} we have $D(p,q)=0$ so the inequality is satisfied. Let $(p,q)\in J_S\backslash \Delta$, by Th.\ \ref{aob} $f(q)-f(p)>0$. We prove the result for $f$ {\em strictly} $\mathscr{F}$-steep temporal function. If $f$ is not so then $\tilde f=(1+\epsilon) f$, $\epsilon>0$, is, then once $D(p,q) \le (1+\epsilon) [f(q)-f(p)]^+$ is proved it is sufficient to take the limit $\epsilon \to 0$.
We are going  to prove that there is a locally Lipschitz proper Lorentz-Finsler space $\mathscr{F}'$, $\mathscr{F}<\mathscr{F}'$, such that $d'(p,q) \le f(q)-f(p)$. Since $f$ is a strictly $\mathscr{F}$-steep temporal function on every tangent space $T_xM$, the indicatrix $\mathscr{I}_x=\mathscr{F}_x^{-1}(1)$ does not intersect the compact section $\dd f^{-1}_x(1)\cap C_x$.  Thus taking the indicatrix of  $\mathscr{I}'$ sufficiently close to  $\mathscr{I}$, the same property is shared by $\mathscr{F}'$ which means that $f$ is a strictly $\mathscr{F}'$-steep temporal function. By Prop.\ \ref{ihg} $\mathscr{F}'{}^o(-\dd f) \ge 1$ and
%
%
%
for
every $C'$-causal curve $\gamma$ connecting $p$ to $q$
\begin{align*}
f(q)-f(p)&=\int \dd f(\dot \gamma) \dd t \ge \int \mathscr{F}'{}^o(-\dd f) \mathscr{F}'(\dot \gamma) \dd t \\
&\ge \int  \mathscr{F}'(\dot \gamma) \dd t \ge  \ell'(\gamma)
\end{align*}
where we used the reverse Cauchy-Schwarz inequality for $(M,\mathscr{F}')$, cf.\ Prop.\ \ref{nut}. Thus taking the supremum over $\gamma$ we find that for every $p,q\in M$, $d'(p,q)\le [f(q)-f(p)]^+$ which implies $D<+\infty$, namely $(M,\mathscr{F})$ is stable. Since $D(p,q)\le d'(p,q)$ we have proved the desired result.
\end{proof}

On Sec.\ \ref{mis} we have  defined on $M^\times=M\times \mathbb{R}$  a cone structure $C^\downarrow$ which at $P=(p,r)$ is given by
\begin{equation}
C{}^\downarrow_{P}=\{(y,z) \colon y\in C_p\cup \{0\}, \ z \le \mathscr{F}(y) \}\backslash \{(0,0)\}.
\end{equation}
There we have shown that  $(M,C^\downarrow)$ is stably causal if and only if $(M,C)$ is stably causal, cf.\ Th.\ \ref{ssh}.

\begin{theorem} \label{mfy}
Let $(M,\mathscr{F})$ be a closed Lorentz-Finsler space. Then $(M,C^\downarrow)$ is globally hyperbolic  if and only if $(M,C)$ is globally hyperbolic.
\end{theorem}

\begin{proof}
Let $f$ be a Cauchy temporal function for $(M,C)$ then $F((p,r))=f(p)-r$ is a Cauchy temporal function for $(M,C^\downarrow)$, thus if $(M,C)$ is globally hyperbolic then $(M,C^\downarrow)$ is globally hyperbolic. For the other direction, the intersection of a Cauchy hypersurface for $(M,C^\downarrow)$ with  $r=0$ provides a Cauchy hypersurface for $(M,C)$.
\end{proof}

\begin{lemma} \label{mwr}
Let $(M,\mathscr{F})$ be a closed Lorentz-Finsler space.
The relation
\[
R=\{((p,r),(p',r'))\colon (p,p')\in J_S \textrm{ and } r'-r \le D(p,p')\}
\]
satisfies $R\subset J^\downarrow _S$ where $J^\downarrow _S$ is the Seifert relation for $(M,C^\downarrow)$.
\end{lemma}

\begin{proof}
Let  $((p,r),(p',r'))\in R$ then for every locally Lipschitz proper Lorentz-Finsler space $\mathscr{F}'>\mathscr{F}$, as $(p,p')\in J_S$ we have $(p,p')\in J'$ and since $r'-r\le D(p,p')$ we have also $r'-r\le d(p,p')$. For every $\epsilon>0$ we can find a continuous $C'$-causal curve $x$ such that $(r'-\epsilon)-r\le \ell(x) <  d(p,p')$, thus $((p,r),(p',r'-\epsilon)) \in J'{}^\downarrow$, which implies  $((p,r),(p',r')) \in \overline{J'{}^\downarrow}$, and hence $R\subset \cap_{\mathscr{F}'>\mathscr{F}}  \overline{J'{}^\downarrow}$. For every locally Lipschitz proper cone structure $D>C^\downarrow$ we can find $\mathscr{F}'>\mathscr{F}$ such that $D> C'{}^\downarrow\ge C^\downarrow$, thus $\cap_{\mathscr{F}'>\mathscr{F}}  \overline{J'{}^\downarrow} \subset \overline{J_D}$ and given the arbitrariness of $D$, we have by Prop.\ \ref{paq}, $R\subset \cap_{\mathscr{F}'>\mathscr{F}}  \overline{J'{}^\downarrow} \subset J_S$.
\end{proof}

The Seifert relation on $(M^\times, C^\downarrow)$ is nicely related with the stable distance.

\begin{theorem} \label{sov}
Let $(M,\mathscr{F})$ be a stably causal closed Lorentz-Finsler space.
The Seifert relation of $(M^\times, C^\downarrow)$ is given by
\begin{equation} \label{vyv}
J_S^\downarrow=\{((p,r),(p',r'))\colon (p,p')\in J_S \textrm{ and } r'-r \le D(p,p')\}.
\end{equation}
\end{theorem}

\begin{proof}
By Th.\ \ref{ssh}  $(M^\times, C^\downarrow)$ is stably causal, thus  by Th.\ \ref{bhd}, $K^\downarrow=J_S^\downarrow$
Let $R$ be the right-hand side of Eq.\ (\ref{vyv}). By Prop.\ \ref{mwr} $R\subset J_S$.
Moreover, $R$ is closed (by the upper semi-continuity of $D$) and transitive (by Th.\ \ref{upq} (c)) and contains $\{((p,r),(p',r'))\colon (p,p')\in J \textrm{ and } r'-r \le d(p,p')\}$ which contains  $J^\downarrow$, thus $K^\downarrow\subset R$. We conclude $R=K^\downarrow=J_S^\downarrow$.
\end{proof}

The next result clarifies why {\em strictly} $\mathscr{F}$-steep temporal functions are of relevance.

\begin{lemma} \label{idf}
Let $(M,\mathscr{F})$ be a closed Lorentz-Finsler space. The smooth temporal functions $F$ on $(M,C^\downarrow)$ whose level sets intersect every $\mathbb{R}$-fiber exactly once and such that $F((q,r_2))-F((q,r_1))=r_1-r_2$, $\forall q\in M$, $r_1,r_2\in \mathbb{R}$, are put in one-to-one correspondence with the smooth strictly $\mathscr{F}$-steep temporal functions $f$ on $(M,\mathscr{F})$, through the condition $F^{-1}(0)=\cup_q (q,f(q))$ or equivalently $F((q,r))=f(q)-r$.
The statement remains valid if we add {\em Cauchy} in front of $F$ and $f$ above.
\end{lemma}

\begin{proof}
It is clear that the condition $F((q,r_2))-F((q,r_1))=r_1-r_2$ determines $F$ provided  the hypersurface $F^{-1}(0)$ is given, in fact $F((q,r))=f(q)-r$. If $F$ is temporal, $\dd F$ is positive on $C^\downarrow$, thus $F^{-1}(0)$ is transverse to every $\mathbb{R}$-fiber, and so $f$ is differentiable and smooth. Moreover, since $\dd F$ is positive on $C^\downarrow$, $0<\dd F((y, \mathscr{F}(y)))$ for $y\in C\backslash 0$, which reads $\dd f(y)> \mathscr{F}(y)$, so $f$ is strictly steep. The converse is analogous, since the steep inequality is strict, $\dd F$ is positive on $C^\downarrow$, and so it is temporal. If $F$ is Cauchy then $f$ is Cauchy, just set $r=0$ so that $F((q,0))=f(q)$ and consider the continuous $C^\downarrow$-causal curves which lie in $r=0$. Conversely, if $f$ is Cauchy $F$ is bound to increase to infinity over every continuous causal curve because the projection of a continuous causal curve is a  continuous causal curve unless function $r$ goes to  infinity over the curve.
\end{proof}

The main idea for proving the distance formula is this: the formula is really the statement on the representability of $J_S^\downarrow$ through smooth temporal functions on $M^\times$ in disguise.

\begin{theorem} \label{aas}
Let $(M,\mathscr{F})$ be a closed Lorentz-Finsler space and let $\mathscr{S}$ be the family of smooth strictly  $\mathscr{F}$-steep temporal functions. The Lorentz-Finsler space  $(M,\mathscr{F})$ is stable if and only if $\mathscr{S}$  is non-empty. In this case $\mathscr{S}$   represents
\begin{itemize}
\item[(a)] the order $J_S$, namely $(p, q)\in J_S \Leftrightarrow f(p)\le f(q), \ \forall f \in \mathscr{S}$;
\item[(b)] the manifold topology, namely for every open set $O\ni p$ we can find $f,h \in \mathscr{S}$ in such a way that
$p\in \{q\colon f(q)>0\}\cap \{q\colon h(q)<0\}\subset O$;
\item[(c)] the stable distance, in the sense that  the distance formula holds true: for every $p,q\in M$
\begin{equation}
 D(p,q)=\mathrm{inf} \big\{[f(q)-f(p)]^+\colon \ f \in \mathscr{S}\big\}.
\end{equation}
\end{itemize}
Moreover, {\em strictly} can be dropped.
\end{theorem}


\begin{proof}
The first statement follows from Th.\ \ref{xzo} and  Th.\ \ref{vkf} in the strict and non-strict case. For the other results we prove them in the strict case which is  stronger.
(a). One direction is Th.\ \ref{aob}. For the other direction, suppose that $(p,q)\notin J_S$, by Th.\ \ref{vkf} there is a smooth strictly $\mathscr{F}$-steep temporal function $f$ such that $f(p)>f(q)$.

(b). This is the last statement of Th.\ \ref{vkf}.

(c). From (a) and Th.\ \ref{upq}(a) we know that the equality holds if $(p,q)\notin J_S$ or $p=q$. Suppose $(p,q)\in J_S\backslash \Delta$, due to Th.\ \ref{xzo} we have just to prove that for every $\epsilon>0$ there is $g \in \mathscr{S}$ such that $g(q)-g(p)\le D(p,q)+\epsilon$. Let $P=(p,0)$, $Q=(q, D(p,q)+\epsilon)$ so that by Th.\ \ref{sov} $(P,Q)\notin J^\downarrow_S$.
We know that on a stably causal closed cone structure the Seifert relation is represented by the smooth temporal functions, cf.\ Th.\ \ref{swq}. Thus there is $H\colon M^\times \to [0,1]$ smooth temporal function on $(M^\times, C^\downarrow)$ such that $H(P)>H(Q)$. 
We know from Theorem \ref{vkf} that there is a smooth strictly $\mathscr{F}$-steep temporal function $f$ on $M$ and by Lemma \ref{idf}   a smooth temporal function $F$  on $M^\times$ which intersects every $\mathbb{R}$-fiber and such that $F((q,r_2))-F((q,r_1))=r_1-r_2$, $\forall q\in M$, $r_1,r_2\in \mathbb{R}$. We can choose it to be positive on $P$, $F(P)>0$. Let $G=H+ kF$, $k>0$, then since $H$ is bounded, the level sets of $G$ intersect every $\mathbb{R}$-fiber and this can happen only once since $G$ is $C^\downarrow$-temporal. Let $0<k< \frac{1}{2} [H(P)-H(Q)]/[1+\vert F(Q)\vert]$, then
\[
G(Q)\le H(Q)+ \tfrac{1}{2} [H(P)-H(Q)]<H(P)<G(P).
\]
We can redefine $G$ by adding a constant to it  and so assume $G(P)=0$.  Let $G^{-1}(0)=\cup_r(r,g(r))$, then $g(p)=0$. Notice that $\tilde G$, the function which shares with $G$ the zero level set but which satisfies $\tilde G((q,r_2))-\tilde G ((q,r_1))=r_1-r_2$ is also  a smooth temporal function, thus by Lemma \ref{idf} $g$ is a smooth strictly $\mathscr{F}$-steep temporal function. Since $G(Q)<0$ we have $D(p,q)+\epsilon> g(q)=g(q)-g(p)$, which is what we wished to prove.
\end{proof}

\begin{remark}
A related result is the following. If $(M,\mathscr{F})$ is  a stably causal  closed Lorentz-Finsler space, then $(M^\times, C^\downarrow)$ is stably causal, thus, by Th.\ \ref{swq}, the topology and order on $M^\times$ are represented by the smooth temporal functions on $(M^\times, C^\downarrow)$, which for that matter can be replaced by the smooth temporal functions $F$ on $(M, C^\downarrow)$, which satisfy $F((q,r_2))-F((q,r_1))=r_1-r_2$, $\forall q\in M$, $r_1,r_2\in \mathbb{R}$. Under the stable condition they can also be chosen so that  their level sets intersect exactly  once every $\mathbb{R}$-fiber.
\end{remark}

If the Lorentz-Finsler space is globally hyperbolic the representing functions can be chosen Cauchy.

\begin{theorem} \label{aad}
Let $(M,\mathscr{F})$ be a globally hyperbolic closed Lorentz-Finsler space and let $\mathscr{S}$ be the family of smooth Cauchy  strictly $\mathscr{F}$-steep temporal functions $f$ which are $h_f$-steep for some complete Riemannian metric dependent on the function.  The Lorentz-Finsler space $(M,\mathscr{F})$ is stable  and  $\mathscr{S}$  is non-empty. Moreover, $\mathscr{S}$   represents
\begin{itemize}
\item[(a)] the causal order $J$, namely $(p, q)\in J \Leftrightarrow f(p)\le f(q), \ \forall f \in \mathscr{S}$;
\item[(b)] the manifold topology, namely for every open set $O\ni p$ we can find $f,h \in \mathscr{S}$ in such a way that
$p\in \{q\colon f(q)>0\}\cap \{q\colon h(q)<0\}\subset O$;
\item[(c)] the  distance, in the sense that  the distance formula holds true: for every $p,q\in M$
\begin{equation}
 d(p,q)=\mathrm{inf} \big\{[f(q)-f(p)]^+\colon \ f \in \mathscr{S}\big\}.
\end{equation}
\end{itemize}
Moreover, {\em strictly} can be dropped.
\end{theorem}

By Th.\ \ref{mfy} $(M,C^\downarrow)$ is globally hyperbolic. Furthermore, by Th.\ \ref{mab} and Th.\ \ref{xbh}  if  $h$ is a complete Riemannian metric then the first statement and (a) and (b) hold with the family  $\mathscr{S}$ restricted to the smooth Cauchy $h$-steep strictly $\mathscr{F}$-steep temporal functions.

\begin{proof}
We prove all results in the strict case which is  stronger. Let $h$ be a complete Riemannian metric. The first statement and (a) and (b) are going to be proved for the smaller family for which $h_f=h$.
The first statement follows from   Th.\ \ref{mab} and Th.\ \ref{xbh}.
(a). One direction is Th.\ \ref{aob}. The other direction is the second statement of Th.\ \ref{xbh}.

(b). This is the last statement of Th.\ \ref{xbh}.

(c). In a globally hyperbolic closed cone structure $d=D$ (Th.\ \ref{upq}) and $J_S=J$ (Th.\ \ref{mbg}).  Let $f$ be a smooth Cauchy $h$-steep strictly $\mathscr{F}$-steep temporal functions which has been already shown to exist, then $ F((p,r))=f(p)-r$ is Cauchy on $(M,C^\downarrow)$ which proves that $(M,C^\downarrow)$ is globally hyperbolic and hence $J^\downarrow=J_S^\downarrow$.
 From (a) and Th.\ \ref{upq}(a) we know that the equality holds if $(p,q)\notin J$ or $p=q$. Suppose $(p,q)\in J\backslash \Delta$, due to Th.\ \ref{xzo} we have just to prove that for every $\epsilon>0$ there is $g \in \mathscr{S}$ such that $g(q)-g(p)\le d(p,q)+\epsilon$. Let $P=(p,0)$, $Q=(q, d(p,q)+\epsilon)$ so that by Th.\ \ref{sov} $(P,Q)\notin J^\downarrow_S$.
We know that on a stably causal closed cone structure the Seifert relation is represented by the smooth temporal functions, cf.\ Th.\ \ref{swq}. Thus there is $H\colon M^\times \to [0,1]$ smooth temporal function on $(M^\times, C^\downarrow)$ such that $H(P)>H(Q)$. 

Function $F$ on $M^\times$  is  smooth temporal and Cauchy. We add to  it a constant so that, $F(P)>0$. Let $G=H+ kF$, $k>0$, then since $H$ is bounded $G$ is Cauchy. In particular the level sets of $G$ intersect every $\mathbb{R}$-fiber exactly once. Let $0<k< \frac{1}{2} [H(P)-H(Q)]/[1+\vert F(Q)\vert]$, then
\[
G(Q)\le H(Q)+ \tfrac{1}{2} [H(P)-H(Q)]<H(P)<G(P).
\]
We can redefine $G$ by adding a constant to it  and so assume $G(P)=0$.  Let $G^{-1}(0)=\cup_r(r, \tilde g(r))$, then $\tilde g(p)=0$. Notice that $\tilde G$, the function which shares with $G$ the zero level set but which satisfies $\tilde G((q,r_2))-\tilde G ((q,r_1))=r_1-r_2$ is also  a smooth Cauchy temporal function, thus, by Lemma \ref{idf}, $\tilde g$ is a smooth Cauchy strictly $\mathscr{F}$-steep temporal function. Since $G(Q)<0$ we have $d(p,q)+\epsilon> \tilde g(q)=\tilde g(q)-\tilde g(p)$. Though $\tilde g$ is not necessarily $h_{\tilde g}$-steep for some complete Riemannian metric $h_{\tilde g}$, the function $g=\tilde g+\delta f$, $\delta>0$, is $h_g$-steep where $h_g=\delta h$ is a complete Riemannian metric. For sufficiently small $\delta$ the inequality $d(p,q)+\epsilon>  g(q)-g(p)$ is satisfied. \ \
\end{proof}

Let us investigate the validity of the distance formula for $d$ replacing $D$.

\begin{lemma} \label{qaq}
Let $(M,\mathscr{F})$ be a locally Lipschitz proper Lorentz-Finsler space such that $\mathscr{F}(\p C)=0$, then on $(M^\times, C^\downarrow)$ the chronological relation is given by
\[
I^\downarrow=\{((p,r),(p',r'))\colon (p,p')\in I \textrm{ and }  r'-r < d(p,p')\} ,
\]
while its closure satisfies
\[
\overline{I^\downarrow}=\overline{J^\downarrow} \supset \{((p,r),(p',r'))\colon (p,p')\in \bar I \textrm{ and }  r'-r \le d(p,p')\} .
\]
\end{lemma}

\begin{proof}
Since $ \mathrm{Int} C^\downarrow_{(p,r)}=\{(y,z) \colon y\in \mathrm{Int} C_p, \ z <  \mathscr{F}(y) \}$,  integration over a $C^\downarrow$-timelike curve on $M^\times$, easily gives the inclusion $\subset$.

For the other direction, by the proof of Th.\ \ref{ddc} for every $0<R<d(p,p')$ we can find a $C^\times$-timelike (and hence $C^\downarrow$-timelike) curve connecting $P=(p,r)$ to $(p', r+R)$, thus, since the $\mathbb{R}$-fibers are  continuous causal curves and $J\circ I\subset I$, all the points of the form $(p', r')$ with $r'-r<d(p,p')$ belong to $(I^\downarrow)^+((p,r))$.

Let $G=\{((p,r),(p',r'))\colon (p,p')\in I \textrm{ and }  r'-r <d(p,p')\}$. By the previous formula $\bar G = \overline{I^\downarrow}\supset I^\downarrow$.
The closure $\bar G$ is given by the pairs $((p,r),(p',r'))$, for which there are sequences  $((p_i,r_i),(p'_i,r'_i)) \to ((p,r),(p',r'))$ with $(p_i,p'_i)\in I$ and $r'_i-r_i <d(p_i,p'_i)$. Notice that if $(p,p')\in \p I$ then $d(p,p')=0$ because if there is a connecting causal curve it is a lightlike geodesic and so it cannot have timelike tangent anywhere (Th.\ \ref{aam}), thus given a sequence $(p_i,p'_i)\in I$, $(p_i,p'_i)\to (p,p')$ we have that $((p_i,r),(p'_i,r'))$ with $r'\le r$ belongs to $G$ and converges to $((p,r),(p',r'))$, thus $\bar G$ includes $\{((p,r),(p',r'))\colon (p,p')\in \bar I \textrm{ and }  r'-r \le d(p,p')\}$.
\end{proof}

\begin{theorem} \label{ape}
Let $(M,\mathscr{F})$ be a distinguishing locally Lipschitz proper Lorentz-Finsler space such that $\mathscr{F}(\p C)=0$. The following conditions are equivalent
\begin{enumerate}
\item $D=d$,
\item  $d$ is upper semi-continuous,
\end{enumerate}
and they imply the causal continuity of $(M,C)$.
\end{theorem}

Under the assumption $d$ is lower semi-continuous but this fact will not be used.

\begin{proof}
$1 \Rightarrow 2$. This direction is immediate from Th.\ \ref{upq}.
$2 \Rightarrow 1$. We know that
\begin{align*}
J^\downarrow & \subset \{((p,r),(p',r'))\colon (p,p')\in J \textrm{ and }  r'-r \le d(p,p')\}\\
&\subset  \{((p,r),(p',r'))\colon (p,p')\in \bar J \textrm{ and }  r'-r \le d(p,p')\}.
\end{align*}
 But the latter relation is closed because $d$ is upper semi-continuous (notice $d$ is defined on $M\times M$, thus also on $\p J$), thus $\overline{J^\downarrow} \subset \{((p,r),(p',r'))\colon (p,p')\in \bar J \textrm{ and }  r'-r \le d(p,p')\}$. The reverse inclusion follows from Lemma \ref{qaq}, so
\begin{equation} \label{jaw}
\overline{J^\downarrow}=\overline{I^\downarrow}=\{((p,r),(p',r'))\colon (p,p')\in \bar I \textrm{ and } r'-r \le d(p,p')\}
\end{equation}

Due to Th.\ \ref{imp} $(M,C)$ is reflective (hence causally continuous and stably causal), thus $\bar J=\bar I$ is transitive (Th.\ \ref{tra}). Since the relation in display is closed and transitive and contains $J^\downarrow$, we have $J^\downarrow_S=K^\downarrow=  \overline{J^\downarrow}$, where the first equality is due to stable causality.
 From causal continuity $J_S=\bar J$, and by the equality $\overline{J^\downarrow}= J_S^\downarrow$ and Th.\ \ref{sov} we get $D=d$.
\end{proof}

\begin{corollary} \label{cot}
Let $(M,\mathscr{F})$ be a distinguishing locally Lipschitz proper Lorentz-Finsler space such that $\mathscr{F}(\p C)=0$ and such that $d$ is finite and continuous, then $(M,\mathscr{F})$ is stable.
\end{corollary}

\begin{proof}
If  $(M,\mathscr{F})$ has  a finite and continuous Lorentz-Finsler  distance $d$, then $(M,C)$ is causally continuous and $D=d$ by Theorem \ref{ape}, and so $D$ is finite. Thus $(M,\mathscr{F})$ is stable.
\end{proof}

The next result clarifies what are the simple spacetimes in the sense of Parfionov and Zapatrin \cite{parfionov00}.
\begin{theorem} \label{xhg}
Let $(M,\mathscr{F})$ be a distinguishing locally Lipschitz proper Lorentz-Finsler space such that $\mathscr{F}(\p C)=0$, and let $\mathscr{S}$ be the family of   $\mathscr{F}$-steep temporal functions.
The distance formula
\begin{equation} \label{dor}
 d(p,q)=\mathrm{inf} \{[f(q)-f(p)]^+\colon f \in \mathscr{S}  \}.
\end{equation}
holds if and only if  $(M,\mathscr{F})$  has  a finite and continuous Lorentz-Finsler distance $d$ (hence it is causally continuous), in which case $d=D$.
\end{theorem}

It should be observed that under global hyperbolicity the distance formula does not require the locally Lipschitz proper condition, cf.\ Th.\ \ref{aad}.

\begin{proof}
If  $(M,\mathscr{F})$ has  a finite and continuous Lorentz-Finsler  distance $d$, then $(M,C)$ is causally continuous and $D=d$ by Theorem \ref{ape}, and so $D$ is finite. Thus $(M,\mathscr{F})$ is stable and by Theorem \ref{aas} the distance formula holds true. For the converse, distinction implies that $d(p,q)<+\infty$ for some $p,q\in M$ (by Prop.\ \ref{iiu}), thus since the distance formula holds, the family of steep temporal functions is non-empty and hence  the spacetime $(M,\mathscr{F})$ is stably causal and $d$ is finite.
The right-hand side of Eq.\ (\ref{dor}) being the infimum of a family of continuous function is upper semi-continuous, thus $d$ is upper semi-continuous.
By Th.\ \ref{ddb} $d$ is lower semi-continuous and by  Prop.\ \ref{ape} $d=D$ and $(M,C)$ is causally continuous.
\end{proof}

\subsubsection{Distance formula for stably causal spacetimes}
In this section we give the distance formula for stably causal closed Lorentz-Finsler spaces. Contrary to the formula for stable spacetimes it requires unbounded representing functions.

Given a preorder $R$ on $M$ we have defined the notion of $R$-isotone and $R$-utility function  $f\colon M\to \mathbb{R}$. The definition can be easily generalized to functions $f\colon M\to [-\infty, \infty]$ since the order $\le$ of the real line can be trivially extended. Such functions will be called extended $R$-isotone functions or extended $R$-utility functions.

 \begin{lemma} \label{kqa}
 Let $(M,C)$ be a stably causal closed cone structure, and let $t$ be an extended continuous $J$-isotone function. For every  $a,b\in\mathbb{R}$ and $\epsilon>0$ we can find a locally Lipschitz proper cone structure $C'>C$ such that $J_{C'}^+(t^{-1}([a,\infty]))  \subset t^{-1}([a-\epsilon, \infty])$ and $ J_{C'}^-(t^{-1}([-\infty,b])) \subset t^{-1}([-\infty, b+\epsilon])$.
 \end{lemma}

\begin{proof}
Let $h$ be a complete Riemannian metric such that the $d^h$ distance between $t^{-1}(a)$ and $t^{-1}(a-\epsilon)$, and between $t^{-1}(b)$ and $t^{-1}(b+\epsilon)$ is larger than $\epsilon$. Let $\tau\colon M\to[-1,1] $ be a time function, then $\tilde \tau=\tanh t+c \tau$, $c>0$, being the sum of a continuous isotone function and a time function is a time function, where we have set by convention $\tanh \pm \infty=\pm 1$. We choose $c$ so small that $ \tanh a -2 c>\tanh (a-\epsilon)$, and $ \tanh b +2 c<\tanh (b+\epsilon)$.  By Lemma \ref{lmn} there is a locally Lipschitz proper cone structure $C'>C$ such that, having defined $R_\epsilon=\{(p,q)\colon d^h(p,q)<\epsilon\}$,  we have ${J_{C'}}\backslash R_\epsilon\subset \{(p,q)\colon \tilde\tau (p)< \tilde \tau(q)\}$.
Let $p\in t^{-1}([a,\infty])$ and $r\in J_{C'}^+(p)$. If $d^h(p,r)<\epsilon$ then $r\in  t^{-1}([a-\epsilon, \infty])$ and we have finished. Otherwise   $(p,r)\in J_{C'}\backslash R_\epsilon$,  thus $\tilde\tau(p)<\tilde \tau(r)$. From this inequality we have $\tanh t(r)>\tanh t(p) +c [\tau(p)-\tau(r)] \ge \tanh a -2 c>\tanh (a-\epsilon)$, that is $t(r)>a-\epsilon$. Similarly, let $q\in t^{-1}([-\infty,b])$  and  $r\in  J_{C'}^-(q)$. If $d^h(r,q)<\epsilon$ then $r\in  t^{-1}([-\infty,b+\epsilon])$ and we have finished.  Otherwise, $(r,q)\in J_{C'}\backslash R_\epsilon$,  thus $\tilde \tau(r)<\tilde \tau(q)$. From this inequality we have $ \tanh t(r)<\tanh t(q) +c [\tau(q)-\tau(r)] \le \tanh b +2 c<\tanh (b+\epsilon)$, that is $t(r)<b+\epsilon$, which concludes the proof.
\end{proof}

\begin{lemma}
Let $(M,\mathscr{F})$ be a stably causal closed Lorentz-Finsler space, and let $t$ be an extended continuous $J$-isotone function. Let $B=t^{-1}(\mathbb{R})$ be its finiteness domain, then $D\vert_{B\times B}=D^B$ where $D^B$ is the Lorentz-Finsler stable distance for $(B,\mathscr{F}\vert_B)$.
\end{lemma}

\begin{proof}
Let $p,q\in B$, $a=t(p)$, $b=t(q)$, then by Lemma \ref{kqa} there is $\tilde C>C$ such that every continuous $\tilde C$-causal curve connecting $p$ to $q$ is entirely contained in $B$. As a consequence, for every $\mathscr{F}'>\mathscr{F}$ such that, $C<C'<\tilde C$, $d'(p,q)=d'{}^B(p,q)$ which taking the infimum over $\mathscr{F}'$ gives $D(p,q)=D^B(p,q)$, and given the arbitrariness of $p,q\in B$, $D\vert_{B\times B}=D^B$.
\end{proof}

\begin{theorem} \label{cib}
Let $(M,\mathscr{F})$ be a stably causal closed Lorentz-Finsler space, and let $t$ be an extended continuous $J$-isotone function which is $\mathscr{F}$-steep and temporal wherever it is finite.
With the convention $\infty-\infty=\infty$ we have for every $p,q\in M$
\begin{equation} \label{kff}
D(p,q)\le [ t(q)-t(p)]^+  .
\end{equation}
\end{theorem}

\begin{proof}
Let us consider the case $(p,q)\in J_S$.
The $J_S$-isotone continuous functions coincide with the $J$-isotone continuous functions by Th.\ \ref{ndc}, thus $t(p)\le t(q)$ which implies the validity of the inequality whenever $t(p)$ or $t(q)$ is not finite, since then $[t(q)-t(p)]^+=\infty$.  If they are both finite then defining $B= t^{-1}(\mathbb{R})$, we have $p,q\in B$, which implies $D(p,q)=D^B(p,q)$. But $t\vert_B$ is a $\mathscr{F}$-steep temporal function, thus by Th.\ \ref{xzo}, $D^B(p,q)\le t(q)-t(p)=[t(q)-t(p)]^+$. Finally, if  $(p,q)\notin J_S$, $D=0$ in which case Eq.\ (\ref{kff})  is satisfied.
\end{proof}

We recall that by Th. \ref{sov} if $(M,\mathscr{F})$ is a stably causal closed Lorentz-Finsler space
the Seifert relation of $(M^\times, C^\downarrow)$ is given by
\begin{equation}
J_S^\downarrow=\{((p,r),(p',r'))\colon (p,p')\in J_S \textrm{ and } r'-r \le D(p,p')\}.
\end{equation}

\begin{lemma} \label{mvo}
Let $F$ be a smooth temporal function on $(M^\times, C^\downarrow)$. Let $a\in \mathbb{R}$, let $O$ be the projection to $M$ of $F^{-1}(a)$, and let $O^+\subset M$ (resp.\ $O^- \subset M$) consist of those points $p$ over whose fiber $F>a$ (resp.\ $F<a$). The function $f\colon M\to [-\infty,+\infty]$ defined by $F^{-1}(a)=(p,f(p))$ on $O$, and $\pm \infty$ on $O^
\pm$, is a   continuous $J$-isotone function which is smooth strictly  $\mathscr{F}$-steep temporal on $O$.
\end{lemma}

\begin{proof}
 Let $(p,q)\in J_S$, then for every $c$, $P=(p,c)$ and $Q=(q,c)$ are such that $(P,Q)\in  J_S^\downarrow$, which implies $F(P)\le F(Q)$. Thus, if $p \in O^+$, then it cannot hold that  $q \in O^-$, since in the fiber of the former point $F>a$ while in that of the latter $F<a$. Similarly, if $p\in O$  it cannot hold that $q \in O^-$, since with $c=f(p)$, we would have that $F(P)=a$ and $F(Q)<a$. Analogously, we obtain that if $q\in O^-\cup O$ we cannot have that $p\in O^+$. Thus if any among $p$ or $q$ does not belong to $O$, we have $f(p)\le f(q)$. If they both belong to $O$, with $c=f(p)$ we have $a=F(P)\le F(Q)$ which implies $c\le f(q)$, that is $f(p)\le f(q)$.
 We conclude that $f$ is indeed $J_S$-isotone. Its continuity follows because  the level sets of $F$ are closed, while it is  strictly  $\mathscr{F}$-steep temporal on $O$ due to the temporality of $F$, see also Lemma \ref{idf}. \ \
\end{proof}

\begin{theorem}
Let $(M,\mathscr{F})$ be a stably causal closed Lorentz-Finsler space. Let ${\mathscr{S}}$ be the family of extended continuous $J$-isotone functions $f\colon M\to [-\infty,+\infty]$ which are smooth strictly  $\mathscr{F}$-steep temporal functions wherever they are finite.
The family ${\mathscr{S}}$   represents
\begin{itemize}
\item[(a)] the order $J_S$, namely $(p, q)\in J_S \Leftrightarrow f(p)\le f(q), \ \forall f \in {\mathscr{S}}$;
\item[(b)] the manifold topology, in fact  for every open set $O\ni p$ and $\epsilon>0$ we can find $f,h \in {\mathscr{S}}$  such that
\[p\in \{q\colon f(q)>0\}\cap \{q\colon h(q)<0\}\subset \{q\colon f(q)\ge h(q)\} \subset  O \]
with $\vert f\vert $ and $\vert h\vert$ bounded by $\epsilon$ in neighborhood of $\{q\colon f(q)\ge h(q)\}$;
\item[(c)] the stable distance, in the sense that  the distance formula holds true: for every $p,q\in M$ (with the convention $\infty-\infty=\infty$)
\begin{equation}
 D(p,q)=\mathrm{inf} \big\{[f(q)-f(p)]^+\colon \ f \in {\mathscr{S}}\big\}.
\end{equation}
\end{itemize}
\end{theorem}

\begin{proof}
$(a)$.  One direction is clear since if  $f\in \mathscr{S}$ then $\tanh f$ is a $J$-isotone function and hence a $J_S$-isotone function by Th.\ \ref{ndc}, which implies that $f$ is an extended $J_S$-isotone function. For the other direction let $(p,q)\notin J_S$, by Th.\ \ref{sov} $P=(p,0)$ and $Q=(q,0)$ are such that $(P,Q)\notin J_S^\downarrow$. Thus by Th.\ \ref{vkg} there is a smooth temporal function $F$ on $(M^\times, C^\downarrow)$ such that $F(P)>F(Q)$.  Let $a=F(P)$ and let $f$ be the function of Lemma \ref{mvo}, then $f(p)=0$ and $f\in \mathscr{S}$. Since $F(Q)<a$ we have $f(q)=-\infty$ or $f(q)$ is finite where, due to $F(P)>F(Q)$,  we  must have $f(p)=0>f(q)$. In any case $f(p)>f(q)$.

$(b)$. Let $P=(p,0)$, $\epsilon>0$ and let $O^\times =O\times(-\epsilon,\epsilon)$. By Th.\ \ref{vkg}  we can find smooth temporal functions $\check T, \hat T$ on $(M^\times,C^\downarrow)$  such that $P\in [\{Q\colon \check T(Q)<0\}\cap \{Q\colon \hat T(Q)>0\}]$ and such that $ [\{Q\colon \check T(Q)\le 0\}\cap \{Q\colon \hat T(Q)\ge 0\}]$ is compact and contained in $O^\times$. Let us denote by $\check O$ and $\hat O$ the projections of the level sets of $\check T^{-1}(0)$ and $\hat T^{-1}(0)$.
There is an open set $\tilde O\subset O$ which
satisfies all the above properties for $O$ and additionally $\tilde O\subset \check O\cap \hat O$. We redefine $\tilde O \to O$.
The mentioned level sets over $O$ read $\{q\in O\colon (q,\check f(q))\}$ and $\{q\in O\colon (q,\hat f(q))\}$, where $\check f,\hat f\colon M\to [-\infty,\infty]$ are as in Lemma \ref{mvo}, and finite over $O$. The condition $P\in [\{Q\colon \check T(Q)<0\}\cap \{Q\colon \hat T(Q)>0\}]$ implies $ \check f(p)< 0< \hat f(p)$ while the condition  $[\{Q\colon \check T(Q)\le 0\}\cap \{Q\colon \hat T(Q)\ge 0\}]\subset O^\times$ implies $\{q\colon \hat f(q)\ge \check f(q)\}\subset O$  with $\vert \hat{f} \vert $ and $\vert \check{f} \vert$ bounded by $\epsilon$ in neighborhood of $\{q\colon \hat f(q)\ge \check f(q)\}$. Thus  $ p\in \{q\colon  \hat{f}(q)>0\}\cap \{q\colon \check{f}(q)<0\}\subset \{q\colon \hat f(q)\ge \check f(q)\} \subset O$.

$(c)$. The direction $\le$ is just Th.\ \ref{cib}. For the other direction, suppose that $(p,q)\notin J_S$, then by the proof of $(a)$ there is $f\in \mathscr{S}$ such that  $f(p)=0$ and $f(q)<f(p)$, thus $[f(q)-f(p)]^+=0$, which proves that equality is attained. Let us consider the case $(p,q)\in J_S$. It is sufficient to prove that  if $D(p,q)<\infty$ for every  $\epsilon >0$ we can find $f\in \mathscr{S}$, finite on $p$ and $q$ such that $f(q)-f(p)<D(p,q)+\epsilon$. Let $P=(p,0)$, $Q=(q, D(p,q)+\epsilon)$, then $(P,Q)\notin J_S^\downarrow$, thus there is $F$ smooth temporal function on $(M^\times, C^\downarrow)$ such that $F(P)>F(Q)$. Notice that $Q'=(p, D(p,q))$ is such that $(P,Q')\in J_S^\downarrow$, thus $F(P)\le F(Q')$ by Th.\ \ref{ndc}. Let $a=F(P)$ and let us consider the level set $F^{-1}(a)$ and the associated function $f\in \mathscr{S}$ as in Lemma \ref{mvo}. Then $f(p)=0$, $D(p,q)\le f(q)<D(p,q)+\epsilon$, which implies $[f(q)-f(p)]^+<D(p,q)+\epsilon$. \ \ \,
\end{proof}

\subsection{Lorentzian embeddings}
In this section we solve the problem of characterizing the Lorentzian spacetimes isometrically embeddable in Minkowski. In this section the cone $C$ is round and $\mathscr{F}(v)=\sqrt{-g(v,v)}$, for $v\in C$, where $g$ is a Lorentzian metric. The $\mathscr{F}$-steep temporal functions will just be called ``steep temporal functions''. We also recall that $N_0(n)$ is the Nash dimension.

We recall the result by M\"uller and S\'anchez \cite{muller11} whose proof has been sketched in the Introduction.
\begin{theorem} \label{sob}
Let $(M,g)$ be an $n+1$-dimensional Lorentzian spacetime  endowed with a $C^{k}$, $3\le k\le \infty$, metric. $(M,g)$ admits a $C^k$ isometric embedding in Minkowski spacetime $E^{N,1}$ for some $N$ if and only if $(M,g)$ admits a $C^k$ steep temporal function. In that case $N=N_0(n+1)$ would do.
\end{theorem}

The necessary conditions for the embedding are outlined by the following result.

\begin{lemma} \label{dow}
Let $(M,g)$ be a $C^1$ Lorentzian submanifold of $E^{N,1}$, then $(M,g)$ is  stable.
\end{lemma}
\begin{proof}
Let $\bar g=-(\dd x^0)^2+\sum_{i\ge 1}(\dd x^i)^2$ be the metric on $E^{N,1}$ in its canonical coordinates, $g=i^* \bar g$. For $\epsilon>0$, let $\bar g'=-(1+\epsilon)(\dd x^0)^2+\sum_{i\ge 1}(\dd x^i)^2$ and $g'=i^* \bar g'$, then $g'>g$ but $(M,g')$ is still causal since $(\mathbb{R}^{N+1}, \bar g')$ is. Moreover, for every $p,q\in M$, $d'(p,q)<+\infty$ since $i(p)$ and $i(q)$ are at finite Lorentzian distance on $(\mathbb{R}^{N+1}, \bar g')$ and the $i$-image of a causal connecting curve on $M$ is a causal connecting curve on $(\mathbb{R}^{N+1}, \bar g')$. As a consequence, $D<+\infty$.
\end{proof}

Th.\ \ref{vkf} on the existence of steep time functions on stable spacetimes leads us to our main embedding result which establishes that the stable spacetimes are the Lorentzian submanifolds of Minkowski spacetime.
\begin{theorem}
Let $(M,g)$ be an $n+1$-dimensional Lorentzian spacetime  endowed with a $C^{k}$, $3\le k\le \infty$, metric. $(M,g)$ admits a $C^k$ isometric embedding in Minkowski spacetime $E^{N,1}$, for some $N>0$, if and only if $(M,g)$ is stable. In that case $N=N_0(n+1)$ would do.
\end{theorem}

\begin{proof}
One direction is proved in Lemma \ref{dow}. If $(M,g)$ is stable then by Th.\ \ref{vkf} there is a smooth steep temporal function, and by Theorem \ref{sob} there exists the  embedding.
\end{proof}

\begin{corollary}
Every {$n+1$}-dimensional  distinguishing Lorentzian spacetime $(M,g)$ endowed with a $C^{k}$, $3\le k\le \infty$, metric and for which $d$ is finite and continuous (for instance, a globally hyperbolic spacetime)   admits a $C^k$ isometric embedding into  $E^{N,1}$ for some $N>0$. Here $N=N_0(n+1)$ would do.
\end{corollary}

\begin{proof}
By Cor.\ \ref{cot} $(M,g)$ is stable. It is well known  \cite{beem96} that the  Lorentzian distance function is continuous and finite in  globally hyperbolic spacetimes.
\end{proof}

The next result due to M\"uller and S\'anchez can also be regarded as a consequence of Lemma \ref{dow} and of Th.\ \ref{maa} on the existence of stable representatives in the conformal class of a stably causal spacetime.

\begin{corollary}
A  $C^{k+1}$, $3\le k\le \infty$, Lorentzian spacetime $(M,g)$  is stably causal if and only if it admits a conformal embedding into  $E^{N,1}$ for some $N>0$.
\end{corollary}
\begin{figure}[ht!]
\begin{center}
 \includegraphics[width=5cm]{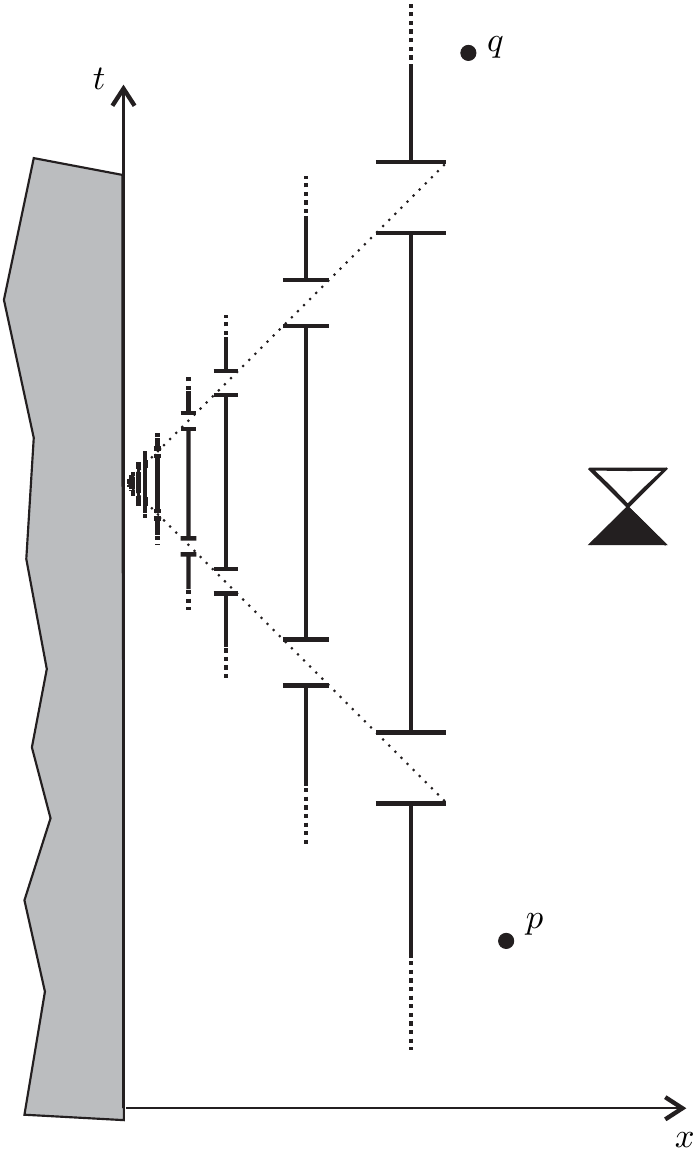}
\end{center}
\caption{An example of stably causal spacetime for which $d$ is finite but unstably so which, therefore, is not isometrically embeddable in Minkowski spacetime. The example is conformal to the displayed subset of 1+1 Minkowski spacetime, where the thick black lines and the region $x\le 0$ have been removed.} \label{acausal}
\end{figure}

\begin{example} \label{cuw}
It is natural to ask whether in stably causal Lorentzian spacetimes the finiteness of $D$ is implied by the finiteness of $d$, and so whether the finiteness of $d$ can be sufficient for the embeddability in Minkowski spacetime. The next example answers this in the negative. Let $M$ be the open set of Minkowski 1+1 spacetime depicted in Figure \ref{acausal}. The metric on $M$ is conformal to Minkowski $g=\varphi(x)(-\dd t^2+\dd x^2)$, $\phi>0$. The open set $M$ is contained in the region $x>0$ from which we remove segments and half-lines. We are really removing the same vertical element, repeated and rescaled a countable number of times.  This vertical element presents two `gates' which the causal curves of $(M,g)$ cannot traverse, thus  $M$ gets separated into a sequence of causally unrelated strips. However, $g'$-causal curves for $g'>g$ can indeed traverse the gates. Let $\varphi$ be bounded on $x>\epsilon$ for every $\epsilon>0$. If $\varphi\to +\infty$ sufficiently fast for $x\to 0$, e.g.\ $\varphi=1/x^4$, then $d$ is finite but $D$ is infinite for some pairs, e.g.\ $D(p,q)=+\infty$. The reason is that by opening the cones there are curves connecting $p$ to $q$ which pass as many gates as desired, acquiring arbitrarily large Lorentzian length thanks to their vertical development near $x=0$. Since $d'(p,q)=+\infty$ for every $g'>g$, we have $D(p,q)=+\infty$.
\end{example}

\begin{example} \label{syu}
In \cite[Sec.\ 6.1]{minguzzi12f} it has been proved that for any constants $a,b>0$, the 3+1-dimensional spacetime $\mathbb{R}^2\times \mathbb{R}^2\backslash\{(0,0)\}$,  $g=a (\dd w^2+\dd z^2) -2 \dd y \dd x+2 \frac{ab}{\sqrt{w^2+z^2}} \dd x^2 $ is causally simple but not globally hyperbolic, and that it has a finite and continuous Lorentzian distance. By the previous corollary it provides a non-trivial example of stable non-globally hyperbolic spacetime embeddable in Minkowski.
\end{example}

\subsection*{Acknowledgments}
The content of this work was presented at the meetings ``Non-regular spacetime geometry", Firenze, June 20-22, 2017, and ``Advances in General Relativity", ESI Vienna, August 28 - September 1, 2017.





\def\cprime{$'$}

\end{document}